\documentclass
  [
    DIV=14,
    fontsize=12,
    abstract=true,
    bibliography=totoc
  ]
  {scrartcl}
\setkomafont{date}{\small\rmfamily}
\setkomafont{author}{\large\scshape}

\usepackage{authblk}

\newcommand
{\NameEmail}[2]
{%
  \stackunder%
  {#1}%
  {%
   \clap{%
   \normalfont%
   \footnotesize%
   \texttt{#2}%
  }}%
  \hspace{-3.5pt}% 
}

\newcommand\keywords[1]{%
  \begingroup
  \renewcommand\thefootnote{}% Remove footnote number
  \footnote{{keywords:} #1}%
  \addtocounter{footnote}{-1}% Reset footnote counter
  \endgroup
}
\newcommand\MSCcodes[1]{%
  \begingroup
  \renewcommand\thefootnote{}% Remove footnote number
  \footnote{{MSC codes:} #1}%
  \addtocounter{footnote}{-1}% Reset footnote counter
  \endgroup
}

\usepackage{scrlayer-scrpage}
\clearpairofpagestyles 
\ihead{Superspace Supergravity} 
\ohead{\headmark}
\automark{section}

\NewCommandCopy{\oldparagraph}{\paragraph}
\RenewDocumentCommand{\paragraph}{s O{#3} m}{%
  \IfBooleanTF{#1}
    {\oldparagraph*{#3.}}       % Handles \paragraph*{...}
    {\oldparagraph[#2]{#3.}}    % Handles \paragraph{...} and \paragraph[...]{...}
}

\usepackage{caption}
\usepackage
  [leftcaption]
  {sidecap}

\usepackage
 [inline]
 {enumitem}
\setlist[enumerate,1]{label=\textbf{\textup{(\roman*)}}}
\setlist{
  topsep=1pt,
  itemsep=0pt,
  leftmargin=1cm,
  before=\leavevmode
}

\usepackage[
  backend=biber,
  style=alphabetic,
  giveninits=true, 
  backref=true
]{biblatex}
\usepackage[T1]{fontenc}
\usepackage{lmodern}
\usepackage{anyfontsize}

\usepackage{multirow}
\usepackage{hhline}

\usepackage{tabularray}
\UseTblrLibrary{booktabs}

\usepackage{adjustbox}
\usepackage{tcolorbox}

\newtcolorbox{standout}{
  colback=gray!15,
  boxrule=0pt,
  left=.3cm,
  right=.3cm,
  top=.18cm,
  bottom=.18cm,
  boxsep=0pt
}

\usepackage{amsmath}
\allowdisplaybreaks
\usepackage{mathtools}
\usepackage{amssymb}
\usepackage{amsthm}
\usepackage{xfrac}

\usepackage{stmaryrd} % for \sslash
\SetSymbolFont{stmry}{bold}{U}{stmry}{m}{n}

\usepackage{slashed} % for \slashed
\usepackage{scalerel} % for scaling brackets
\usepackage{stackengine} % for incremental scaling
\usepackage[safe]{tipa} % for \esh
\DeclareFontFamilySubstitution{T3}{lmr}{cmr}
\usepackage[
  cal=euler,
  scr=dutchcal
]
 {mathalpha} % for control over script fonts

 \newcommand{\bracket}[3]{%
  \stretchleftright
    {#1}
    {%
      \ensurestackMath{\addstackgap[1pt]{#2}}%
      \vrule width 0pt depth 2pt height 0pt
    }
    {#3}%
} 
\newcommand{\scaledbracket}[3]{%
  \ThisStyle{%
    \stretchleftright
      {#1}
      {
        \ensurestackMath{\addstackgap[1pt]{\SavedStyle #2}}%
        \vrule width 0pt depth 1.5pt height 0pt
      }
      {#3}%
  }%
}
\newcommand{\bracketmid}[4]{%
  \stretchleftright{#1}{%
    \ensurestackMath{%
      \addstackgap[2pt]{#2}%
      \,\stretchrel*{|}{\addstackgap[2pt]{#2#3}}\,%
      \addstackgap[2pt]{#3}%
    }%
  }{#4}%
}

\theoremstyle{plain}
\newtheorem{theorem}{Theorem}[section]
\newtheorem{lemma}[theorem]{Lemma}
\newtheorem{proposition}[theorem]{Proposition}

\theoremstyle{definition}
\newtheorem{definition}[theorem]{Definition}
\newtheorem{example}[theorem]{Example}

\theoremstyle{remark}
\newtheorem{remark}[theorem]{Remark}
\usepackage[
  colorlinks=true,
  linkcolor=darkgreen,
  citecolor=darkgreen,
  urlcolor=darkgreen
]{hyperref}
\usepackage{xurl} % allow line breaks in URLs

\usepackage{cleveref}
\crefname{equation}{}{}
\crefname{section}{\S}{\S\S}
\crefname{subsection}{\S}{\S\S}
\crefname{subsubsection}{\S}{\S\S}
\crefname{definition}{Def.}{Defs.}
\crefname{theorem}{Thm.}{Thms.}
\crefname{corollary}{Cor.}{Cors.}
\crefname{lemma}{Lem.}{Lems.}
\crefname{proposition}{Prop.}{Props.}
\crefname{remark}{Rem.}{Rems.}
\crefname{notation}{Ntn.}{Ntns.}
\crefname{fact}{Fact}{Fact}
\crefname{example}{Ex.}{Exs.}
\crefname{figure}{Fig.}{Figs.}
\crefname{table}{Tab.}{Tabs.}
\crefname{footnote}{ftn.}{ftns.}
\Crefname{footnote}{Ftn.}{Ftns.}

\usepackage{tikz}
\usetikzlibrary{
  cd,
  calc,
  nfold,
  arrows.meta,
  backgrounds,
  decorations,
  decorations.pathmorphing, 
}

\tikzcdset{
  arrow style=tikz, 
  diagrams={
    trim left=
    {([xshift=5pt]current bounding box.west)},
    trim right=
    {([xshift=-5pt]current bounding box.east)},
    baseline=
    {([yshift=-2pt]current bounding box.center)},
    >={
      Computer Modern Rightarrow[
        length=4pt, width=4pt
      ]
    },
    hook/.style={
      {Hooks[right, length=2pt, width=8pt]}->
    },
    hook'/.style={
      {Hooks[left, length=2pt, width=8pt]}->
    },
    shorten=0pt,
  }
}

\definecolor{darkblue}{rgb}{0.05,0.25,0.65}
\definecolor{darkgreen}{RGB}{20,140,10}
\definecolor{lightgray}{rgb}{0.9,0.9,0.9}
\definecolor{darkorange}{RGB}{200,100,5}
\definecolor{darkyellow}{rgb}{.91,.91,0}
\definecolor{lightolive}{RGB}{225, 220, 185}

\let\originalsslash\sslash
\renewcommand{\sslash}{\mathord{\originalsslash}}

\makeatletter
\newcommand{\cpt}{\mathpalette\cpt@inner\relax}
\newcommand{\cpt@inner}[2]{%
  \scalebox{0.5}[0.9]{$#1\cup$}% Scales the cup relative to the current style
  #1\{\infty\}% Typesets the infinity set in the current style
}
\makeatother

\makeatletter
\newcommand{\plus}{\mathpalette\sqcpt@inner\relax}
\newcommand{\sqcpt@inner}[2]{%
  \scalebox{0.5}[0.9]{$#1\sqcup$}% Scales the cup relative to the current style
  #1\{\infty\}% Typesets the infinity set in the current style
}
\makeatother

\newcommand{\grayunderbrace}[2]{\mathcolor{gray}{\underbrace{\mathcolor{black}{#1}}}_{\mathcolor{gray}{#2}}}
\newcommand{\grayoverbrace}[2]{\mathcolor{gray}{\overbrace{\mathcolor{black}{#1}}}^{\mathcolor{gray}{#2}}}

\tikzset{
  snake left/.style={
    rounded corners,
    to path={
      let \p1 = (\tikztostart.east),
          \p2 = (\tikztotarget.west),
          \p3 = ($(\p1)!0.5!(\p2)$),
          \n1 = {8pt} 
      in
      (\p1)
      -- (\x1 + \n1, \y1)
      -- (\x1 + \n1, \y3)
      -- (\x2 - \n1, \y3) \tikztonodes
      -- (\x2 - \n1, \y2)
      -- (\p2)
    }
  }
}

\tikzset{
  uphordown/.style={
    rounded corners,
    to path={
      let \p1 = (\tikztostart.north),
          \p2 = (\tikztotarget.north),
          \n1 = {max(\y1,\y2) + 8pt}
      in
      (\p1)
      -- (\x1, \n1)
      -- (\x2, \n1) \tikztonodes 
      -- (\p2)
    }
  }
}

\tikzset{
  downhorup/.style={
    rounded corners,
    to path={
      let \p1 = (\tikztostart.south),
          \p2 = (\tikztotarget.south),
          \n1 = {min(\y1,\y2) - 8pt}
      in
      (\p1)
      -- (\x1, \n1)
      -- (\x2, \n1) \tikztonodes 
      -- (\p2)
    }
  }
}

\tikzset{
  rightvertleft/.style={
    rounded corners,
    to path={
      let \p1 = (\tikztostart.east),
          \p2 = (\tikztotarget.east),
          \n1 = {max(\x1,\x2) + 8pt}
      in
      (\p1)
      -- (\n1, \y1)
      -- (\n1, \y2) \tikztonodes 
      -- (\p2)
    }
  }
}

\tikzset{
  leftvertright/.style={
    rounded corners,
    to path={
      let \p1 = (\tikztostart.west),
          \p2 = (\tikztotarget.west),
          \n1 = {min(\x1,\x2) - 8pt}
      in
      (\p1)
      -- (\n1, \y1)
      -- (\n1, \y2) \tikztonodes 
      -- (\p2)
    }
  }
}

\newcommand{\inlinetikzcd}[1]{\begin{tikzcd}[sep=small, ampersand replacement=\&]#1\end{tikzcd}}

\newcommand{\CyclicGroup}[1]{\mathbb{Z}_{/#1}}

\renewcommand{\setminus}{-}

\newcommand{\defneq}{\equiv}

\newcommand{\shape}{%
  {\mathord{\scalerel*{\raisebox{0.1ex}{\textup{\textesh}}}{f}}%
  \mkern 2mu}
}

\newcommand{\ii}{\mathrm{i}}

\newcommand{\dir}[2]{\overset{\scalebox{.54}{$\mathclap{\;#1}$}}{#2}}

\newcommand{\dd}{\mathrm{d}}

\newcommand{\FR}{\mathbb{R}}

\newcommand\bosonic[1]{\mathstrut\mkern0mu#1\mkern-14mu\raise1.7ex%
  \hbox{$\scriptstyle\rightsquigarrow$}}

\newcommand{\evencoordinateindex}{r}
\newcommand{\oddcoordinateindex}{\rho}

\newcommand{\BosonicCounit}{{\epsilon^{\rightsquigarrow}}}

\newcommand{\FullSpacetimeFlux}{F}

\usepackage{amsmath}

\makeatletter
\newcommand\makebig[2]{%
  \@xp\newcommand\@xp*\csname#1\endcsname{\bBigg@{#2}}%
  \@xp\newcommand\@xp*\csname#1l\endcsname{\@xp\mathopen\csname#1\endcsname}%
  \@xp\newcommand\@xp*\csname#1r\endcsname{\@xp\mathclose\csname#1\endcsname}%
}
\makeatother

\makebig{biggg} {3.0}
\makebig{Biggg} {3.5}
\makebig{bigggg}{4.0}
\makebig{Bigggg}{4.5}
\makebig{biggggg}{5.0}
\makebig{Biggggg}{5.5}

\begin{document}
%%%%%%%%%%%%%%%%%%%%%%%%%

%%%%%%%%%%%%%%%%%%%%%%%%%%%%%%%%%%
%vertical spacing around displayed equations %
\setlength{\abovedisplayskip}{3.5pt}
\setlength{\belowdisplayskip}{3.5pt}
\setlength{\abovedisplayshortskip}{-6pt}
\setlength{\belowdisplayshortskip}{4pt}
%%%%%%%%%%%%%%%%%%%%%%%%%%%%%%%%%%%%%

\titlehead{
  \href{https://ncatlab.org/nlab/show/Center+for+Quantum+and+Topological+Systems}{CQTS} Lecture Notes
}

\subject{
  Introduction to modernized
}
\title{
  Higher Superspace Supergravity
}
\subtitle{
  \& its IR Completions by Flux Quantization
  \\[.3cm]
  \phantom{A}
}

\author[1]
{\NameEmail
  {Grigorios Giotopoulos}
  {gg2658@nyu.edu}
}

\author[1,2]
{\NameEmail
  {Hisham Sati}
  {hsati@nyu.edu}
}

\author[1]
{\NameEmail
  {Urs Schreiber}
  {us13@nyu.edu}
}

\affil[1]{Mathematics Program and Center for Quantum and Topological Systems (CQTS), \newline New York University Abu Dhabi, UAE}
\affil[2]{The Courant Institute for Mathematical Sciences, \newline New York University, New York, USA}

\maketitle

\begin{abstract}
  It is an old idea that higher-dimensional super-gravity (SuGra) is put on-shell just by imposing Bianchi identities on super-field strengths over super-spacetimes subject to super-torsion constraints.
  We give a modernized, rigorous account and review recent developments, pointing out how this perspective lends itself to the construction of infrared completions of SuGra by electromagnetic flux quantization in differential nonabelian cohomology theories $\mathcal{A}_{\mathrm{dff}}$ lifting the coefficient $L_\infty$-algebra $\mathfrak{l}\mathcal{A}$ of the super-Bianchi identities. 
  After surveying necessary background, we
  highlight:
  \begin{enumerate}
  \item our recent proof that solutions of 11D SuGra are equivalent to the $\mathfrak{l}S^4$-Bianchi identities on super $C$-field flux over super-torsion-free 11D super-spacetimes, 

  \item how from this the nonlinearly self-dual gauge sector on 6D M5-brane worldvolumes is equivalent to the $\mathfrak{l}_{S^4}S^7$-Bianchi identity on super $B$-field flux,
  
  \item and the recent complete discussion of superspace dimensional reduction of this situation to 10D IIA SuGra via $\mathfrak{l}\mathrm{Cyc}(S^4)$-Bianchi identities on super NS/RR-flux, together with its extension to further reduction to 9 SuGra via $\mathfrak{l}\mathrm{Tor}(S^4)$-Bianchis on the corresponding super-fluxes. 
  \end{enumerate}
  We close by indicating how this implies consistent infrared completions 
  of 11D SuGra by $C$-field flux quantization in 4-Cohomotopy, 
  of 10D IIA SuGra by NS/RR-field flux-quantization in a form of twisted unstable K-theory, 
  and
  of M5-brane worldvolumes by $B$-field flux-quantization in twisted (twistorial) relative Cohomotopy. The latter admits geometric engineering of experimentally relevant topological quantum orders.
\end{abstract}


\MSCcodes{
  Primary:
  83E50, %Supergravity	
  58A50, %Supermanifolds and graded manifolds
  81T30; %String and superstring theories; other extended objects (e.g., branes) in quantum field theory 
  Secondary:
  55P62, %Rational homotopy theory
  55Q55, %Cohomotopy groups
  81T70, %Quantization in field theory; cohomological methods 
  17B55 %Homological methods in Lie (super)algebras
}

\setcounter{tocdepth}{2}
\tableofcontents

\vfill

%%%%%%%%%%%%
\paragraph
{Acknowledgements}
%%%%%%%%%%%
We thank the organizers and the participants of the workshop \emph{\href{https://www.theory-challenges.eu/events/spqaecis}{Symmetries, Protected Quantities, and Exact Computations in Supergravity}} (Nesselwang, 29/06--03/07 2026), where parts of these lectures were first presented. 

The authors acknowledge funding by \textit{Tamkeen UAE} under the \textit{Abu Dhabi Research Institute Grant} \texttt{CG008}.

\newpage

%%%%%%%%%%%%%
%%%%%%%%%%%%%
\section
{Overview}
%%%%%%%%%%%%%

Super-gravity (SuGra) originated over 50 years ago with \cite{FreedmanEtAl1976,DeserZumino1976} following \cite{VolkovAkulov1972,VolkovSoroka1973}. Its superspace formulation was pioneered by \cite{AkulovVolkovSoroka1977,WessZumino1977,GrimmWessZumino1979,SiegelGates1979,GatesStelleWest1980,Howe1982,DAuriaFre1982,Gates1983,CarrGatesOerter1987,GatesLuOerter1989}. General introductions and reviews include: \parencites{vanNieuwenhuizen1981}{Cremmer1982}{West1986}{DuffNilssonPope1986}{SalamSezgin1989}{CDF1991}{WessBagger1992}{BuchbinderKuzenko1995}[\S31]{Weinberg2000}{FreedmanVanProeyen2012}
[\S6--9]{Fre2013}{Ortin2015}{TraubenbergEtAl2020}{Sezgin2023}{DomingoTraubenberg2023}. We follow \cite{GSS24-SuGra,GSS24-M5,GSS25-Embedding,GiotopoulosSati2026}.

%%%%%%%%%%%%%%%
\subsection
{Motivation}
%%%%%%%%%%%%%%%

\textbf{The theoretical idea of SuGra} is rather suggestive: Where ordinary gravity is (cf. \parencites[\S 3]{Krasnov2020}) Cartan geometry for the Spin subgroups of Poincar{\'e} groups (\emph{local Lorentz symmetry}), the latter have natural supergroup extensions (\emph{super-symmetry}, SuSy, cf. \parencites{Freed1999}{Varadarajan2004}), and super-gravity is the Cartan geometry for the Spin subgroups of these super-Poincar{\'e} groups (\emph{local super-symmetry}) combined with higher gauge fields dictated by the Chevalley-Eilenberg structure of these groups. We discuss this in \cref{Supergravity}.

%%%%%%%%%%%%%
\paragraph
{Physical motivations}
%%%%%%%%%%%%
Less widely appreciated may be that discussion of super-gravity has \textit{concrete physical motivations} (beyond the widely discussed application of \emph{holography} to solid-state physics and quark-gluon plasmas):
\begin{description}

\item[\textbf{In fundamental physics},] the high-energy local SuSy embodied by SuGra remains phenomenologically viable, in contrast to the low-energy global SuSy that has notoriously not been seen in experiment. In particular, SuGra is known \cite{nLab:CosmicInflationInSupergravity} to improve models of inflationary cosmology, by naturally furnishing not only the required scalar fields themselves but also explaining their observationally favored plateau potentials (the \emph{$\alpha$-attractor} mechanism) and their stability against quantum corrections (resolving the \emph{$\eta$-problem}). 

\item[\textbf{In particle physics},] despite common perception, (approximate) super-symmetry is experimentally seen in the form of \emph{hadron super-symmetry} \cite{nLab:HadronSupersymmetry}, in the ``flavor sector'' instead of the commonly considered ``color sector'' of the standard model. Here it naturally fits into holographic models realizing quantum chromodynamics on supergravitational branes.

\item[\textbf{In solid-state physics},] \emph{topological order} in quantum materials is a 21st-century holy grail, plausibly necessary for stabilizing future utility-scale quantum computing hardware. To date, a single class of candidate systems has been experimentally realized with certainty: fractional quantum Hall (FQH) liquids. 

Remarkably, it has been established (cf. \cite{nLab:SuSyInFQHReferences} and see \cref{Applications} below) that the gapped excitations of certain non-abelian FQH liquids fall into multiplets comprising a (massive graviton-like) \emph{spin-2 magneto-roton} and a (massive gravitino-like) \emph{spin-3/2 neutral fermion mode} which are partners under an emergent super-symmetry, suggestive of an effective (2D massive) super-gravity description.

Furthermore, the effective symmetry of FQH excitations is ``special'' diffeomorphism invariance ($W_\infty$-symmetry). In combination with SuSy this matches the characteristic symmetry on super $p$-brane probes of super-gravity backgrounds, suggesting that the effective field theory of FQH liquids may naturally be understood via \emph{geometric engineering} of strongly coupled quantum systems on branes probing super-gravity orbifolds \parencites{SS25-Srni}{SS25-Seifert}.

\item[\textbf{In mathematical metaphysics},] the \emph{brane bouquet} \parencites{HS18-Superpoint}[Fig. 1]{FSS19-RationalM}[Prop. 4.14]{FSS15-WZW} shows that out of the \emph{superpoint} --- the classifying space for fermion fields \emph{per se} --- grows a hierarchy of (the local structures of) ever higher-dimensional super-gravity theories and their super $p$-brane content, culminating in the local structure of 11D super-gravity with its M2/M5-branes. Therefore, super-gravity and branes are fundamentally hard-coded into the BIOS of mathematical physics. The surprising emergence of SuGra in the effective theories of hadrons and of FQH liquids, noted above, may be an instance of this universality. 
\end{description}

We further relate to some of these applications in \cref{Applications}.  We will arrive there by revisiting the theoretical foundations related to the appearance of \emph{higher gauge sectors} in \emph{higher-dimensional} formulations of super-gravity (cf. \cref{LogicOfSuGra}).

\begin{figure}[htb]
\caption{\label{LogicOfSuGra}
It is widely appreciated that lower-dimensional SuGra tends to be best understood as dimensional reduction of higher-dimensional SuGra (cf. \cref{SomeReductionsOf11DSuGra}) involving higher gauge fields. Less attention has been devoted to the fact that the latter only make global sense after a choice of electromagnetic flux quantization. Here we discuss how such admissible (infrared) completions of SuGra are controlled by the Bianchi identities on superspace.
}
\centering
\adjustbox{
  rndfbox=4pt,
  scale=1.2
}{
\;\;\;\;
\begin{tikzcd}[
  column sep=50pt
]
 \colorbox{lightgray}{$\substack{
  \text{low dim} 
  \\
  \text{SuGra}
}$}
\ar[
  r,
  "{
    \text{reduction of}
  }"
]
&
 \colorbox{lightgray}{$
\substack{
  \text{higher-dim} 
  \\
  \text{SuGra}
}$}
\ar[
  r,
  "{
    \text{governed by}
  }"
]
& 
\colorbox{lightgray}{$
\substack{
  \text{higher}
  \\
  \text{gauge sector}
}$}
\ar[
  dll,
  snake left,
  "{
    \text{requiring}
  }"{description}
]
\\
 \colorbox{lightgray}{$
\substack{
  \text{global completion}
  \\
  \text{by flux quantization}
}
$}
\ar[
  r,
  "{ \text{controlled by} }"
]
&
 \colorbox{lightgray}{$
\substack{
  \text{Bianchi identities}
  \\
  \text{on superspace}
}
$}
\end{tikzcd}
\;\;\;\;
}
\end{figure}
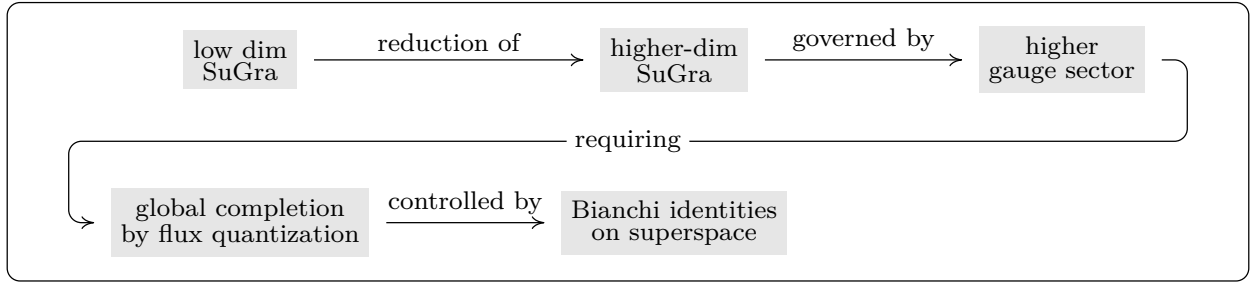

%%%%%%%%%%%%%
\paragraph
{Higher dimensions}
%%%%%%%%%%%%%
Local SuSy drastically constrains the possible field content and couplings of super-gravity theories. 
As a result, it is a famous phenomenon that lower-dimensional SuGra theories tend to be most transparently understood as (Kaluza-Klein, KK) \emph{dimensional reductions} of higher-dimensional SuGra theories (cf. \parencites{Cremmer1982}{DuffNilssonPope1986}[\S22.1]{Ortin2015}{DuffEtAl2026}). Here the extra dimensions serve as a neat geometric encoding of consistent fields and couplings in lower dimensions (cf. \cref{SomeReductionsOf11DSuGra}).

At the top of this hierarchy appears $D=11$ $\mathcal{N}=1$ super-gravity, usually just called \emph{11D SuGra} (cf. \cref{Supergravity}).

\begin{SCfigure}[.5][htb]
\caption{\label{SomeReductionsOf11DSuGra}
A small selection of super-gravity theories arising by dimensional reduction from 11D SuGra. The label $\mathcal{N}$ is the ``number of super-symmetries'' in the given dimension, indicating the $\mathrm{Spin}$-representation on the gravitino field. The label on an arrow indicates the nature of the relevant fiber space.
\newline
\newline
(Not all super-gravity theories descend this way from 11D SuGra. For instance, 10D heterotic SuGra does not.)
}
\centering
\adjustbox{
  scale=.9,
  rndfbox=4pt
}{
\begin{tikzcd}[row sep=15pt, 
  column sep=-2pt
]
  D = 11
  &[+5pt] 
  &[-10pt] 
  \adjustbox{
    rndfbox=4pt,
    bgcolor=lightgray
  }{$\mathcal{N}=1$}
  \ar[
    to=D11N2A,
    "{ 
      S^1 
    }"{swap}
  ]
  \ar[
    to=D9N2,
    "{ 
      T^2 
    }"{description}
  ]
  \ar[
    to=D4N1,
    "{ 
      \text{$G_2$-Mfd} 
    }"{description}
  ]
  \ar[
    to=D4N2,
    "{ 
      \mathrm{CY}_3 \times S^1 
    }"{description, pos=.6}
  ]
  \ar[
    to=D4N4,
    "{ 
      \mathrm{K3} \times T^3 
    }"{description}
  ]
  \ar[
    to=D4N8,
    "{ 
      T^7 
    }"{description, pos=.42}
  ]
  \ar[
    to=D3N1,
    "{ 
      \text{$\mathrm{Spin}(7)$-Mfd} 
    }"{description, pos=.65}
  ]
  \\
  D = 10
  &
  |[alias=D11N2A]|
  \adjustbox{
    rndfbox=4pt,
    bgcolor=lightgray
  }{$\mathcal{N}=(1,1)$}  
  \ar[
    to=D4N2,
    crossing over,
    bend right=9,
    "{ 
      \mathrm{CY}_3 
    }"{description, pos=.2}
  ]
  \\
  D=9
  &&&&& % Pushes the node precisely into the 6th column
  |[alias=D9N2]|
  \adjustbox{
    rndfbox=4pt,
    bgcolor=lightgray
  }{$\mathcal{N}=2$}
  \ar[
    to=D4N8,
    "{ 
      T^5 
    }"{description}
  ]
  \\
  D=8
  \\
  D=7
  \\
  D=6
  \\
  D=5
  \\
  D=4
  &&
  |[alias=D4N1]|
  \adjustbox{
    rndfbox=4pt,
    bgcolor=lightgray
  }{$\mathcal{N}=1$}   
  \ar[
    to=D3N2,
    "{ 
      S^1 
    }"{pos=.45}
  ]
  & 
  |[alias=D4N2]|
  \adjustbox{
    rndfbox=4pt,
    bgcolor=lightgray
  }{$\mathcal{N}=2$}   
  \ar[
    to=D3N4,
    "{ 
      S^1 
    }"{pos=.45}
  ]
  & 
  |[alias=D4N4]|
  \adjustbox{
    rndfbox=4pt,
    bgcolor=lightgray
  }{$\mathcal{N}=4$}   
  \ar[
    to=D3N8,
    "{ 
      S^1 
    }"{ pos=.45}
  ]
  & 
  |[alias=D4N8]|
  \adjustbox{
    rndfbox=4pt,
    bgcolor=lightgray
  }{$\mathcal{N}=8$}   
  \ar[
    to=D3N16,
    "{ 
      S^1 
    }"{pos=.45}
  ]
  \\
  D=3
  &
  |[alias=D3N1]|
  \adjustbox{
    rndfbox=4pt,
    bgcolor=lightgray
  }{$\mathcal{N}=1$}   
  &
  |[alias=D3N2]|
  \adjustbox{
    rndfbox=4pt,
    bgcolor=lightgray
  }{$\mathcal{N}=2$}   
  &
  |[alias=D3N4]|
  \adjustbox{
    rndfbox=4pt,
    bgcolor=lightgray
  }{$\mathcal{N}=4$}   
  &
  |[alias=D3N8]|
  \adjustbox{
    rndfbox=4pt,
    bgcolor=lightgray
  }{$\mathcal{N}=8$}   
  &
  |[alias=D3N16]|
  \adjustbox{
    rndfbox=4pt,
    bgcolor=lightgray
  }{$\mathcal{N}=16$}   
\end{tikzcd}
}
\end{SCfigure}

Hence, even if one is \emph{a priori} interested in low-dimensional SuGra, analysis of the theory tends to involve higher-dimensional versions.

%%%%%%%%%%%%
\paragraph
{Higher gauge sectors}
%%%%%%%%%%%%

Now, a key phenomenon of higher-dimensional SuGra is the necessary appearance of further bosonic fields%
\footnote{
A quick way to understand the necessity of further fields is that the dimension $\sim 2^{\lfloor D/2 \rfloor}$ of the spacetime spinor representations grows faster with the spacetime dimension $D$ than the number of (metric) components $\sim D(D-1)/2$ of the gravitational field, hence requiring further bosonic fields to satisfy super-symmetry between bosonic and fermionic fields.}, specifically of \emph{higher gauge fields} (cf. \parencites{Szabo2013}{SS26-HigherGauge}).

Concretely, what appears in the equations of motion of higher-dimensional SuGra are \emph{flux density} differential forms $F_{p}$ (cf. \cite[\S 2]{SS25-Flux}) of higher degree $p$. 

For instance, 11D SuGra involves a 4-flux density $G_4$ and a 7-flux density $G_7$ (the flux of the \emph{$C$-field}) subject to these equations of motion (\cite[\S 3.1.3]{MiemiecSchnakenburg2006}, cf. \cref{lS4IdentityFailsOnSpacetimeSalvagedOnSuperspace}):
\begin{equation}
\label
{CFieldEoM}
  \substack{%
    \color{gray}%
    \text{duality-symmetric}
    \\
    \color{gray}%
    \text{Bianchi identities}
  }
  \;\;\;
  \begin{aligned}
    \mathrm{d}\, G_4 & = 0
    \mathrlap{\,,}
    \\
    \mathrm{d}\, G_7 & =
    \tfrac{1}{2} G_4 \wedge G_4
    \mathrlap{\,,}
  \end{aligned}
  \;\;\;\;\;\;
  G_7 = \star G_4
  \;\;\;
  \substack{
    \color{gray}
    \text{duality constraint}
  }
  \mathrlap{\,.}
\end{equation}
This \cref{CFieldEoM} is clearly a higher-degree version of \emph{Maxwell's equations} for the electromagnetic flux density $F_2$ in 4D (the flux of the \emph{$A$-field}), which famously read as follows \cite[\S80]{Cartan1924}:
\begin{equation}
\label
{BosonicCFieldEoMs}
  \substack{%
    \color{gray}%
    \text{duality-symmetric}
    \\
    \color{gray}%
    \text{Bianchi identities}
  }
  \;\;\;
  \begin{aligned}
    \mathrm{d}\, F_2 & = 0
    \mathrlap{\,,}
    \\
    \mathrm{d}\, G_2
    & = J_3
    \mathrlap{\,,}
  \end{aligned}
  \;\;\;\;\;\;\;\;\;\;\;\;\;
  G_2 = \star F_2
  \;\;\;
  \substack{
    \color{gray}
    \text{duality constraint}
  }
\end{equation}
(for some prescribed \emph{electric current density} $J_3$).

It follows that \emph{locally} on a chart $\inlinetikzcd{ U \ar[r, hook, "{\iota}"] \& X^{1,10} }$ of spacetime the higher gauge field of 11D SuGra is given by a differential 3-form $C_3$ and a differential 6-form $C_6$:
\begin{equation}
\label
{CFieldPotential3Form}
  \begin{aligned}
    \mathrm{d}\, C_3 
      &= \iota^\ast G_4
      \mathrlap{\,,}
    \\
    \mathrm{d}\, C_6 + \tfrac{1}{2} C_3 \wedge \iota^* G_4 
      &= \iota^\ast G_7
      \mathrlap{\,,}
  \end{aligned}
\end{equation}
(originally called the \emph{3-index photon} and its dual, now usually just called the \emph{C-field}) in higher analogy to how the electromagnetic field is locally
on a chart $\inlinetikzcd{ U \ar[r, hook, "{\iota}"] \& X^{1,3} }$ given by a differential 1-form $A_1$ (the \emph{vector potential}) with 
\begin{equation}
\label
{OrdinaryVectorPotential}
  \mathrm{d}A_1 = \iota^\ast F_2
  \mathrlap{\,.}
\end{equation}

%%%%%%%%%%%%%
\paragraph
{Infrared completions}
%%%%%%%%%%%%%

However, this raises a crucial question that has remained largely underappreciated.
What is actually needed in applications (at least to topological quantum orders) are \emph{globally defined} SuGra theories --- meaning defined at large scales, jargon: in the \emph{infrared} (IR) --- namely with well-defined topological (brane) charges.

For the example of the ordinary electromagnetic field, it is well-known that in addition to the local (chartwise) vector potentials $A_1$ \cref{OrdinaryVectorPotential}, it involves further data, given by specified gauge transformations (\emph{transition functions}) between gauge potentials associated with different charts where these overlap. Moreover, where three charts overlap, these gauge transformations are required to satisfy a consistency constraint that implements the \emph{quantization} (in the sense of: discretization) of electromagnetic flux/charge.

This highlights that the global definition of the C-field in 11D SuGra also must involve more data than the local 3-form potentials $C_3$ \cref{CFieldPotential3Form} that are usually considered, exhibiting a higher form of \emph{flux quantization}.

Hence, already classically (before entering the topic of quantum field theory) we find that:
\begin{standout}
  To completely define a higher-dimensional super-gravity theory requires an admissible choice of global completion of its higher gauge field sector by flux quantization.
\end{standout}

Since such completions concern the behavior of fields at large spacetime scales (not coverable by a single coordinate chart), we may speak of \emph{infrared completions}.

Hence, after 55 years of SuGra theory, it is time to ask for the actual globally complete formulation of these theories --- providing for super-gravity what Dirac's charge quantization argument began to provide in 1931 for Maxwell's theory (from 1865).

%%%%%%%%%%%%%
\paragraph
{Questions and Answers}
%%%%%%%%%%%%%

Hence here we ask -- \textit{and answer} -- the following questions:

\vspace{-2mm} 
\begin{center}
\emph{What is the $C$-field, globally (``model of the C-field'')?}
\end{center}
Hence:
\vspace{-5mm} 
\begin{center}
\emph{What is 11D SuGra globally, what are its IR-completions?}
\end{center}
\vspace{-5mm} 
\begin{center}
\emph{What are its compatibly IR-complete dimensional reductions?}
\end{center}
Then:
\vspace{-5mm} 
\begin{center}
\emph{What is the ``self-dual'' field on M5-branes, globally?}
\end{center}
Hence: 
\vspace{-5mm} 
\begin{center}
\emph{What is the M5 globally, what are its IR-completions (``M5-brane model'')?}
\end{center}

\vspace{-3mm} 
%%%%%%%%%%%%%%
\paragraph
{Flux Quantization on Superspace}
%%%%%%%%%%%%%%

This leads us to revisit the higher gauge sectors of higher-dimensional SuGra. Here we encounter a miracle 
\cite[Thm. 3.1]{GSS24-SuGra}\cite[Thm. 3.4]{GiotopoulosSati2026} previously not fully appreciated: 

\begin{standout}
When formulated \emph{on superspace}, the equations of motion of 11D SuGra and its KK-descendants are \emph{equivalent} to the gauge sector equations \cref{BosonicCFieldEoMs}, thus allowing for infrared completions by covariant flux quantization.
\end{standout}

We next survey what this means.

%%%%%%%%%%%%%%
\subsection
{Formulation}
%%%%%%%%%%%%%%

%%%%%%%%%%%%
\paragraph
{Superspace super-gravity}
%%%%%%%%%%%%%
The SuGra field content and equations of motion tend to appear intricate and opaque when expressed in terms of fields over ordinary spacetime.

But just as plain gravity is intrinsically a theory of \emph{spacetime geometry} rather than a theory of fields over spacetime, so super-gravity becomes conceptually more transparent as a theory of \emph{super}-spacetime geometry (cf. \cref{Superspace}). This perspective geometrizes the symmetry principle, realizing local super-symmetry as super-diffeomorphism invariance.   

Or rather, such a picture is what the practice of super-gravity suggests \cite{CDF1991}, even if precise formulations have remained elusive in the physics literature.

%%%%%%%%%%%%
\paragraph
{Super-Spacetime}
%%%%%%%%%%%%

More concretely, for $\mathbf{N}$ a real linear representation of $\mathrm{Spin}(1,d)$ with equivariant pairing 
\begin{equation}
  \bracket({
    \overline{(-)}
    \Gamma
    (-)
  })
  :
  \inlinetikzcd{
    \mathbf{N} 
      \otimes_{_{\mathbb{R}}}
    \mathbf{N}
    \ar[r]
    \&
    \mathbb{R}^{1,d}
    \mathrlap{\,,}
  }
\end{equation}
the \emph{$(1,d\vert \mathbf{N})$-dimensional super-Minkowski spacetime} is the superspace whose algebra of smooth functions is freely generated from $1 + d$ even-graded coordinate functions $x^a$ and $N$ odd-graded coordinate functions $\theta^i$:
\begin{equation}
\label{SuperMinkowskiFunctionAlgebra}
  C^\infty\bracket({
    \mathbb{R}^{1,d\vert \mathbf{N}}
  })
  \defneq
  C^\infty\bracket({
    \mathbb{R}^{1,d}
  })
  \otimes_{_{\mathbb{R}}}
  \wedge^\bullet \bracket({
    \mathbf{N}^\ast
  })
  \mathrlap{\,,}
\end{equation}
and equipped with the super-coframe field $(E,\Psi)$ given by
\begin{equation}\label{FlatSuperCoframe}
  \left.
  \begin{aligned}
    E^a 
    & :=
    \mathrm{d} x^a
    +
    \bracket({
      \overline{\theta}
      \Gamma^a
      \mathrm{d}\theta
    })
    \\
    \Psi^\alpha
    & :=
    \mathrm{d} \theta^\alpha
    \mathrlap{\,,}
  \end{aligned}
  \right\}
  \;\;\;
  \text{hence satisfying}
  \;\;\;
  \left\{
  \begin{aligned}
    \mathrm{d}\,
    E^a
    & =
    \bracket({
      \overline{\Psi}
      \Gamma^a
      \Psi
    })
    \\
    \mathrm{d}\,
    \Psi^\alpha
    & = 0
    \mathrlap{\,.}
  \end{aligned}
  \right.
\end{equation}
Regarded as left-invariant super 1-forms, this witnesses $\mathbb{R}^{1,d\vert\mathbf{N}}$ as a supergroup (cf. \cite[\S 3.1]{GSS26-Hidden}), and as such, the identities from \eqref{FlatSuperCoframe} are precisely the defining cochain conditions of its tangent super-Lie algebra structure. Its semidirect product $\mathbb{R}^{1,d\vert \mathbf{N}} \rtimes \mathrm{Spin}(1,d)$ is the \emph{super-Poincar{\'e} group} in dimension $(1,d\vert \mathbf{N})$, traditionally called ``$D = 1+d$, $\mathcal{N} = \mathbf{N}$ \emph{super-symmetry}''.   

A general $(1,d\vert \mathbf{N})$-dimensional super-spacetime $X^{1,d\vert \mathbf{N}}$ (\cref{Frames}) is a super-manifold equipped with a super coframe
\begin{equation}
\label
{SuperCoframeInIntro}
  \bracket({
    E,\Psi
  })_{(-)}
  :
  \inlinetikzcd{
    T_{(-)} X^{1,d\vert \mathbf{N}}
    \ar[
      rr,
      "{ \sim }"
    ]
    \&\&
    \mathbb{R}^{1,d\vert \mathbf{N}} \, ,
  }
\end{equation}
and a Spin-connection $\Omega$ that satisfies the \textit{(super-)torsion-free} condition in the bosonic direction
\begin{equation}\label{TorsionCondition}
\dd\, E^{a} + \Omega^{a}{}_b\, E^b \, = \, \bracket({
      \overline{\Psi}
      \Gamma^a
      \Psi
    }) \, . 
\end{equation}

We stress that this is nothing but the curved and globalized Cartan geometric extension of the corresponding (bosonic) tangent-space-wise cochain condition from \eqref{FlatSuperCoframe}, now to arbitrary super-manifolds.

% {\color{gray}G: A small fact that I've noticed is the following. All constraints we impose, be it the torsion free constraint, or the Bianchis, are all imposed on derivatives of *bosonic* coframes and fluxes. Perhaps we can also state  this as a principle. Namely, one takes the flat spacetime cocycle conditions of 11D, single out the bosonic ones, canonically extend them to curved spacetimes by adding the natural terms.}

% U: Maybe briefly like this:

In this vein it is natural to also consider flat-space avatars for the bosonic flux densities and then promote these to curved spacetime:

%%%%%%%%%%%%%
\paragraph
{Avatar super-flux}
%%%%%%%%%%%%%

Remarkably, super-Minkowski spacetimes turn out to carry super-symmetric super-forms whose de Rham differential models the duality-symmetric Bianchi identities of the corresponding super-gravity theories. 

For instance, in 11D the following invariant superforms (we notationally suppress the wedge product symbol)
\begin{equation}
\label
{11DAvatarSuperFluxInIntro}
  \left.
  \begin{aligned}
    G_4^0
    & 
    :=
    \mathcolor{darkorange}{\tfrac{1}{2}\bracket({
      \overline{\Psi}
      \Gamma_{a_1 a_2}
      \Psi
    })
    E^{a_1} E^{a_2}
    }
    \\
    G_7^0
    & 
    :=
    \mathcolor{darkorange}{\tfrac{1}{5!}\bracket({
      \overline{\Psi}
      \Gamma_{a_1 \cdots a_5}
      \Psi
    })
    E^{a_1} \cdots E^{a_5}
    }
  \end{aligned}
  \right\}
  \;\;\;
  \text{on $\mathbb{R}^{1,10\vert \mathbf{32}}$}
\end{equation}
miraculously satisfy the following equations (\cref{TheAvatarSuperFLuxesIn11D}, cf. \cite[\S 2]{FSS17-Sphere}):
\begin{equation}
\label
{BianchiOf11DSuperFluxesInIntro}
  \begin{aligned}
    \mathrm{d}\, 
    G_4^0 & = 0
    \\
    \mathrm{d}\,
    G_7^0 & =
    \tfrac{1}{2}
    G_4^0 \wedge G_4^0
    \mathrlap{\,,}
  \end{aligned}
\end{equation}
of the same form as \cref{CFieldEoM}!

%%%%%%%%%%%%
\paragraph
{The Role of Higher Super Cartan Geometry}
%%%%%%%%%%%%%
% Equivalently, these forms \cref{11DAvatarSuperFluxInIntro} are left-invariant extensions of corresponding cochains on the super-Lie algebra at the tangent space of the identity of the super-group $\mathbb{R}^{1,10\vert \mathbf{32}} \rtimes \mathrm{Spin}(1,10)$. In this sense, we may equivalently think of this (...) 

Besides solving the expected equations of motion, these avatar super-flux forms
\cref{11DAvatarSuperFluxInIntro} and their analogs in other super-dimensions
turn out to exhibit expected \emph{duality} relations (cf. \parencites[\S 6]{Duff1999World}[\S 17]{West2012}{Polchinski2014}), such as \emph{M/IIA duality} \cite{FSS17-Sphere}, \emph{T-duality} \cite{FSS18-TD,GSS25-TD}, \emph{S-duality} \cite[Prop. 8.6]{FSS18-TD}, and \emph{U-duality} \parencites[\S 4]{SatiVoronov2024}{GSS24-Exceptional}. Hence, in some sense the avatar super-fluxes preconfigure the web of super-gravity theories (exhibited by the \emph{brane bouquet} \parencites[Prop. 4.14]{FSS15-WZW}[Fig. 3]{HSS2019}[Fig. 1]{FSS19-RationalM}).

Viewing these relations through the lens of Cartan geometry with respect to the subgroup inclusion $\inlinetikzcd{\mathrm{Spin}(1,10) \ar[r, hook] \& \mathrm{Iso}(\mathbb{R}^{1,10\vert \mathbf{32}})}$, whereby a super-manifold $X^{1,10\vert \mathbf{32}}$ is supplied with a super-coframe $(E,\Psi)$, we may interpret them as  {\color{darkorange}``tangent-space-wise relations''} when pulled back along the corresponding point-wise isomorphism \eqref{SuperCoframeInIntro}.

In this vein, it is natural to ask for \textit{global extensions} of the tangent-space-wise defined avatar forms \cref{11DAvatarSuperFluxInIntro} by ``curved'' \textit{purely bosonic} components:
\begin{equation}
\label
{11DSuperFluxFormsInIntro}
  \left.
  \begin{aligned}
    G_4^s
    &:=
    \tfrac{1}{4!}
    \bracket({
      G_4
    })_{a_1 \cdots a_4}
    E^{a_1} \cdots E^{a_4}
    +
    \mathcolor{darkorange}{
    \tfrac{1}{2}
    \bracket({
      \overline{\Psi}
      \Gamma_{a_1 a_2}
      \Psi
    })
    E^{a_1} E^{a_2}
    }
    \\
    G_7^s
    &:=
    \tfrac{1}{7!}
    \bracket({
      G_7
    })_{a_1 \cdots a_7}
    E^{a_1} \cdots E^{a_7}
    +
    \mathcolor{darkorange}{
    \tfrac{1}{5!}
    \bracket({
      \overline{\Psi}
      \Gamma_{a_1 \cdots a_5}
      \Psi
    })
    E^{a_1} \cdots E^{a_5}}
  \end{aligned}
  \right\}
  \;\;\;
  \text{on $X^{1,10\vert \mathbf{32}}$.}
\end{equation}

%%%%%%%%%%%%%
\paragraph
{The Miracle}
%%%%%%%%%%%%%

It now turns out (\cite[Thm. 3.1]{GSS24-SuGra}, reviewed as \cref{11DSuGraAslS4BianchiOnSuperspace} in \cref{Equations}) that solutions of 11D super-gravity are equivalent simply to super-torsion-free super-coframes \cref{SuperCoframeInIntro} (encoding the super-gravitational field) supporting super-flux forms \cref{11DSuperFluxFormsInIntro} (encoding the $C$-field) subject to the Bianchi identities \cref{CFieldEoM}:
\begin{equation}
\label
{TheMiracleEquivalence}
\adjustbox{
  scale=1.3
}{$
\left\{
\colorbox{lightgray}{$
\substack{
\text{Super-torsion-free $(1,10\vert\mathbf{32})$-coframe}
\\
\text{\& super-fluxes satisfying Bianchis \cref{CFieldEoM}}
}
$}
\right\}
$}
\;\;
\Leftrightarrow
\;\;
\adjustbox{
  scale=1.3
}{$
\left\{
\colorbox{lightgray}{$
\substack{
  \text{Solutions to 11D SuGra}
}
$}
\right\}
$}
\mathrlap{\,.}
\end{equation}
This equivalence is a veritable miracle: Its proof involves pages of analysis and excessive Clifford algebra computations that a contemporary personal computer takes about 11 minutes to verify \cite{GSS24-SuGraAncillary}: After analysis of the many components of the superspace Bianchi identities, the result still hinges on a massive conspiracy of combinatorial prefactor cancellations. 

%%%%%%%%%%%%%
\paragraph
{Rheonomy}
%%%%%%%%%%%%
Moreover, this result \cref{TheMiracleEquivalence} involves a phenomenon called \emph{rheonomy} (so named in \cite[\S III.3.3]{CDF1991}, discussed in \cref{Rheonomy}): The gravitational and flux fields on super-spacetime \emph{a priori} have many more components than the actual fields on ordinary spacetime. But the super-Bianchi identities imply \emph{differential flow equations} (whence ``rheonomy'') along the odd-graded directions of super-spacetime, whereby the superfield solutions turn out to be uniquely determined by their restriction to ordinary spacetime.

%%%%%%%%%%%%%%
\paragraph
{Role of the Bianchi Identities}
%%%%%%%%%%%%%%

It was early on realized that super-gravity equations of motion tend to be at least close to being equivalent to superspace Bianchi identities subject to certain torsion constraints (cf. \cite{GrimmWessZumino1979}). This concerns the classical Bianchi identities \cite{Cartan1923} satisfied by the curvature and the torsion tensor of a Spin connection (now on superspace), but also the analogous equations for the higher flux densities (cf. \cref{TypesOfBianchiIdentities}).

Here the classical Bianchi identities (on curvature and torsion) are of course implied (solved) once a Spin-connection is given, which we will take to be part of the data of a super-spacetime. The remaining Bianchi identities on the flux densities are the equations of motion in the higher gauge sector.

\begin{SCtable}[.9][htb]
\caption{\label{TypesOfBianchiIdentities}
The field content of higher-dimensional super-gravity has a gravity sector (modeled by a super-coframe field and Spin-connection) and a higher gauge sector. 
}
\adjustbox{
  rndfbox=4pt
}{
\begin{tikzcd}[
  sep=2pt
]
  &[-6pt]
  \text{Gravity sector}
  &
  \text{Gauge sector}
  \\
    \text{Bianchis on:}
  &
  \substack{
    \text{curvature}
    \\
    \text{\& torsion}
  }
  &
  \substack{
    \text{duality-symmetric}
    \\
    \text{flux densities}
  }
\end{tikzcd}
}
\end{SCtable}

%%%%%%%%%%%%%%%
\paragraph
{Role of Duality-Symmetry}
%%%%%%%%%%%%%%
What drives the statement \cref{TheMiracleEquivalence} is the imposition of the superspace Bianchi identities specifically on the ``\emph{duality symmetric}'' flux densities \cref{CFieldEoM}. While this possibility has briefly been considered in \cite[\S III.8.5]{CDF1991} (cf. also \cite[\S 6]{CandielloLechner1994}), the understanding that the duality-symmetric super-flux Bianchi is the master constraint to put the theory on-shell has only been realized in \parencites[Thm. 3.1]{GSS24-SuGra}{GiotopoulosSati2026}. This perspective puts the higher gauge sector into the center of attention and hence lends itself to the further global completion of the theory (\cref{LogicOfSuGra}).

%%%%%%%%%%%%%%
\paragraph
{Characteristic $L_\infty$-Algebras}
%%%%%%%%%%%%%%

To systematize this, a key observation (\cref{Characteristics}) is that duality-symmetric Bianchi identities are equivalently the closure conditions for differential forms with coefficients in a \emph{characteristic $L_\infty$-algebra}% 
\footnote{
  Basic background on $L_\infty$-algebras is recalled in  \cref{LInfinityAlgebra}.
  Beware (\cref{DifferentRolesOfLInfinityAlgebras}) that these characteristic $L_\infty$-algebras $\mathfrak{a}$
  \cref{IdeaOfCharacteristicLInfinityAlgebras} serve as coefficients for the (globally defined) flux densities, and thus play a different role than the super $L_\infty$-algebras which may be identified \cite{FSS19-RationalM,FSS15-WZW} as implicit in \cite{CDF1991}, where they serve as coefficients for (locally defined) gauge potentials, instead.
}
$\mathfrak{a}$:
\begin{equation}
\label
{IdeaOfCharacteristicLInfinityAlgebras}
  \adjustbox{
    scale=1.3
  }{$
  \left\{
  \colorbox{lightgray}{$
  \substack{
    \text{Duality-symmetric}
    \\[2pt]
    \text{Bianchi identities}
  }
  $}
  \right\}
  $}
  \Leftrightarrow
  \adjustbox{
    scale=1.3
  }{$
  \left\{
  \colorbox{lightgray}{$
  \substack{
    \text{Closed forms with} 
    \\[2pt]
    \text{coefficients in $\mathfrak{a}$}
  }
  $}
  \right\}
  $}
\end{equation}

For instance, the flux Bianchi identities \cref{CFieldEoM} for 11D SuGra are characterized by the  $L_\infty$-algebra  $\mathfrak{a} = \mathfrak{l}S^4$, also known as the \emph{C-field gauge algebra} \cite[(2.6)]{CremmerEtAl1998}\cite[(3.4)]{LavrinenkoEtAl999}\cite[(75)]{KalkkinenStelle2003}\cite[(86)]{BandosEtAl2004}\cite[(4.9)]{Sati2010}:
\begin{equation}
\label
{lS4InIntroduction}
  \mathfrak{l}S^4
  \;\simeq\;
  \mathrm{Free}_{ L_\infty \mathrm{Alg}}\!%
  \left(
  \begin{aligned}
    & v_3
    \\
    & v_6
  \end{aligned}
  \right)
  \big/
  \left(
    \begin{aligned}
      [v_3] & = 0
      \mathrlap{}
      \\
      [v_6] & = 0
      \mathrlap{\,,}
    \end{aligned}
    \;
    \begin{aligned}
      [v_3, v_3] &= v_6
      \\
      [v_6, -] &= 0
    \end{aligned}
  \right)
\end{equation}
in that:
\vspace{1mm} 
\begin{equation}
\label
{11DSuGraFluxesAslS4ClosedFormsInIntro}
  \left\{
  \begin{aligned}
    & G_4
    \in 
    \Omega^4_{\mathrm{dR}}(X)
    \\
    & G_7
    \in 
    \Omega^7_{\mathrm{dR}}(X)
  \end{aligned}
  \;\middle\vert\;
  \begin{aligned}
    \mathrm{d}\, G_4 & = 0
    \\
    \mathrm{d}\, G_7 & = 
    \tfrac{1}{2} G_4 \wedge G_4
  \end{aligned}
  \right\}
  \simeq
  \Omega^1_{\mathrm{cl}}\bracket({
    X; 
    \mathfrak{l}S^4
  })
  \;\;
  \substack{
    \color{gray}%
    \text{closed (flat, Maurer-Cartan)}
    \\
    \color{gray}%
    \text{$L_\infty$-valued differential forms}
  }
  \mathrlap{\,.}
\end{equation}

As the notation indicates, these characteristic $L_\infty$-algebras tend to be associated with \emph{classifying spaces} $\mathcal{A}$, 
\begin{equation}
\label
{RealHomotopyGroupsInIntro}
  \bracket({
    \mathfrak{l}\mathcal{A}
  })_\bullet
  \defneq
  \pi_\bullet\bracket({
    \Omega \mathcal{A}
  })
  \otimes_{_{\mathbb{Z}}}
  \mathbb{R}
  \mathrlap{\,,}
\end{equation}
carrying the \emph{higher Whitehead brackets} on the real homotopy groups of the loop space. This encodes the \emph{$\mathbb{R}$-rational homotopy type} of $\mathcal{A}$ (\cref{Charges}).

%%%%%%%%%%%%%
\paragraph
{Dimensional Reduction}
%%%%%%%%%%%%%

One insight gained from identifying the characteristic $L_\infty$-algebra \cref{IdeaOfCharacteristicLInfinityAlgebras} is that it determines the minimal set of super-flux Bianchi identities which puts the theory on-shell after dimensional reduction (\cref{Reductions}).

Namely, if $\mathfrak{a} = \mathfrak{l}\mathcal{A}$ \cref{RealHomotopyGroupsInIntro} is the characteristic $L_\infty$-algebra of a SuGra theory, then the Bianchi identities of its dimensional reduction along a circle fiber are characterized by $\mathfrak{l}\mathrm{Cyc}\bracket({\mathcal{A}})$, 
where
\begin{equation}
  \mathrm{Cyc}\bracket({
    \mathcal{A}
  })
  :=
  \mathrm{Map}\bracket({
    S^1, 
    \mathcal{A}
  }) 
  \sslash
  S^1
\end{equation}
is the \emph{cyclic loop space} (the homotopy quotient of free loop space by rigid rotation of loops).

Generally, for dimensional reduction along the fibers of a $G$-principal bundle, the resulting Bianchi identities are characterized by $\mathfrak{l}\mathrm{Cyc}_G(\mathcal{A})$, where
\begin{equation}
  \mathrm{Cyc}_G\bracket({
    \mathcal{A}
  })
  :=
  \mathrm{Map}\bracket({
    G, 
    \mathcal{A}
  }) 
  \sslash
  G
  \mathrlap{\,.}
\end{equation}
Nevertheless, here we shall be concerned only with toroidal reductions.

For example, closed differential forms with coefficients in the cyclification of \cref{lS4InIntroduction} yield
\begin{equation}
\label{cycS4ClosedForms}
  \Omega^1_{\mathrm{cl}}
  \bracket({
    X;
    \mathfrak{l}
    \mathrm{Cyc}\bracket({S^4})
  })
  \simeq
  \left\{
  \begin{aligned}
    F_2 & 
    \in \Omega^2_{\mathrm{dR}}(X)
    \\
    F_4 & 
    \in \Omega^4_{\mathrm{dR}}(X)
    \\
    F_6 & 
    \in \Omega^6_{\mathrm{dR}}(X)
    \\
    H_3 & 
    \in \Omega^3_{\mathrm{dR}}(X)
    \\
    H_7 & 
    \in \Omega^7_{\mathrm{dR}}(X)
  \end{aligned}
  \;\middle\vert\;
  \begin{aligned}
    \mathrm{d}\, F_2
    & = 0
    \\
    \mathrm{d}\, F_4
    & = -  F_2 \wedge H_3
    \\
    \mathrm{d}\, F_6
    & = F_4 \wedge H_3
    \\
    \mathrm{d}\, H_3
    & = 0
    \\
    \mathrm{d}\, H_7
    & =
    \tfrac{1}{2}
    F_4 \wedge F_4
    +
    F_2 \wedge F_6
  \end{aligned}
  \right\}
  \mathrlap{\,,}
\end{equation}
which are the solutions to the (duality-symmetric) NS/RR-flux Bianchi identities of 10D $\mathcal{N} = (1,1)$ (type IIA) SuGra (cf. \cref{SomeReductionsOf11DSuGra}). 

Dimensionally reducing once more, or better, directly from 11D to 9D, yields closed differential forms with coefficients in the \emph{toroidification} of \cref{lS4InIntroduction} \parencites[p. 10]{SatiVoronov2025}[\S 2.4]{GSS25-TD}:
\begin{equation}\label{torS4ClosedForms}
  \Omega^1_{\mathrm{cl}}
  \bracket({
    X;
    \mathfrak{l}
    \mathrm{Tor}\bracket({S^4})
  })
  \simeq
  \left\{
  \begin{aligned}
     \dir{2}{F}_2 & \in \Omega^2_{\mathrm{dR}}(X)
    \\
    \dir{1}{F}_2 & \in \Omega^2_{\mathrm{dR}}(X)
    \\
    F_4 & \in \Omega^4_{\mathrm{dR}}(X)
    \\
    \dir{2}{H}_3 & \in \Omega^3_{\mathrm{dR}}(X)
    \\
    \dir{1}{H}_3 & \in \Omega^3_{\mathrm{dR}}(X)
    \\
    H_2 & \in \Omega^2_{\mathrm{dR}}(X)
    \\
    H_7 & \in \Omega^7_{\mathrm{dR}}(X)
    \\
    \dir{2}{F}_6 & \in \Omega^6_{\mathrm{dR}}(X)
    \\
    \dir{1}{F}_6 & \in \Omega^6_{\mathrm{dR}}(X)
    \\
    F_5 & \in \Omega^5_{\mathrm{dR}}(X)
  \end{aligned}
  \;\middle\vert\;
  \begin{aligned}
    \mathrm{d}\, \dir{1}{F}_2
    & = 0
    \\
    \mathrm{d}\, \dir{2}{F}_2
    & = 0
    \\
    \mathrm{d}\, F_4
    & = \dir{1}{F}_2 \wedge \dir{1}{H}_3 + \dir{2}{F}_2 \wedge \dir{2}{H}_3
    \\  
    \mathrm{d}\, \dir{2}{H}_3 
    & = -  \dir{1}{F}_2 \wedge H_2
    \\
     \mathrm{d}\, \dir{1}{H}_3 
    & =   \dir{2}{F}_2 \wedge H_2
    \\
    \mathrm{d}\, H_2
    & = 0
    \\
    \mathrm{d}\, H_7
    & = \tfrac{1}{2}F_4\wedge F_4 + \dir{1}{F}_2 \wedge \dir{1}{F}_6  + \dir{2}{F}_2 \wedge \dir{2}{F}_6 
    \\
    \mathrm{d}\, \dir{2}{F}_6
    & = -
    F_4 \wedge \dir{2}{H}_3
    -
    \dir{1}{F}_2 \wedge F_5
    \\
    \mathrm{d}\, \dir{1}{F}_6
    & = -
    F_4 \wedge \dir{1}{H}_3
    +
    \dir{2}{F}_2 \wedge F_5
    \\
    \mathrm{d}\, F_5
    & = \dir{2}{H}_3 \wedge \dir{1}{H}_3
    +
    F_4 \wedge H_2
  \end{aligned}
  \right\}
  \mathrlap{\,.}
\end{equation}
These are the solutions to the duality-symmetric flux Bianchis of 9D $\mathcal{N}=2$ SuGra, which serve as the basis for T-duality between 10D IIA and IIB supergravities. 

This pattern persists in a straightforward manner when reducing to the gauge sector of lower $(11-k)$-dimensional SuGras, albeit with a huge proliferation in the number of (higher) gauge fields, with the corresponding characteristic $L_\infty$-algebras being the rational \emph{$k$-toroidi-}\emph{fications} of $S^4$. 

\medskip 
In this fashion, superspace geometry together with the miracle properties of super-flux yields a transparent description of the equations of motion of SuGra theories descending from 11D SuGra. But perhaps most importantly, it paves the way to the global (IR) completion of these theories: \cref{Globalization}.

%%%%%%%%%%%%%%%%
\subsection
{Globalization}
\label
{Globalization}
%%%%%%%%%%%%%%%%

While the above discussion concerns a modernized streamlining of decade-old ideas and completion of their proofs, we now come to a crucial aspect that has previously received little attention: The global definition of the higher gauge fields appearing in higher-dimensional SuGra, and with it the global (IR) completion of the entire super-gravity theory.

\medskip

For the remainder of this section, we need to assume that the reader has at least a rough idea of higher groupoids --- we recall basics in \cref{Charges}. Further exposition of the following construction is in \cite{SS25-Flux,SS26-HigherGauge}.

%%%%%%%%%%%%%%%
\paragraph
{Global Nature of Higher Gauge Fields}
%%%%%%%%%%%%%%

Given the flux densities and their Bianchi identities in a higher-dimensional super-gravity theory, the task is to specify the nature of corresponding full field data (cf. \cref{GluingDataSchematics}):
\begin{enumerate}
\item
\emph{gauge potentials} $A$ on an open cover $\tilde X$ of spacetime by charts, 
\item
\emph{gauge transformations} $g$ where pairs of charts overlap, 
\item
\emph{gauge-of-gauge transformations} between these where triples of charts overlap,
\item and so on,
eventually subject to a condition which ``quantizes'' (discretizes) the \emph{topological charge} embodied by this data,
\end{enumerate}
and all compatible with the flux densities and their Bianchi identities.

\begin{figure}[htb]
\caption{%
\label{GluingDataSchematics}%
A higher gauge field is only locally (on charts of an open cover) given by gauge potentials. Its global structure furthermore involves gauge transformations between these where pairs of charts overlap, gauge-of-gauge transformations between those where triples of charts overlap, and so on.
}
\adjustbox{
  rndfbox=5pt
}{
\includegraphics
  [width=16cm]
  {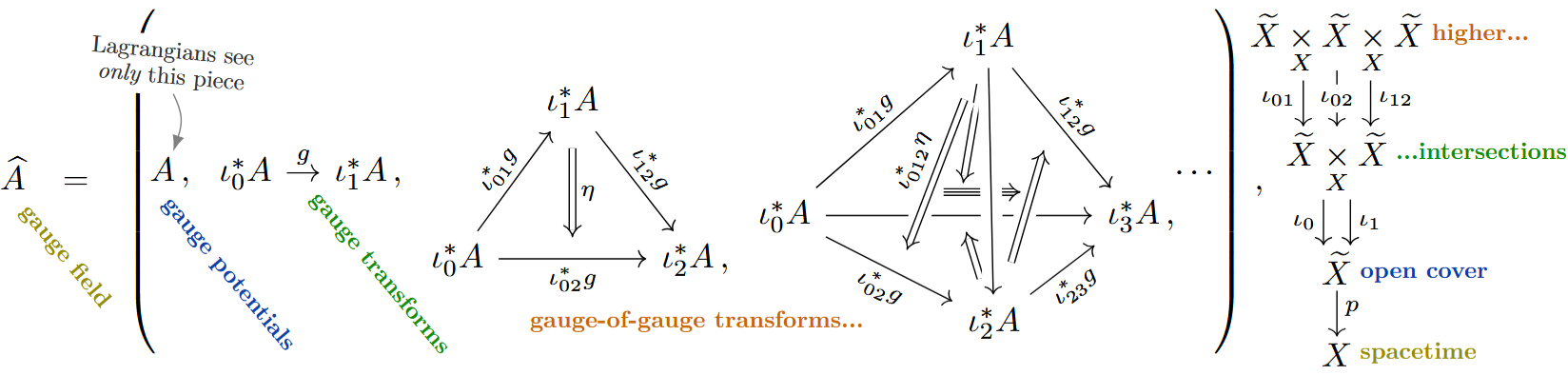}
}
\end{figure}

Such global completion of higher gauge fields is necessary to make global sense of (``cancel anomalies'' in) the Lagrangian mechanics of branes charged under these fields.

For example:
\begin{enumerate}
\item
Dirac's quantization of magnetic flux 
is motivated and justified by the fact (cf. \parencites{Alvarez1985}[\S 2]{Gawedzki1988}) that it gives consistent global meaning to the topological term 
  \begin{equation}
  \label
  {ChargedParticleTopologicalTerms}
    \text{``{}}
    \exp\bracket({
    2\pi
    \mathrm{i} \, \int_{\Sigma^{1,0}} \phi^\ast {A_1}
    })
    \text{{}''}
  \end{equation}
  in the exponentiated action functional for an electron worldline $\inlinetikzcd{\Sigma^{1,0} \ar[r, "{ \phi }"] \& X^{1,d}}$ propagating in the given magnetically fluxed background.

\item
In the same vein, one should ask:
\emph{What flux quantization condition on the combined C-field \textup{(locally: $C_3$, $C_6$)} in 11D SuGra and the tensor field \textup{(locally: $B_2$)} on an M5-brane} worldvolume
  $\inlinetikzcd{ \Sigma^{1,5} \ar[r, "{ \Phi }"] \& X^{1,10} }$ 
  ensures that the M5's exponentiated Hopf-WZW action
  \cite[(1)]{BandosEtAl1997}
  \begin{equation}
  \label
  {HopfWZTerm}
    \text{``{}}
    \exp\bracket({ 
    4\pi \mathrm{i}
    \,
    \int_{\Sigma^{1,5}} \bracket({ \Phi^\ast C_6 - \tfrac{1}{2} \mathrm{d}B_2 \wedge \Phi^\ast C_3 }) 
    })
    \text{{}''}
    \mathrlap{\,,}
  \end{equation} 
  \emph{is globally well-defined?} 

  This question was only answered in \cite{FSS21-Hopf}, using the method \cite{FSS23-Char,SS25-Flux} that we are concerned with here.

\end{enumerate}

%%%%%%%%%%%%%%
\paragraph
{Gauge Potentials as Concordance}
%%%%%%%%%%%%%%

To motivate how this works, we may easily observe (cf. \cite[Prop. 6.4]{FSS23-Char}) that an ordinary abelian higher gauge potential, namely a $p$-form potential $A_p$ on $\widetilde X$ for a flux density $F_{p+1}$ on $X$, is equivalently a \emph{concordance} $\widehat{F}$ of closed forms from $0$ to $F_{p+1}$, namely:
\begin{equation}
\label
{AbelianNullConcordance}
  \left.
  \begin{aligned}
  & A_p \in 
  \Omega^p_{\mathrm{dR}}\bracket({
    \widetilde X
  })
  \\
  & \text{s.t.}\;
  \mathrm{d}\, A_p
  =
  \iota^\ast F_{p+1}
  \end{aligned}
  \right\}
  \;\leftrightarrow\;
  \left\{
  \begin{aligned}
  &
  \widehat{F}
  \in
  \Omega^{p+1}_{\mathrm{cl}}
  \bracket({
    [0,1]
      \times
    \widetilde X
  })
  \\
  & 
  \text{s.t.}\;
  \left\{
  \begin{aligned}
    \widehat{F}_{\vert 0}
    & = 0
    \\
    \widehat{F}_{\vert 1}
    & = 
    \iota^\ast F_{p+1}
    \mathrlap{\,.}
  \end{aligned}
  \right.
  \end{aligned}
  \right.
\end{equation}
Here one side exists iff the other does, and the choices correspond to each other up to (finite) gauge transformation ($A_p \mapsto A_p + \mathrm{d} A_{p-1}$) and concordance-of-concordances, respectively.

But the concordance picture on the right of \cref{AbelianNullConcordance} immediately generalizes from ordinary closed differential forms to the closed $\mathfrak{a}$-valued forms \cref{IdeaOfCharacteristicLInfinityAlgebras}: 
\begin{equation}
\label
{CorrspndncBetweenGaugePotentialsAndConcordances}
  \left.
  \adjustbox{scale=1.3}{$
  \colorbox{lightgray}{$  \substack{
    \text{Duality-symmetric gauge potentials for}
    \\
    \text{flux densities $\vec F$ subject to $\mathfrak{a}$-Bianchis}
  }
  $}
  $}
  \right\}
  \;\leftrightarrow\;
  \left\{
  \begin{aligned}
    & 
    \widehat{F}
    \in
    \Omega^1_{\mathrm{cl}}\bracket({
      [0,1] \times 
      \widetilde{X};
      \mathfrak{a}
    })
    \\
    &
    \text{s.t.\;}
    \left\{
    \begin{aligned}
      \widehat{F}_{\vert 0}
      & 
      = 0
      \\
      \widehat{F}_{\vert 1}
      & =
      \iota^\ast \vec F
      \mathrlap{\,.}
    \end{aligned}
    \right.
  \end{aligned}
  \right.
\end{equation}
Analogously, gauge transformations between such gauge potentials correspond to concordances between such concordances (\cite[Def. 2.46]{GSS24-SuGra}), and so on.

One checks that these concordances correspond to the expected local expressions (\cref{Potentials}). For the case $\mathfrak{a} = \mathfrak{l}S^4$ \cref{lS4InIntroduction} of 11D SuGra we have  (\cref{ConcordancesEquivalentToGaugeTransf}, following \parencites[Prop. 2.48]{GSS24-SuGra}, cf. \cite[\S 3.1]{Banerjee2025-Potentials}):
\begin{equation}
  \begin{tikzcd}[
    row sep=-5pt, column sep=large 
  ]
 \colorbox{lightgray}{$ \substack{
    \text{gauge}
    \\
    \text{potentials}
  }
  $}
  \ar[
    rr,
    "{
     \colorbox{lightgray}{$ \substack{
        \text{gauge}
        \\
        \text{transformation}
      }$}
    }"{description, pos=.35}
  ]
  &&
\colorbox{lightgray}{$  \substack{
    \text{gauge}
    \\
    \text{potentials}
  }
  $}
  \\
    \bracket({
      \substack{
        C_3
        \\
        C_6
      }
    })
    \ar[
      dr,
      |->
    ]
    \ar[
      rr,
      |->,
      "{\color{darkblue}
        \scaledbracket({
        \substack{
          B_2
          \\
          B_5
        }
        })
      }"{description}
    ]
    &&
    \bracket({
      \begin{subarray}{l}
        C'_3
        \,=\,
        C_3 + \mathrm{d} B_2
        \\
        C'_6
        \,=\,
        C_6 
          + \mathrm{d} B_5
          + \sfrac{1}{2}\,C'_3 C_3
      \end{subarray}
    })
      \mathrlap{\,.}
    \ar[
      dl,
      |->
    ]
    \\[10pt]
    &
    \bracket({
      \begin{subarray}{l}
        G_4
        \,=\, \mathrm{d}\, C_3
        \\
        G_7
        \,=\,
        \mathrm{d}\, C_6 +
        \sfrac{1}{2}\, C_3 G_4
      \end{subarray}
    })
  \end{tikzcd}
\end{equation}

%%%%%%%%%%%%%%%
\paragraph
{Gauge Potentials as Homotopy}
%%%%%%%%%%%%%%%
To understand how to consistently glue such local data to a global gauge field as in \cref{GluingDataSchematics}, we observe that the system of higher $\mathfrak{a}$-concordances forms a higher smooth groupoid, denoted%
\footnote{
  The \emph{esh} symbol $\shape(-)$ stands for ``shape'' as in ``shape theory'' (referring to homotopy types of spaces), which in the phonetic alphabet reads: \textipa{/SeIp/}. 
}
\begin{equation}
\label
{ShapeOfSmoothSetOfClosedFormsInIntro}
  \shape 
  \mathbf{\Omega}^1_{\mathrm{cl}}
  ({
    \ast;
    \mathfrak{a}
  })
  \;\;\;
  \colorbox{lightgray}{$
  \substack{
    \text{whose $n$-morphisms are}
    \\
    \text{$n$-fold concordances of}
    \\
    \text{closed $\mathfrak{a}$-valued forms.}
  }
  $}
\end{equation}
This is the object that classifies concordances in that
\begin{enumerate}
  \item smooth maps $\inlinetikzcd{\vec F :  X \ar[r]\& \shape \mathbf{\Omega}^1_{\mathrm{cl}}(\ast; \mathfrak{a})}$
  correspond to 
  $\vec F \in \Omega^1_{\mathrm{cl}}
  \bracket({
    X;
    \mathfrak{a}
  })$;

  \item
  homotopies
  $
    \begin{tikzcd}
      X
      \ar[
        rr,
        bend left=20,
        "{ \vec F }"{
          description,
          name=source
        }
      ]
      \ar[
        rr,
        bend right=20,
        "{ \vec F' }"{
          description,
          name=target
        }
      ]
      \ar[
        from=source,
        to=target,
        Rightarrow, 
        "{\color{darkgreen} \, \widehat{F} }"
      ]
      &&
      \shape
      \mathbf{\Omega}^1_{\mathrm{cl}}
      (\ast; \mathfrak{a})
    \end{tikzcd}
  $
  correspond to concordances 
  $\widehat{F}$ between these;
  
  \item and so on. 
\end{enumerate}

Therefore, gauge potentials as in \cref{CorrspndncBetweenGaugePotentialsAndConcordances} correspond to homotopies of this form:
\begin{equation}
  \begin{tikzcd}
    &&
    \ast
    \ar[
      dd,
      "{ 
        0 
      }"
    ]
    \\
    \widetilde{X}
    \ar[
      d,
      "{ \iota }"
    ]
    \ar[
      dr,
      "{
        \iota^\ast \vec F
      }"{
        name=flux
      }
    ]
    \ar[
      urr,
      bend left=20,
      "{ 0 }"{
        swap,
        name=charge,
        pos=.35
      }    
    ]
    \ar[
      from=charge,
      to=flux,
      Rightarrow,
      "{\color{darkgreen} A }"
    ]
    \\[+5pt]
    X
    \ar[
      r,
      "{ \vec F }"
    ]
    &
    \mathbf{\Omega}^1_{\mathrm{cl}}
    ({
      \ast;
      \mathfrak{a}
    })
    \ar[
      r
    ]
    &
    \shape
    \mathbf{\Omega}^1_{\mathrm{cl}}
    ({
      \ast;
      \mathfrak{a}
    })\mathrlap{\,.}
  \end{tikzcd}
\end{equation}

%%%%%%%%%%%%%%
\paragraph
{Gluing Gauge Transformations}
%%%%%%%%%%%%%%

Next, on double intersections of charts, hence over the fiber product $\widetilde{X}^{\times^2_X} := \widetilde X \times_X \widetilde X$, these gauge potentials need to be glued by gauge transformations $g$ (still according to \cref{GluingDataSchematics}). In terms of the classifying object \cref{ShapeOfSmoothSetOfClosedFormsInIntro}, these correspond to homotopies-of-homotopies as shown on the left here:
\begin{equation}
  \begin{tikzcd}
    &&
    \ast
    \ar[
      dd,
      "{ 
        0 
      }"
    ]
    \\
    \widetilde{X}^{\times_X^2}
    \ar[
      d,
      "{ \iota }"
    ]
    \ar[
      dr,
      "{
        \iota^\ast \vec F
      }"{
        name=flux
      }
    ]
    \ar[
      urr,
      bend left=20,
      "{ 0 }"{
        swap,
        name=charge,
        pos=.35
      }    
    ]
    \ar[
      from=charge,
      to=flux,
      Rightarrow,
      bend right=40, 
      "{ \color{darkgreen}
        \iota_0^\ast A 
      }"{
        description,
        name=i0A
      }
    ]
    \ar[
      from=charge,
      to=flux,
      Rightarrow,
      bend left=80,
      "{ \color{darkgreen}
        \iota_1^\ast A 
      }"{
        description,
        name=i1A
      }
    ]
    \ar[
      from=i0A,
      to=i1A,
      Rightarrow,
      shorten=-2pt,
      "{\color{darkorange}
        g
      }"{
        pos=.3
      }
    ]
    \ar[
      from=i0A,
      to=i1A,
      -,
      shorten <=-2pt
    ]
    \\[+5pt]
    X
    \ar[
      r,
      "{ \vec F }"
    ]
    &
    \mathbf{\Omega}^1_{\mathrm{cl}}
    ({
      \ast;
      \mathfrak{a}
    })
    \ar[
      r
    ]
    &
    \shape
    \mathbf{\Omega}^1_{\mathrm{cl}}
    ({
      \ast;
      \mathfrak{a}
    })
    \mathrlap{\,.}
  \end{tikzcd}
\end{equation}
Moreover, these gauge transformations need to be compatible, in that transforming in two steps through a third overlapping chart is equivalent (concordant) to transforming directly between a pair of overlapping charts.

This may be expressed by understanding the double overlaps $\widetilde X^{\times^2_X}$ as morphisms in a smooth groupoid, the \emph{{\v C}ech 1-groupoid}
\begin{equation}
  C\bracket({
    \widetilde X
  })_1
  :=
  \Big(
  \begin{tikzcd}
    \widetilde X^{\times^2_X}
    \ar[
      rr,
      shift left=8pt,
      "{
        \iota_0
      }"{description}
    ]
    \ar[
      rr,
      <-,
      "{ \Delta }"{description}
    ]
    \ar[
      rr,
      shift right=8pt,
      "{
        \iota_1
      }"{description}
    ]
    &&
    \widetilde X
  \end{tikzcd}
  \Big)
\end{equation}
and then asking for a single homotopy but now over this {\v C}ech 1-groupoid:
\begin{equation}
  \begin{tikzcd}
    &[-10pt]
    &
    \ast
    \ar[
      dd,
      "{ 
        0 
      }"
    ]
    \\
    C\bracket({
      \widetilde{X}
    })_1
    \ar[
      d,
      "{ \iota }"
    ]
    \ar[
      dr,
      "{
        \iota^\ast \vec F
      }"{
        name=flux
      }
    ]
    \ar[
      urr,
      bend left=20,
      "{ 0 }"{
        swap,
        name=charge,
        pos=.35
      }    
    ]
    \ar[
      from=charge,
      to=flux,
      Rightarrow,
      "{ \color{darkgreen}
        (A,g) 
      }"{
        description,
      }
    ]
    \\[+2pt]
    X
    \ar[
      r,
      "{ \vec F }"
    ]
    &
    \mathbf{\Omega}^1_{\mathrm{cl}}
    ({
      \ast;
      \mathfrak{a}
    })
    \ar[
      r
    ]
    &
    \shape
    \mathbf{\Omega}^1_{\mathrm{cl}}
    ({
      \ast;
      \mathfrak{a}
    })
    \mathrlap{\,.}
  \end{tikzcd}
\end{equation}

In this formulation it is clear how to deal at once with the whole tower of higher gauge transformations over higher intersections of charts:  We simply ask for such a homotopy over the full {\v C}ech $\infty$-groupoid $C\bracket({\widetilde{X}})$. But this is \emph{equivalent} to $X$ itself (since it is a \emph{resolution} of $X$ and is a good resolution if $\widetilde X$ is a good open cover, cf. \parencites[\S4.1.7]{SS26-Orb}[Ex. 4.3.45]{SS26-Bun}):
\begin{equation}
  C\bracket({
    \widetilde X
  })
  :=
  \begin{tikzcd}
  \Bigg(
    \cdots
    \ar[
      rr,
      phantom,
      shift left=16pt,
      "{ \vdots }"{
        scale=.7,
        description
      }
    ]
    \ar[
      rr,
      shift left=8pt,
      dotted
    ]
    \ar[
      rr,
      <-,
      dotted
    ]
    \ar[
      rr,
      shift right=8pt,
      dotted
    ]
    \ar[
      rr,
      phantom,
      shift right=14pt,
      "{ \vdots }"{
        scale=.7,
        description
      }
    ]
    &&
    \widetilde X^{\times^3_X}
    \ar[
      rr,
      shift left=18pt,
      "{  \iota_{01} }"{description}
    ]
    \ar[
      rr,
      <-,
      shift left=9pt,
      "{ \Delta_{01} }"{description}
    ]
    \ar[
      rr,
      "{  \iota_{02} }"{description}
    ]
    \ar[
      rr,
      <-,
      shift right=9pt,
      "{ \Delta_{12} }"{description}
    ]
    \ar[
      rr,
      shift right=18pt,
      "{  \iota_{12} }"{description}
    ]
    &&
    \widetilde X^{\times^2_X}
    \ar[
      rr,
      shift left=8pt,
      "{
        \iota_0
      }"{description}
    ]
    \ar[
      rr,
      <-,
      "{ \Delta }"{description}
    ]
    \ar[
      rr,
      shift right=8pt,
      "{
        \iota_1
      }"{description}
    ]
    &&
    \widetilde X
  \Bigg)
  \ar[
    r,
    "{
      \iota
    }",
    "{ \sim }"{swap}
  ]
  &
  X
  \mathrlap{\,.}
  \end{tikzcd}
\end{equation}
This way we find that the gauge potentials (on any open cover $\widetilde X$, which we do not even need to prescribe anymore), together with their whole tower of higher transition data by higher gauge transformations over higher intersections, correspond simply to a homotopy in the following diagram of smooth $\infty$-groupoids:
\begin{equation}
\label
{GaugePotentialsWithVanishingTopologicalCharge}
  \begin{tikzcd}
    &&
    \ast
    \ar[
      dd,
      "{ 
        0 
      }"
    ]
    \\
    X
    \ar[
      dr,
      "{
        \vec F
      }"{
        name=flux
      },
      "{
        \text{\color{darkblue}fluxes}
      }"{swap, sloped}
    ]
    \ar[
      urr,
      bend left=20,
      "{ 0 }"{
        swap,
        name=charge,
        pos=.35
      }    
    ]
    \ar[
      from=charge,
      to=flux,
      Rightarrow,
      "{ \widehat{A} }"{swap},
      "{
        \substack{
        \text{\color{darkgreen}potentials}
        }
      }"{swap, sloped, yshift=-2pt}
    ]
    \\[+5pt]
    &
    \mathbf{\Omega}^1_{\mathrm{cl}}
    ({
      \ast;
      \mathfrak{a}
    })
    \ar[
      r
    ]
    &
    \shape
    \mathbf{\Omega}^1_{\mathrm{cl}}
    ({
      \ast;
      \mathfrak{a}
    })
  \end{tikzcd}
\end{equation}
Or almost: This covers the case of \emph{trivial topological charge}. But in this form it is now straightforward to also bring in this last piece of data.

%%%%%%%%%%%%%%%%%
\paragraph
{The Case of Ordinary Magnetic Charge}
%%%%%%%%%%%%%%%%%

In the example of ordinary magnetic flux, where $\mathfrak{a} = \mathfrak{l}B^2 \mathbb{Z}$, the data encoded in \cref{GaugePotentialsWithVanishingTopologicalCharge} is equivalently of this form:
\begin{equation}
  \begin{aligned}
    A \in \Omega^1_{\mathrm{dR}}
    \bracket({
      \widetilde X
    })
    &
    \text{ s.t. }
    \mathrm{d}\, A
    =
    \iota^\ast F_2
    \\
    \lambda \in 
    \Omega^0_{\mathrm{dR}}
    \bracket({
      \widetilde X^{\times^2_X}
    })
    &
    \text{ s.t. }
    \mathrm{d}\lambda
    =
    \iota_1^\ast A
    -
    \iota_0^\ast A
    \\
    &
    \phantom{\text{ s.t. }}
    \;0
    =
    \iota_{12}^\ast \lambda
    -
    \iota_{02}^\ast \lambda
    +
    \iota_{01}^\ast \lambda
    \mathrlap{\,.}
  \end{aligned}
\end{equation}
But the general electromagnetic field contains more data, as follows (a \emph{{\v C}ech-Deligne 2-cocycle}, cf. \parencites{Alvarez1985}[\S 2]{Gawedzki1988} and \cref{0CocycleInCechDeligne2Cohomology} below): 
\begin{equation}
\label
{CechDeligne2CocyclesInIntro}
  \begin{aligned}
    A \in \Omega^1_{\mathrm{dR}}
    \bracket({
      \widetilde X
    })
    &
    \text{ s.t. }
    \mathrm{d}\, A
    =
    \iota^\ast F_2
    \\
    \lambda \in 
    \Omega^0_{\mathrm{dR}}
    \bracket({
      \widetilde X^{\times^2_X}
    })
    &
    \text{ s.t. }
    \mathrm{d}\lambda
    =
    \iota_1^\ast A
    -
    \iota_0^\ast A
    \\
    \mathcolor{purple}{
    n \in 
    C^\infty\bracket({
      \widetilde X^{\times_X^3};
      \mathbb{Z}
    })
    }
    &
    \text{ s.t. }
    \;
    \mathcolor{purple}{n}
    =
    \iota_{12}^\ast \lambda
    -
    \iota_{02}^\ast \lambda
    +
    \iota_{01}^\ast \lambda
    \\
    &
    \phantom{\text{ s.t. }}
    \;
    \mathcolor{purple}{
    0 =
    \iota_{123}^\ast n
    -
    \iota_{023}^\ast n
    +
    \iota_{013}^\ast n
    -
    \iota_{012}^\ast n
    }
    \mathrlap{\,.}
  \end{aligned}
\end{equation}
Here the cocycle condition for the real-valued gauge transformations $\lambda$ needs to hold only up to a locally constant integer $n$ (this makes these $\mathbb{R}$-valued functions represent the expected $\mathrm{U}(1)$-valued transition functions) and that integer function itself satisfies a cocycle condition on quadruple overlaps, which makes it a representative of a class in integral cohomology
\begin{equation}
  [n]
  \in
  H^2\bracket({
    X; \mathbb{Z}
  })
  \simeq
  \pi_0\, \mathrm{Map}\bracket({
    X,
    B^2 \mathbb{Z}
  })
  \mathrlap{\,,}
\end{equation}
this being the ``quantized'' (discretized) magnetic charge reflected by the electromagnetic field.

In the usual 2-groupoid model for $B^2 \mathbb{Z}$ there is an evident map
\begin{equation}
  \inlinetikzcd{
    B^2 \mathbb{Z}
    \ar[
      rrr,
      "{
        \mathbf{ch}^{B^2 \mathbb{Z}}
      }"
    ]
    \&\&\&
    \shape
    \mathbf{\Omega}^1_{\mathrm{cl}}
    \bracket({
      \ast;
      \mathfrak{l}
      B^2 \mathbb{Z}
    })
  }
\end{equation}
which sends the 2-morphism labeled $n$ to a differential 2-form on the 2-simplex integrating to $n$. This map models the insertion seen in \cref{CechDeligne2CocyclesInIntro}, making this data be equivalent to a homotopy in the following diagram of groupoids:
\begin{equation}
\label
{MagneticGaugePotentialsAsHomotopy}
  \begin{tikzcd}[
    column sep=20pt
  ]
    &&
    \mathcolor{purple}{B^2 \mathbb{Z}}
    \ar[
      dd,
      "{ 
        \mathcolor{purple}
        {\mathbf{ch}^{B^2 \mathbb{Z}}}
      }"
    ]
    \\
    X
    \ar[
      dr,
      "{
        \vec F
      }"{
        name=flux
      },
      "{
        \text{\color{darkblue}fluxes}
      }"{swap, sloped}
    ]
    \ar[
      urr,
      bend left=20,
      "{ 
        \mathcolor{purple}{\chi} 
      }"{
        swap,
        name=charge,
        pos=.35
      },
      "{
        \text{\color{purple}charges}
      }"{
        sloped,
        pos=.35
      }
    ]
    \ar[
      from=charge,
      to=flux,
      Rightarrow,
      "{ \widehat{A} }"{swap},
      "{
        \substack{
        \text{\color{darkgreen}potentials}
        }
      }"{swap, sloped, yshift=-2pt}
    ]
    \\[+5pt]
    &
    \mathbf{\Omega}^1_{\mathrm{cl}}
    ({
      \ast;
      \mathfrak{l}B^2\mathbb{Z}
    })
    \ar[
      r
    ]
    &
    \shape
    \mathbf{\Omega}^1_{\mathrm{cl}}
    ({
      \ast;
      \mathfrak{l}B^2\mathbb{Z}
    })
    \mathrlap{\,.}
  \end{tikzcd}
\end{equation}

%%%%%%%%%%%%%%
\paragraph
{General Flux Quantization}
%%%%%%%%%%%%%

From this example, the general case becomes clear: We need --- and there canonically exists (\cref{Charges}) --- for every (connected, nilpotent, of rational finite type)
classifying space $\mathcal{A}$ a map
\begin{equation}
\label
{CharacterMapInIntroduction}
  \begin{tikzcd}
    \mathcal{A}
    \ar[
      r,
      "{ \mathbf{ch}^{\mathcal{A}} }"
    ]
    &
    \shape
    \mathbf{\Omega}^1_{\mathrm{cl}}
    ({
      \ast;
      \mathfrak{l}\mathcal{A}
    })
  \end{tikzcd}
\end{equation}
which embeds $\mathcal{A}$-cocycle data into concordance data.
(Concretely, this map is essentially what, in rational homotopy theory, is the \emph{rationalization unit} on $\mathcal{A}$.)

Thereby, a choice of classifying space $\mathcal{A}$ is \emph{admissible} for duality-symmetric Bianchi identities characterized by $\mathfrak{a}$ if
\begin{equation}
\label
{AdmissibilityConditionInIntro}
  \mathfrak{l}\mathcal{A}
  \simeq
  \mathfrak{a}
  \mathrlap{\,,}
\end{equation}
in which case it defines globally completed gauge fields as the data embodied by dashed cones in the following diagram, encoding cocycles in \emph{differential nonabelian $\mathcal{A}$-cohomology} (\cite[\S9]{FSS23-Char} discussed in \cref{Charges}):
\begin{equation}
\label
{FullFieldDiagramInIntro}
  \begin{tikzcd}[column sep=30pt]
    &&
    \mathcal{A}
    \ar[
      dd,
      "{ 
        \mathbf{ch}^{\mathcal{A}}
      }"
    ]
    \\
    X
    \ar[
      dr,
      dashed,
      "{
        \vec F
      }"{
        name=flux
      },
      "{
        \text{\color{darkblue}fluxes}
      }"{swap, sloped}
    ]
    \ar[
      urr,
      dashed,
      bend left=20,
      "{ 
        \chi 
      }"{
        swap,
        name=charge,
        pos=.35
      },
      "{
        \text{\color{purple}charges}
      }"{
        sloped,
        pos=.35
      }
    ]
    \ar[
      from=charge,
      to=flux,
      Rightarrow,
      dashed,
      "{ 
        \widehat{A} 
      }"{swap},
      "{
        \text{\color{darkgreen}potential}
      }"{swap,sloped, yshift=-2pt}
    ]
    \\[+5pt]
    &
    \mathbf{\Omega}^1_{\mathrm{cl}}
    ({
      \ast;
      \mathfrak{a}
    })
    \ar[
      r
    ]
    &
    \shape
    \mathbf{\Omega}^1_{\mathrm{cl}}
    ({
      \ast;
      \mathfrak{a}
    })
    \mathrlap{\,.}
  \end{tikzcd}
\end{equation}

In the case that $\mathcal{A} \simeq E_n$ is a stage in a \emph{spectrum of spaces} $E_\bullet$, this construction reduces \cite[Ex. 9.1]{FSS23-Char} to exhibiting the gauge fields as cocycles in Whitehead-generalized \emph{differential abelian cohomology} $E^n_{\mathrm{dff}}$ (\cite[\S 4]{HopkinsSinger2005}, cf. \cite{Bunke2012}), such as differential K-theory (cf. \parencites{Freed2002}{Szabo2013}{GS22-KTheory}). But beware that these abelian quantization laws are admissible \cref{AdmissibilityConditionInIntro} only for linear Bianchi identities (cf. \cite[Ex. 5.6]{FSS23-Char}), hence not for examples like \cref{CFieldEoM} of interest here.
\vspace{2mm} 
\begin{equation}
  \adjustbox{scale=1.3}{$
  \colorbox{lightgray}{$
  \substack{
    \text{Nonlinear}
    \\[2pt]
    \text{Bianchis}
  }
  $}
  $}
  \;\;\Rightarrow\;\;
  \adjustbox{scale=1.3}{$
  \colorbox{lightgray}{$
  \substack{
    \text{Whitehead-generalized}
    \\[1pt]
    \text{differential abelian cohomology}
    \\[1pt]
    \text{\emph{not} admissible as flux quantization.}
  }
  $}
  $}
\end{equation}
Indeed, the popular idea that 10D IIA SuGra should be flux quantized in K-theory (cf. \parencites[\S 3.7]{Freed2001}[\S 4]{SS25-Flux}) tacitly ignores the non-linear Bianchi identity of the $H_7$-flux (cf. \parencites{BaSS26-UnstableK}{GiotopoulosSati2026}{BMSS2019}).

This is the crucial advance provided by the above construction \cref{FullFieldDiagramInIntro}, that it generalizes flux quantization to flux with non-linear Bianchis as they actually appear in higher super-gravity.

%%%%%%%%%%%%%
\paragraph
{IR Completed 11D super-gravity}
%%%%%%%%%%%%%

And this construction works perfectly well in the generality that $X$ is a supermanifold. Therefore it combines with the miracle \cref{TheMiracleEquivalence} to the following remarkable statement:

\begin{standout}
  The global IR completions of 11D SuGra are parameterized by classifying spaces $\mathcal{A}$ of the rational homotopy type of the 4-sphere,
  $
    \mathfrak{l}\mathcal{A}
    \simeq
    \mathfrak{l}S^4
  $.
\end{standout}
\begin{standout}
  Given such a choice of \emph{C-field flux quantization}, we have:
\begin{equation}
\hspace{-4mm}
  \Bigg\{
  \adjustbox{scale=1.3}{$
  \substack{
    \text{Global solutions of 11D SuGra}
    \\[2pt]
    \text{with C-field fluxes $(G_4,G_7)$}
    \\[1pt]
    \text{quantized in $\mathcal{A}$-cohomology }
  }
  $}
  \Bigg\}
  \,
  \leftrightarrow
  \,
  \left\{
  \begin{aligned}
  &
  X^{1,10\vert \mathbf{32}} 
  \in 
  \left\{
  \substack{
    \text{super-torsion-free}
    \\
    \text{super-spacetimes}
  }
  \right\},
  \\
  &
  \begin{tikzcd}[column sep=35pt, 
    ampersand replacement=\&
  ]
    \&\&[-8pt]
    \mathcal{A}
    \ar[
      dd,
      "{ 
        \mathbf{ch}^{\mathcal{A}}
      }"
    ]
    \\
    X^{\mathrlap{1,10\vert \mathbf{32}}} \;\;\;
    \ar[
      dr,
      dashed,
      "{
        (G^s_4, G^s_7)
      }"{
        sloped, pos=.6, 
        name=flux
      },
      "{
        \text{\color{darkblue}fluxes}
      }"{swap, sloped}
    ]
    \ar[
      urr,
      dashed,
      bend left=20,
      "{ 
        \chi 
      }"{
        swap,
        name=charge,
        pos=.35
      },
      "{
        \text{\color{purple}M-brane\phantom{\clap{g}}}
      }"{
        sloped,
        pos=.23
      },
      "{
        \text{\color{purple}charges}
      }"{
        sloped,
        pos=.45
      }
    ]
    \ar[
      from=charge,
      to=flux,
      Rightarrow,
      dashed,
      "{ 
        \widehat{C} 
      }"{swap},
      "{
        \text{\color{darkgreen}C-field}
      }"{swap,sloped, yshift=-2pt}
    ]
    \\[+5pt]
    \&
    \mathbf{\Omega}^1_{\mathrm{cl}}
    ({
      \ast;
      \mathfrak{l}S^4
    })
    \ar[
      r
    ]
    \&
    \shape
    \mathbf{\Omega}^1_{\mathrm{cl}}
    ({
      \ast;
      \mathfrak{l}S^4
    })
    \mathrlap{\,.}
  \end{tikzcd}
  \end{aligned}
  \right.
\end{equation}  
\end{standout}

This baseline understanding of global completion of 11D super-gravity by C-field flux quantization generalizes to various kinds of further structure, notably to the presence of probe M5-branes \cite{GSS24-M5,FSS21-StrStruc} and to orbifold spacetimes \cite{SS20-Tad,SS25-Seifert}.

%%%%%%%%%%%%%%%
\paragraph
{Brane Charges in Nonabelian Cohomology}
%%%%%%%%%%%%%%

With the global completion of SuGra in hand, we have access to the \emph{topological charges} of the theory, which for 11D SuGra are the charges of M-branes.

One finds \parencites[\S2.1]{SS25-Complete}[(189)]{SS26-HigherGauge} that these constitute the \emph{nonabelian cohomology} (\parencites[Def. 6.0.6]{Toen2002}[Def. 6]{Lurie2014}[\S 2]{FSS23-Char}, cf. \parencites[\S 1]{SS25-TEC}[\S 4]{SS26-Orb}) of spacetime with coefficients in $\Omega \mathcal{A}$:
\begin{equation}
  H^1\bracket({
    X; 
    \Omega
    \mathcal{A}
  })
  := 
  \pi_0\, \mathrm{Map}\bracket({
    X, 
    \mathcal{A}
  })
  \mathrlap{\,.}
\end{equation}

In particular this means that the charges of singular (``black'') $p$-branes (those whose singular worldvolume is deleted from spacetime, cf. \cite[\S 2.2]{SS25-Flux}) are in the $\Omega\mathcal{A}$-cohomology of the $d-p-1$-spheres surrounding them and hence
fall into the homotopy group of $\mathcal{A}$ in degree $d-p-1$ (cf. \cite{SS23-Mf}), in that (for simply connected $\mathcal{A}$):
\begin{equation}
  \begin{aligned}
  H^1\bracket({
    \mathbb{R}^{1,d}
    \setminus
    \mathbb{R}^{1,p};
    \Omega \mathcal{A}
  })
  & \defneq
  \pi_0\, \mathrm{Map}\bracket({
    \mathbb{R}^{1,d}
    \setminus
    \mathbb{R}^{1,p},
    \mathcal{A}
  })
  \\
  & \simeq
  \pi_0\, \mathrm{Map}\bracket({
    S^{d-p-1},
    \mathcal{A}
  })
  \\
  &
  \simeq 
  \pi_{d-p-1}\bracket({
    \mathcal{A}
  })
  \mathrlap{\,.}
  \end{aligned}
\end{equation}
Due to the admissibility condition \cref{AdmissibilityConditionInIntro} and by basic facts of rational homotopy theory, these charge groups may have integer subgroups $\mathbb{Z}$ --- signifying non-fractional brane species --- precisely if $d-p-1$ equals the degree of one of the flux densities $F_{d-p-1}$ in the theory. 

This is what is expected in the physics literature: That non-fractional (non-torsion) $p$-brane charges are witnessed by a nontrivial integral of the appropriate flux density over any $(d-p-1)$-brane surrounding them. 

In particular, Dirac's original argument for quantization of the charge of magnetic monopoles (the 0-branes of electromagnetism) here comes down to:
\begin{equation}
  \mathllap{
  \substack{
    \text{\color{gray}Dirac monopole}
    \\
    \text{\color{gray}charges}
  }}
  \;\;
  H^1\bracket({
    \mathbb{R}^{1,3}
    \setminus
    \mathbb{R}^{1,0};
    \Omega 
    B^2 \mathbb{Z}
  })
  \simeq
  \pi_0\, 
  \mathrm{Map}\bracket({
    S^2, B^2 \mathbb{Z}
  })
  \simeq
  \pi_2\bracket({B^2\mathbb{Z}})
  \simeq
  \mathbb{Z}
  \mathrlap{\,.}
\end{equation}

%%%%%%%%%%%%%
\paragraph
{Global Completions of 11D SuGra}
%%%%%%%%%%%%%
Of the (infinitely) many admissible choices for the global completion of 11D SuGra, we here briefly mention the following three (further discussed in \cref{Applications}):

For the $\mathfrak{l}S^4$-Bianchi identity \cref{11DSuGraFluxesAslS4ClosedFormsInIntro}, admissible flux quantization laws $\mathcal{A}$ \cref{AdmissibilityConditionInIntro}, $\mathfrak{l}\mathcal{A} \simeq \mathfrak{l}S^4$, include:  
\begin{enumerate}

\item
 \colorbox{lightgray}{
the homotopy fiber space of the rationalized cup square on $B^4 \mathbb{Z}$}:
\begin{equation}
\label
{TheBaselineChoice}
  \mathcal{A}
  \defneq
  \mathrm{hfib}\Big(
    \inlinetikzcd{
      B^4 \mathbb{Z}
      \ar[
        rr,
        "{
          (-)^{\cup_2}  
        }"
      ]
      \&\&
      B^8 \mathbb{Z}
      \ar[r]
      \&
      B^8 \mathbb{Q}
    }
  \Big)
  \mathrlap{\,.}
\end{equation}

This choice  completes the $G_4$-flux to a differential integral 4-cocycle (a ``$\mathrm{U}(1)$-bundle 2-gerbe connection'', as direct higher generalization of how $F_2$ is quantized in \cref{CechDeligne2CocyclesInIntro} by a $\mathrm{U}(1)$-bundle connection)  and essentially imposes the non-linear second Bianchi identity \cref{CFieldEoM} only on $G_7$ itself, but not on its integral charge structure.

Essentially, this choice is tacitly the C-field flux quantization traditionally considered in the literature  
\cite{AschieriJurco2004,DFM2007,FSS15-ModuliStack}.

Of course, with $G_7$ remaining essentially unquantized with this choice, it cannot imply the (Page) charge quantization of M2-branes.
In any case, there are other choices one may consider:

\item
 \colorbox{lightgray}{the 4-sphere itself}:
\begin{equation}
\label
{ChoiceOfS4}
  \mathcal{A}
  \defneq
  S^4
  \mathrlap{\,.}
\end{equation}
This choice, originally proposed in \cite[\S 2.5]{Sati2018}, is the \emph{minimal} choice in number of cells (hence it is ``universal'' in that its charges map into the charges for all other choices). 

In \cite{FSS20-H,FSS21-Hopf,SS23-Mf} it was shown that this choice implies a whole list of subtle topological conditions that are expected in M-theory, notably it implies:
\begin{enumerate}
\item
\cite[Prop. 3.13]{FSS20-H}: the $\tfrac{1}{4}p_1$-shifted integral quantization of $G_4$ argued in \cite{Witten1997Flux},

\item
\cite[Thm. 4.8]{FSS21-Hopf}: global consistency of the M5's topological term \cref{HopfWZTerm}.

\end{enumerate}

Closely related to the second point is the more immediate statement that $\mathcal{A} \defneq S^4$ quantizes not just the charge of black M5-branes but also of black M2-branes, because,
\begin{equation}
\label
{CohomotopicalMBraneChargesInIntro}
  \begin{aligned}
  \pi_0\, \mathrm{Map}\bracket({
    \mathbb{R}^{1,10}
    \setminus
    \mathbb{R}^{1,5},
    S^4
  })
  & 
  \simeq
  \pi_4\bracket({S^4})
  \simeq
  \mathbb{Z}
  \,,
  \\
  \pi_0\, \mathrm{Map}\bracket({
    \mathbb{R}^{1,10}
    \setminus
    \mathbb{R}^{1,2},
    S^4
  })
  & 
  \simeq
  \pi_7\bracket({S^4})
  \simeq
  \mathbb{Z}
  \mathcolor{gray}{
    \oplus \mathbb{Z}_{/12}
  }
  \mathrlap{\,.}
  \end{aligned}
\end{equation}

These results support the hypothesis that \cref{ChoiceOfS4} is the ``right'' choice of global completion for the purpose of ``M-theory'' and as such the choice \cref{ChoiceOfS4} has been called \emph{Hypothesis H} (\cite{FSS20-H}, since it implies that M-brane charges are in the nonabelian cohomology called \emph{co-Homotopy} \cite[\S VII]{STHu59}\cite[Ex. 2.7]{FSS23-Char}, hence ``H-cohomology'' analogous to the more familiar K-theoretic ``K-cohomology''). 

This is the choice we will be concerned with below. But there are still other choices of interest:

\item
 \colorbox{lightgray}{ the homotopy fiber of the squared second $\mathrm{SU}(2)$-Chern class}:
\begin{equation}
  \mathcal{A}
  \defneq
  \mathrm{hfib}\Big(
    \inlinetikzcd{
      B \mathrm{SU}(2)
      \ar[
        r,
        "{ c_2 }"
      ]
      \&
      B^4 \mathbb{Z}
      \ar[
        rr,
        "{
          (-)^{\cup_2}
        }"
      ]
      \&\&
      B^8 \mathbb{Z}
    }
  \Big)
  \mathrlap{\,.}
\end{equation}
Since $S^4 \simeq \mathbb{H}P^1$ and $B \mathrm{SU}(2) \simeq \mathbb{H}P^\infty$ (quaternionic projective spaces), one may understand this choice as a kind of hybrid of the previous two choices \cref{ChoiceOfS4,TheBaselineChoice}.

Its noteworthy property is \cite{BaSS26-UnstableK} that under dimensional reduction to 10D, it implies a flux quantization of type IIA super-gravity which is a form of twisted unstable K-theory. 

The topological charges predicted by this choice are actually indistinguishable from those predicted by Hypothesis H on spacetime compactifications of the form $X^{1,10} \simeq \mathbb{R}^{1,3} \times X^7$, since the canonical comparison map
\begin{equation}
  \begin{tikzcd}[
    column sep=-3pt
  ]
    S^4 
    \ar[r]
    &[25pt]
    \mathrm{hfib}
    \Big(
      B \mathrm{SU}(2)
      \ar[
        r,
        "{ c_2 }"
      ]
      &[15pt]
      B^4 \mathbb{Z}
      \ar[
        r,
        "{ (-)^{\cup_2} }"
      ]
      &[30pt]
      B^8 \mathbb{Z}
    \Big)
  \end{tikzcd}
\end{equation}
is a 7-equivalence \cite[Thm. 2.6]{BaSS26-UnstableK}. In particular, this choice predicts the same M-brane charges as in \cref{CohomotopicalMBraneChargesInIntro}.
\end{enumerate}

%%%%%%%%%%%%%
\paragraph
{Conclusions}
%%%%%%%%%%%%
\begin{enumerate}
\item
Global (infrared) completion of (higher) gauge fields is part of what it means to actually define a field theory, including its topological (brane) charges. 

\item 
For theories like higher-dimensional super-gravity with its non-linear Bianchi identities (in duality-symmetric form) such completion requires flux/charge quantization in differential \emph{nonabelian} cohomology going beyond differential abelian (stable) cohomology theories like ordinary cohomology and K-theory.

\item The choice of global completion/flux quantization has direct implications on physical observables, notably it determines the fractional brane species through the torsion charges these carry. 

\item
Hence different global completions of the same equations of motion yield distinct field theories. If one takes the physics to be described as fixed (such as: ``M-theory'') then a choice of global completion of its field content is a \emph{hypothesis} about its correct theoretical description.

\item
This IR-resolved theory space is waiting to be explored. Notably it remains to be determined which IR-completion of 11D SuGra is the actual low-energy limit of ``M-theory''.

\item Once a choice of global completion is made in higher dimensions (notably in 11D SuGra), this implies compatible global completions of the dimensionally reduced descendants (at least for reductions along abelian principal fibers).

\end{enumerate}

%%%%%%%%%%%%

%%%%%%%%%%%%%%
%%%%%%%%%%%%%%%
\section
{Superspace}
\label
{Superspace}
%%%%%%%%%%%%%%%

We begin the technical part of these notes by surveying a treatment of super-geometry that is both mathematically rigorous as well as practical, and which seamlessly lends itself to the higher-geometric generalization needed to discuss flux quantization in \cref{Completions}.

The impatient reader may want to skip ahead to the discussion of super-gravity in \cref{Supergravity} and come back here only as need arises.

%%%%%%%%%%%%
\paragraph
{Outline}
%%%%%%%%%%%
In the time-honored spirit of Grothendieck's \emph{functorial geometry} \cite{nLab:FunctorialGeometry} we proceed to:
\begin{enumerate}
\item discuss simple building-block superspaces (\cref{Charts}), 
\item
bootstrap from these a good notion of generalized super-spaces (\cref{Spaces}), 

among them the smooth \emph{super-manifolds} (Ex. \ref{SmoothSupermanifoldsAsSmoothSuperSet}),
\item 
which become super-spacetimes when equipped with suitable super-coframes (\cref{Frames}). 
\end{enumerate}

%%%%%%%%%%%%%%%%
\paragraph
{Terminology}
%%%%%%%%%%%%%%%
Beware that the qualifier ``super'' carries different connotations in the math and the physics literature.

In mathematics, \emph{super-algebra} refers generally to algebra internal to the symmetric monoidal category of super-vector spaces with its $\mathbb{Z}_{/2}$-graded commutativity. Dually, \emph{super-geometry} in mathematics is broadly concerned with generalized spaces whose algebras of functions are super-algebras. In particular, the phase spaces of all field theories with fermions are super-spaces in this sense (cf. \cite{BerezinMarinov1977,Schmitt1997a,Schmitt1997b}), irrespective of whether these theories exhibit super-symmetry. Indeed,  parts of the mathematical literature are concerned with super-geometry without even considering fermions. (In mathematical generality, a \emph{super-group} is any group internal to super-manifolds or super-schemes; cf. \parencites[\S 7.1]{Varadarajan2004}[\S 3.1]{GSS26-Hidden}). This generality of \emph{super-geometry} we consider here first. 

In contrast, in physics, \emph{super-symmetry} refers much more narrowly to super-geometric extensions specifically of Poincar{\'e} groups and to field theories with such \emph{super Poincar{\'e} groups} as (global or local) symmetry groups. We come to this further below in \cref{11DSuGra}, cf. \cref{11DSuperMinkowski}.

%%%%%%%%%%%%%%
\paragraph
{Literature}
%%%%%%%%%%%%%%
Early influential discussion of super-geometry (super-manifolds) includes \parencites{BerezinLeites1975}{Batchelor1979}{DeWitt1984}{Manin1985}{Berezin1987}.
Comprehensive review includes \parencites{DeligneMorgan1999}{Varadarajan2004}{Rogers2007}{CarmeliEtAl2011}{HohnholdStolzTeichner2011}.
The powerful functorial perspective that we use was first articulated by \parencites{Schwarz1984}{Molotkov1984}[\S A]{KonechnySchwarz1998}, further discussed in \parencites{sachse2008}{BalduzziEtAl2010}.
We follow \parencites[\S4.6]{Sc13-dcct}{Sc18-Durham}{JurcoEtAl2019}{Schreiber2025}[\S2]{GSS24-SuGra}[\S 9.1.3]{GSS26-Hidden}{SS26-Orb}.

The super-Poincar{\'e} Lie algebra originates with \cite{GolfandLikhtman1972} (reprinted in \cite[pp. 44--53]{Shifman2000}). Early review of super(-Poincar{\'e}) symmetry includes \parencites{FayetFerrara1977}.
Mathematical discussion of super-Poincar{\'e} groups includes \parencites[\S 1.1]{DeligneFreed1999}[\S 5]{Freed1999}[\S 7.5]{Varadarajan2004}[\S 6]{Freed2006}.

%%%%%%%%%%%%%%%%
\subsection
{Charts}
\label
{Charts}
%%%%%%%%%%%%%%%%

%%%%%%%%%%%%%%
\paragraph
{Differential geometry}
%%%%%%%%%%%%%%
Ordinary differential geometry is, of course, modeled on the Cartesian spaces $\mathbb{R}^n$ with smooth maps between them. We denote this category by
\begin{equation}
\label
{CartesianSpaces}
  \mathrm{CrtSp}
  :=
  \big\{
  \inlinetikzcd{
    \mathbb{R}^n
    \ar[rrr, "{ \text{smooth} }"]
    \&\&\&
    \mathbb{R}^{n'}
  }
  \big\}
  \mathrlap{\,.}
\end{equation}

Remarkably, the functor that sends these to the formal duals \cref{AffineSchemes} of their commutative $\mathbb{R}$-algebras of smooth functions is fully faithful \cref{FullSubcategoryEmbedding} --- a statement sometimes known as ``Milnor's exercise'', cf. \cite{nLab:MilnorExercise}:
\begin{equation}
\label
{MilnorExercise}
  \begin{tikzcd}[
  row sep=-4pt,  column sep=10pt
  ]
    \mathrm{CrtSp}
    \ar[
      rr,
      hook
    ]
    &&
    \mathrm{CAlg}^{\mathrm{op}}_{_{\mathbb{R}}}
    \\
    \mathbb{R}^n
    \ar[
      d,
      "{ f }"
    ]
    &\longmapsto&
    C^\infty\bracket({\mathbb{R}^n})
    \\[20pt]
    \mathbb{R}^{n'}
    &\longmapsto&
    C^\infty\bracket({\mathbb{R}^{n'}})
    \ar[
      u,
      "{ f^\ast }"'
    ]
  \end{tikzcd}
  \;\;\;
  \text{and generally:}
  \;\;
  \begin{tikzcd}[
 row sep=-4pt,  column sep=10pt
  ]
    \mathrm{SmthMfd}
    \ar[
      rr,
      hook
    ]
    &&
    \mathrm{CAlg}^{\mathrm{op}}_{_{\mathbb{R}}}
    \\
    X
    \ar[
      d,
      "{ f }"
    ]
    &\longmapsto&
    C^\infty\bracket({X})
    \\[20pt]
    X'
    &\longmapsto&
    C^\infty\bracket({X'})
    \mathrlap{\,.}
    \ar[
      u,
      "{ f^\ast }"'
    ]
  \end{tikzcd}
\end{equation}

This provides a dual algebraic incarnation \cref{AffineSchemes} of smooth manifolds in which form that may be generalized.

%%%%%%%%%%%%%%
\paragraph
{Infinitesimal Geometry}
%%%%%%%%%%%%%
For instance, to make sense of the \emph{$k$-th order infinitesimal neighbourhood of the origin in $\mathbb{R}^n$},
\begin{equation}
\label
{InfinitesimalDisk}
  \mathbb{D}^n_k
  \subset 
  \mathbb{R}^n
  \mathrlap{\,,}
\end{equation}
we observe that its algebra of smooth functions ought to be
\begin{equation}
  \begin{aligned}
  C^\infty\bracket({
    \mathbb{D}^n_k
  })
  &
  \defneq
  C^\infty\bracket({
    \mathbb{R}^n
  })/
  \bracket({
    (x^1, \cdots, x^n)^{k+1}
  })
  \\
  &
  \simeq
  \mathbb{R}\bracket[{
    \epsilon^1,
    \cdots,
    \epsilon^n
  }]/
  \bracket({
    (\epsilon^1, \cdots, \epsilon^n)^{k+1}
  }) 
  \mathrlap{\,,}
  \end{aligned}
\end{equation}
and then understand the \emph{infinitesimal halo} $\mathbb{D}^n_k$ as its formal dual, hence as the corresponding object in the opposite category:
\begin{equation}
  \begin{tikzcd}[
    row sep=-3pt, column sep=0pt
  ]
    \mathrm{Hl}
    \ar[
      rr,
      hook
    ]
    &&
    \mathrm{CAlg}^{\mathrm{op}}_{_{\mathbb{R}}}
    \\
    \mathbb{D}^n_k
    &\longmapsto&
    C^\infty\bracket({
      \mathbb{D}^n_k
    })
    \mathrlap{\,.}
  \end{tikzcd}
\end{equation}
More generally, we may take \emph{infinitesimally thickened Cartesian spaces} $\mathbb{R}^n \times \mathbb{D}^{n'}_k$ to be the formal duals of the tensor products of these algebras:
\begin{equation}
  \begin{tikzcd}[row sep=-4pt, 
    column sep=0pt
  ]
    \mathrm{CrtSp}
    \ar[
      rr,
      hook
    ]
    &&
    \mathrm{HldCrtSp}
    \ar[
      rr,
      hook
    ]
    &&
    \mathrm{CAlg}^{\mathrm{op}}_{_{\mathbb{R}}}
    \\
    \mathbb{R}^n
      &\longmapsto&
    \mathbb{R}^n
    \phantom{
      \times
      \mathbb{D}^{n'}_k
    }
    \\
    &&
    \mathbb{R}^n 
    \times
    \mathbb{D}^{n'}_k
    &\longmapsto&
    C^\infty\bracket({
      \mathbb{R}^n
    })
    \otimes_{_{\mathbb{R}}}
    C^\infty\bracket({
      \mathbb{D}^{n'}_k
    }).
  \end{tikzcd}
\end{equation}

%%%%%%%%%%%%%%%%%
\paragraph
{Supergeometry}
%%%%%%%%%%%%%%%%%
Yet more generally we want to describe would-be spaces whose infinitesimal coordinate functions skew-commute among themselves. 

Write 
\begin{equation}
  \begin{tikzcd}
    \mathrm{CAlg}_{_{\mathbb{R}}}
    \ar[
      r,
      hook
    ]
    &
    \mathrm{sCAlg}_{_{\mathbb{R}}}
  \end{tikzcd}
\end{equation}
for the further full inclusion into the category of super-commutative $\mathbb{R}$-algebras, 
hence of $\mathbb{Z}_{/2} \simeq \{\mathrm{evn}, \mathrm{odd}\}$-graded algebras whose elements of homogeneous degree $n \in \mathbb{Z}_{/2}$ satisfy the \emph{super-com\-mu\-ta\-tiv\-i\-ty condition}:
\begin{equation}
  a \cdot b
  =
  (-1)^{n_a n_b}
  \,
  b \cdot a
  \mathrlap{\,.}
\end{equation}

Archetypical examples are the \emph{Grassmann algebras}
\begin{equation}
\label
{GrassmannAlgebra}
  C^\infty\bracket({
    \mathbb{R}^{0\vert q}
  })
  :=
  \wedge^\bullet_{\mathbb{R}}
  \bracket({
    \mathbb{R}^q
  })^\ast
  \in 
  \mathrm{sCAlg}_{_{\mathbb{R}}}
\end{equation}
which are those free on $q$ generators $\bracket({\theta_i})_{i=1}^q$ in odd degree, hence satisfying 
\begin{equation}
  \theta_i \cdot \theta_j 
    =
  - \theta_j \cdot \theta_i
  \mathrlap{\,.}
\end{equation}
This implies in particular that $(\theta_i)^2 = 0$, hence that these Grassmann algebras are akin to smooth functions on the halo $\mathbb{D}^q_1$ \cref{InfinitesimalDisk}, except that they also skew-commute among each other.

As before, we may take a \emph{super-Cartesian space} 
\begin{equation}
  \mathbb{R}^{n\vert q}
  \defneq
  \mathbb{R}^n
  \times
  \mathbb{R}^{0\vert q}
\end{equation}
to be the formal dual to its intended algebra of functions:
\begin{equation}
\label
{SuperCartesianSpaces}
  \begin{tikzcd}[row sep=-4pt, 
  column   sep=0pt
  ]
    \mathrm{sCrtSp}
    \ar[
      rr,
      hook
    ]
    &&
    \mathrm{sCAlg}_{_{\mathbb{R}}}^{\mathrm{op}}
    \\
    \mathbb{R}^{n\vert q}
    &\longmapsto&
    C^\infty\bracket({
      \mathbb{R}^n
    })
    \otimes_{_{\mathbb{R}}}
    C^\infty\bracket({
      \mathbb{R}^{0\vert q}
    })
    \mathrlap{\,.}
  \end{tikzcd}
\end{equation}

%%%%%%%%%%%%%%%%%
\subsection
{Spaces}
\label
{Spaces}
%%%%%%%%%%%%%%%%%

The notion of (super) \emph{smooth sets} that we turn to now originates as such in \cite{Sc13-dcct}. Further development as a foundation for field theory is in \parencites{GS25-FieldsI,GS25-FieldsII}; further exposition includes \cite{Giotopoulos2025, Schreiber2025,IbortMas2025}.

%%%%%%%%%%%%%%
\paragraph
{Smooth Sets}
%%%%%%%%%%%%%%
A general smooth space $\mathbf{X}$ modeled on the Cartesian spaces \cref{CartesianSpaces} should be characterized by the sets $\mathrm{Plt}\bracket({ \mathbb{R}^n, \mathbf{X}})$ of smooth maps (\emph{plots}) $\inlinetikzcd{ \mathbb{R}^n \ar[r] \& \mathbf{X} }$  from Cartesian \emph{probe} spaces into it, subject to the conditions that: 
\begin{enumerate}
\item
plots may consistently be precomposed with smooth maps of Cartesian spaces,
\item
the space is already determined by germs of plots.
\end{enumerate}
In mathematical jargon, the first condition says that $\mathrm{Plt}\bracket({-,\mathbf{X}})$ is a \emph{presheaf} on $\mathrm{CrtSp}$ \cref{CategoryOfPresheaves} and the second condition says that it is in fact a \emph{sheaf} \cref{CategoryOfSheaves}.

Since its system $\mathrm{Plt}\bracket({-;\mathbf{X}})$ of plots is meant to fully characterize $\mathbf{X}$ we say that the category of such \emph{smooth sets} is the sheaf topos over $\mathrm{CrtSp}$:
\begin{equation}
\label
{CategoryOfSmoothSets}
\begin{tikzcd}[row sep=-3pt,
  column sep=0pt
]
  \mathrm{SmthSet}
  & := &
  \mathrm{Sh}\bracket({\mathrm{CrtSp}})
  \\
  \mathbf{X}
  &\leftrightarrow&
  \mathrm{Plt}\bracket({
    -,
    \mathbf{X}
  })
  \mathrlap{\,.}
\end{tikzcd}
\end{equation}

\begin{example}
\label[example]
{SmoothManifoldsAsSmoothSets}
A smooth manifold becomes a smooth set via the ordinary smooth functions into it:
\begin{equation}
  \begin{tikzcd}[row sep=-2pt, 
  column sep=0pt
  ]
    \mathrm{SmthMfd}
    \ar[
      rr,
      hook
    ]
    &&
    \mathrm{SmthSet}
    \\
    X 
      &\longmapsto&
    C^\infty(-,X)
    \mathrlap{\,.}
  \end{tikzcd}
\end{equation}
Recall here that by \cref{MilnorExercise} we have, for our smooth manifold $X$:
\begin{equation}
\label
{PlotsOfSmoothManifold}
  \mathrm{Plt}\bracket({
    \mathbb{R}^n,
    X
  })
  \defneq
  C^\infty\bracket({
    \mathbb{R}^n,
    X
  })
  \simeq
  \mathrm{Hom}\bracket({
    C^\infty(X),
    C^\infty\bracket({\mathbb{R}^n})
  })
  \mathrlap{\,.}
\end{equation}
\end{example}

%%%%%%%%%%%%%
\paragraph
{Smooth Supersets}
%%%%%%%%%%%%

In this fashion it is clear that supergeometric smooth sets should be characterized 
via their sheaves of plots by super-Cartesian probe spaces:
\begin{equation}
\label
{SmoothSuperSets}
  \mathrm{sSmthSet}
  :=
  \mathrm{Sh}\bracket({
    \mathrm{sCrtSp}
  })
\end{equation}

\begin{example}
\label[example]
{SmoothSupermanifoldsAsSmoothSuperSet}
A \emph{smooth super-manifold} is a smooth superset $\mathbf{X}$ \cref{SmoothSuperSets} for which there exists a smooth vector bundle $\inlinetikzcd{E \ar[r] \& X}$ (in smooth manifolds) such that, in generalization of \cref{PlotsOfSmoothManifold}:
\begin{equation}
  \mathrm{Plt}\bracket({
    \mathbb{R}^{n\vert q}
    ,
    \mathbf{X}
  })
  \simeq
  \mathrm{Hom}
  \bracket({
    \wedge^\bullet_{C^\infty(X)} 
    \Gamma\bracket({
      E^\ast
    }),
    C^\infty\bracket({
      \mathbb{R}^{n \vert q}
    })
  })
  \mathrlap{\,.}
\end{equation}
Namely, this is equivalent to any other definition of smooth super-manifolds by Batchelor's theorem (\parencites{Batchelor1979}[\S1.13]{Batchelor1984}, cf. \parencites[\S8.2]{Rogers2007}) and using \cref{MilnorExercise}.
\end{example}

\begin{example}
[Super Minkowski Spacetime]
\label[example]
{SuperMinkowskiSpacetime}
  Any super Minkowski spacetime $\FR^{1,d\vert \mathbf{N}}$ \eqref{SuperMinkowskiFunctionAlgebra} is, in particular, a smooth super-manifold corresponding to the trivial smooth vector bundle $\mathbf{N}^\ast_{\FR^{1,d}} :=  \FR^{1,d} \times \mathbf{N}^\ast \rightarrow \FR^{1,d}$, so that 
  $$   
\wedge^\bullet_{C^\infty(\FR^{1,d})} 
    \Gamma\bracket({
      \mathbf{N}^\ast_{\FR^{1,d}}
  }) \simeq C^\infty\bracket({
    \mathbb{R}^{1,d}
  })
  \otimes_{_{\mathbb{R}}}
  \wedge^\bullet \bracket({
    \mathbf{N}^\ast
  }) \defneq C^\infty (\FR^{1,d\vert\mathbf{N}})  
  $$
  hence also being a smooth superset via
\begin{equation}
  \mathrm{Plt}\bracket({
    \mathbb{R}^{n\vert q}
    ,
    \mathbf{X}
  })
  \simeq
  \mathrm{Hom}
  \bracket({
    C^\infty(\FR^{1,d\vert\mathbf{N}}) ,
    C^\infty\bracket({
      \mathbb{R}^{n \vert q}
    })
  })
  \mathrlap{\,.}
\end{equation}

More generally, for spinor bundles $\inlinetikzcd{E \ar[r]\& X}$ with typical fiber $\mathbf{N}$, we denote the corresponding super-manifolds \cref{SmoothSupermanifoldsAsSmoothSuperSet} generally by $X^{1,d\vert \mathbf{N}}$. We say more about these in \cref{Coframes}.
\end{example}

\begin{example}
  For $\mathbf{X} \times \mathbf{Y} \in \mathrm{sSmthSet}$, their \emph{Cartesian product} has as plots pairs of plots into the factor space:
  $$
    \begin{aligned}
      \mathbf{X} \times \mathbf{Y} 
      & \in \mathrm{sSmthSet}
      \\
      \mathrm{Plt}\bracket({
        \mathbb{R}^{n \vert q}
        ,
        \mathbf{X}
        \times
        \mathbf{Y}
      })
      & =
      \mathrm{Plt}\bracket({
        \mathbb{R}^{n \vert q}
        ,
        \mathbf{X}
      })
      \times
      \mathrm{Plt}\bracket({
        \mathbb{R}^{n \vert q}
        ,
        \mathbf{Y}
      })
    \end{aligned}
  $$
\end{example}

\begin{example}
  For $\mathbf{X}, \mathbf{Y} \in \mathrm{sSmthSet}$, their \emph{mapping space} is
  \begin{equation}
  \label
  {MappingSpace}
    \begin{aligned}
      \mathbf{Map}\bracket({
        \mathbf{X},
        \mathbf{Y}
      })
      & 
      \in 
      \mathrm{sSmthSet}
      \\
      \mathrm{Plt}\bracket({
        \mathbb{R}^{n\vert q}
        ,
        \mathbf{Map}\bracket({
          \mathbf{X},
          \mathbf{Y}
        })
      })
      & 
      :=
      \mathrm{Hom}\bracket({
        \mathbb{R}^{n\vert q}
        \times
        \mathbf{X},
        \mathbf{Y}
      })
      \mathrlap{\,.}
    \end{aligned}
  \end{equation}
\end{example}

%%%%%%%%%%%%%%%%
\subsubsection
{Cohesion}
%%%%%%%%%%%%%%%

The canonical inclusion of Cartesian spaces \cref{CartesianSpaces} among super-Cartesian spaces \cref{SuperCartesianSpaces} evidently has a right adjoint \cref{AdjointFunctors}
\begin{equation}
  \left.
  \begin{tikzcd}[
    sep=0pt
  ]
    \mathrm{sCrtSp}
    \ar[
      rr,
      phantom,
      "{ \bot }"{scale=.7}
    ]
    \ar[
      rr,
      shift right=5pt,
      "{
        \mathrm{bos}
      }"'
    ]
    &&
    \mathrm{CrtSp}
    \ar[
      ll,
      hook',
      shift right=5pt
    ]
    \\
    \mathbb{R}^{n\vert q}
    &\mapsto&
    \mathbb{R}^n
  \end{tikzcd}
  \right\}
  \;\;\Leftrightarrow\;\;
  \left\{
  \begin{aligned}
  &
  \mathrm{Hom}\bracket({
    \mathbb{R}^{n'\vert 0},
    \mathbb{R}^{n\vert q}
  })
  \\
  \underset{\mathrm{ntrl}}{\simeq}
  \;
  &
  \mathrm{Hom}\bracket({
    \mathbb{R}^{n'},
    \mathbb{R}^n
  })
  \mathrlap{\,.}
  \end{aligned}
  \right.
\end{equation}
By left Kan extension, this induces a corresponding operation on smooth super-sets \cref{SmoothSuperSets}:
\begin{equation}
  \begin{tikzcd}[
    sep=15pt
  ]
    \mathrm{sSmthSet}
    \ar[
      rr,
      phantom,
      "{ \bot }"{scale=.7}
    ]
    \ar[
      rr,
      shift right=5pt,
      "{
        \mathrm{bos}
      }"'
    ]
    &&
    \mathrm{SmthSet}
    \ar[
      ll,
      hook',
      shift right=5pt
    ]
  \end{tikzcd}
  \,,\;\;\;
  \begin{tikzcd}
    \mathrm{sSmthSet}
    \ar[
      rr,
      uphordown,
      "{ 
        \text{$%
          \mathbf{X} 
            \mapsto 
          \bosonic{\mathbf{X}}%
        $} 
      }"{description}
    ]
    \ar[
      r,
      "{ \mathrm{bos} }"
    ]
    &
    \mathrm{SmthSet}
    \ar[
      r,
      hook
    ]
    &
    \mathrm{sSmthSet}
    \mathrlap{\,,}
  \end{tikzcd}
\end{equation}
equipped with a natural transformation, being the adjunction counit \cref{AdjunctionUnit}:
\begin{equation}
  \inlinetikzcd{
    \BosonicCounit
    :
    \bosonic{\mathbf{X}}
    \ar[r]
    \&
    \mathbf{X}
    \mathrlap{\,.}
  }
\end{equation}

For example, on a super-manifold $X^{1,d\vert \mathbf{N}}$ (\cref{SuperMinkowskiSpacetime}) this is the inclusion of its \emph{bosonic body}
\begin{equation}
\label
{IncludingBosonicBodyOfSupermanifold}
  \inlinetikzcd{
  X^{1,d}
  \ar[
    rr,
    "{
      \BosonicCounit
    }"
  ] 
  \&\&
  X^{1,d\vert \mathbf{N}}
  }
  \mathrlap{\,.}
\end{equation}

%%%%%%%%%%%%%%
\subsection
{Coframes}
\label
{Frames}
%%%%%%%%%%%%%

Supergravity is naturally formulated in \emph{first-order formalism for gravity} (cf. \cite{nLab:FirstOrderFormulationOfGravity}) where not the metric tensor directly but a \emph{coframe field} (aka ``vielbein'') and \emph{spin connection} represent the field of gravity, locally represented by a differential 1-form with coefficients in the (super) Poincar{\'e} Lie algebra.

%%%%%%%%%%%%%%%
\subsubsection
{Forms}
%%%%%%%%%%%%%%%

%%%%%%%%%%%%%
\paragraph
{Ordinary Differential Forms}
%%%%%%%%%%%%

To warm up before considering super-differential forms:

Write $\mathrm{dgCAlg}_{_{\mathbb{R}}}$ for \emph{differential graded-commutative algebras} over the real numbers. This means that these are $\mathbb{Z}$-graded algebras such that for elements of homogeneous degree $n \in \mathbb{Z}$ their product satisfies
\begin{equation}
  a \cdot b
  =
  (-1)^{n_a n_b}
  \,
  b \cdot a
  \mathrlap{\,,}
\end{equation}
and equipped with a differential $\mathrm{d}$ of degree +1,
\begin{equation}
  n_{\mathrm{d}a}
  =
  n_a + 1
  \mathrlap{\,,}
\end{equation}
which thus satisfies the graded Leibniz rule:
\begin{equation}
  \mathrm{d}(a \cdot b)
  \,=\,
  (\mathrm{d}a) \cdot b
  \,+\,
  (-1)^{n_b}\, 
  a \cdot \mathrm{d}b
\end{equation}
A basic example is the algebra of smooth differential forms on a Cartesian space,
which (emphasizing this basic fact for the following generalization) we may understand as 
\begin{equation}
  \Omega^1_{\mathrm{dR}}\bracket({
    \mathbb{R}^n
  })
  \simeq
    C^\infty\bracket({
      \mathbb{R}^n
    })\bracket[{
      \mathrm{d}x^1,
      \cdots,
      \mathrm{d}x^n
    }]
  \mathrlap{\,,}
\end{equation}
with generators $\mathrm{d}x^i$ in degree 1, and 
with the differential defined on generators by
\begin{equation}
  \begin{tikzcd}[row sep=-2pt, column sep=0pt]
    \Omega^\bullet_{\mathrm{dR}}
    \bracket({
      \mathbb{R}^n
    })
    \ar[
      rr,
      "{ \mathrm{d} }"
    ]
    &&
    \Omega^\bullet_{\mathrm{dR}}
    \bracket({
      \mathbb{R}^n
    })
    \\
    \mathllap{
    C^\infty\bracket({\mathbb{R}^n})
    \ni
    }
    f
    &\longmapsto&
    \sum_{i =1}^n
    \frac{\partial f}{\partial x^i}
    \mathrm{d}x^i
    \\
    \mathrm{d}x^i
    &\longmapsto&
    0
    \mathrlap{\,.}
  \end{tikzcd}
\end{equation}
With respect to pullback along smooth maps $\inlinetikzcd{\mathbb{R}^n \ar[r]\& \mathbb{R}^m}$, differential forms evidently form a sheaf, which as such is a smooth set \cref{CategoryOfSmoothSets} (in fact a smooth dgc-algebra), the \emph{smooth moduli set} of differential forms:
\begin{equation}
  \begin{aligned}
    \mathbf{\Omega}^\bullet_{\mathrm{dR}}(\ast)
    &
    \in
    \mathrm{dgCAlg}\bracket({\mathrm{SmthSet}})
    \\
    \mathrm{Plt}\bracket({
      \mathbb{R}^n,
      \mathbf{\Omega}^\bullet_{\mathrm{dR}}(\ast)
    })
    & :=
    \Omega^\bullet_{\mathrm{dR}}\bracket({
      \mathbb{R}^n
    })
    \mathrlap{\,.}
  \end{aligned}
\end{equation}

From this we obtain the definition of differential forms on any smooth set 
$\mathbf{X} \in \mathrm{SmthSet}$:
\begin{equation}
  \begin{aligned}
  \Omega^\bullet_{\mathrm{dR}}\bracket({
    \mathbf{X}
  })
  & 
  :=
  \mathrm{Hom}\bracket({
    \mathbf{X},
    \mathbf{\Omega}^\bullet_{\mathrm{dR}}(\ast)
  })
  \;
  \in
  \mathrm{dgCAlg}\bracket({
    \mathrm{Set}
  })
  \\
  \mathbf{\Omega}^\bullet_{\mathrm{dR}}
  \bracket({
    \mathbf{X}
  })
  &
  :=
  \mathbf{Map}\bracket({
    \mathbf{X},
    \mathbf{\Omega}^\bullet_{\mathrm{dR}}({
      \ast
    })
  })
  \;
  \in
  \mathrm{dgCAlg}\bracket({
    \mathrm{SmthSet}
  })
  \\
  \mathrm{Plt}\bracket({
    \mathbb{R}^n,
    \mathbf{\Omega}^\bullet_{\mathrm{dR}}\bracket({
      \mathbf{X}
    })
  })
  & \simeq
  \Omega^\bullet_{\mathrm{dR}}\bracket({
    \mathbb{R}^n
      \times 
    X 
  })
  \mathrlap{\,.}
  \end{aligned}
\end{equation}
It is a good exercise to verify that for $\mathbf{X} = X$ a smooth manifold (\cref{SmoothManifoldsAsSmoothSets}) this reduces to the traditional definition of differential forms on $X$.

%%%%%%%%%%%%%%%
\paragraph
{Super Differential Forms}
%%%%%%%%%%%%%%%

Write $\mathrm{dgcCAlg}$ for \emph{differential-graded super-commutative algebras}. These are $(\mathbb{Z} \times \mathbb{Z}_{/2})$-graded algebras, such that for elements of homogeneous bidegree $(n,\sigma) \in \mathbb{Z} \times \mathbb{Z}_{/2}$ their product satisfies the following \emph{homological super-sign rule}%
\footnote{
  Beware that there is also another sign convention used by some authors in super-mathematics, see \cite{nLab:SignInSupergeometry} for more discussion.
}
\begin{equation}
\label
{TheSuperHomologicalSignRule}
  a \cdot b
  =
  (-1)^{
    n_a n_b
    +
    \sigma_a \sigma_b
  }
  \,
  b \cdot a
  \mathrlap{\,,}
\end{equation}
and equipped with a differential of bidegree $(1,\mathrm{evn})$
\begin{equation}
  \begin{aligned}
    n_{\mathrm{d}a}
    &=
    n_a + 1
    \mathrlap{\,,}
    \\
    \sigma_{\mathrm{d}a}
    & =
    \sigma_a
    \mathrlap{\,,}
  \end{aligned}
\end{equation}
thus satisfying the graded Leibniz rule:
\begin{equation}
  \mathrm{d}(a \cdot b)
  \,=\,
  \bracket({\mathrm{d}a})
  \cdot b
  \,+\,
  (-1)^{n_a}
  \,
  a 
    \cdot 
  \bracket({\mathrm{d}b})
  \mathrlap{\,.}
\end{equation}

A basic example is the dgc-algebra of differential super-forms on a super-Cartesian space \cref{SuperCartesianSpaces}: This is 
\begin{equation}
  \Omega^\bullet_{\mathrm{dR}}
  \bracket({
    \mathbb{R}^{n\vert q}
  })
    \simeq
    C^\infty\bracket({
      \mathbb{R}^{k \vert q}
    })\bracket[{
      \substack{
        \mathrm{d}x^1,
        \cdots,
        \mathrm{d}x^k,
        \\
        \mathrm{d}\theta^1,
        \cdots,
        \mathrm{d}\theta^q
      }
    }]
\end{equation}
with generators in bi-degrees
\begin{equation}
\label
{BidegreeIfGeneratingSuperForms}
  \begin{alignedat}{2}
    n_{\mathrm{d}x^i}
    & = 
    1
    \mathrlap{\,,}
    &\;\;\;
    \sigma_{\mathrm{d}x^i}
    & =
    \mathrm{evn}
    \mathrlap{\,,}
    \\
    n_{\mathrm{d}\theta^i}
    & = 
    1
    \mathrlap{\,,}
    &
    \mathrlap{\,,}
    \sigma_{\mathrm{d}\theta^i}
    & =
    \mathrm{odd}
    \mathrlap{\,.}
    \mathrlap{\,.}
  \end{alignedat}
\end{equation}
and with differential defined on these generators as
\begin{equation}
  \begin{tikzcd}[
    row sep=-2pt, column sep=5pt
  ]
    \Omega^\bullet_{\mathrm{dR}}
    \bracket({
      \mathbb{R}^{n\vert q}
    })
    \ar[
      rr,
      "{ \mathrm{d} }"
    ]
    &&
    \Omega^\bullet_{\mathrm{dR}}
    \bracket({
      \mathbb{R}^{n\vert q}
    })
    \\
    \mathllap{
    C^\infty\bracket({
      \mathbb{R}^{n\vert q}
    })
    \ni
    \;
    }
    f 
      &\longmapsto& 
    \begin{aligned}
    &
    \textstyle{\sum_{i=1}^n}
    \tfrac{\partial f}{\partial x^i}
    \,
    \mathrm{d}x^i
    \\
    +
    &
    \textstyle{\sum_{\iota=1}^q}
    \tfrac{\partial f}{\partial \theta^\iota}
    \,
    \mathrm{d}\theta^\iota 
    \end{aligned}
    \\
    \mathrm{d}x^i &\longmapsto& 0
    \\
    \mathrm{d}\theta^\iota &\longmapsto& 0
    \mathrlap{\,.}
  \end{tikzcd}
\end{equation}
By \cref{TheSuperHomologicalSignRule,BidegreeIfGeneratingSuperForms} these generators have the following commutation relations:
\begin{equation}
\label
{CommutationRelationsOfBasicSuper1Forms}
  \begin{aligned}
    x^i 
    \wedge 
    \mathrm{d}x^j
    &
    =
    +
    \mathrm{d}x^j 
    \wedge 
    x^i
    \\
    \theta^i 
    \wedge 
    \mathrm{d}\theta^j
    &
    =
    -
    \mathrm{d}\theta^j 
    \wedge 
    \theta^i   
    \\
    x^i 
    \wedge 
    \mathrm{d}\theta^j
    &
    =
    +
    \mathrm{d}\theta^j 
    \wedge 
    x^i       
  \end{aligned}
  \;\;\;\;\;\;
  \begin{aligned}
    \mathrm{d}x^i 
    \wedge 
    \mathrm{d}x^j
    &
    =
    -
    \mathrm{d}x^j 
    \wedge 
    \mathrm{d}x^i
    \\
    \mathrm{d}\theta^i 
    \wedge 
    \mathrm{d}\theta^j
    &
    =
    \mathcolor{purple}{+}
    \mathrm{d}\theta^j 
    \wedge 
    \mathrm{d}\theta^i   
    \\
    \mathrm{d}x^i 
    \wedge 
    \mathrm{d}\theta^j
    &
    =
    -
    \mathrm{d}\theta^j 
    \wedge 
    \mathrm{d}x^i     .  
  \end{aligned}
\end{equation}
The fact that odd super-forms like $\mathrm{d}\theta^i$ commute with each other is of central importance for super-space super-gravity in \cref{Equations}.

These super-differential forms pull back along smooth maps $\inlinetikzcd{\mathbb{R}^{n\vert q} \ar[r] \& \mathbb{R}^{n'\vert q'}}$ and as such form a smooth super set \cref{SmoothSuperSets}, the \emph{smooth moduli super set} of super-differential forms:
\begin{equation}
  \begin{aligned}
    \mathbf{\Omega}^\bullet_{\mathrm{dR}}(\ast)
    &
    \in
    \mathrm{dgsCAlg}\bracket({\mathrm{sSmthSet}})
    \\
    \mathrm{Plt}\bracket({
      \mathbb{R}^{k\vert q},
      \mathbf{\Omega}^\bullet_{\mathrm{dR}}(\ast)
    })
    & :=
    \Omega^\bullet_{\mathrm{dR}}\bracket({
      \mathbb{R}^{k\vert q}
    })
  \end{aligned}
\end{equation}
Generally for $\mathbf{X} \in \mathrm{sSmthSet}$:
\begin{equation}
  \begin{aligned}
    \Omega^\bullet_{\mathrm{dR}}
    \bracket({
      \mathbf{X}
    })
    &
    :=
    \mathrm{Hom}\bracket({
      \mathbf{X},
      \mathbf{\Omega}^\bullet_{\mathrm{dR}}
      \bracket({
        \mathbf{X}
      })
    })
    \;
    \in
    \mathrm{dgsCAlg}\bracket({
      \mathrm{Set}
    })
    \\
    \mathbf{\Omega}^\bullet_{\mathrm{dR}}
    \bracket({
      \mathbf{X}
    })
    &
    :=
    \mathbf{Map}\bracket({
      \mathbf{X},
      \mathbf{\Omega}^\bullet_{\mathrm{dR}}
      \bracket({
        \mathbf{X}
      })
    })
    \; \in
    \mathrm{dgsCAlg}\bracket({
      \mathrm{sSmthSet}
    })
    \\
    \mathrm{Plt}\bracket({
      \mathbb{R}^{n\vert q},
      \mathbf{\Omega}^\bullet_{\mathrm{dR}}
      \bracket({
        \mathbf{X}
      })
    })
    &
    \simeq
    \Omega^\bullet_{\mathrm{dR}}\bracket({
      \mathbb{R}^{n\vert q} 
      \times
      \mathbf{X}
    })
    \mathrlap{\,.}
  \end{aligned}
\end{equation}

%%%%%%%%%%%%%%%
\subsubsection
{Coframes}
\label
{Coframes}
%%%%%%%%%%%%%

%%%%%%%%%%%%%%%%%%
\paragraph
{Super-symmetry}
%%%%%%%%%%%%%%%%%%

A flavor of super-symmetry is determined by a real spinor representation
\begin{equation}
  \mathbf{N}
  \in
  \mathrm{Rep}_{_{\mathbb{R}}}
  \bracket({
    \mathrm{Spin}(1,d)
  })
\end{equation}
equipped with a spin-equivariant bilinear map
\begin{equation}
  \begin{tikzcd}[row sep=-3pt, 
   column sep=0pt
  ]
    \mathbf{N} 
      \otimes_{_{\mathbb{R}}}
    \mathbf{N}
    \ar[
      rr
    ]
    &&
    \mathbb{R}^{1,d}
    \\
    \bracket({
      \Phi_1, \Phi_2
    })
    &\longmapsto&
    \bracket({
      \overline{\Phi}_1
      \,\Gamma^a\,
      \Phi_2
    })_{a = 0}^{d}
    \mathrlap{\,.}
  \end{tikzcd}
\end{equation}

%%%%%%%%%%%%%
\paragraph
{Coframe Fields}
%%%%%%%%%%%%%
A \emph{coframe field} on an ordinary manifold is a reduction of its structure group along $\inlinetikzcd{O(1,d) \ar[r, hook] \& \mathrm{GL}\bracket({ \mathbb{R}^{1,d} })}$. On a supermanifold it is a reduction of its structure group along $\inlinetikzcd{\mathrm{Spin}(1,d) \ar[r, hook] \& \mathrm{GL}\bracket({ \mathbb{R}^{1,d\vert N} }) }$.
This means equivalently that: 
\begin{definition}
\label[definition]
{SuperCoframeField}
A \emph{super-coframe field} in super-dimension $(1,d\vert\mathbf{N})$ is:
\begin{enumerate}
\item
an open cover
$
  \big\{
    \inlinetikzcd{
    U_i^{1,d\vert \mathbf{N}}
    \ar[
      rr,
      hook,
      "{ \iota_i }"
    ]
    \&\&   
    X^{1,d\vert \mathbf{N}}
    }
  \big\}_{i \in I}
  \,,
$
\item
isomorphisms \;
$
  \begin{tikzcd}[row sep=0pt, column sep=12pt]
    T U_i
    \ar[
      rr,
      "{
        (E,\Psi)_i
      }",
      "{ \sim }"{swap}
    ]
    \ar[
      dr,
      shorten =-3pt,
    ]
    &&
    U_i
    \times
    \mathbb{R}^{1,d\vert \mathbf{N}}
    \ar[dl]
    \mathrlap{\,,}
    \\
    &
    U_i
  \end{tikzcd}
$
\item
transition functions
$
  \inlinetikzcd{
    U_i \cap U_j
    \ar[
      r,
      "{ g_{i j} }"
    ]
    \&
    \mathrm{Spin}(1,d)
    \mathrlap{\,,}
  }
$
\item
and cocycle conditions
$
  (E,\Psi)_j
  =
  g_{i j}
  \cdot
  (E,\Psi)_i
  \text{ on $U_i \cap U_j$.}
$
\end{enumerate}

Concretely, these isomorphisms correspond to tuples of 1-forms
\begin{equation}
\label
{SuperCoframeForms}
  \bracket({
    E_i, \Psi_i
  })
  =
  \left({
  \begin{aligned}
  &
  \bracket({
    E^a_i
    \in
    \Omega^{1,\mathrm{evn}}
      _{\mathrm{dR}}
    \bracket({U_i})  
  })_{a=0}^{d},
  \\  
  &
  \bracket({
    \Psi^\alpha_i
    \in
    \Omega^{1,\mathrm{odd}}
      _{\mathrm{dR}}
    \bracket({U_i})  
  })_{\alpha = 1}^N
  \end{aligned}
  }\right)
  \mathrlap{,}
\end{equation}
on which the transition functions act via the given $\mathrm{Spin}(1,d)$-representation:
\begin{equation}
  \begin{aligned}
  E_j^a 
  &=
  ({
    g_{i j}
  })^a{}_{a'} 
  E_i^{a'}
  \\
  \Psi_j^\alpha 
  &=
  ({
    g_{i j}
  })^\alpha{}_{\alpha'} 
  \Psi_i^{a'}
  \end{aligned}
  \;\;\;
  \text{on $U_i \cap U_j$.}
\end{equation}
\end{definition}
\begin{remark}
Since ultimately one is focused on Spin-invariant expressions in these forms, it is safe and common to notationally suppress the chart index, as in \cref{CoordinateComponentsOfCoframeField} and all of the following.
\end{remark}

\begin{remark}
  Alternatively, it is natural to write ``$E^\alpha$'' for ``$\Psi^\alpha$'' and this is a common convention. However, the convention we adopt allows for tighter notation where spinor indices in Spin-invariant expression need not be expanded out.
\end{remark}

%%%%%%%%%%%%%%
\paragraph
{Spin-connection}
%%%%%%%%%%%%%%

Given a coframe field as above, a \emph{spin connection} is a principal connection on the corresponding $\mathrm{Spin}(1,d)$-principal bundles.
In terms of a choice of open cover as above, this means equivalently:
\begin{definition}
\label[definition]
{SpinConnection}
A Spin-connection consists of:
\begin{enumerate}
  \item
  connection forms
  $
    \bracket({\Omega_i})^{ab}
    =
    -
    \bracket({\Omega_i})^{ba}
    \;\;
    \text{in $\Omega^{1,\mathrm{evn}}_{\mathrm{dR}}\bracket({U_i})$}
  $,
  \item
  cocycle conditions
  $
    \Omega_{j}
    =
    g_{i j}
    \bracket({
      \Omega_i
      +
      \mathrm{d}
    })
    g_{ij}^{-1}
  $
  \;\;
  \text{ on $U_i \cap U_j$.}
\end{enumerate}
\end{definition}

%%%%%%%%%%%%%
\paragraph
{Component expansion}
%%%%%%%%%%%%%

On a super-chart with coordinates $(X,\Theta)$ we write the expansion of the super-coframe and spin-connection as follows (cf. \cref{IndexConvention}):
\begin{equation}
\label
{CoordinateComponentsOfCoframeField}
  \begin{alignedat}{2}
  E^a 
  &=:
  E^a_{\evencoordinateindex}
  \,
  \mathrm{d} X^\evencoordinateindex
  &+
  E^a_{\oddcoordinateindex}
  \,
  \mathrm{d} \Theta^\oddcoordinateindex
  \,
  \\
  \Psi^\alpha 
  &=:
  \Psi^\alpha_{\evencoordinateindex}
  \,
  \mathrm{d} X^\evencoordinateindex
  &+
  \Psi^\alpha_{\oddcoordinateindex}
  \,
  \mathrm{d} \Theta^\oddcoordinateindex
  \\
  \Omega^{a b}
  &=:
  \Omega^{a b}_{\evencoordinateindex}
  \,
  \mathrm{d} X^\evencoordinateindex
  &+
  \Omega^{a b}_{\oddcoordinateindex}
  \,
  \mathrm{d} \Theta^\oddcoordinateindex \,.
  \end{alignedat}
\end{equation}

\begin{SCtable}[.8][htb]
\caption{\label
{IndexConvention}%
  \textbf{Index convention,} cf. \cref{CoordinateComponentsOfCoframeField}. Here the even coordinate index ``$r$'' takes the role of the traditional ``$\mu$'', which we avoid so that Greek indices are reserved for odd components.
}
\begin{tblr}{ 
  colspec = {c|cc},
  colsep = 5pt,
  stretch = 1.1,
  hline{1,2,Z} = {solid},   
  row{2} = {bg=lightgray}, 
}
  & \textbf{Even} & \textbf{Odd} 
  \\
 \small  \textbf{Frame} 
    & $a \in \{0, \dots, d\}$ 
    & $\alpha \in \{1, \dots, N\}$ 
  \\
 \small  \textbf{Coord.} 
    & $\evencoordinateindex \in \{0, \dots, d\}$ 
    & $\oddcoordinateindex \in \{1, \dots, N\}$ 
    \\
\end{tblr}
\end{SCtable}

These form components \cref{CoordinateComponentsOfCoframeField} further expand into a (terminating) Taylor series in the odd coordinate function $\Theta^\rho$:
\begin{equation}
\label
{OddTaylorExpansionOfSuperfields}
  \begin{alignedat}{2}
    E^a_{r/\rho}
    & 
    =:
    \textstyle{\sum_{n=0}^N}
    \bracket({
      E^{(n)}
    })\rule{0pt}{11pt}^a_{r/\rho}
    && =:
    \textstyle{\sum_{n=1}^{N}}
    \tfrac{1}{n!}
    \Theta^{\rho_1}
    \cdots
    \Theta^{\rho_n}
    \bracket({
      E
        ^{(n)}
        _{\rho_1 \cdots \rho_n}
    })\rule[-4pt]{0pt}{15pt}
      ^a_{r/\rho}
    \\[4pt]
    \Psi^\alpha_{r/\rho}
    & 
    =:
    \textstyle{\sum_{n=0}^N}
    \bracket({
      \Psi^{(n)}
    })\rule{0pt}{11pt}
      ^\alpha_{r/\rho}
    && =:
    \textstyle{\sum_{n=1}^{N}}
    \tfrac{1}{n!}
    \Theta^{\rho_1}
    \cdots
    \Theta^{\rho_n}
    \bracket({
      \Psi
        ^{(n)}
        _{\rho_1 \cdots \rho_n}
    })\rule[-4pt]{0pt}{15pt}
      ^\alpha_{r/\rho}
    \\[4pt]
    \Omega^{ab}_{r/\rho}
    & 
    =:
    \textstyle{\sum_{n=0}^N}
    \bracket({
      \Omega^{(n)}
    })\rule{0pt}{11pt}
      ^{ab}_{r/\rho}
    && =:
    \textstyle{\sum_{n=1}^{N}}
    \tfrac{1}{n!}
    \Theta^{\rho_1}
    \cdots
    \Theta^{\rho_n}
    \bracket({
      \Omega
        ^{(n)}
        _{\rho_1 \cdots \rho_n}
    })\rule[-4pt]{0pt}{15pt}
      ^{ab}_{r/\rho}
    \mathrlap{\,,}
  \end{alignedat}
\end{equation}
whose components (on the far right) are functions on the ordinary underlying spacetime.

\begin{remark}
  The restriction of super-fields to ordinary spacetime $\inlinetikzcd{ U^{1,d} \ar[r, hook, "{ \epsilon^{\mathrlap{\rightsquigarrow}} }"] \& U^{1,d\vert \mathbf{N}} }$ underlying the given chart retains exactly the 0th order of the $r$-components in this expansion \cref{OddTaylorExpansionOfSuperfields}:
  \begin{equation}
  \label
  {PullbackOfExpandedFieldsToBosonic}
    \begin{aligned}
    ({
      \BosonicCounit
    })^\ast
    E
    & = E^{(0)}_r \, \mathrm{d}X^r
    \\
    ({
      \BosonicCounit
    })^\ast
    \Psi
    & = \Psi^{(0)}_r \, \mathrm{d}X^r
    \\
    ({
      \BosonicCounit
    })^\ast
    \Omega
    & = \Omega^{(0)}_r \, \mathrm{d}X^r
    \mathrlap{\,.}
    \end{aligned}
  \end{equation}
  Conversely, this means that any super-space formulation of super-gravity on $X^{1,d}$ needs to involve constraints that fix the higher components \cref{OddTaylorExpansionOfSuperfields} in terms of the 0th order spacetime component \cref{PullbackOfExpandedFieldsToBosonic}. We discuss this issue of \emph{rheonomy} in \cref{Rheonomy}.
\end{remark}

%%%%%%%%%%%%%%
\paragraph
{Super-spacetime}
%%%%%%%%%%%%%

In conclusion, we say that:
\begin{definition}
\label[definition]
{SuperSpacetime}
A \emph{super-spacetime} $X^{1,d\vert \mathbf{N}}$ of dimension $(1,d\vert\mathbf{N})$ is a supermanifold $X$ (\cref{SmoothSupermanifoldsAsSmoothSuperSet})  equipped with a $(1,d\vert\mathbf{N})$ super-coframe (\cref{SuperCoframeField}) and a Spin-connection (\cref{SpinConnection}).
\end{definition}
\begin{remark}
\label[remark]
{SuperspacetimeAsSupergravityFieldConfiguration}
Super-spacetimes (\cref{SuperSpacetime}) model the off-shell field content of supergravity:
\begin{enumerate}
\item  $X$ -- the \emph{spacetime},

\item $E$ -- the \emph{graviton field},

\item $\Psi$ -- the \emph{gravitino field},

\item $\Omega$ -- the \emph{auxiliary field}.
\end{enumerate}
\end{remark}

We proceed to discuss in \cref{Supergravity} the conditions on such data to make it the on-shell configurations of supergravity, hence satisfying the SuGra equations of motion.

%%%%%%%%%%%%%%

%%%%%%%%%%%%%%
%%%%%%%%%%%%%%%%
\section
{Supergravity}
\label
{Supergravity}
%%%%%%%%%%%%%%%

We focus on 11D SuGra (in \cref{Equations}) and its dimensional reductions (in \cref{Reductions}, cf. \cref{LogicOfSuGra}) and its probe M5-branes (in \cref{Probes}).

%%%%%%%%%%%%%%
\subsection
{Equations}
\label
{Equations}
%%%%%%%%%%%%%%

%%%%%%%%%%%%%%%%%
\subsubsection
{Higher Maxwell-Type Equations}
\label
{HigherMaxwellEquations}
%%%%%%%%%%%%%%%%%

Before discussing the concrete (super-)spacetime equations of motion of 11D SuGra (\cref{11DSuGra}) and its reductions (\cref{Reductions}) and branes (\cref{Probes}), here we briefly look at the general structure that these equations take (in the bosonic sector and generally on super-spacetime, cf. \cref{lS4IdentityFailsOnSpacetimeSalvagedOnSuperspace}). Namely these are equations of ``higher Maxwell-type'' \cite[\S 2.1]{SS24-Phase}, in that: 
\begin{enumerate}
  \item The Bianchi identities are independent of gauge potentials.

  This entails that we exclude non-abelian 1-form gauge fields from the discussion (whose Bianchi identities, $\mathrm{d}F^i = f^i{}_{j k}\, A^j \wedge F^k$, depend on gauge potentials $A^i$). In particular our discussion here does not apply to 10D heterotic SuGra. (But does apply to gauged SuGras with KK-gauge fields.)
  
  \item
  However, the Bianchi identities may involve non-linear self-sourcing terms.

  This means that electric and magnetic fluxes may couple to each other non-linearly as in \cref{CFieldEoM}. In the infrared completion of these theories (\cref{Completions}) such non-linearity entails that electromagnetic flux quantization must take place in \emph{non-abelian cohomology theories}, beyond the Whitehead-generalized cohomology theories (like K-theory) that have been considered for this purpose in the past.

\end{enumerate}

%%%%%%%%%%%%%%
\paragraph
{Plain Situation Without Sources}
%%%%%%%%%%%%%

First consider the plain situation,  without external sources but allowing self-sourcing of flux.

\begin{definition}
\label[definition]
{OnHigherMaxwellTypeBianchiIdentities}
Given a tuple of (duality-symmetric) flux densities
\begin{equation}
  \vec F
  :=
  \bracket({
  F^{(i)}
  \in
  \Omega
    ^{\mathrm{deg}_i}
    _{\mathrm{dR}}
  \bracket({X})
  })_{i \in I}
  \mathrlap{\,,}
\end{equation}
we say that \emph{Maxwell-type Bianchi identities} are differential equations of the form
\begin{equation}
\label
{HigherMaxwellTypeBianchiIdentities}
  \mathllap{
    \forall_{i \in I}
    \;\;\;
  }
  \mathrm{d}
  \,
  F^{(i)}
  \,=\,
  P^{(i)}\bracket({\vec F})
  \mathrlap{\,,}
\end{equation}
for graded-symmetric polynomials $P^{(i)}$ ($i \in I$) subject to the condition that:

The image of these equations under $\mathrm{d}$ is no further condition, hence $\mathrm{d} P^{(i)} \bracket({ \vec F }) = 0$ holds independently of special properties of $\vec F$, hence (cf. \parencites[Lem. 3.2.1]{SS26-HigherGauge}): 
\begin{equation}
\label
{IntegrabilityConditionOnBianchis}
  \forall_{i \in I}
  \;\;\;\;
  \textstyle{\sum_{j \in I}}
  P^{(j)}
  \frac{\partial}{\partial F^{(j)}}
  P^{(i)}
  =
  0
  \mathrlap{\,.}
\end{equation}
where the partial derivative is itself understood as being a graded derivation. 
\end{definition}

\begin{example}
  Examples of higher Maxwell-type duality-symmetric Bianchi identities \cref{HigherMaxwellTypeBianchiIdentities} include the following (cf. \cite[\S 2.4]{SS25-Flux}):
  \begin{description}
    \item[Vacuum Maxwell theory:]
    \begin{equation}
      \begin{aligned}
        \mathrm{d}\, F_2 
        & = 0
        \\
        \mathrm{d}\, G_2 
        & =
        0 \,.
      \end{aligned}
    \end{equation}

    \item[Gauge sector of 5D SuGra:]
    \begin{equation}
      \begin{aligned}
        \mathrm{d}\, F_2 & = 0
        \\
        \mathrm{d}\, F_3 
        & =
        \tfrac{1}{2} F_2 F_2 \,.
      \end{aligned}
    \end{equation}

    \item[Gauge sector of 10D IIA SuGra:]
    \begin{equation}\label{IIAGaugeSector}
      \begin{aligned}
        \mathrm{d}\, 
        F_2 & = 0
        \\
        \mathrm{d}\, 
        F_4 & = H_3 \wedge F_2        
        \\
        \mathrm{d}\, 
        F_6 & = - H_3 \wedge F_4    
        \\
        \mathrm{d}\, H_3 & = 0
        \\
        \mathrm{d}\, H_7
        & =
        \tfrac{1}{2}
        F_4 F_4 + F_2 F_6 \, ,
      \end{aligned}
    \end{equation}
which dimensionally reduces to the gauge sector of $(\mathcal{N}=2)$ 9D SuGra \eqref{torS4ClosedForms}, and arises itself as the reduction of the 
    \item[gauge sector of 11D SuGra:] \cref{CFieldEoM}
    \begin{equation}
      \begin{aligned}
        \mathrm{d}\, G_4 & = 0
        \\
        \mathrm{d}\, G_7 
        & =
        \tfrac{1}{2} G_4 G_4\,.
      \end{aligned}
    \end{equation}
  
  \end{description}
\end{example}

\begin{definition}
\label[definition]
{OnTheCharacteristicLInfinityAlgebra}
  Given higher Maxwell-type Bianchi identities \cref{HigherMaxwellTypeBianchiIdentities}, we say that their \emph{characteristic $L_\infty$-algebra} $\mathfrak{a} \in L_\infty \mathrm{Alg}_{_{\mathbb{R}}}$ is that whose Chevalley-Eilenberg algebra \cref{LInfinityFaithfulInDGCAlgOp} is,
  \begin{equation}
  \label
  {TheCharacteristicLInfinityAlgebra}
    \mathrm{CE}(\mathfrak{a})
    \defneq
    \mathrm{Free}_{\mathrm{dgcAlg}}
    \bracket({
    \bracket({b^{(i)}})_{i \in I}
    })
    \big/
    \bracket({
    \mathrm{d}\, b^{(i)}
    =
    P^{(i)}\bracket({\vec b\,})
    })_{i \in I}
    \mathrlap{\,,}
  \end{equation}
  for generators $b^{(i)}$ in degree $\mathrm{deg}_i$.

  That this is well-defined, in that $\mathrm{d}^2 = 0$, is equivalently the integrability condition \cref{IntegrabilityConditionOnBianchis}.
\end{definition}

It is now essentially tautological but important that:
\begin{proposition}
[Characteristic $L_\infty$-algebras]
\label[proposition]
{Characteristics}
Solutions to higher Maxwell-type Bianchi identities \cref{HigherMaxwellTypeBianchiIdentities} are in natural bijection with closed $L_\infty$-valued differential forms \cref{ClosedLInfinityValuedDifferentialForms} with coefficients in the characteristic $L_\infty$-algebra \cref{TheCharacteristicLInfinityAlgebra}:
\begin{equation}
  \bracketmid\{{
    \vec F
    \in 
    \Omega^\bullet_{\mathrm{dR}}(X)
  }{
    \mathrm{d}\, \vec F
    =
    \vec P\bracket({\vec F})
  }\}
  \,\simeq\,
  \Omega^1_{\mathrm{cl}}\bracket({
    X;
    \mathfrak{a}
  })
  \mathrlap{\,.}
\end{equation}
\end{proposition}

\begin{example}
\label[example]
{CharacteristicLInfinityOf11DSuGra}
The characteristic $L_\infty$-algebra of the Bianchi identities \cref{CFieldEoM}, which we denote $\mathfrak{l}S^4$ (for reasons explained in \cref{Charges}), has generators $v_3$ and $v_6$ subject to one nontrivial Lie bracket relation:
\begin{equation}
\label
{MTheoryGaugeAlgebra}
  \mathfrak{l}
  S^4
  \,\simeq\,
  \mathrm{Free}_{ L_\infty \mathrm{Alg}}(
    v_3, v_6
  )
  \big/
  \left({
    \begin{aligned}
    [v_3, v_3]_2 & = v_6, 
    \\
    [-,\cdots,-]_{\neq 2} & = 0 
    \end{aligned}
  }\right)
  \mathrlap{.}
\end{equation}
This is known as the \emph{M-theory gauge algebra} \parencites[(2.6)]{CremmerEtAl1998}[(3.4)]{LavrinenkoEtAl999}[(75)]{KalkkinenStelle2003}[(86)]{BandosEtAl2004}[(4.9)]{Sati2010}.

Accordingly, we may refer to the form of the Bianchi identities \cref{CFieldEoM} of 11D Sugra as \emph{$\mathfrak{l}S^4$-Bianchis}:
\begin{equation}
\label
{lS4BianchisAsClosedlS4ValuedForms}
  \bracketmid\{{
  \begin{aligned}
    G_4
    & \in
    \Omega^4_{\mathrm{dR}}(X)
    \\
    G_7
    & \in
    \Omega^7_{\mathrm{dR}}(X)
  \end{aligned}
  }{
  \begin{aligned}
    \mathrm{d}\,
    G_4
    & =
    0
    \\
    \mathrm{d}\,
    G_7
    & =
    \tfrac{1}{2}
    G_4 G_4
  \end{aligned}
  }\}
  \simeq
  \Omega^1_{\mathrm{cl}}\bracket({
    X; \mathfrak{l}S^4
  })
  \mathrlap{\,.}
\end{equation}
\end{example}

\begin{remark}
[Different roles played by $L_\infty$-algebras]
\label[remark]
{DifferentRolesOfLInfinityAlgebras}
Beware (cf. \cref{TwoApproachesToHigherGaugeTheory}) that the role of $L_\infty$-algebras here is different from the ``higher principal connection'' approach to higher gauge theory (\parencites[\S 11-12]{Schreiber2005}{BaezSchreiber2007}{FSSt12-DiffClasses}{SSS12}{SchreiberWaldorf2013}{BorstenEtAl2025}[\S 2]{Alfonsi2025HigherGeometry}{RistSaemannWolf2026}{Waldorf2026}{BunkEtAl2026}) which in the supergravity literature corresponds (\parencites{FSS15-WZW}{FSS19-RationalM}) to the (Cartan-) geometric approach of \parencites{DAuriaFre1982}{CDF1991}: 

There $L_\infty$-algebras serve as the coefficients of gauge potential forms which are generically not flat/closed --- while here they serve as coefficients of the flux densities and their closure is no less than the electromagnetic Gauss law constraint putting the theory on shell.

Indeed, supergravity does not justify a principality constraint on its globally completed higher gauge fields. For instance differential K-theory (a candidate IR-completion of the RR-field in 10D type II SuGra, cf. \parencites{Freed2002}[\S 3]{Szabo2013}{GS22-KTheory}) is not a higher principal connection theory, but is an example (\cite[Exs. 9.2, 11.2]{FSS23-Char}) of the differential nonabelian cohomology considered here (\cref{OnGeneralChargesInDiffNonabCohomology}).
\end{remark}

\begin{SCfigure}[.85][htb]
\caption{\label
{TwoApproachesToHigherGaugeTheory} 
There are two approaches to the global construction of higher gauge fields: Either as higher generalization of connections on principal/fiber bundles, or as nonabelian generalization of ordinary differential cohomology. Here we are concerned with the latter (from \cite[\S 9]{FSS23-Char}, cf. \cite{SS26-HigherGauge}).
}
\adjustbox{
  rndfbox=4pt, scale=0.9
}{
\hspace{-6pt}
\begin{tblr}{
  colspec={c||c|c},
  row{odd} = {gray!40},
  colsep=3pt
}
  &
  \textbf{Principal}
  &
  \textbf{Cohomological}
  \\
  \hline
  \def\arraystretch{.85}
  \begin{tabular}{@{}c@{}}
   Specified
    \\
   $L_\infty$-algebra
  \end{tabular}
  &
  \def\arraystretch{.85}
  \begin{tabular}{@{}c@{}}
    coefficient of
    \\
    gauge potentials
    \\
    (connections)
  \end{tabular}
  &
  \def\arraystretch{.85}
  \begin{tabular}{@{}c@{}}
    coefficient of
    \\
    flux densities
    \\
    (curvatures)
  \end{tabular}
  \\
  \def\arraystretch{.85}
  \begin{tabular}{@{}c@{}}
    Flatness /
    \\
    closure
  \end{tabular}
  &
  non-generic
  &
  \def\arraystretch{.85}
  \begin{tabular}{@{}c@{}}
    Gauss law /
    \\
    Bianchi identity
  \end{tabular}
  \\
  Examples
  &
  \def\arraystretch{.85}
  \begin{tabular}{@{}c@{}}
    heterotic SuGra
  \end{tabular}
  &
  \def\arraystretch{.85}
  \begin{tabular}{@{}c@{}}
    type I/II SuGra 
    \\
    / 11D SuGra
  \end{tabular}
\end{tblr}
}

\end{SCfigure}

%%%%%%%%%%%%%%%
\paragraph
{Relative Situation with Sources}
%%%%%%%%%%%%%%%

More generally:
\begin{definition}
\label[definition]
{TwistedRelativeBianchis}
Given ``background'' fluxes $\vec F$ on $X$  satisfying their higher Maxwell-type equations \cref{HigherMaxwellTypeBianchiIdentities}, consider a smooth map 
\begin{equation}
  \inlinetikzcd{
    X'
    \ar[r, "{ \Phi }"]
    \&
    X
  }
\end{equation}
from another smooth manifold $X'$.
Then on a tuple of further (duality-symmetric) flux densities on this submanifold,
\begin{equation}
  \vec F'
  :=
  \bracket\{{
    F'^{{i'}}
    \in
    \Omega^{\mathrm{deg}_{i'}}_{\mathrm{dR}}
    \bracket({
      \Sigma
    })
  }\}
    _{i' \in I'}
  \mathrlap{\,,}
\end{equation}
we say that \emph{$\vec F$-twisted $\Phi$-relative} Maxwell-type Bianchi identities are differential equations of the form
\begin{equation}
\label
{TwistedRelativeMaxwellEquations}
  \forall_{i' \in I'}
  \;\;\;\;
  \mathrm{d}\,
  F'^{(i')}
  =
  P'^{i'}\bracket({
    \vec F', \Phi^\ast\vec F
  })
\end{equation}
for graded-symmetric polynomials $P'^{i'}$ ($i' \in I'$) subject to the condition that:

The image of these equations under $\mathrm{d}$ is no further condition, in that 
\begin{equation}
  \forall_{i' \in I'}
  \;\;\;\;
  \textstyle{\sum_{j' \in I'}}
  P'^{(j')}
  \frac{\partial}{\partial F^{(j')}}
  P'^{(i')}
  +
  \textstyle{\sum_{j \in I}}
  P^{(j)}
  \frac{\partial}{\partial \Phi^\ast F^{(j)}}
  P^{(i')}
  =
  0
  \mathrlap{\,.}
\end{equation}
\end{definition}

In straightforward generalization of \cref{Characteristics} we have:
\begin{proposition}
[Characteristic relative $L_\infty$-algebras]
\label[proposition]
{RelativeCharacteristics}
  Solutions to twisted relative Maxwell-type equations
  \cref{TwistedRelativeMaxwellEquations} are in natural bijection with the set
  $
    \Omega^1_{\mathrm{cl}}\bracket({
      \Phi; \pi
    })
  $ of
  \begin{enumerate}
  \item 
  those
  closed $L_\infty$-algebra valued differential forms \cref{ClosedLInfinityValuedDifferentialForms} with coefficients in the \emph{relative characteristic $L_\infty$-algebra} $\mathfrak{a}'$ given by
  $$
    \mathrm{CE}(\mathfrak{a}')
    \defneq
    \mathrm{Free}_{\mathrm{dgcAlg}}
    \left({
      \begin{aligned}
        &
        \bracket({
          b'^{(i')}
        })_{i' \in I'}
        \\
        &
        \bracket({
          b^{(i)}
        })_{i \in I}
      \end{aligned}
    }\right)
    \Big/
    \left(
    \begin{aligned}
     \mathrm{d}\, b'^{(i')}
     &=
     P^{(i')}\bracket({
       \vec b',
       \vec b \,
     })
     \\
     \mathrm{d}\, b^{(i)}
     &=
     P^{(i)}\bracket({
       \vec b \,
     })
    \end{aligned}
    \right)
  $$
  and hence naturally fibered over the characteristic $L_\infty$-algebra \cref{TheCharacteristicLInfinityAlgebra} of the background fluxes:
  \begin{equation}
  \label
  {TheFibrationOfCharacteristicLInftyAlgebras}
    \inlinetikzcd{
      \mathfrak{a}'
      \ar[rr, ->>, "{ \pi }"]
      \&\&
      \mathfrak{a}
      \mathrlap{\,,}
    }
  \end{equation}
  \item
  whose projection to $\mathfrak{a}$-coefficients along \cref{TheFibrationOfCharacteristicLInftyAlgebras} equals the pullback of the background fluxes along $\Phi$:
  \begin{equation}
    \Omega^1_{\mathrm{cl}}\bracket({
      \Phi; \pi
    })
    :=
    \bracketmid\{{
    \vec F' \in 
    \Omega^1_{\mathrm{cl}}\bracket({
      X'; \mathfrak{a'}
    })
    }{
      \pi_\ast \vec F'
      =
      \Phi^\ast \vec F
    }
    \}\,.
  \end{equation}
  \end{enumerate}
\end{proposition}

\begin{example}
  Examples of twisted relative higher Maxwell-type duality-symmetric Bianchi identities (\cref{TwistedRelativeBianchis}) include the following:
  \begin{description}

  \item[Maxwell theory with sources:]
    In the inhomogeneous Maxwell equation 
    \begin{equation}
      \mathrm{d}\,
      F_2
      =
      J_3\,,
    \end{equation}
    the electric current density $J_3$ is in itself actually to be considered with compact spatial support and appearing in the above equation after forgetting this condition. 
    
    Hence on 4D Minkowski spacetime this is modeled by taking the map $\Phi$ to be the inclusion into the spatial 1-point compactification
    \begin{equation}
       \begin{tikzcd}[row sep=13pt, 
         column sep=-6pt,
         /tikz/column 5/.append style={anchor=base west}
       ]
         X'
         \ar[
           d,
           "{ \Phi }"
         ]
         & 
         :=
         &
         \mathbb{R}^{1,0} 
         &
           \times
          \ar[
             d,
             hook
           ]
         & 
        \mathbb{R}^3 
        &[15pt]
         \\
         X
         &
         :=
         &
         \mathbb{R}^{1,0}
         &\times&
         \mathbb{R}^3_{\cpt}
       \end{tikzcd}
    \end{equation}
    and taking the $L_\infty$-fibration to be as follows:
    \begin{equation}
      \begin{tikzcd}[row sep=13pt, 
        column sep=-6pt,
         /tikz/column 8/.append style={anchor=base west}
      ]
        \mathfrak{a}'
        \ar[
          d,
          ->>,
          "{ \pi }"{pos=.4}
        ]
        &:=&
        \mathfrak{l}_{_{B^3 \mathbb{Z}}}
        &[-5pt]
        E B^2 \mathbb{Z}
        \ar[
          d,
          ->>
        ]
        &
        \text{  where  }
        &
        \mathrm{CE}\bracket({
          \mathfrak{l}_{_{B^3 \mathbb{Z}}}
          EB^2 \mathbb{Z}
        })
        &\defneq&
        \mathrm{Free}_{\mathrm{dgcAlg}}
        \binom{f_2}{j_3}
        \big/
        \binom{
          \mathrm{d}
          f_2 = j_3}{
          \mathrm{d}
          j_3  = 0
      }
        \\
        \mathfrak{a}
        &:=&
        \mathfrak{l}
        &
        B B^2 \mathbb{Z}
        &
        \text{  where  }
        &
        \mathrm{CE}\bracket({
          \mathfrak{l}
          B^3 \mathbb{Z}
        })
        \ar[
          u,
          hook
        ]
        &\defneq&
        \mathrm{Free}_{\mathrm{dgcAlg}}
        (j_3)
        /
        \bracket({
          \mathrm{d}\, 
          j_3 = 0
        })
        \mathrlap{\,.}
      \end{tikzcd}
    \end{equation}

    Here the symbols indicate that this $L_\infty$-algebraic fibration is a shadow (the ``rationalization'') of the \emph{universal principal circle 2-bundle} (the universal \emph{$\mathrm{U}(1)$-bundle gerbe}); we explain this in \cref{Charges}.

  \item[Self-dual Flux on M5-Branes:]
  Here
  \begin{equation}
  \label
  {BianchiForM5WorldvolumeFlux}
    \mathrm{d}\, H_3
    =
    \Phi^\ast G_4
    \mathrlap{\,,}
  \end{equation}
  where $\Phi$ is the embedding of the 5-brane worldvolume into 11D spacetime (\cref{M5braneProbes}):
  \begin{equation}
    \begin{tikzcd}[row sep=13pt]
      \Sigma^{1,5}
      \ar[
        d,
        hook,
        "{ \Phi }"
      ]
      \\
      X^{1,10}
      \mathrlap{\,.}
    \end{tikzcd}
  \end{equation}
  This is characterized by the following fibration (cf. \cite[\S 3.7]{FSS20-H}):
  \begin{equation}
  \label
  {lS7RelativeToS4}
      \begin{tikzcd}[row sep=8pt,
        ampersand replacement=\&,
        column sep=-6pt,
         /tikz/column 8/.append style={anchor=base west}
      ]
      \mathfrak{a}'
      \ar[
        d,
        ->>
      ]
      \&:=\&
      \mathfrak{l}_{_{S^4}}
      \&[-5pt]S^7
      \ar[
        d,
        "{ 
          \mathfrak{l}
          h_{\mathbb{H}} 
        }"{swap} 
      ]
      \&\text{ where }\&
      \mathrm{CE}\bracket({
        \mathfrak{l}_{_{S^4}}S^7
      })
      \&\defneq\&
      \mathrm{Free}_{\mathrm{dgcAlg}}
      \left(
        \begin{subarray}{l}
        h_3
        \\
        g_4
        \\
        g_7
        \end{subarray}
      \right)
      \Big/
      \left(
      \begin{subarray}{l}
        \mathrm{d}\,
        h_3 = g_4
        \\
        \mathrm{d}\,
        g_4 = 0
        \\
        \mathrm{d}\,
        g_7 = \frac{1}{2}\, g_4 g_4
      \end{subarray}
      \right)
      \\
      \mathfrak{a}
      \&:=\&
      \mathfrak{l}
      \&[-5pt]
      S^4
      \&
      \text{ where }
      \&
      \mathrm{CE}\bracket({
        \mathfrak{l}S^4
      })
      \ar[
        u,
        hook
      ]
      \&\defneq\&
      \mathrm{Free}_{\mathrm{dgcAlg}}
    \binom{
        g_4}{
        g_7
      }
      \big/
      \left(
      \begin{subarray}{l}
        \mathrm{d}\,g_4
        = 0
        \\
        \mathrm{d}\, g_7
        = 
        \frac{1}{2}\, g_4 g_4
      \end{subarray}
      \right).
      \end{tikzcd}
  \end{equation}

  Here, the symbols indicate that this $L_\infty$-algebraic fibration is a shadow (the ``rationalization'') of the \emph{quaternionic Hopf fibration} $h_{\mathbb{H}}$ from the 7-sphere $S^7$ to the 4-sphere $S^4$; we explain this in \cref{Charges}.

\end{description}

\end{example}

%%%%%%%%%%%%%%%%
\subsubsection
{11D SuGra}
\label
{11DSuGra}
%%%%%%%%%%%%%%%%

We discuss the (duality-symmetric) superspace construction of 11D SuGra (\cref{11DSuGraAslS4BianchiOnSuperspace}).

%%%%%%%%%%%%%
\paragraph
{Literature on 11D SuGra}
%%%%%%%%%%%%%
The discovery or construction of 11D SuGra is attributed to \parencites{CremmerJuliaScherk1978}, early development includes \cite{DuffNilssonPope1986}. General reviews include \parencites[\S I]{Duff1999World}{MiemiecSchnakenburg2006}[\S 10]{FreedmanVanProeyen2012}[\S22.1.1]{Ortin2015}{DuffEtAl2026}.
Duality-symmetric formulations were considered in \parencites{BandosEtAl1998}[\S 2]{CremmerEtAl1998}[\S2]{BandosEtAl2004}.
The superspace constraint formulation originates with \cite{CremmerFerrara1980,BrinkHowe1980,DAuriaFre1982}, review includes \parencites[\S 2]{Cederwalletal2005}. (Alternatively there is a Lagrangian pure spinor superfield formulation due to  \parencites{Cederwall2010}, cf. \parencites{HahnerSaberi2025}.) The possibility of a duality-symmetric formulation on superspace was (after being discounted in \cite[p. 120]{DAuriaFre1982}) briefly touched upon in \parencites[\S III.8.5]{CDF1991}[\S 6]{CandielloLechner1994}. 
We follow the completed construction in \cite[\S3]{GSS24-SuGra}.

%%%%%%%%%%%%%
\paragraph
{Spinors in 11D}
%%%%%%%%%%%%%
To set the scene and establish notation, we briefly recall 11D spin geometry.

\begin{proposition}[\textbf{The 11D Spin Representation}, {cf. \parencites{BaezHuerta2010}[\S2.5]{MiemiecSchnakenburg2006}[\S2.2.1]{GSS24-SuGra}}]
There exists an irreducible real pinor representation (the \emph{Majorana representation}) 
\begin{equation}
  \mathbf{32}
  \in
  \mathrm{Rep}_{_{\mathbb{R}}}
  \bracket({
    \mathrm{Pin}^+\bracket({
      1,10
    })
  })
  \mathrlap{\,,}
\end{equation}
equipped with a $\mathrm{Spin}(1,10)$-invariant, skew-symmetric, bilinear map
\begin{equation}
\label
{TheSpinorPairing}
  \begin{tikzcd}[row sep=-3pt, column sep=0pt]
    \mathbf{32}
    \times_{_{\mathbb{R}}}
    \mathbf{32}
    \ar[
      rr
    ]
    &&
    \mathbb{R}
    \\
    \bracket({
      \psi, \phi
    })
    &\longmapsto&
    \bracket({
      \overline{\psi}
      \,
      \phi
    })
  \end{tikzcd}
\end{equation}
such that the Clifford generators $\bracket({\Gamma_a})_{a=0}^{10}$ satisfy:
\begin{equation}
\begin{aligned}
  \Gamma_a \circ \Gamma_b
  +
  \Gamma_b \circ \Gamma_a
  & =
  + 2\, \eta_{ab} 
  \\
  &
  :=
  + 2 \, 
  \mathrm{diag}\bracket({
    -1, +1, +1, \cdots, +1
  })_{a b}
  \mathrlap{\,,}
\end{aligned}
\end{equation}
and are skew-adjoint with respect to
\cref{TheSpinorPairing}:
\begin{equation}
  \overline{\Gamma_a}
  =
  -\Gamma_a
  \mathrlap{\,.}
\end{equation}
\end{proposition}

When we speak of an $(1,10\vert\mathbf{32})$-dimensional super-spacetime (\cref{SuperSpacetime}), it means that its tangent super-spaces are regarded as equipped with this Pin-action.

%%%%%%%%%%%%%%
\paragraph
{The 11D Super-Minkowski spacetime}
%%%%%%%%%%%%%%

\begin{example}
\label
{11DSuperMinkowski}
The $(1,10\vert \mathbf{32})$-dimensional super-Minkowski spacetime super-manifold $\mathbb{R}^{1,10\vert \mathbf{32}}$ has canonical coordinates
\begin{equation}
  \begin{tikzcd}[
    row sep=0pt
  ]
    X^a : 
    \mathbb{R}^{1,10\vert \mathbf{32}}
    \ar[r]
    &
    \mathbb{R}^1
    \,,
    &
    a \in \{0,\cdots, 10\}
    \\
    \Theta^\alpha : 
    \mathbb{R}^{1,10\vert \mathbf{32}}
    \ar[r]
    &
    \mathbb{R}^{0\vert 1}
    \,,
    &
    \alpha \in \{1,\cdots, 32\}
  \end{tikzcd}
\end{equation}
and super-coframe $\bracket({E,\Psi})$ given by:
\begin{equation}
\begin{aligned}
  E^a & :=
  \mathrm{d} x^a
  + 
  \bracket({
   \overline{\Theta}
   \Gamma^a
   \mathrm{d}\Theta
  })
  \\
  \mathrm{\Psi}^\alpha
  & :=
  \mathrm{d}\Theta^\alpha\,,
\end{aligned}
\end{equation}
hence satisfying:
\begin{equation}
\label
{DifferentialOfFlat11DSuperCoframe}
 \begin{aligned}
   \mathrm{d}\,
   E^a
   & = 
   \bracket({
     \overline{\Psi}
     \Gamma^a
     \Psi
   })
   \mathrlap{\,,}
   \\
   \mathrm{d}\,
   \Psi^\alpha
   & =
   0
   \mathrlap{\,.}
 \end{aligned}
\end{equation}

By construction there is a canonical $\mathrm{Spin}(1,10)$ action on $\mathbb{R}^{1,10\vert \mathbf{32}}$. The corresponding semidirect product group is the \emph{super-Poincar{\'e}} or \emph{super-symmetry} group in dimension $(1,10\vert \mathbf{32})$:
\begin{equation}
  \mathrm{Iso}\bracket({
    \mathbb{R}^{1,10\vert \mathbf{32}}
  })
    \simeq
    \mathbb{R}^{1,10\vert \mathbf{32}}
    \rtimes
    \mathrm{Spin}\bracket({1,10})
    \mathrlap{\,.}
\end{equation}
\end{example}

\begin{proposition}
\label[proposition]
{TheAvatarSuperFLuxesIn11D}
  The differential super-forms
  \begin{equation}
  \label{Flat11Dcocycles}
    \left.
    \begin{aligned}
      G_4^0
      & :=
      \tfrac{1}{2}
      \bracket({
        \overline{\Psi}
        \Gamma_{a_1 a_2}
        \Psi
      })
      E^{a_1}
      E^{a_2}
      \\
      G_7^0
      & :=
      \tfrac{1}{5!}
      \bracket({
        \overline{\Psi}
        \Gamma_{a_1 \cdots a_5}
        \Psi
      })
      E^{a_1}
      \cdots
      E^{a_5}
    \end{aligned}
    \right\} \, 
    \in \, 
    \Omega^\bullet_{\mathrm{dR}}
    \bracket({
      \mathbb{R}^{1,10\vert \mathbf{32}}
    }) \, \simeq \mathrm{CE}(\FR^{1,10\vert \mathbf{32}})
  \end{equation}
  satisfy the equations \cref{BianchiOf11DSuperFluxesInIntro}:
  \begin{equation}
    \begin{aligned}
      \mathrm{d}\,
      G_4^0 & = 0
      \mathrlap{\,,}
      \\
      \mathrm{d}\,
      G_7^0 & =
      \tfrac{1}{2}
      G_4^0
      G_4^0
      \mathrlap{\,.}
    \end{aligned}
  \end{equation}
\end{proposition}
\begin{proof}
  By \cref{DifferentialOfFlat11DSuperCoframe}, the claim is equivalent to the statement that the following \emph{Fierz identities}
\begin{subequations}
\label
{TheFierzIdentitiesIn11D}
  \begin{align}
    \label{FirstFierzIdentityIn11D}
    \bracket({
      \overline{\Psi}
      \Gamma_{a b}
      \Psi
    })
    \bracket({
      \overline{\Psi}
      \Gamma^{a}
      \Psi
    })
    &
    = 0
    \mathrlap{\,,}
    \\
    \label{SecondFierzIdentityIn11D}
    \tfrac{1}{4!}
    \bracket({
      \overline{\Psi}
      \Gamma_{a [b_1\cdots b_4]}
      \Psi
    })
    \bracket({
      \overline{\Psi}
      \Gamma^{a}
      \Psi
    })
    &
    =
    \tfrac{1}{2}
    \,
    \tfrac{1}{2}
    \bracket({
      \overline{\Psi}
      \Gamma_{[b_1 b_2}
      \Psi
    })
    \tfrac{1}{2}
    \bracket({
      \overline{\Psi}
      \Gamma_{b_3 b_4]}
      \Psi
    })
    \mathrlap{\,,}
  \end{align}
\end{subequations}
hold for all \emph{commuting} \cref{CommutationRelationsOfBasicSuper1Forms} spinors $\Psi \in \mathbf{32}$ (and all $b_i \in \{0,\cdots, 10\}$).

A direct algebraic proof of these identities is given in \cite[(2.28)]{NaitoOsadaFukui1986}. A representation-theoretic proof is indicated in \parencites[(3.13), (3.28a)]{DAuriaFre1982} using \cite{DAuriaEtAl1982}.
\end{proof}

%%%%%%%%%%%%%%
\paragraph
{Curved 11D Super-Spacetime}
%%%%%%%%%%%%%

\begin{definition}[Cartan structural equations]
\label[definition]
{CartanStructuralEquationsIn11D}
  Given an $(1,10\vert\mathbf{32})$-dimensional super-spacetime $X^{1,10\vert\mathbf{32}}$ with super-coframe $(E,\Psi)$ and spin-connection $\Omega$, its \emph{super-torsion}, \emph{curvature}, and \emph{gravitino field strength} are:
  \begin{subequations}
  \label
  {TheCartanStructuralEquationsIn11D}
  \begin{alignat}{4}
  \label{SuperTorsionIn11D}
     \substack{
       \text{\color{darkblue}super-}
       \\
       \text{\color{darkblue}torsion}
     }
     &
     \;\;\big(
     &
     T^a
     & 
     :=
     \mathrm{d}\, E^a
     &&
     +
     \Omega^a_b \, E^b
     -
     \bracket({
       \overline{\Psi}
       \Gamma^a
       \Psi
     })
     \big)_{a = 0}^{10}
     \mathrlap{\,,}
    \\
     \substack{
       \text{\color{darkblue}curvature}
     }
    &
     \;\;\big(
    &
    R^{ab}
    & :=
    \mathrm{d}\, 
    \Omega^{ab}
    &&
    +\Omega^a{}_c \Omega^{c b}
    \big)_{a,b = 0}^{10}
     \mathrlap{\,,}
    \\
     \substack{
       \text{\color{darkblue}gravitino}
       \\
       \text{\color{darkblue}field-strength}
     }
    &
     \;\;\big(
    &
    \rho
    & 
    :=
    \mathrm{d}\, \Psi
    &&
    +
    \tfrac{1}{4}
    \Omega^{a b}
    \Gamma_{a b}
    \Psi
    \big)_{\alpha=1}^{32}
    \mathrlap{\,.}
  \end{alignat}
  \end{subequations}
\end{definition}

\begin{definition}
\label[definition]
{SuperTorsionFreeIn11D}  
  We say that the super-spacetime is \emph{super-torsion-free} if its super-torsion tensor \cref{SuperTorsionIn11D} vanishes.
\end{definition}

\begin{definition}
\label[definition]
{OnSuperFluxFormsIn11D}
  Given a super-spacetime $X^{1,10\vert\mathbf{32}}$, we say that $C$-field \emph{super-flux} forms are pairs of differential forms $G^s_4, G^s_7 \in \Omega^\bullet_{\mathrm{dR}}\bracket({X^{1,10\vert \mathbf{32}}})$ with the following specific local expansion in the given super-coframe field $\bracket({E,\Psi})$: 
  \begin{equation}
  \label
  {SuperFluxFormsIn11D}
    \begin{aligned}
      G_4^s
      & =
      \grayoverbrace{
      \tfrac{1}{4!}
      \bracket({
        G_4
      })_{a_1 \cdots a_4}
      E^{a_1}
      \cdots
      E^{a_4}
      }{G_4}
      +
      \grayoverbrace{
      \tfrac{1}{2}
      \bracket({
        \overline{\Psi}
        \Gamma_{a_1 a_2}
        \Psi
      })
      E^{a_1}
      E^{a_2}
      }{
        G_4^0
      }
      \\[4pt]
      G_7^s
      & :=
      \grayunderbrace{
      \tfrac{1}{7!}
      \bracket({
        G_7
      })_{a_1 \cdots a_7}
      E^{a_1}
      \cdots
      E^{a_7}
      }{G_7}
      +
      \grayunderbrace{
      \tfrac{1}{5!}
      \bracket({
        \overline{\Psi}
        \Gamma_{a_1 \cdots a_5}
        \Psi
      })
      E^{a_1}
      \cdots
      E^{a_5}
      }{G_7^0}
      \mathrlap{\,,}
    \end{aligned}
  \end{equation}
  hence over charts determined by component super-functions
  \begin{equation}
    \bracket({G_4})_{a_1 \cdots a_4}
    ,\;
    \bracket({G_7})_{a_1 \cdots a_7}
    \in
    C^\infty\bracket({
      U^{1,10\vert \mathbf{32}}
    })
    \mathrlap{\,,}
  \end{equation}
  and with the odd coframe components $\bracket({G^s_4, G^s_7})$ from \cref{Flat11Dcocycles}.
\end{definition}

%%%%%%%%%%%%%
\paragraph
{The Equations of Motion}
%%%%%%%%%%%%%
With these preliminaries in hand, we now have the following most remarkable statement, the \emph{miracle} \cref{TheMiracleEquivalence} announced above:
\begin{theorem}
[{\cite[Thm. 3.1, Cor. 3.12]{GSS24-SuGra}}]
\label[theorem]
{11DSuGraAslS4BianchiOnSuperspace}
On a super-torsion-free \textup{(\cref{SuperTorsionFreeIn11D})} super-spacetime $X^{1,10\vert\mathbf{32}}$ \textup{(\cref{SuperSpacetime})}, super-flux densities \cref{SuperFluxFormsIn11D}
satisfy the $\mathfrak{l}S^4$-Bianchi identities \textup{(\cref{CharacteristicLInfinityOf11DSuGra})} iff the equations of motion of 11D supergravity hold rheonomically on super-spacetime:
\begin{subequations}
\label
{11DSugraSuperspaceEoM}
\hspace{-8mm}
  \begin{align}
  \label{SuperSpaceS4Bianchi}  
  &
  \left\{
  \begin{aligned}
    \mathrm{d}\, G_4^s
    & = 0
    \\
    \mathrm{d}\, 
    G_7^s
    & = 
    \tfrac{1}{2}
    G_4^s
    G_4^s
  \end{aligned}
  \right.
  \;\;\;
  \textup{\color{darkblue}($\mathfrak{l}S^4$-Bianchi})
  \\[+5pt]
  \label{EinsteinRSMaxwell}
  \Leftrightarrow
  &
  \left\{
  \begin{alignedat}{2}
      \partial
        \rule[-3.5pt]{0pt}{0pt}_{[b}\Omega^{ac}_{c]}
      +
      \Omega^{a e}_{[b}
      \,
      \Omega^{e' c}_{c]}
      \eta_{e e'}
    & =
      \substack{
      \phantom{-}
      \tfrac{1}{12}
      \bracket({G_4})
        _{a c_1 c_2 c_3}
      \bracket({G_4})
        _{b}{}^{c_1 c_2 c_3}
      \\
      -
      \tfrac{1}{144}
      \bracket({G_4})
        _{c_1 \cdots c_4}
      \bracket({G_4})
        ^{c_1 \cdots c_4}
      \eta_{a b}
      }
    \;\;&
    \textup{\color{darkblue}(Einstein)}
    \\
    \Gamma^{a b_1 b_2}
    \bracket({\nabla \Psi})
      _{b_1 b_2}
    & =
    0
    &
    \textup{\color{darkblue}\llap{(Rarita-Schwinger)}}
    \\
    \nabla_{[b} 
    \bracket({G_4})
      _{a_1 \cdots a_4]} 
    & 
    =\, 0
    \\
    \nabla_{[c}
    \bracket({
      G_7
    })_{a_1 \cdots a_7]}
    &
    = 
    \tfrac{1}{2}
    \bracket({
      G_4 G_4
    })_{c a_1 \cdots a_7}
    &
    \textup{\color{darkblue}(Maxwell-Chern-Simons)}
    \\
    \bracket({
      G_7
    })_{a_1 \cdots a_7}
    & 
    =
    \tfrac{1}{4!}
    \epsilon
     _{
       a_1 \cdots a_7 
       \,
       b_1 \cdots b_4
    }
    \bracket({G_4})
      ^{b_1 \cdots b_4}
    \\
    \mathrlap{\textup{rheonomically (\cref{Rheonomy})}}
  \end{alignedat}
  \right.
  \end{align}
\end{subequations}
In the case of vanishing gravitino field \cref{CoordinateComponentsOfCoframeField} on ordinary spacetime \cref{IncludingBosonicBodyOfSupermanifold}, $\bracket({\BosonicCounit})^\ast \Psi = 0$, this is equivalent there to the ordinary 11D Einstein-Maxwell-Chern-Simons equations, cf. \textup{\cref{FermionicSourceTerms}}.
\end{theorem}

Various remarks are in order:
\begin{remark}[\textbf{State of the literature}]
\begin{itemize}
\item[{\textbf{(i)}}]
  The equations \cref{EinsteinRSMaxwell} may be compared to \parencites[p. 131]{DAuriaFre1982}[p. 910]{CDF1991}; for vanishing gravitino cf. also \cite[(24)]{FigueroaOFarrill2013}.

\item[{\textbf{(ii)}}]
  A claim akin to \cref{11DSuGraAslS4BianchiOnSuperspace} appeared in \cite[\S III.8.5]{CDF1991}, but verification was offered there only for the simplest of a long list of increasingly intricate conditions that need to be checked. A similar claim without duality-symmetry (but imposing  ``superspace constraints'')  is due to \parencites{CremmerFerrara1980}{BrinkHowe1980}, where less verification was offered. 

  \item[{\textbf{(iii)}}]
  The full proof in \cite[Thm. 3.1, Cor. 3.12]{GSS24-SuGra} treats all required conditions and uses computer algebra \cite{GSS24-SuGraAncillary} to verify the heavy combinatorial identities involved. This shows in particular that various assumptions made in the literature are already implied (cf. \cite[Rem. 3.10]{GSS24-SuGra}, such as the scaling condition and ``rheonomic parameterization'' of the curvatures).
  \end{itemize}
\end{remark}

\begin{remark}[\textbf{Fermionic source terms}]
\label[remark]
{FermionicSourceTerms}
  Remarkably, the above equations \cref{EinsteinRSMaxwell} have the \emph{form} of the purely bosonic Einstein-Maxwell-Chern-Simons equations. The fermionic gravitino source terms are absorbed into the super-geometry.
  
  Namely, these equations \cref{EinsteinRSMaxwell} are on superspace and hence with respect to the super-coframe $(E,\Psi)$ \cref{SuperCoframeForms}. This entails in particular that the (covariant) derivatives on the left are formed with the coframe's inverse, so that with respect to a coordinate basis of local vector fields (cf. \cite[Rem. 2.83]{GSS24-SuGra}):
  \begin{equation}
  \label
  {RecallingCovariantDerivativeInCoframe}
    \nabla_a
    =
    \bracket({
      (E,\Psi)^{-1}
    })
    \mathclap{\phantom{\vert}}_a^r
    \,
    \nabla_r
    +
    \bracket({(E,\Psi)^{-1}})
      \mathclap{\phantom{\vert}}
        _a^\rho
    \,
    \nabla_\rho
    \mathrlap{\,,}
  \end{equation}
  hence with the inverse of the matrix
  \begin{equation}
    \bracket({
      (E,\Psi)^{a/\alpha}_{r/\rho}
    })
    \;=\;
    \left(
    \begin{matrix}
      \bracket({
        E^a_r
      })
      &
      \bracket({
        \Psi^\alpha_r
      })
      \\
      \bracket({
        E^a_\rho
      })
      &
      \bracket({
        \Psi^\alpha_\rho
      })
    \end{matrix}
    \right)
    \mathrlap{.}
  \end{equation}
  It is only when the spacetime gravitino field  vanishes, $\bracket({ \Psi^\alpha_r }) = 0$, that
  \begin{equation}
    \begin{aligned}
    \bracket({ \Psi^\alpha_r }) = 0
    \;
    \Leftrightarrow
    \;\;\;
    &
    \bracket({
      (E,\Psi)^{a/\alpha}_{r/\rho}
    })
    =
    \left(
    \begin{matrix}
      \bracket({
        E^a_r
      })
      &
      0^\alpha_r
      \\
      \bracket({
        E^a_\rho
      })
      &
      \bracket({
        \Psi^\alpha_\rho
      })
    \end{matrix}    
    \right)
    \\
    \Rightarrow
    \;\;\;
    &
    \bracket({
      \bracket({
        (E,\Psi)^{-1}
      })
      ^{r/\rho}_{a/\alpha}
    })
    \;=\;
    \left(
    \begin{matrix}
      \bracket({
        \bracket({E^{-1}})^r_a
      })
      &
      0^\rho_a
      \\
      \bracket({
      -
      ({
        E^{-1}
      })^r_a
      \,
        E^a_\rho 
      \,
      ({
        \Psi^{-1}
      })^\rho_\alpha
      })
      &
      \bracket({
        ({
          \Psi^{-1}
        })^\rho_\alpha
      })
    \end{matrix}    
    \right)
    \\
    \Rightarrow
    \;\;\;
    &
    \nabla_a 
      = 
    \bracket({
      E^{-1}
    })^r_a \nabla_r
    \mathrlap{\,.}
    \end{aligned}
  \end{equation}
  Hence only when the gravitino field vanishes 
  do the (covariant) super-derivatives \cref{RecallingCovariantDerivativeInCoframe} restrict to ordinary covariant derivatives,
  \begin{equation}
  \label
  {VanishingGravitinoImpliesOrdinaryDerivatives}
    \bracket({\BosonicCounit})
      ^\ast
    \Psi =  0
    \;\;\;\;\;
    \Rightarrow
    \;\;\;\;\;
    \nabla_a 
    =
    \bracket({E^{-1}})^r_a
    \nabla_r
    \;\;
    \text{on $X^{1,10}$,}
  \end{equation}
  and hence only then are the left hand sides of \cref{EinsteinRSMaxwell} the ordinary field strength tensors of bosonic Einstein-Maxwell-Chern-Simons theory.
  
  In general, the gravitino field induces corrections to \cref{VanishingGravitinoImpliesOrdinaryDerivatives} which conspire to produce fermionic source terms in the Einstein-Maxwell-Chern-Simons theory. We expand on this in \cref{ComputingFermionicSourceTerms}.
\end{remark}

\begin{remark}[\textbf{Rheonomy}]
\label[remark]
{RheonomyInTheEoMs}
  The crucial qualifier ``rheonomically'' in \cref{EinsteinRSMaxwell} refers to further differential equations (not yet explicitly stated here) specifying the odd coordinate derivatives (and hence: dependence) of the super-fields. 
  We expand on this in 
  \cref{Rheonomy}.
\end{remark}

%%%%%%%%%%%%%%%%
\subsubsection
{Fermionic Source Terms}
\label
{ComputingFermionicSourceTerms}
%%%%%%%%%%%%%%%%

As highlighted in \cref{FermionicSourceTerms}, the super-space Einstein-Maxwell-Chern-Simons equations \cref{EinsteinRSMaxwell} induce the ordinary such equations on ordinary spacetime only when the gravitino field vanishes there, in that $\bracket({\BosonicCounit})^\ast \Psi = 0$.

%%%%%%%%%%%%%%
\paragraph
{State of the Literature}
%%%%%%%%%%%%%%
While this special case of vanishing gravitino profile is predominantly considered in the literature (cf. \parencites[\S1.2.2]{Marolf2012}[\S12.6]{FreedmanVanProeyen2012}), SuGra solutions with non-vanishing gravitino field are of interest and have been discussed, at least in lower dimensions (cf. \cite{AichelburgGuven1983,AichelburgEmbacher1986,Aichelburg1994,BrooksKalloshOrtin1995,DereliSenikoglu2024}).
 
It may be underappreciated that, for nontrivial gravitino field, the 11D Sugra Einstein-Maxwell-Chern-Simons equations on spacetime pick up fermionic corrections. This was originally expressed 
\parencites[(6)]{CremmerFerrara1980}[(2.2.12-14)]{DuffNilssonPope1986} through absorbing the gravitino terms in a torsionful spin-connection (and via an accordingly modified spacetime 4-form flux), and more explicitly in  \cite[p. 4]{BandosEtAl1998}, whose formula we quickly rederive now from \cref{11DSuGraAslS4BianchiOnSuperspace}.

%%%%%%%%%%%%%%
\paragraph
{Fermion Terms in the Maxwell-Chern-Simons Equation}
%%%%%%%%%%%%%%

The systematic way would be to insert the covariant derivatives \cref{RecallingCovariantDerivativeInCoframe} into the Maxwell-Chern-Simons equation \cref{EinsteinRSMaxwell} and simplify. But this is tedious, not the least because it requires the rheonomic recursion (which we discuss in \cref{Rheonomy} but make explicit there only for the case of trivial gravitino field strength).

However, since the Maxwell-Chern-Simons equation on the flux density $G_4$ is \emph{close} to the $\mathfrak{l}S^4$-Bianchi on its super-flux version \cref{SuperSpaceS4Bianchi}, there is a shortcut:

To this end, consider the pullback of the superfluxes \eqref{SuperFluxFormsIn11D} to ordinary spacetime, using \cref{CoordinateComponentsOfCoframeField}, hence the \emph{bosonic fluxes}
\begin{subequations}
  \begin{alignat}{2}
    G^b_4 & := 
    \bracket({\BosonicCounit})^\ast
    G_4
    &&
    =
    \tfrac{1}{4!}\bracket({
      G_4
    })
      ^{\vert \Theta=0}
      _{r_1 \cdots r_4}
    \,
    \mathrm{d}X^1 \cdots \mathrm{d}X^4
    \\
    G^b_7 & := 
    \bracket({\BosonicCounit})^\ast
    G_7
    &&
    =
    \tfrac{1}{7!}\bracket({
      G_7
    })
      ^{\vert \Theta=0}
      _{r_1 \cdots r_7}
    \,
    \mathrm{d}X^1 \cdots \mathrm{d}X^7
    \mathrlap{\,,}
  \end{alignat}
\end{subequations}
and the \emph{fermionic fluxes}
\begin{subequations}
\label
{FermionicSpacetimeFluxes}
  \begin{alignat}{2}
    G^f_4
    &
    :=
    \bracket({\BosonicCounit})^\ast
    G_4^0
    &&
    =
    \tfrac{1}{2}
    \bracket({
      \overline{\Psi}_{r_1}
      \Gamma_{r_2 r_3}
      \Psi_{r_4}
    })^{\vert \Theta=0}
    \,
    \mathrm{d}X^1 \cdots \mathrm{d}X^4
    \\
    G^f_7
    &
    :=
    \bracket({\BosonicCounit})^\ast
    G_7^0
    \mathrlap{\,,}
    &&
    =
    \tfrac{1}{5!}
    \bracket({
      \overline{\Psi}_{r_1}
      \Gamma_{r_2 \cdots r_6}
      \Psi_{r_7}
    })^{\vert \Theta=0}
    \,
    \mathrm{d}X^1 \cdots \mathrm{d}X^7
    \mathrlap{\,,}
  \end{alignat}
\end{subequations}
and their combinations,
\begin{subequations}
\label
{FullSpacetimeFlux}
  \begin{alignat}{2}
  \FullSpacetimeFlux_4  
  & :=
  \bracket({\BosonicCounit})^\ast G_4^s   
  &&
  =
  G^b_4 + G^f_4
  \\
  \FullSpacetimeFlux_7 
  &:= \bracket({\BosonicCounit})^\ast G_7^s
  &&
  =
  G^b_7 + G^f_7
  \,,
  \end{alignat}
\end{subequations}
which may be understood as the ``actual'' spacetime flux densities.

Now, since the bosonic fluxes on superspace are ``super-Hodge dual'' by \cref{EinsteinRSMaxwell}, in that
\begin{equation}
  G_7 = \star \, G_4
  \,,
  \;\;\;
  \text{meaning:}
  \;\;\;
    \bracket({
      G_7
    })_{a_1 \cdots a_7}
    \,=\,
    \tfrac{1}{4!}
    \epsilon
     _{
       a_1 \cdots a_7 
       \,
       b_1 \cdots b_4
    }
    \bracket({G_4})
      ^{b_1 \cdots b_4}
    \,,
\end{equation}
and since this pulls back to actual Hodge duality on ordinary spacetime:
\begin{equation}
  \bracket(\BosonicCounit)^\ast
  \star G_4
  =
  \star \,
  \bracket({\BosonicCounit})^\ast 
  G_4
  =
  \star \, G_4^b
  \mathrlap{\,,}
\end{equation}
it follows that the full spacetime fluxes \cref{FullSpacetimeFlux} satisfy Hodge duality up to a correction by the fermionic spacetime fluxes \cref{FermionicSpacetimeFluxes}:
\begin{align}
  \label{SpacetimeFermionicHodgeDuality}
  \FullSpacetimeFlux_7
    \, = \, 
  \star 
  \FullSpacetimeFlux_4
    + 
  \bracket({
    G^f_7  
      - 
    \star 
    G^f_4
  })
  \,.
\end{align}

Therefore,  pulling back the super-spacetime $\mathfrak{l}S^4$-Bianchi identity \eqref{SuperSpaceS4Bianchi} to ordinary spacetime
yields the following \emph{gravitino-corrected Maxwell-Chern-Simons equations}:
\begin{subequations}
\label
{FullSpacetimeBianchiIdentities}
\begin{align}
  \dd\, \FullSpacetimeFlux_4 \, &= \, 0 
  \\
  \dd \star \FullSpacetimeFlux_4 
  &= \, 
  \tfrac{1}{2} \FullSpacetimeFlux_4 \, \FullSpacetimeFlux_4 +
  \dd 
  \bracket({ 
     \star G^f_4 - G^f_7
  }) 
  \mathrlap{\,.}
\end{align}
\end{subequations}
This coincides with the field equation obtained  in \cite[above (2.4)]{BandosEtAl1998} from the original Lagrangian density of \cite{CremmerJuliaScherk1978} (recognizing our $\FullSpacetimeFlux_4$ as ``$F^{(4)}$'' in \cite[(2.3), (7.1)]{BandosEtAl1998}).

\begin{remark}
\label[remark]
{lS4IdentityFailsOnSpacetimeSalvagedOnSuperspace}
  Equation \cref{FullSpacetimeBianchiIdentities} highlights that the spacetime flux densities $\bracket({\FullSpacetimeFlux_4, \star \FullSpacetimeFlux_4})$ do actually \emph{not generally} satisfy the $\mathfrak{l}S^4$-equations of motion \cref{CFieldEoM,lS4BianchisAsClosedlS4ValuedForms} (and  $\bracket({\FullSpacetimeFlux_4, \FullSpacetimeFlux_7})$ are not in general Hodge duals on spacetime). Indeed, equation \cref{FullSpacetimeBianchiIdentities} is not of higher Maxwell-type in the sense of \cref{OnHigherMaxwellTypeBianchiIdentities}; hence, when taken at face value, does not admit a global completion according to \cref{Completions}.

  The magic of \cref{11DSuGraAslS4BianchiOnSuperspace} is to restore all these properties \emph{on superspace} (where the offending gravitino terms get absorbed into the super-geometry) and thereby allow for IR-completions of 11D SuGra (discussed in \cref{Charges}).
\end{remark}

%%%%%%%%%%%%%%%%%%%
\subsubsection
{Rheonomy in 11D SuGra}
\label
{Rheonomy}
%%%%%%%%%%%%%%%%%%%

The idea of \emph{rheonomy} is evident, but the technical implementation requires care.

%%%%%%%%%%%%%
\paragraph
{The Idea}
%%%%%%%%%%%%%
Super-functions on a super-spacetime $X^{1,d\vert \mathbf{N}}$ have $2^N-1$ field components \cref{OddTaylorExpansionOfSuperfields} beyond their restriction to the ordinary spacetime $\inlinetikzcd{ \bosonic{X}^{1,d} \ar[r, hook, "{ \epsilon^{\mathrlap{\rightsquigarrow}} }"] \& X^{1,d\vert \mathbf{N}} }$. Super 1-forms have $N\bracket({2^N-1}) + N 2^N$ more components than their restriction, etc.
Hence for a super-space formulation to be equivalent to a field theory on $\bosonic{X}^{1,d}$, its on-shell conditions must include constraints that fix all these extraneous components. 

Concretely, this is the case if the on-shell super-fields are subject to differential equations prescribing their derivatives \emph{normal} to the bosonic spacetimes $\bosonic{X}{}^{1,d}$ inside super-spacetime $X^{1,d\vert \mathbf{N}}$, hence their odd coordinate derivatives in any chart.
Since this means to say that on-shell fields satisfy \emph{flow equations} away from ordinary spacetime through superspace, this idea has been called \emph{rheonomy} in \cite[\S III.3.3]{CDF1991}. 
It has been claimed \cite[p. 12]{DallAgataEtAl1999} that this is equivalent to what other authors call ``superspace constraints'', though details vary across authors.

%%%%%%%%%%%%%
\paragraph
{Concrete implementation}
%%%%%%%%%%%%%
The remaining question is how to make this idea precise. 
For instance, in \cite{CDF1991} rheonomy is meant to be imposed by hand, via imposition of ``rheonomic parameterization of the curvatures''. In \cref{11DSuGraAslS4BianchiOnSuperspace} however it is a rigorous consequence of the (duality-symmetric) Bianchi identities.

The precise implementation of the idea of rheonomy in 11D SuGra is due to \cite{Tsimpis2004} (who does not use that term, though). This relies on the existence of super-normal coordinates for the underlying $\mathrm{Spin}(1,10)$-structure, a fact that seems to not be noted in the original statement of rheonomy from \cite[\S III.3.3]{CDF1991}.

We proceed to briefly survey how this works, following the account in \cite[\S 2]{GSS25-Embedding}.

%%%%%%%%%%%%%%
\paragraph
{Wess-Zumino-Tsimpis Gauge}
%%%%%%%%%%%%%%

First --- due to the local $\mathrm{Spin}(1,10)$-gauge and diffeomorphism symmetry of the theory --- we need to fix a gauge in order to be able to talk about unique extensions of superfields over super-spacetime. A good gauge is the analog of \emph{Riemann normal coordinates} in first order formalism, namely of the \emph{radial gauge} $X^r \Omega_r = 0$, but now imposed only on the \emph{odd} coordinates. This is due to \cite[\S 3]{Tsimpis2004} (cf. related discussion by \parencites[\S A.3-4]{McArthur1984}[\S 3]{Atick1987}). 

In \cite[p. 2]{Tsimpis2004} this is motivated as ``the most natural generalization of the Wess-Zumino gauge'', whence we refer to it in \cite[Def. 2.1]{GSS25-Embedding} as \emph{Wess-Zumino-Tsimpis gauge}:
\begin{equation}
\label
{WZTGauge}
  \begin{aligned}
    \bracket({
      E^{(0)}
    })\rule[-0pt]{0pt}{10pt}
     ^a_\rho
    & \defneq
    0
    \\
    \bracket({
      \Psi^{(0)}
    })\rule[-0pt]{0pt}{10pt}
     ^\alpha_\rho
    & 
    \defneq 
    \delta^\alpha_{\rho}
    \\
    \bracket({
      \Omega^{(0)}
    })\rule[-0pt]{0pt}{10pt}
     ^{ab}_\rho
     & \defneq 0
  \end{aligned}
  \;\;\;\;
  \text{and}
  \;\;\;\;
  \forall_n 
  \;
  \left\{
  \begin{aligned}
    \bracket({
      E^{(n)}
        _{[\rho_1 \cdots \rho_n}
    })\rule[-3pt]{0pt}{15pt}
      _{\rho]}^a
    & 
    = 0
    \\
    \bracket({
      \Psi^{(n)}
        _{[\rho_1 \cdots \rho_n}
    })\rule[-3pt]{0pt}{15pt}
      _{\rho]}^{\alpha}
    & =
    0
    \\
    \bracket({
      \Omega^{(n)}
        _{[\rho_1 \cdots \rho_n}
    })\rule[-3pt]{0pt}{15pt}
      _{\rho]}^{ab}
    & =
    0
    \mathrlap{\,.}
  \end{aligned}
  \right.
\end{equation}
That these gauge conditions may always be satisfied on any chart, by applying super-symmetry and local Lorentz transformations to given fields (already off-shell), is shown in \parencites[\S 3]{Tsimpis2004}.

%%%%%%%%%%%%%
\paragraph
{Solving the Cartan Structural Equations}
%%%%%%%%%%%%%

Now it happens that in WZT gauge \cref{WZTGauge}, the Cartan structural equations \cref{TheCartanStructuralEquationsIn11D} may be solved for the odd coordinate derivatives.
For example (cf. \parencites[(59)]{Tsimpis2004}[Lem. 2.5]{GSS25-Embedding}):
\begin{equation}
\hspace{-3mm} 
  \begin{alignedat}{3}
    &\;\;&
    \mathrm{d}\, E^a
    &
    =
    - 
    \Omega^a{}_b\,
    E^b
    +
    \bracket({
      \overline{\Psi}
      \,\Gamma^a\,
      \Psi
    })
    &\;\;\;&
    \begin{subarray}{l}
      \text{by vanishing}
      \\
      \text{super-torsion}
    \end{subarray}
    \\
    \Rightarrow
    &&
    \Theta^\rho
    \partial_\rho
    E^a_r
    & =
    2
    \Theta^\rho
    \Psi^\alpha_\rho
    \Psi^{\alpha'}_r
    \Gamma^a_{\alpha\alpha'}
    +
    \Theta^\rho
    \partial_r E^a_\rho
    -
    \Theta^\rho
    \bracket({
      \Omega^a{}_b
    })_\rho
    E^b_r
    +
    \Theta^\rho
    \bracket({
      \Omega^{a}{}_b
    })_r
    E^b_\rho
    \\
    &&&
    = 2 \Theta^\rho
    \delta^\alpha_\rho
    \Gamma^a_{\alpha\alpha'}
    \Psi^{\alpha'}_r
    &&
    \begin{subarray}{l}
      \text{by gauge}
      \\
      \text{condition \cref{WZTGauge}.}
    \end{subarray}
  \end{alignedat}
\end{equation}

%%%%%%%%%%%%%
\paragraph
{Component Recursion}
%%%%%%%%%%%%%%

This way, in the WZT gauge \cref{WZTGauge}, one finds for any solution to the super-space equations \cref{11DSugraSuperspaceEoM} explicit recursion relations that express the higher orders \cref{OddTaylorExpansionOfSuperfields} of the super-fields in terms of the lower orders and of the flux density $G_4$ on ordinary spacetime \cite{Tsimpis2004}.

Concretely, for solutions with vanishing gravitino field strength, $\bracket({\nabla \Psi})_{ab} = 0$, one finds the following component recursion relation \cite[Lem. 2.5-7]{GSS25-Embedding}:
\begin{equation}
  \begin{aligned}
    \bracket({
      E^{(n\mathcolor{purple}{+1})}
    })\rule{0pt}{11pt}
      ^a_\rho
    \;\;
    & 
    =
    \tfrac{2}{n+2}
    \bracket({
      \overline{\Theta}
      \,\Gamma^a\,
      \Psi^{(n)}_\rho
    })
    \\
    \bracket({
      E^{(n\mathcolor{purple}{+1})}
    })\rule{0pt}{11pt}
      ^a_{r}
    \;\;
    & =
    \tfrac{2}{n+1}
    \bracket({
      \overline{\Theta}
      \,\Gamma^a\,
      \Psi^{(n)}_r
    })
    \\
    \bracket({
      \Psi^{(n\mathcolor{purple}{+1})}
    })\rule{0pt}{11pt}
      ^\alpha_{\rho}
    \;\;
    & =
    -
    \tfrac{1}{n+2}
    \tfrac{1}{4}
    \bracket({
      \Gamma_{ab}
      \Theta
    })\rule{0pt}{11pt}
      ^\alpha
    \bracket({
      \Omega^{(n)}
    })\rule[-2pt]{0pt}{13pt}
      ^{ab}_\rho
    +
    \tfrac{1}{n+2}
    \bracket({
      H_a \Theta
    })^\alpha
    \bracket({
      E^{(n)}
    })\rule[-2pt]{0pt}{13pt}
      ^a_\rho
    \\
    \bracket({
      \Psi^{(n\mathcolor{purple}{+1})}
    })\rule{0pt}{11pt}
      ^\alpha_{r}
    \;\;
    & =
    -
    \tfrac{1}{n+2}
    \tfrac{1}{4}
    \bracket({
      \Gamma_{ab}
      \Theta
    })\rule{0pt}{11pt}
      ^\alpha
    \bracket({
      \Omega^{(n)}
    })\rule[-2pt]{0pt}{13pt}
      ^{ab}_r
    +
    \tfrac{1}{n+1}
    \bracket({
      H_a \Theta
    })^\alpha
    \bracket({
      E^{(n)}
    })\rule[-2pt]{0pt}{13pt}
      ^a_r    
    \\
    \bracket({
      \Omega^{(n\mathcolor{purple}{+1})}
    })\rule{0pt}{11pt}
      ^{a_1 a_2}_{\rho}
    & =
    \tfrac{2}{n+2}
    \bracket({
      \overline{\Theta}
      \,K^{a_1 a_2}\,
      \Psi^{(n)}_\rho
    })
    \\
    \bracket({
      \Omega^{(n\mathcolor{purple}{+1})}
    })\rule{0pt}{11pt}
      ^{a_1 a_2}_{r}
    & =
    \tfrac{2}{n+1}
    \bracket({
      \overline{\Theta}
      \,K^{a_1 a_2}\,
      \Psi^{(n)}_r
    })
    \mathrlap{\,,}
  \end{aligned}
\end{equation}
where (cf. \cite[(135, 162)]{GSS24-SuGra})
\begin{equation}
  \begin{aligned}
    H_a\;\;
    & =
    \tfrac{1}{6}
    \tfrac{1}{3!}
    \bracket({G_4})_{a\,  b_1 b_2 b_3}
    \Gamma^{b_1 b_2 b_3}
    -
    \tfrac{1}{12}
    \tfrac{1}{4!}
    \bracket({G_4})^{b_1 \cdots b_4}
    \Gamma_{a \, b_1 \cdots b_4}
    \\
    K^{a_1 a_2}
    & =
    \tfrac{1}{6}
    \bracket({G_4})^{a_1 a_2 b_1 b_2}
    \Gamma_{b_1 b_2}
    +
    \tfrac{1}{6}
    \tfrac{1}{4!}
    \bracket({G_4})_{b_1 \cdots b_4}
    \Gamma^{a_1 a_2\, b_1 \cdots b_4}
  \end{aligned}
\end{equation}
and where the superspace flux field components are independent of the odd coordinates
(\cite[Lem. 2.4]{GSS25-Embedding}, still a consequence of the assumption that $(\nabla \Psi)_{ab} = 0$):
\begin{equation}
  \partial_{\Theta^\alpha}
  \bracket({
  \bracket({G_a})_{a_1 \cdots a_4}
  })
  =
  0
  \mathrlap{\,.}
\end{equation}

\begin{example}
[\textbf{Near-Horizon M-Brane Solutions}]
  Using this recursion one may for instance lift the M-brane near horizon spacetimes $\mathrm{AdS}^{4/7} \times S^{7/4}$ to super-space solutions of \cref{11DSugraSuperspaceEoM}. This is worked out in \cite[\S3.2, \S4.2]{GSS25-Embedding}.
\end{example}

%%%%%%%%%%%%%%%
\subsection
{Reductions}
\label
{Reductions}
%%%%%%%%%%%%%%%

We discuss (following \cite{GiotopoulosSati2026}) how the above formulation (\cref{11DSuGraAslS4BianchiOnSuperspace}) of 11D SuGra on superspace induces analogous formulations of its dimensional reductions, in particular to 10D IIA SuGra (in \cref{IIASuperspaceSuGra}, cf. \cref{SomeReductionsOf11DSuGra}).

%%%%%%%%%%%%%%%%%%
\subsubsection
{General Torus Reductions}
%%%%%%%%%%%%%%%%%%

%%%%%%%%%%%%
\paragraph
{Torus-fibered Super-Spacetimes}
%%%%%%%%%%%

We consider super-spacetimes $Y^{1,d\vert \mathbf{N}}$ which are principal bundles (cf. \parencites{CarmeliFioresiVaradarajan2018}{Eder2021}) 
\begin{equation}
  \label
  {TheKKBundle}
  \begin{tikzcd}[
    sep=17pt
  ]
    T^{k}
    \ar[
      r,
      hook
    ]
    &
    Y^{\mathrlap{1,d\vert \mathbf{N}}}
    \ar[
      d,
      ->>,
      "{ \pi }"{pos=.35}
    ]
    \\
    &
    X^{\mathrlap{1,d-k\vert \mathbf{N}}}
  \end{tikzcd}
\end{equation}
with typical fiber the $k$-torus $T^k := ({S^1})^k$,
regarded as a supermanifold, $T^k \defneq T^{k \vert 0}$.
Accordingly, the fermionic dimension $N$ remains unreduced, but in general branches as a $\mathrm{Spin}$-representation $\mathbf{N}$ upon reducing the $\mathrm{Spin}(1,d)$ structure group on $Y$ to be adapted along the principal $T^k$-structure (cf. Lem. \ref{GiotopoulosGauge}). 

Recall that a super-coframe $\bracket({E,\Psi})$ on $X^{1,d-k\vert \mathbf{N}}$ is actually defined on an open cover $\inlinetikzcd{ \tilde X \ar[rr, ->>, "{ (\iota_i)_{i \in I}\; }"] \&\& X }$ by charts $\big\{{ \inlinetikzcd{ U_i^{1,d-k\vert \mathbf{N}} \ar[r, hook, "{ \iota_i }"] \& X^{1,d-k\vert \mathbf{N}}  } }\big\}_{i \in I}$, $\tilde X := \coprod_i U_i$. We assume without restriction of generality that the $T^k$-bundle \cref{TheKKBundle} has been trivialized  over this same open cover $\tilde X$ (if not, we may refine the cover) by local sections $\sigma_i$: 
\begin{equation}
  \begin{tikzcd}
    &
    Y^{\mathrlap{1,d\vert \mathbf{N}}}
    \ar[
      d,
      ->>,
      "{ \pi }"{pos=.4}
    ]
    \\
    U^{1,d-k\vert \mathbf{N}}_i
    \ar[
      r,
      hook
    ]
    \ar[
      ur,
      "{ \sigma_i }"
    ]
    &
    X^{\mathrlap{1,d-k \vert \mathbf{N}}}
  \end{tikzcd}
\end{equation}
This means that a principal Ehresmann connection form $E'$ on $Y$ has a chartwise pullback (along the chosen trivializing sections $\sigma_i$) to $\tilde X$. Conversely, its curvature 2-form $F_2$ is a basic 2-form on $X$ whose  pullback to $Y$ (along $\pi$) we shall also often denote by the same symbol.

%%%%%%%%%%%%%%%
\paragraph
{Traditional formulation}
%%%%%%%%%%%%%%%
The traditional approach to dimensional reduction along $S^1$-fibers of supergravity theories in first-order coframe formulation may be found in \parencites{ScherkSchwarz1979}{Cremmer1982}{DuffNilssonPope1986}.
These authors essentially declare that a super-spacetime of the form $\big(X\times S^1, (E,\Psi,\Omega)\big)$ is (infinitesimally) symmetric with respect to the corresponding $\mathrm{U}(1)$-action if the fields in question are ``independent of the corresponding (locally defined) adapted coordinate''. This comes down  to saying that they are invariant under the Lie derivative of the corresponding (even) fundamental vector field $\xi$ flowing along the circle fiber:
\begin{equation}
\label{NaiveLieVanishing}
\bracket({
  L_{\xi} E,\, 
  L_{\xi} \Psi, \, 
  L_{\xi} \Omega 
})
  \, = \, 
  0 
  \mathrlap{\,.}
\end{equation}

But beware that this condition in itself is ill-defined on arbitrary topologies, since both the super-coframe $(E,\Psi)$ and the Spin-connection $\Omega$
are only locally defined (as forms valued in plain vector spaces, cf. \cref{Coframes}) and hence the naive Lie derivative does not yield a well-defined global operation on such objects, as it does not commute with the action of local Spin-transition functions (cf. \cite[\S1]{Giotopoulos2026}). 

%%%%%%%%%%%%%%
\paragraph
{Proper formulation}
%%%%%%%%%%%%%%

Instead, the a-priori well-defined ``covariant'' Lie derivative operator which appropriately describes\footnote{In particular, this covariant Lie derivative has the property that $L_\xi^K E = 0 $ if and only if $L_\xi g =0$ for the corresponding (even) metric $g= \eta_{ab} E^a\otimes E^b$. This is as expected, since a symmetry of a super-gravitational background should be, in particular, an isometry.} 
globally the symmetry condition on arbitrary topologies is that of Lichnerowicz and Kosmann \parencites{Lichnerowicz1963}{Kosmann1972}, originally introduced only for the action on spinor fields. 
We will not review details of the conceptual resolution of this issue here, but simply note that:
\begin{lemma}
[{\cite[Cor. 3.16]{Giotopoulos2026}}]
\label[lemma]
{GiotopoulosGauge}
On $(1,d\vert \mathbf{N})$-dimensional super-spacetimes which are principal $k$-torus bundles \cref{TheKKBundle} with space-like principal action,
there always exists a (partial) gauge fixing reducing the structure group along 
\begin{equation}
  \begin{tikzcd}[
    sep=small
  ]
    \mathrm{Spin}(1,d-k)
      \times 
    \mathrm{Spin}(k)
    \ar[
      r,
      hook
    ]
    &
    \mathrm{Spin}(1,d)
    \mathrlap{\,,}
  \end{tikzcd}
\end{equation}
in which the covariant symmetry condition does indeed collapse to the naive Lie derivative vanishing condition \eqref{NaiveLieVanishing}. 
\end{lemma}

%%%%%%%%%%%%%
\paragraph
{Kaluza-Klein Ansatz}
%%%%%%%%%%%%%

Hence in this gauge, from \cref{GiotopoulosGauge}, the super-gravitational field $(E,\Psi,\Omega)$ on $Y$ (\cref{SuperspacetimeAsSupergravityFieldConfiguration}) takes the form of the standard ``Kaluza--Klein ansatz'', but now justifiably applicable in arbitrary such topologies. That is:
\begin{enumerate}
\item the bosonic coframe on $Y$ is
decomposed as
\begin{equation}
\label
{CoframeKaluzaKleinGauge}
  E 
    = 
  \bracket({ 
    E_X
    ,\, 
    \Phi \cdot E'
  })
  \mathrlap{\,,}
\end{equation}
where 
\begin{enumerate}

\item 
  $E_X$ is a bosonic coframe on the base $X$ (pulled back along the projection $\pi$), 

\item 
$\ii E' \in \Omega^{1}_{\mathrm{dR}}\bracket({Y;\, \ii \FR^k})$ is an Ehresmann $\mathrm{U}(1)^{k}$-connection form on \cref{TheKKBundle} with curvature $F_2 \in \Omega^2_\mathrm{cl}\bracket({X; \FR^k})$
\begin{equation}
  \pi^\ast F_2 
    \,:=\, 
  \dd E'  
  \mathrlap{\,,}
\end{equation}

\item
$\Phi \in C^\infty_\mathrm{evn}\bracket({\tilde X;\,  \FR^k\otimes\bracket({ {\FR^\ast}})^{k}})$ is the dilaton field,  being  a ``matrix-valued'' field in that
\[
  \bracket({
    \Phi \cdot E'
  })^m 
    \,=\, 
  \Phi^m{}_J \, E'^J
  \mathrlap{\,,} 
\] 
where $m,J=1,\cdots, k$,%
\footnote{To distinguish between the two indices, we denote the indices of the tangent Euclidean $\FR^k$-fiber directions by $m,n,\dots$, while those of the internal Lie algebra $\mathfrak{u}(1)^{\times k}\simeq \ii \FR^k$ by $I,J,\cdots$. The point here is that the former transform in the fundamental $\mathrm{SO}(k)$-representation of the induced $\mathrm{Spin}(k)$-structure group, while the latter in its trivial representation.}  
\end{enumerate}

\item
the fermionic coframe on $Y$ is decomposed as
\begin{equation}
  \Psi 
    = 
    \Psi_X  \, + \, \lambda \cdot E'
  \mathrlap{\,,}
\end{equation}
where 
\begin{enumerate}
\item
$\Psi_X$ is a fermionic coframe on $X$, 
\item
$\lambda \in C^\infty_\mathrm{odd}\bracket({\tilde X;\,  \mathbf{N}\vert_{\mathrm{Spin}(1,d-k)}\otimes\bracket({ {\FR^\ast}})^{k} })$ is the dilatino superpartner of $\Phi$, with components forming a linear map valued in the corresponding branching of the 11D Majorana representation, i.e.,
\begin{equation}
  \lambda \cdot E' 
    = 
  \lambda_J \, E'^J
\end{equation}
for $\big\{\lambda_J \in C^\infty_\mathrm{odd}\bracket({\tilde X,\,  \mathbf{N}\vert_{\mathrm{Spin}(1,d-k)} }) \big\}_{J = 1,\cdots,k}$,
\end{enumerate}

\item
the matrix components of Spin$(1,10)$-connection $\Omega$ on $Y$ similarly decompose into basic components as
\begin{equation}
  \Omega^{ab} 
    \,=\, 
  \Omega^{ab}_{X} 
    + 
  {\widetilde{\Omega}_X}{}_{J}{}^{ab} 
    \, 
  E'^J 
  \mathrlap{\,.}
\end{equation}

\end{enumerate}

It follows that the components $\bracket\{{\Omega_X^{ab}}\}_{a,b=0,\cdots, d-k}$ encode a Spin$(1,d-k)$-connection on $X$, while those given by $\bracket\{{\hat{\Omega}_{X}^{mn}}\}_{m,n= 1,\cdots k} := \bracket\{{\Omega_X^{ab}}\}_{a,b=d-k+1,\cdots, d}$ encode an auxiliary $\mathrm{Spin}(k)$-connection on $X$. Naturally, these correspond to the remaining local $\mathrm{Spin}(1,d-k)$-gauge invariance of the basic super-coframe $(E_X,\Psi_X)$ and the local $\mathrm{Spin}(k)$-invariance of the dilaton/dilatino and $\mathrm{U}(1)^{k}$-connection combination $(\Phi, \lambda, E')$, respectively.

%%%%%%%%%%%%%
\paragraph
{Kaluza-Klein Reduction of Torsion-free Super-Spacetimes}
%%%%%%%%%%%%%

It is not hard to see (cf. \cite[Lem. 3.1]{Giotopoulos2026} for the $S^1$-bundle case) that, when expressed in the above partially fixed gauge (\cref{GiotopoulosGauge}), the condition of vanishing super-torsion (\cref{SuperTorsionFreeIn11D}),
\[
  \dd E^a 
     + \Omega^{a}{}_b E^b 
    - (\overline{\Psi}\, \Gamma^a \, \Psi) 
  \, = \, 0
  \mathrlap{\,,}
\]
is equivalent to the following list of conditions on fields of the base super-spacetime $X$:
\begin{enumerate}
\item The basic super-torsion-free condition 
\begin{equation}
  \dd E^a_X 
    + {\Omega_X}^{a}{}_b \,  E^b_X 
    - \bracket({
      \overline{\Psi}_X\, \Gamma^a \, \Psi_X
    }) 
    \,=\, 
    0
    \mathrlap{\,.} 
\end{equation}
\item The super-coframe expansion of the $\mathrm{U}(1)^{k}$-curvature $F_2$ component 2-forms 
\begin{equation}
\label{Curvature2FormExpansion}
\hspace{-.5cm}
\begin{aligned}
  F_2^I 
    &= 
    \bracket({
      \Phi^{-1}
    })^I{}_n 
    \bracket({ 
      - \Omega_X{}^{n}{}_{a} \wedge E_{X}^{a} 
      + \bracket({
          \overline{\Psi}_X \Gamma^n \Psi_X
        })  
    })
  \\
  &= 
  \bracket({
    \Phi^{-1}
  })^I{}_n 
  \bracket({   
   \Omega_X{}_{b,a}{}^{n} E_{X}^b \, E_{X}^{a} 
   + 2\delta^{nm}
     \bracket({\Phi^{-1}})^J{}_m\,  
     \bracket({
       \overline{\Psi}_X \Gamma_a \lambda_J
     }) 
     \, E_{X}^{a} 
     +  
     \bracket({
       \overline{\Psi}_X \Gamma^{10-(n-1)} \Psi_X
      }) 
 })
 \mathrlap{\,.}
 \end{aligned}
\end{equation}
\item The following super-coframe expansion of the dilaton $\mathrm{SO}(k
)$-covariant derivative on X:
\begin{equation}
  \begin{aligned}
  \dd^{\hat{\Omega}} \Phi^m 
  &
  \equiv\,  
  \dd \Phi^m + \hat{\Omega}^m{}_n \, \Phi^n  \,  
  \\
  &  
  =\,
  \widetilde{\Omega}_{X}{}^m{}_a E^a_X 
  + 2\bracket({\overline{\Psi}_X\Gamma^m \lambda})
  \,.
  \end{aligned}
\end{equation}
\item The following prescription for the reconstruction of the remaining components of the 11D Spin connection:
\begin{equation}
\begin{aligned}
  \widetilde{\Omega}_X{}_J{}^{a}{}_{b} 
  &
  = \Phi^m{}_{J} \, \Omega_{X}{}_{b}{}^{a}{}_{m}
  \mathrlap{\,,}
  \\
  \Omega_{X \,\beta}{}^{am} \Psi_X^\beta 
  & = 
  2 \delta^{mn} 
  \bracket({\Phi^{-1}})^{I}{}_{m} 
  \bracket({
    \overline{\Psi}_X \Gamma^a \lambda_I
  }) 
  \mathrlap{\,.}
\end{aligned}
\end{equation}
\end{enumerate}
This basic calculation shows that the data of any symmetric super-torsion-free super-spacetime of arbitrary principal $T^k$-bundle topology is essentially equivalent, via the partial gauge from \cref{GiotopoulosGauge}, to the data of: 
\begin{enumerate}
\item a torsion-free super-spacetime structure on $X$ 
\item an integral (curvature) $2$-form $F_2 \in \Omega^{2}_\mathrm{dR}(X;\, \FR^k)$ of specific coframe expansion \eqref{Curvature2FormExpansion}, encoding the topology of the bundle, 
\item a dilaton/dilatino  pair $(\Phi, \lambda)$ with an auxiliary Spin$(k)$-connection $\hat{\Omega}$, satisfying the differential relation
\begin{equation}
  \dd^{\hat{\Omega}} \Phi^m 
    = 
  \partial_a \Phi^m \, E_X^a  
  + 
  2 
  \bracket({
    \overline{\Psi}_X \, \Gamma^{10-(m-1)} \, \lambda
  }) 
  \mathrlap{\,.}
\end{equation}
\end{enumerate}

\begin{equation}
\label
{ReductionOf11DTorsionFreeSuperSpacetimes}
  \adjustbox{
    scale=1.3
  }{$
  \left\{
  \colorbox{lightgray}{$
  \substack{\text{Symmetric principal $T^k$-bundle} 
  \\
    \text{super-torsion-free}
    \\
    11D \, \text{super-spacetimes}
  }
  $}
  \right\}
  $}
  \Leftrightarrow
  \adjustbox{
    scale=1.3
  }{$
  \left\{
  \colorbox{lightgray}{$
  \substack{
    \text{Super-torsion-free} 
    \\
    (11-k)D \,\text{ super-spacetimes with}
    \\
    \text{integral 2-flux and dilaton/dilatino pair}
  }
  $}
  \right\}
  $}
  .
\end{equation}

This describes the dimensional reduction (and reoxidation) of the off-shell \textit{super-gravitational field} content of 11D SuGra to the off-shell super-gravitational field content of its lower $(11-k)$-dimensional descendants, along with the additional $(F_2, \Phi, \lambda)$ fields due to ``winding'' along the toroidal fibers. The remaining task is to reduce the full, fluxed on-shell conditions (Thm. \ref{11DSuGraAslS4BianchiOnSuperspace}), in a manner consistent with flux quantization choices. This is achieved via the following process of \textit{rational cyclification}, and more generally \textit{k-toroidification}, of the corresponding classifying $L_\infty$-algebra $\mathfrak{l}S^4$.  

%%%%%%%%%%%%%%
\paragraph
{Reduction of closed $L_\infty$-algebra-valued forms}
%%%%%%%%%%%%%%
Following  \cite[Prop. 2.2, Outlook]{GiotopoulosSati2026}, the $L_\infty$-algebraic toroidification/oxidation bijection from \cref{ExtCycAdjunction} has a straightforward extension to the setting of \textit{symmetric} closed $L_\infty$-algebra-valued forms on super-principal $T^k$-bundle spacetimes \cref{TheKKBundle}, by utilizing a choice of connection on the total space. 

\begin{proposition}[\textbf{Supermanifold reduction/oxidation for torus bundles}] \label{NonTrivialTorusBundleToroidification/Oxidation}
Let $X$ be a supermanifold supplied with $k$ even, integral closed but potentially non-exact 2-forms $ \bracket\{{ F_2^I\in \Omega^2_{\mathrm{dR}}(X) }\}_{I=1,\cdots, k}$ and $\mathfrak{h}$ a super-$L_\infty$-algebra of finite type. For any representative $T^k$-principal bundle \cref{TheKKBundle}
with de Rham Chern classes $\bracket\{{ [F_2^I]\in H^2_{\mathrm{dR}}(X) }\}_{I=1,\cdots, k}$, fix a choice of Ehresmann connection $\ii E ' \in \Omega^1_{\mathrm{dR}}\bracket({ Y; \, i \FR^k })$
trivializing the (pullback) curvature 2-form $F_2 = (F_2^1,\cdots, F_2^k)$
\[
\dd E' \, = \, \pi^* F_2\, ,
\]
hence in particular s.t.
\[
 \iota_{\xi_J} E ' \, = \, t_J 
\]
for $\{\xi_J\}_{J=1,\cdots, k}$ the fundamental (vertical) vector fields generating the $T^k$-action on $Y$, with respect to the Lie algebra basis $\{t_I\}_{I=1,\cdots, k}\subset i \FR^k$. Then:
There is a canonically associated bijection between:
\begin{itemize}[leftmargin=4mm]  
\item $T^k$-invariant maps of super-$L_\infty$-algebroids out of the tangent Lie algebroid $TY$ into $\mathfrak{h}$,  
  
\item maps of super-$L_\infty$-algebroids out of $TX \cong T\bracket({Y/T^k})$ into $\mathrm{tor}^k(\mathfrak{h})$ that preserve the curvature 2-form $F_2$,
\vspace{-2mm} 
  \begin{equation}
    \label{LinftyAlgebroidToroidificationHomIsomorphism}
    \Big\{
    \begin{tikzcd}
      TY
      \ar[
        r,
        "{ F }"
      ]
      &
      \mathfrak{h}
    \end{tikzcd}
    \Big\}
    \begin{tikzcd}[
      column sep=50pt
    ]
      \ar[
        r,
        shift left=5pt,
        "{  
          \substack{
            \textup{\color{darkgreen}%
              reduction
            }
            \\[+1pt]
            \mathrm{rdc}_{F_2}
          }
        }",
        "{ \sim }"{swap, yshift=-2pt}
      ]
      \ar[
        r,
        <-,
        shift right=5pt,
        "{ 
          \substack{
            \mathrm{oxd}_{F_2}
            \\
            \textup{\color{darkgreen}%
              oxidation  
            }
          }
        }"{swap},
      ]
      &
      {}
    \end{tikzcd}
    \bigg\{
    \begin{tikzcd}[
      row sep=-3pt, 
      column sep=8pt
    ]
      TX
      \ar[
        rr,
        "{ \widetilde F }"
      ]
      \ar[
        dr,
        shorten=-2pt,
        "{ F_2 }"{swap}
      ]
      &&
      \mathrm{tor}^k(\mathfrak{h})
      \ar[
        dl,
        shorten=-2pt,
        "{ \omega_2 }"
      ]
      \\
      &
      b \mathbb{R}^k
    \end{tikzcd}
    \bigg\}
  \end{equation}
  given by
\begin{equation}
\label
{LinftyAlgebroidToroidificationHomBijection}
  \hspace{-.6cm}
    \left.
  \begin{tikzcd}[
    ampersand replacement=\&,
    sep=0pt
  ]
    TY
    \ar[
      rr,
      "{ F }"
    ]
    \&\&
    \mathfrak{h}
    \\
    \begin{aligned}
    F^i 
    &=  
    F^i_{\mathrm{nw}} 
    \\ 
    &
    + 
    \sum\limits_{
      \mathclap{
      \substack{
        1\leq r \leq k 
        \\ 
        1 \leq I_1 < \cdots < I_r \leq k
      }
      }
    } 
    E'^{I_r}\cdots E'^{I_1} 
    F_{I_r\cdots I_1 }  
    \end{aligned}
    \&\mapsfrom\&
    e^i
  \end{tikzcd}
  \right\}
  \leftrightsquigarrow
  \left\{
  \begin{tikzcd}[
    row sep=-4pt, 
    column sep=0pt]
    TX
    \ar[
      rr,
      "{ 
        \widetilde{F} 
      }"
    ]
    &&
    \mathrm{tor}^k(\mathfrak{h})
    \\
    F^i_{\mathrm{nw}}
    &\mapsfrom&
    e^i
    \\
    (-1)^{r(r+1)/2} F_{I_r\cdots I_1 }
    &\mapsfrom&
    \dir{I_r}{\mathrm{s}}\cdots \dir{I_1}{\mathrm{s}}\,e^i
    \\
    F_2
    &\mapsfrom&
    \omega_2
    \mathrlap{\,, }
  \end{tikzcd}
  \right.
\end{equation}
where $ F_{I_r\cdots I_1 } := \iota_{\xi_{I_r}} \cdots \iota_{\xi_{I_1}} F^i $.

\end{itemize} 
That is, associated to any choice of $\mathrm{U}(1)^k$-connection $E'$ on $Y$, there is a bijection of $T^k$-invariant, closed $\mathfrak{h}$-valued forms on $Y$ and closed $\mathrm{tor}^k(\mathfrak{h})$-valued forms on $X\cong Y/T^k$ with fixed 2-form component $F_2 := \dd E'$\footnote{This extends to a (strict) isomorphism of the corresponding $\infty$-groupoids of fluxes with their higher (globally defined) concordances 
\begin{equation*}
  \inlinetikzcd{
  \shape  
  \Omega^1_{\mathrm{cl}}
  \bracket({
    Y; \,\mathfrak{h}
  })_{\mathrm{inv}} 
  \, \cong_{E'}  
  \shape
  \Omega^1_{\mathrm{cl}}
  \bracket({
    X; \,\mathrm{tor}^k(\mathfrak{h})
  })_{F_2} 
  \ar[r, hook]
  \&
  \shape
  \Omega^1_{\mathrm{cl}}
  \bracket({
    X; \,\mathrm{tor}^k(\mathfrak{h})
  })  
  \mathrlap{\,.}
  }
\end{equation*}
}
\begin{equation}
  \Omega^1_{\mathrm{cl}}
  \bracket({
    Y; \,\mathfrak{h}
  })_{\mathrm{inv}} 
  \;\; \cong_{E'} \;\; 
  \Omega^1_{\mathrm{cl}}
  \bracket({
    X; \,\mathrm{tor}^k(\mathfrak{h})
  })_{F_2}\, .
\end{equation}
\end{proposition}
\begin{proof}
Morphisms of $L_\infty$-algebroids out of $TY$ into $\mathfrak{h}$ are equivalently morphisms of sCDGAs $\mathrm{CE} (\mathfrak{h}) \rightarrow \Omega^\bullet_{\mathrm{dR}}(Y)$, and hence, for $\mathfrak{h}$ of finite type, are completely determined by the image of generators of $\mathrm{CE}(\mathfrak{h})$
$$
e^i \longmapsto  F^i \,.
$$
Now, given a (fixed) choice of connection $E'= E^I t_I$ on $Y\rightarrow X$, any $T^k$-\textit{invariant} $i$-form decomposes as 
\begin{equation}
  F^i 
  \, = \, 
  F^i_\mathrm{nw} 
    +
  \;\; 
  \sum\limits_{\mathclap{\substack{
    1\leq r \leq k 
    \\ 
    1 \leq I_1< \cdots < I_r \leq k
  }}}  
  \;\;
  E'^{I_r}\cdots E'^{I_1} \, F_{I_r\cdots I_1 }\, ,
\end{equation}
where, by construction, the coefficients are \textit{basic} differential forms
\[
  \bracket\{{
    F^i_\mathrm{nw}, F_{I_r\cdots I_1 }
      := 
    \iota_{\xi_{I_r}} 
      \cdots 
    \iota_{\xi_{I_1}} F^i
  }\} 
    \, \subset \, 
  \Omega^\bullet_\mathrm{dR}(X)
  \mathrlap{\,.}
\]

The proof that the forms $\{F^i\}$ satisfy the required $\mathfrak{h}$-Bianchi identities if and only if the forms $\bracket\{{ F^i_\mathrm{nw}, F_{I_r\cdots I_1 } F^i }\}$ satisfy the $\mathrm{tor}^k(\mathfrak{h})$-Bianchis follows by the same formal algebraic steps as that of \cite[Prop 2.60]{GSS25-TD}, i.e., the source being an $L_\infty$-algebra (an $L_\infty$-algebroid over a point).

The extension  to the $\infty$-groupoids of higher concordances (cf. ftn. 7) follows by lifting the bijection \eqref{LinftyAlgebroidToroidificationHomIsomorphism} to have instead as domains the $L_\infty$-algebroids $T(Y\times [0,1]^{\times n})$ and $T(X\times [0,1]^{\times n})$ that support (higher) concordances, using the induced pulled back $\mathrm{U}(1)^k$-connection.
\end{proof}

\medskip 
We next make explicit the formulas \cref{LinftyAlgebroidToroidificationHomBijection}
for $k=1,2$ (the cases of \emph{cyclification} and \emph{toroidification}), and for target $L_\infty$-algebra specialized to the characteristic $\mathfrak{h} \defneq \mathfrak{l}S^4$ of 11D SuGra (\cref{CharacteristicLInfinityOf11DSuGra}):
% Indeed, the existence of the Kaluza--Klein (partial) gauge for $k$-toroidal symmetric super-spacetimes naturally provides the crucial ingredient of a connection \eqref{CoframeKaluzaKleinGauge}, necessary to 
We apply Prop. \ref{NonTrivialTorusBundleToroidification/Oxidation} to the 11D superspace super-fluxes $(G^s_4,G_7^s)$ \eqref{SuperFluxFormsIn11D}. Together with the correspondence from \eqref{ReductionOf11DTorsionFreeSuperSpacetimes}, this propagates the miracle of \cref{11DSuGraAslS4BianchiOnSuperspace} to all $(11-k)$D descendant supergravity theories obtained by reduction on $k$-torus fibers! We proceed with the ``miracles'' corresponding to the cases of 10D IIA and 9D supergravities, respectively.

%%%%%%%%%%%%%%%%%%
\subsubsection
{IIA superspace SuGra via cyc$(\mathfrak{l}S^4)$}
\label
{IIASuperspaceSuGra}
%%%%%%%%%%%%%%%%
The cyclification (\cref{RationalCyclification}) of the real Whitehead $L_\infty$-algebra of the 4-sphere \cref{MTheoryGaugeAlgebra} is given (cf. \cite[Ex. 2.26]{GSS25-TD}\cite[Ex. 3.3]{FSS17-Sphere}) by:
\begin{equation}\label{CEcycS4}
  \mathrm{CE}\big(
    \mathrm{cyc}(
      \mathfrak{l}S^4
    )
  \big)
  \;\simeq\;
  \FR_\dd 
  \left[
  \def\arraystretch{1}
  \def\arraycolsep{2pt}
  \begin{array}{c}
    \omega_2
    \\
    g_4
    \\
    \mathrm{s}g_4
    \\
    g_7
    \\
    \mathrm{s}g_7
  \end{array}
  \right]
  \Big/
  \left(
  \def\arraystretch{1}
  \def\arraycolsep{2pt}
  \begin{array}{ccl}
    \mathrm{d}\,
    \omega_2 
    &=&
    0
    \\
    \mathrm{d}\,
    g_4
    &=&
    \omega_2\, \mathrm{s}g_4
    \\
    \mathrm{d}
    \,
    \mathrm{s}g_4
    &=&
    0
    \\
    \mathrm{d}
    \,
    g_7
    &=&
    \tfrac{1}{2}
    g_4\, g_4
    +
    \omega_2 \, \mathrm{s}g_7
    \\
    \mathrm{d}\,
    \mathrm{s}g_7
    &=&
    -g_4\, \mathrm{s}g_4
  \end{array}
  \!\right)
  \mathrlap{\,.}
\end{equation}
As indicated earlier in  \eqref{cycS4ClosedForms}, this is the classifying $L_\infty$-algebra for the gauge sector of IIA SuGra \eqref{IIAGaugeSector}, as seen by a direct comparison with the corresponding duality symmetric NS/RR-flux density Bianchi identities.

Moreover, by the $L_\infty$-algebraic (``tangent-space-wise'') version (\cref{ExtCycAdjunction}) of Prop. \ref{NonTrivialTorusBundleToroidification/Oxidation} it follows \cite[Ex 3.10]{GSS25-TD} that the $\mathfrak{l}S^4$-cocycle $(G_4^0, G_7^0)$ \eqref{Flat11Dcocycles} of $11D$ super-Minkowski spacetime $\FR^{1,10\vert \mathbf{32}}$ dimensionally reduces to 
the $\mathrm{cyc}(\mathfrak{l}S^4)$-cocycle $(F_2^0,H_3^0, F_4^0,F_6^0, H_7^0)$ on flat IIA super-Minkowski $\FR^{1,9\vert \mathbf{16}\oplus \overline{\mathbf{16}}}$, given as: 
\begin{equation}
\label{FlatNSRRcocycles}
\hspace{-1cm} 
 \begin{tikzcd}[
    row sep=-4pt
    ,column sep=25pt
    ,/tikz/column 1/.append style={anchor=base east}
  ]
    \mathbb{R}^{
      1,9\,\vert\,
      \mathbf{16} 
        \oplus 
      \overline{\mathbf{16}}
    }
    \ar[
      rr,
      "{ 
          \begin{array}{c}
            (F_2^0, F_4^0, F_6^0, H_3^0, H_7^0)
          \end{array}
      }"
    ]
    &&
\mathrm{cyc}(
      \mathfrak{l}S^4
    )
    \\
    {\color{darkorange}  \big(\hspace{1pt}
      \overline{\psi}
      \,\Gamma_{1\!0}
      \psi
    \big) }
    \;=:\;
    \mathrm{F}_2^0
    &\longmapsfrom&
    \omega_2
    \\
     {\color{darkorange}  \tfrac{1}{2}\big(\hspace{1pt}
      \overline{\psi}
      \,\Gamma_{a_1 a_2}\,
      \psi
    \big)
    e^{a_1} e^{a_2} }
    \;=:\;
    \mathrm{F}_4^0
    &\longmapsfrom&
    g_4
    \\
    {\color{darkorange} -  \big(\hspace{1pt}
      \overline{\psi}
      \,
      \Gamma_{1\!0}
      \Gamma_{a}
      \,
      \psi
    \big)
    e^a ) }
    \;=:\;
    H_3^0
    &\longmapsfrom&
    \mathrm{s}g_4
    \\
    {\color{darkorange}  \tfrac{1}{5!}
    \big(\hspace{1pt}
      \overline{\psi}
      \,\Gamma_{a_1 \cdots a_5}\,
      \psi
    \big)
    e^{a_1} \cdots e^{a_5} }
    \;=:\;
    H^0_7
    &\longmapsfrom&
    g_7
    \\
   {\color{darkorange}  -
    \tfrac{1}{4!}
    \big(\hspace{1pt}
      \overline{\psi}
      \,\Gamma_{1\!0}
      \Gamma_{a_1 \cdots a_4}
      \,
      \psi
    \big)
    e^{a_1} \cdots e^{a_4} }
    \;=:\;
    F_6^0
    &\longmapsfrom&
    \mathrm{s}g_7 \, ,
\end{tikzcd}
\end{equation}
hence also satisfying the IIA Bianchi identities \eqref{IIAGaugeSector}.

%%%%%%%%%%%%%
\paragraph
{10D IIA Supergravity via Super-flux Bianchis}
%%%%%%%%%%%%%
The miracle from 11D superspace SuGra (Thm. \ref{11DSuGraAslS4BianchiOnSuperspace}) suggests a similar result should hold in 10D: Namely, an appropriate combination of these {\color{darkorange}\textit{flat supersymmetric}} cocycles with the \textit{curved bosonic} duality symmetric NS/RR flux densities on super-torsion-free $(10\vert \mathbf{16}\oplus\overline{\mathbf{16}})$-dimensional super-spacetimes $X$ should correspond precisely to on-shell IIA supergravity. 

Indeed, by applying Prop. \ref{NonTrivialTorusBundleToroidification/Oxidation} on $S^1$-symmetric super-fluxes $(G_4^s,G_7^s)$ \eqref{SuperFluxFormsIn11D} living on $S^1$-principal bundles $Y\rightarrow X$, in the (partial) gauge fixing resulting in the Kaluza--Klein form for its super-gravitational field $(E,\Psi,\Omega)$, one obtains the following form of super-fluxes 
on the 
 $(10\vert \mathbf{16}\oplus \overline{\mathbf{16}})$-dimensional base super-spacetime $X$:   
\begin{equation}
\label
{CurvedNS/RRCocycles}
  \begin{alignedat}{4}
    F^s_2
    & :=
    \tfrac{1}{2}
    \bracket({F_2})_{a_1 a_2} 
    \, 
    E_X^{a_1} E_X^{a_2} 
    && + 
    \mathcolor{olive}{%
      \tfrac{2}{\Phi^2} 
      \bracket({
        \overline{\Psi}_X \, \Gamma_{a} \lambda
      })\, E^a_X }
    && 
    + 
    \mathcolor{olive}{ 
      \tfrac{1}{\Phi}
    }  
    \mathcolor{darkorange}{  
      \bracket({
        \overline{\Psi}_X \,\Gamma_{1\!0}\, \Psi_X
      })
    }
    \\
    H^s_3
    &:=
    \tfrac{1}{3!}(H_{3})_{a_1 a_2 a_3} 
    \, E_X^{a_1}E_X^{a_2}E_X^{a_3} 
    &&
    +  
    \mathcolor{olive}{ 
    \bracket({
      \overline{\Psi}_X
      \,\Gamma_{a_1 a_2}
      \lambda 
      })\,
    E_X^{a_1}  E_X^{a_2} } 
    &&
    -
    \mathcolor{olive}
    { \Phi } \, 
    \mathcolor{darkorange}{  
      \bracket({
        \overline{\Psi}_X\, \Gamma_{1\!0 a}\, \Psi_X
      })
      \, E_X^a 
    }
    \\
    F^s_4
    & :=
    \tfrac{1}{4!}
    \bracket({F_4})_{a_1 \cdots a_4} 
    \,
    E_X^{a_1}\!\cdots\!E_X^{a_4}  
    && 
    &&
    +
    \mathcolor{darkorange}{ 
      \tfrac{1}{2}
      \bracket({
        \overline{\Psi}_X
        \,\Gamma_{a_1 a_2}\, 
        \Psi_X
      })
      \,  E_X^{a_1} E_X^{a_2} 
    }
    \\
    F^s_6
    & := 
    \tfrac{1}{6!}
    \bracket({F_6})_{a_1 \cdots a_6}
    \,  E_X^{a_1}\cdots E_X^{a_6} 
    && +
    \mathcolor{olive}{
    \tfrac{2}{5!}
    \bracket({
      \overline{\Psi}_X
      \,\Gamma_{a_1 \cdots a_5}
      \lambda
    }) 
    \, 
    E_X^{a_1}\!\cdots\!E_X^{a_5} } 
    && -
    \mathcolor{olive}{
      \tfrac{\Phi}{4!} 
    } \,  
    \mathcolor{darkorange}{
      \bracket({
        \overline{\Psi}_X
        \,\Gamma_{1\!0 \, a_1 \cdots a_4}\, 
        \Psi_X
      }) 
      \,  
      E_X^{a_1}\!\cdots\!E_X^{a_4}  
    }
    \\
    H^s_7
    &:= 
    \tfrac{1}{7!}
    \bracket({H_7})_{a_1 \cdots a_7}
    \, 
    E_X^{a_1} \cdots E_X^{a_7} 
    && &&
    + 
    \mathcolor{darkorange}{ 
    \tfrac{1}{5!}
    \bracket({
      \overline{\Psi}_X
      \,\Gamma_{a_1 \cdots a_5}\,
      \Psi_X
    }) 
    \, 
    E_X^{a_1} \cdots E_X^{a_5}  
    }
    \mathrlap{\,,}
  \end{alignedat}
\end{equation}
together with an integrality condition on the super de Rham (hence Chern) class of the 2-flux 
\[
  [F_2^s] \, \in  \,H^2_\mathrm{dR}(X)
\] 
and the derivative relation for the induced dilaton and dilatino fields on $X$
\[
  \dd \Phi 
  \, = \, 
  \partial_a\Phi\, 
  E_X^a \, + \, 2 
  \bracket({
    \overline{\Psi}_X 
    \,\Gamma^{1\!0-(m-1)}\,
    \lambda
  }) 
  \, ,
\] 
with the additional terms in {\color{olive} olive} arising due to the {\color{olive}winding} of the 11D super-coframe along the ``shrinking'' {\color{olive}M-theoretic circle}. Recalling then that on-shell 10D IIA SuGra on spacetime is precisely equivalent to the dimensional reduction of on-shell 11D supergravity along circle fibers, the miracle from Thm. \ref{11DSuGraAslS4BianchiOnSuperspace} immediately guarantees the following.
\begin{proposition}
[{\cite[Thm. 3.4]{GiotopoulosSati2026}}]
\label{IIASuGraAsCyclS4BianchiOnSuperspace}
On a super-torsion-free super-spacetime $X^{1,9\vert\mathbf{16}\oplus\overline{\mathbf{16}}}$ \textup{(\cref{SuperSpacetime})}, super-flux densities \cref{CurvedNS/RRCocycles}
satisfy the cyc$(\mathfrak{l}S^4)$-Bianchi identities \eqref{cycS4ClosedForms} iff the equations of motion of 10D IIA supergravity hold rheonomically on super-spacetime.
\end{proposition}

%%%%%%%%%%%%%%%%%
\subsubsection
{9D Super-space SuGra via $\mathrm{tor}(\mathfrak{l}S^4)$}
%%%%%%%%%%%%%
The rational toroidification of the real Whitehead $L_\infty$-algebra of the 4-sphere is given by (cf. 
%Ex 2.26 
\parencites[\S 2]{SatiVoronov2025}[Ex. 2.54]{GSS25-TD}):
\begin{equation}
\label{torS4}
  \mathrm{CE}\big(
    \mathrm{tor}^2(\mathfrak{l}S^4)
  \big)
  \;\;
  \simeq
  \;\;
  \mathbb{R}_\dd
  \left[
  \def\arraystretch{1.3}
  \begin{array}{c}
    \dir{2}{\omega}_2
    \\
    \dir{1}{\omega}_2
    \\
    g_4
    \\
    \dir{2}{\mathrm{s}}
    g_4
    \\
    \dir{1}{\mathrm{s}}g_4
    \\
    \dir{2}{\mathrm{s}}
    \dir{1}{\mathrm{s}}g_4
    \\
    g_7
    \\
    \dir{2}{\mathrm{s}}g_7
    \\
    \dir{1}{\mathrm{s}}g_7
    \\
    \dir{2}{\mathrm{s}}
    \dir{1}{\mathrm{s}}g_7
  \end{array}
  \right]
  \Big/
  \left(\!
  \def\arraystretch{1.3}
  \begin{array}{ccc}
    \mathrm{d}
    \,
    \dir{2}{\omega}_2
    &=&
    0
    \\
    \mathrm{d}
    \,
    \dir{1}{\omega}_2
    &=&
    0
    \\
    \mathrm{d}
    \,
    g_4
    &=&
    \dir{1}{\omega}_2
    \,
    \dir{1}{\mathrm{s}}g_4
    +
    \dir{2}{\omega}_2
    \,
    \dir{2}{\mathrm{s}}g_4
    \\
    \mathrm{d}
    \,
    \dir{2}{\mathrm{s}}g_4
    &=&
    -
    \dir{1}{\omega}_2
    \,
    \dir{2}{\mathrm{s}}
    \dir{1}{\mathrm{s}}
    g_4
    \\
    \mathrm{d}
    \,
    \dir{1}{\mathrm{s}}g_4
    &=&
    \dir{2}{\omega}_2\,
    \dir{2}{\mathrm{s}}
    \dir{1}{\mathrm{s}}g_4
    \\
    \mathrm{d}
    \,
    \dir{2}{\mathrm{s}}
    \dir{1}{\mathrm{s}}g_4
    &=&
    0
    \\
    \mathrm{d}
    \,
    g_7
    &=&
    \tfrac{1}{2}
    g_4 g_4
    +
    \dir{1}{\omega}_2
    \,
    \dir{1}{\mathrm{s}}g_7
    +
    \dir{2}{\omega}_2
    \,
    \dir{2}{\mathrm{s}}g_7
    \\
    \mathrm{d}\,
    \dir{2}{\mathrm{s}}
    g_7
    &=&
    -g_4 \, 
      \dir{2}{\mathrm{s}}g_4
    -\dir{1}{\omega}_2\, 
     \dir{2}{\mathrm{s}}
     \dir{1}{\mathrm{s}}g_7
    \\
    \mathrm{d}
    \,\dir{1}{\mathrm{s}}g_7
    &=&
    -g_4
    \,
    \dir{1}{\mathrm{s}}g_4
    +
    \dir{2}{\omega}_2
    \,
    \dir{2}{\mathrm{s}}
    \dir{1}{\mathrm{s}}g_7
    \\
    \mathrm{d}
    \,
    \dir{2}{\mathrm{s}}
    \dir{1}{\mathrm{s}}
    g_7
    &=&
    (\dir{2}{\mathrm{s}}g_4)
    (\dir{1}{\mathrm{s}}g_4)
    +
    g_4
    \,
    \dir{2}{\mathrm{s}}
    \dir{1}{\mathrm{s}}g_4
  \end{array}
  \!\! \right),
\end{equation}
which is, again, seen to be the classifying $L_\infty$-algebra for the gauge sector of $\mathcal{N}=2$ 9D SuGra \eqref{torS4ClosedForms}.

Forgoing the tangent-space-wise motivation of the corresponding flat supercocycles on $\FR^{1,8\vert \mathbf{16}\oplus \mathbf{16}}$, which holds along the lines of \eqref{Flat11Dcocycles} and \eqref{FlatNSRRcocycles}, we apply directly the $k=2$ case of Prop. \ref{NonTrivialTorusBundleToroidification/Oxidation} to the $\mathrm{U}(1)^2$-symmetric super-fluxes \eqref{SuperFluxFormsIn11D} of 11D supergravity on $\mathrm{U}(1)^2$-principal bundles $Y^{11}\rightarrow X^{9}$.\footnote{Or, equivalently, the $k=1$ case to $S^1$-symmetric IIA super-fluxes \eqref{CurvedNS/RRCocycles}. In doing so, however, one must remember that the 2-flux $F_2^s$ from \eqref{CurvedNS/RRCocycles} actually arises as the curvature of a connection from 11D, now additionally symmetric along the second $S^1$-fiber of the total torus in 11D, hence with no legs along the corresponding $\mathrm{U}(1)$-connection; in other words, that it is basic with respect to the 9D base super-spacetime.} This yields the following combination of {\color{darkorange}\textit{flat supersymmetric}} cocycles with \textit{curved bosonic} duality-symmetric fluxes of 9D supergravity:
\begin{equation}
\label
{Curved9DCocycles}
  \begin{alignedat}{4}
    (F^s_2)^I
    :=
    &
    \tfrac{1}{2}
    \bracket({F_2})^I_{a_1 a_2} 
    \, 
    E_X^{a_1} E_X^{a_2} 
    \\
    &
    + 
    \mathcolor{olive}{%
 2(\Phi^{-1})^I{}_n \delta^{nm} (\Phi^{-1})^J{}_m  (\overline{\Psi}_X \Gamma_a \lambda_J)\, E^a_X }
    && 
    + 
    \mathcolor{olive}{ 
     (\Phi^{-1})^I{}_m
    }  
    \mathcolor{darkorange}{  
      \bracket({
        \overline{\Psi}_X \,\Gamma^{\hat{m}}\, \Psi_X
      })
    }
    \\[3pt]
    \dir{12}{F^s_2}
    :=
    &
    \tfrac{1}{2}
    \dir{12}{(F_2)}_{a_1 a_2} 
    \, 
    E_X^{a_1} E_X^{a_2} 
    \\
    & + 
    \mathcolor{olive}{%
 (\overline{\lambda}_1 \Gamma_{ab} \lambda_2) \, E^a_X E^b_X }
    && 
    + 
    \mathcolor{olive}{ 
     \tfrac{1}{2} \epsilon_{mn} \Phi^m{}_1 \Phi^n{}_2
    }  
    \mathcolor{darkorange}{  
      \bracket({
        \overline{\Psi}_X \,\Gamma^{1\!0\, 9}\, \Psi_X
      })
    }
    \\
    & -{\color{olive}  \Phi^m{}_1 (\overline{\Psi}_X \Gamma_{a \hat{m}} \lambda_2) E_X^a }  && 
    \\
    & +{\color{olive}  \Phi^m{}_2 (\overline{\psi}_X \Gamma_{a \hat{m}} \lambda_1) E_X^a}  && 
    \\[3pt]
    \dir{I}{H^s_3}
    :=
    &
    \tfrac{1}{3!}(\dir{I}{H_{3}})_{a_1 a_2 a_3} 
    \, E_X^{a_1}E_X^{a_2}E_X^{a_3} 
    \\
    &
    +  
    \mathcolor{olive}{ 
    \bracket({
      \overline{\Psi}_X
      \,\Gamma_{a_1 a_2}
      \lambda_I 
      })\,
    E_X^{a_1}  E_X^{a_2} } 
    &&
    -
    \mathcolor{olive}
    { \Phi^{m}{}_I } \, 
    \mathcolor{darkorange}{  
      \bracket({
        \overline{\Psi}_X\, \Gamma_{\hat{m}a}\, \Psi_X
      })
      \, E_X^a 
    }
    \\[3pt]
    F^s_4
    :=
    &
    \tfrac{1}{4!}
    \bracket({F_4})_{a_1 \cdots a_4} 
    \,
    E_X^{a_1}\!\cdots\!E_X^{a_4}  
    \\
    &
    &&
    +
    \mathcolor{darkorange}{ 
      \tfrac{1}{2}
      \bracket({
        \overline{\Psi}_X
        \,\Gamma_{a_1 a_2}\, 
        \Psi_X
      })
      \,  E_X^{a_1} E_X^{a_2} 
    }
    \\[3pt]
    \dir{12}{F^s_5}
    := 
    &
    \tfrac{1}{5!}
    \dir{12}{\bracket({F_5})}_{a_1 \cdots a_6}
    \,  E_X^{a_1}\cdots E_X^{a_6} 
    \\
    &
    +
    \mathcolor{olive}{
    \tfrac{2}{5!}
    \bracket({
      \overline{\lambda}_1
      \,\Gamma_{a_1 \cdots a_5}
      \lambda_2
    }) 
    \, 
    E_X^{a_1}\!\cdots\!E_X^{a_5} } 
    && \hspace{-1.5cm} -
    \mathcolor{olive}{
      \tfrac{\epsilon_{mn}}{3!} \Phi^m{}_1 \Phi^{m}{}_2 
    } \,  
    \mathcolor{darkorange}{
      \bracket({
        \overline{\Psi}_X
        \,\Gamma_{1\!0 \,9 a_1 \cdots a_3}\, 
        \Psi_X
      }) 
      \,  
      E_X^{a_1}\!\cdots\!E_X^{a_3}  
    }
    \\
    & - {\color{olive} \tfrac{2}{4!} \Phi^m{}_1 (\overline{\Psi}_X \Gamma_{a_1 \dots a_4 \hat{m}} \lambda_2) E_X^{a_1} \dots E_X^{a_4} }  &&
    \\
    &  + {\color{olive} \tfrac{2}{4!} \Phi^m{}_2 (\overline{\Psi}_X \Gamma_{a_1 \dots a_4 \hat{m}} \lambda_1) E_X^{a_1} \dots E_X^{a_4} } &&
    \\[3pt]
    \dir{I}{F^s_6}
    :=
    &
    \tfrac{1}{6!}
    \dir{I}{\bracket({F_6})}_{a_1 \cdots a_6}
    \,  E_X^{a_1}\cdots E_X^{a_6} 
    \\
    & +
    \mathcolor{olive}{
    \tfrac{2}{5!}
    \bracket({
      \overline{\Psi}_X
      \,\Gamma_{a_1 \cdots a_5}
      \lambda_I
    }) 
    \, 
    E_X^{a_1}\!\cdots\!E_X^{a_5} } 
    && \hspace{-1.5cm} -
    \mathcolor{olive}{
      \tfrac{1}{4!} \Phi^{m}{}_I
    } \,  
    \mathcolor{darkorange}{
      \bracket({
        \overline{\Psi}_X
        \,\Gamma_{\hat{m} \, a_1 \cdots a_4}\, 
        \Psi_X
      }) 
      \,  
      E_X^{a_1}\!\cdots\!E_X^{a_4}  
    }
    \\[3pt]
    H^s_7
    :=
    &
    \tfrac{1}{7!}
    \bracket({H_7})_{a_1 \cdots a_7}
    \, 
    E_X^{a_1} \cdots E_X^{a_7} 
    \\
    & &&
    + 
    \mathcolor{darkorange}{ 
    \tfrac{1}{5!}
    \bracket({
      \overline{\Psi}_X
      \,\Gamma_{a_1 \cdots a_5}\,
      \Psi_X
    }) 
    \, 
    E_X^{a_1} \cdots E_X^{a_5}  
    }
    \mathrlap{\,,}
  \end{alignedat}
\end{equation}
where $\hat{m}:= 10 -m +1$, together with integrality conditions on the de Rham classes of the super 2-fluxes 
\[
  \bracket\{{ 
    \bracket[{(F_2^s)^1}]
    ,\, 
    \bracket[{(F_2^s)^2}] 
  }\} 
    \, \subset \,
  H^2_\mathrm{dR}\bracket({
    X})
  \mathrlap{\,,}
\] 
and the derivative relation for the induced dilaton and dilatino fields on $X$
\[
  \dd^{\hat{\Omega}} \Phi^m 
    \, = \, 
  \partial_a\Phi^m\, E_X^a 
    \, + \, 
  2 \bracket({
    \overline{\Psi}_X 
    \,\Gamma^{1\!0-(m-1)}\,
    \lambda
  }) 
  \mathrlap{\,,}
\] 
with the additional terms in {\color{olive} olive} arising due to the {\color{olive}winding} of the 11D super-coframe along the ``shrinking'' {\color{olive}M-theoretic torus}.

Since on-shell  $\mathcal{N}=2$ 9D SuGra is precisely equivalent to 11D SuGra compactified along 2-torus fibers, and equivalently, to IIA 10D SuGra along circle fibers, the miracle of Thm. \ref{11DSuGraAslS4BianchiOnSuperspace} further implies that:

\begin{proposition}
\label
{N29DSuGraAsTorlS4BianchiOnSuperspace}
On a super-torsion-free super-spacetime $X^{1,8\vert\mathbf{16}\oplus\mathbf{16}}$ \textup{(\cref{SuperSpacetime})}, super-flux densities \eqref{Curved9DCocycles}
satisfy the $\mathrm{tor}(\mathfrak{l}S^4)$-Bianchi identities \eqref{torS4ClosedForms} iff the equations of motion of 9D $\mathcal{N}=2$ supergravity hold rheonomically on super-spacetime.
\end{proposition}
\begin{proof}
This follows analogously to the $k=1$ case from IIA supergravity, spelled out in detail in \cite[Thm. 3.4]{GiotopoulosSati2026}, whose (slight) generalization we briefly summarize here as follows:

Firstly, the form of the super-fluxes \eqref{Curved9DCocycles} is obtained, via a long but straightforward calculation, by expanding the explicit form \eqref{SuperFluxFormsIn11D} of the (symmetric) 11D super-fluxes $(G_4^s, G_7^s)$ in the (partial) Kaluza-Klein gauge, where $E= (E_X, \Phi^1{}_J E'^J + \Phi^2{}_I E'^I)$ and $\Psi= \Psi_X + \lambda_1 E'^1 + \lambda_2 E'^2$. Then, by the $k=2$ case of Prop. \ref{NonTrivialTorusBundleToroidification/Oxidation}, the super-fluxes \eqref{Curved9DCocycles} on $X^{9}$ satisfy the $\mathrm{tor}^k(\mathfrak{l}S^4)$-Bianchis if and only if the corresponding $\mathrm{U}(1)^2$-symmetric fluxes $(G_4^s, G_7^s)$ satisfy the $\mathfrak{l}S^4$ Bianchis on $Y^{11}$. But by Thm. \eqref{11DSuGraAslS4BianchiOnSuperspace}, the latter holds if and only if these rheonomically encode (symmetric) 11D supergravity solutions on $\bosonic{Y}\hookrightarrow Y$. Such symmetric solutions along toric fibers, however, correspond precisely to 9D SuGra solutions on the base $\bosonic{X}\hookrightarrow{X}$. Finally, the result follows by the bi-directionality of these implications and the commutativity of the diagram 
\begin{equation}
\label
{CommutingSquareOfBosonicInclusions}
  \begin{tikzcd}[row sep=10pt]
    \bosonic{Y}^{\, 11} 
      \ar[
        r, 
        hook, 
        "{
          \epsilon_Y^
            {\mathrlap{\rightsquigarrow}}
        }"
      ] 
      \ar[
        d,
        shorten=-2pt
      ]
    & 
    Y^{11}
      \ar[d]
    \\
    \bosonic{X}^{\, 9} \ar[r, hook, "{\epsilon_X^{\mathrlap{\rightsquigarrow}}}"] & X^9 
  \end{tikzcd}
\end{equation}
of bosonic- and super-spacetime bundles. 
\end{proof}
Given that the result from Prop. \ref{NonTrivialTorusBundleToroidification/Oxidation} holds for any $k\in \mathbb{N}$, this pattern may be obviously extended to produce superspace formulations of even lower dimensional supergravities. Of course, as can already be seen from the $k=2$ case \eqref{Curved9DCocycles}, the explicit formulas for the super-fluxes will become rather unwieldy in number and form and so we do not pursue to write these out in these notes. Nevertheless, we can immediately deduce the following important statement.

%%%%%%%%%%%%%%
\paragraph
{Conclusion} 
%%%%%%%%%%%%%%
Higher $k$-toroidal dimensional reductions of the 11D super-fluxes yields a compatible tree of superspace formulations for the resulting $(11-k)$-dimensional supergravities via (uniquely determined) duality-symmetric super-fluxes valued in the rational $k$-toroidification of the 4-sphere,
$\mathfrak{l}\mathrm{Tor}^k(S^4)$.

%%%%%%%%%%%%%%%
\subsection
{Branes}
\label
{Probes}
%%%%%%%%%%%%%%%

We briefly discuss how for M5-brane probes of 11D SuGra spacetimes there is an analogue of \cref{11DSuGraAslS4BianchiOnSuperspace}, where the dynamics is equivalent to a Bianchi identity imposed on the worldvolume super-flux over super-space.

%%%%%%%%%%%%%
\paragraph
{Black and Probe Branes in Supergravity}
%%%%%%%%%%%%%

The notion of ``branes'' (a famous riff on the word \emph{membranes}, cf. \cite{Musaev2026}) in higher-dimensional SuGra is the higher-dimensional generalization of (magnetic or electric) \emph{monopoles} (which are the ``0-branes'' of ordinary electromagnetism, cf. \cite[\S 2.1]{SS25-Flux}).
Here a \emph{$p$-brane} is a $p$-dimensional object, tracing out a $1+p$-dimensional trajectory (``worldvolume'') in spacetime.

Branes come in (at least) two flavors (cf. \parencites[\S 7]{Fre2013}[\S 2]{HSS2019}):
\begin{description}

\item[Black branes] are solutions to equations of motion of gravitational theories that are higher dimensional generalizations of charged black hole spacetimes:

Here the would-be locus of the brane's worldvolume is a ``spacetime singularity'' that is actually excised from the spacetime manifold, and whose presence is felt only through the gravitational and higher gauge field emanating from it. The corresponding gauge flux density (and often, but not necessarily, also the spacetime curvature) would diverge (be ``singular'') at the locus of a black brane's worldvolume.

\item[\textbf{Probe branes}] (or: \emph{sigma-model} branes) are maps from an abstract worldvolume manifold into a spacetime that are higher-dimensional generalizations of trajectories of elementary charged particles, propagating subject to the (gravitational and higher Lorentz) forces exerted by the ambient spacetime fields, but \emph{not backreacting} on these.

\end{description}

The physical intuition is that probe branes are idealized approximations to black branes in the limit in which their mass and charge is negligible in comparison to the background fields. 

Remarkably, this intuition is brought out in fine technical detail for those branes which are preserve ``half of the maximal supersymmetry'' of the theory and are ``asymptotically Minkowskian'': For these \emph{$\sfrac{1}{2}$-BPS branes} one finds that their structure --- their worldvolume super-dimension and type of charges carried --- matches exactly between admissible black and probe brane species, even though their mathematical definitions are quite distinct.

This match is so ingrained into terminology that physics authors often do not distinguish between black and probe (and other incarnations of) branes. 

In this sense, 11D SuGra famously admits/contains:
\begin{enumerate}
\item
two species of charged branes (higher dimensional analogs of electric and magnetic monopoles, respectively, sensitive to IR-completion via flux-quantization):
\begin{enumerate}
\item 
2-branes: \emph{M2-branes} (cf. \cite[\S 2]{Duff1999World}), 

\item 5-branes: \emph{M5-branes} (cf. \cite[\S 3]{Duff1999World}), 
\end{enumerate}
\item 
two species of non-charged purely gravitational branes (cf. \cite{Figueroa-OFarrillSimon2002}):
\begin{enumerate}
\item 0-branes: \emph{M-waves} (cf. \cite{Townsend1997}) whose corresponding probe brane is the \emph{super-particle} in 11D,
\item 6-branes: \emph{KK-monopoles} (cf. \cite{HanKoh1985}), which are however not asymptotically Minkowskiand and do not have corresponding probe brane incarnation.
\end{enumerate}
\end{enumerate}
We are next concerned with the probe brane incarnation of the M5-brane. (For super-space construction of their black brane cousins see instead \cite{GSS25-Embedding}.)

%%%%%%%%%%%%%%%%%
\subsubsection
{Super-Embeddings}
%%%%%%%%%%%%%%%%%

The super-geometry that geometrizes the intricacies of 11D SuGra on super-spacetime (\cref{Equations}) turns out to also naturally know about its probe branes, via their super-worldvolume \emph{super-embeddings} (really: super-immersions).

%%%%%%%%%%%%%
\paragraph
{Literature on Super-Embedding Formalism}
%%%%%%%%%%%%%
The \emph{super-embedding formalism} for super $p$-brane probes (sigma-models) is due to \parencites{HoweSezgin1997-Superbranes}{HoweSezginWest1998}{HoweRaetzelSezgin1998}{Sorokin2000} following \cite{GatesNishino1986,BrooksMuhammaGates1986}.
Review includes \cite{BandosPastiSorokinTonin1998,Bandos1999,Sorokin2001,Bandos2011,BandosSorokin2023}.
Analysis specifically for the M5-brane originates with 
\cite{HoweSezgin1997,HoweSezginWest1997}, further discussed in \cite[\S 5.2]{Sorokin2000}. We follow \cite{GSS24-M5,GSS25-Embedding,BaSS26-MString}.

%%%%%%%%%%%%
\paragraph
{Idea of Super $p$-Brane Super-Embeddings}
%%%%%%%%%%%%

Historically (cf. \cite{Sorokin2001}), the \emph{superstring} was famously first found in its Ramond-Neveu-Schwarz (RNS) formulation. While this has manifest super-symmetry on the worldsheet, the super-symmetry of its spacetime spectrum is not manifest and came as a surprise (which fueled the early interest in string theory). Later, the Green-Schwarz (GS) superstring model made spacetime super-symmetry manifest, but now at the cost of losing manifest worldsheet super-symmetry and introducing a somewhat more unwieldy ``$\kappa$-symmetry'' constraint.
The \emph{superembedding} formulation arose as a  synthesis of the lessons learned from the RNS and the GS models: 

Here the string worldsheet is modeled as a supermanifold of superdimension $(1,1\vert 8 \cdot \mathbf{2})$, as is the target spacetime (for type IIA, say), of super-dimension $(1,9\vert \mathbf{16} \oplus \overline{\mathbf{16}})$, and the world- 
volume fields are understood as constituting supergeometric embeddings (super-immersions, really, cf. \parencites[p. 14]{GSS24-M5}) of supermanifolds:
\begin{equation}
  \substack{
    \text{\color{gray}super-worldvolume}
    \\
    \text{\color{gray} of super-string}
  }
  \;\;
  \inlinetikzcd{
    \Sigma^{1,1\vert 8 \cdot \mathbf{2}}
    \ar[
      rr,
      "{
        \text{\color{gray}super-embedding}
      }",
      "{ \Phi }"'
    ]
    \&\phantom{----}\&
    X^{1,9\vert \mathbf{16} \oplus \overline{\mathbf{16}}}
  \;\;
  \substack{
      \text{\color{gray}super-spacetime}
      \\
      \text{\color{gray} in 10D IIA SuGra}
  }
    \mathrlap{\,.}
  }
\end{equation}
Remarkably, the equations of motion of the superstring then turn out to be equivalent to the requirement that this super-embedding is ``tangent space-wise $\sfrac{1}{2}$-BPS'' in a suitable sense (stated in a moment, originally known as the \emph{geometrodynamical condition}, then as the \emph{basic embedding condition} and eventually as the \emph{superembedding condition} \cite{HoweSezginWest1998,Sorokin2000}, overview in \cite[Rem. 2.24]{GSS24-M5}). With this, super-symmetry of the model is manifest both on the worldsheet and on the target space, and the GS ``$\kappa$-symmetry'' is geometrized as the evident super-diffeomorphism invariance of the formulation \cite[\S4.3]{Sorokin2000}.

Analogous statements may be made for essentially all the super $p$-branes one encounters in string/M-theory. This is particularly striking for those branes that carry higher flux densities (``tensor fields'') on their worldvolume, such as the M5-brane with its $H_3$ flux, to which we turn in \cref{M5braneProbes}.

%%%%%%%%%%%%
\paragraph
{The Superembedding Condition}
%%%%%%%%%%%%

More in detail, the strong form of the superembedding condition that is relevant here (cf. \cite[Rem. 3.20]{GSS24-M5} and \cref{M5SuperembeddingComparisonToLiterature} below) says that the superembedding must be \emph{$\sfrac{1}{2}$-BPS} in the following sense \cite[Def. 2.19, Lem. 2.20]{GSS24-M5}: 

On  the super-vector space $\mathbb{R}^{1,d\vert \mathbf{N}}$ which is the local model of superspacetime $X^{1,d\vert \mathbf{N}}$, let 
\begin{equation}
\label
{ProjectionOperator}
  P
  :
  \inlinetikzcd{
    \mathbb{R}^{1,d\vert \mathbf{N}}
    \ar[
      rr,
      ->>
    ]
    \&\&
    \mathbb{R}^{1,p\vert \mathbf{n}}
    \ar[
      rr,
      hook
    ]
    \&\&
    \mathbb{R}^{1,d\vert \mathbf{N}}
  }
\end{equation}
be the super-projection operator onto the local model space $\mathbb{R}^{1,p\vert \mathbf{n}}$ of the super $p$-brane, with $n = N/2$ (the \emph{tangential} projection on the fixed space of a \emph{$p$-brane involution} \cite[Def. 4.4]{HSS2019}, exhibiting the local halving of the ``number of super-symmetries''). Write $P_{\!\!{}_\perp\!}$ for the orthogonal super-projector (the \emph{transversal} projection).

Then a super-immersion
\begin{equation}
\label
{TowardsHalfBPSSuperEmbedding}
  \inlinetikzcd{
    \Sigma^{1,p\vert \mathbf{n}}
    \ar[
      rr,
      "{ \Phi }"
    ]
    \&\&
    X^{1,d\vert \mathbf{N}}
  }
\end{equation}
is said to satisfy the \emph{$\sfrac{1}{2}$-BPS embedding condition} \cite[Def. 2.19, Lem. 2.20]{GSS24-M5} iff there exists an orthonormal super-coframe $\bracket({ \bracket({E^a})_{a=0}^d, \bracket({ \Psi^\alpha })_{\alpha = 1}^{N} })$ on $X^{1,d\vert \mathbf{N}}$ 
and a \emph{super-shear map}
\begin{equation}
\label
{TheShearMap}
  \mathrm{Sh}
  \in
  C^\infty\bracket({
    \widetilde{\Sigma};
    \mathrm{Hom}_{\mathbb{R}}\bracket({
      P\bracket({
        \mathbb{R}^{1,d\vert \mathbf{N}}
      }),
      P_{\!\!{}_\perp\!}\bracket({
        \mathbb{R}^{1,d\vert \mathbf{N}}
      })   
    })
  })\,,
\end{equation}
(defined on the open cover $\widetilde{\Sigma}$ of the worldvolume on which its super-coframe is defined, cf. \cite[\S 2.1]{GSS24-M5})
such that:
\begin{enumerate}
  \item
  The pullback of the tangential part of the target super-coframe is a super-coframe $(e,\psi)$ on the brane:
  \begin{equation}
  \label
  {TangentialSuperembeddingCondition}
    \Phi^\ast\bracket({
      P\bracket({
        \bracket({E^a})_{a=0}^d,
        \bracket({\Psi^\alpha})_{\alpha=1}^N
      })
    })
    =
    \bracket({
      \bracket({e^a})_{a=0}^{p},
      \bracket({ \psi^\alpha })_{\alpha=1}^{n}
    })\,.
  \end{equation}
  
  \item
  The pullback of the transversal part is expanded in this coframe by the given shear map \cref{TheShearMap}:
  \begin{equation}
  \label
  {TransversalSuperembeddingCondition}
    \begin{aligned}
    \Phi^\ast
    \bracket({
      P_{\!\!{}_\perp\!}\bracket({
        \bracket({E^a})_{a=0}^d,
        \bracket({\Psi^\alpha})_{\alpha=1}^N
      })
    })
    & = \;
    \mathrm{Sh} \cdot 
    \Phi^\ast
    \bracket({
      P\bracket({
        \bracket({E^a})_{a=0}^d,
        \bracket({\Psi^\alpha})_{\alpha=1}^N
      })
    })
    \\
    & 
    \underset{\mathclap{
      \text{\cref{TangentialSuperembeddingCondition}}
    }}{=} \;
    \mathrm{Sh} \cdot 
    \bracket({
        \bracket({e^a})_{a=0}^p,
        \bracket({\psi^\alpha})_{\alpha=1}^n
    })
    \end{aligned}
  \end{equation}
  and \emph{equivariantly so} with the transversal $\mathrm{Spin}(d-p)$-action.
\end{enumerate}

Remarkably, the transverse fermionic shear in \cref{TransversalSuperembeddingCondition} is, when it exists, the origin of nonvanishing worldvolume flux densities, in addition to the ``embedding fields'' that constitute the tangential component \cref{TangentialSuperembeddingCondition}. But the transversal equivariance condition on the shear means, by Schur's lemma, that it can only be non-vanishing if the $\mathrm{Spin}(d-p)$-representations $P(\mathbf{N})$ and $P_{\!\!{}_\perp\!}(\mathbf{N})$ share isomorphic irrep summands (cf. \cite[Rem. 2.2]{GSS24-M5}).

%%%%%%%%%%%%%%%%
\subsubsection
{M5-Brane Probes}
\label
{M5braneProbes}
%%%%%%%%%%%%%%%%

We focus on the higher gauge sector on M5-brane probes. This famously involves a 3-form flux that is often referred to as ``self-dual'' (cf. \parencites[p. 3]{Witten1997Fivebrane}). In reality, it is self-dual only in a subtle non-linear sense \parencites{HoweSezginWest1997}[Rem. 3.19]{GSS24-M5} or else asymptotically in the small field limit (linearizing the equations of motion, cf. \cite{AndrianopoliEtAl2022}).

%%%%%%%%%%%%%%
\paragraph
{The M5 Super-Worldvolume}
%%%%%%%%%%%%%

Instead of considering matrix representations of the 6d spinors, 
it is more efficient to algebraically 
carve them out  from the 11D spinors
\parencites[\S 3.2]{GSS24-M5} (cf. also \cite[\S A]{LambertPapageorgakis2010}):
\begin{equation}
  \label{TangentialTransversalDecomposition}
  \def\arraystretch{1.3}
  \def\tabcolsep{4pt}
  \begin{array}{cccccccccccl}
    &&
    \mathclap{
      \;\;\;\,
      \overbrace{
        \phantom{----.----}
      }^{
        \scalebox{.7}{
          tangential directions
        }
      }
    }
    &&&&&&
    \mathclap{
      \!\!
      \overbrace{
        \phantom{----.---}
      }^{
        \scalebox{.7}{
          transversal directions
        }
      }
    }
    \\[-10pt]
    0 
    &
    1 
    &
    2 
    &
    3 
    &
    4 
    &
    5 
    &
    5' 
    &
    6 
    &
    7 
    &
    8 
    &
    9 
    &
    \\
    \Gamma_0
    &
    \Gamma_1
    &
    \Gamma_2
    &
    \Gamma_3
    &
    \Gamma_4
    &
    \Gamma_5
    &
    \Gamma_{5'}
    &
    \Gamma_{6}
    &
    \Gamma_{7}
    &
    \Gamma_{8}
    &
    \Gamma_{9}
    &
    \in 
    \;
    \mathrm{Pin}^+(1,10)
    \;\subset\; 
    \mathrm{End}_{\mathbb{R}}(\mathbf{32})
    \\
    \gamma_0
    &
    \gamma_1
    &
    \gamma_2
    &
    \gamma_3
    &
    \gamma_4
    &
    \gamma_5
    & &&&&
    &
    \in
    \;
    \mathrm{Pin}^+(1,5)
    \;\;\,\subset\;
    \mathrm{End}_{\mathbb{R}}
    \bracket({
      2 
        \cdot 
      \mathbf{8}_+
      \oplus
      2 \cdot
      \mathbf{8}_-
    })
    \mathrlap{\, ,}
  \end{array}
\end{equation}
where
\begin{equation}
  P
  :=
  \tfrac{1}{2}
  \bracket({
    1 + 
    \Gamma_{012345}
  })
\end{equation}
serves as the 5-brane projection operator \cref{ProjectionOperator} 
and the tangent Clifford algebra generators are simply the tangential projections of the bulk Clifford algebra generators:
\begin{equation}
\label
{ProjectionOfCliffordElements}
\begin{aligned}
  \overline{P} \, \Gamma^a \, P
  & =
  \left\{\!
  \def\arraystretch{1.3}
  \begin{array}{ll}
    \gamma^a{}_{\vert_{2 \cdot \mathbf{8}_+}}
    &
    \mbox{for tangential $a$}
    \\
    0 & \mbox{for transversal $a$}
    \,,
  \end{array}
  \right.
  \\
  P \, \Gamma^a \, \overline{P}
  & =
  \left\{\!
  \def\arraystretch{1.3}
  \begin{array}{ll}
    \gamma^a{}_{\vert_{2 \cdot \mathbf{8}_-}}
    &
    \mbox{for tangential $a$}
    \\
    0 & \mbox{for transversal $a$}.
  \end{array}
  \right.
  \end{aligned}
\end{equation}

Hence the super-worldvolume of a $\sfrac{1}{2}$-BPS M5-brane has super-dimension $(1,5\vert 2 \cdot \mathbf{8}_+)$ \cite[(92)]{GSS24-M5}, commonly denoted: $D=6$, $\mathcal{N} = (2,0)$.

%%%%%%%%%%%%%
\paragraph
{The M5 Super-Embedding}
%%%%%%%%%%%%%

We are thus looking at a $\sfrac{1}{2}$-BPS super-embeddings \cref{TowardsHalfBPSSuperEmbedding} of the form
\begin{equation}
\label
{M5SuperEmbedding}
  \inlinetikzcd{
    \Sigma^{1,5\vert \mathbf{8}_+}
    \ar[rr, "{ \Phi }"]
    \&\&
    X^{1,10\vert \mathbf{32}}
  }.
\end{equation}

Remarkably,
the fermionic shear maps \cref{TheShearMap} 
for an M5 super-embedding
correspond exactly to self-dual super-forms $\tilde H_3$ on  $\Sigma^{1,5\vert 2 \cdot \mathbf{8}_+}$ (\cite[Lem. 3.16(iii)]{GSS24-M5}, cf. previous arguments via matrix representations in \parencites[(15)]{HoweSezgin1997}[(5.66)]{Sorokin2000}):
\begin{equation}
  \begin{tikzcd}[row sep=small, column sep=40pt]
    &[20pt]
    \hspace{-28pt}
    \mathrlap{
      \overbrace{
      2 \cdot \mathbf{8}_+
      \oplus
      2 \cdot \mathbf{8}_-
      }^{ 
        \mathbf{32} 
        \,\in\, 
        \mathrm{Spin}(1,10)
      }
    }
    \ar[
      d,
      ->>
    ]
    \\
    T \Sigma^{
     1,5\vert 2\cdot \mathbf{8}_+
    }
    \ar[
      ur,
      bend left=10,
      shorten >=15pt,
      "{
        \Phi^\ast \Psi
        =
        \scaledbracket({
          \psi, 
          \slashed{\tilde H}_3
          \psi
        })
      }"{
        sloped,
        pos=.3
      }
    ]
    \ar[
      r,
      "{
        \psi
      }"
    ]
    &
    2 \cdot \mathbf{8}_+
  \end{tikzcd}
  \hspace{2cm}
  \text{for}
  \;\;\;
  \begin{aligned}
  \bracket({
    \tilde H_3
  })_{a_1 a_2 a_3}
  & 
  =
  \tfrac{1}{3!}
  \epsilon_{
    a_1 a_2 a_3
    \,
    b_1 b_2 b_3
  }
  \bracket({
    \tilde H_3
  })^{b_1 b_2 b_3}
  \\
  \slashed{\tilde H}_3
  & \defneq
  \tfrac{1}{3!}
  \bracket({
    \tilde H_3
  })_{a_1 a_2 a_3}
  \gamma^{a_1 a_2 a_3}.
  \end{aligned}
\end{equation}

%%%%%%%%%%%%%
\paragraph
{The M5 Worldvolume Flux}
%%%%%%%%%%%%%

In analogy with the bulk super-flux forms from \cref{OnSuperFluxFormsIn11D,CurvedNS/RRCocycles}, we say that \emph{M5 worldvolume super-flux forms} are of this form:
\begin{equation}
\label
{M5WorldvolumeSuperFLux}
  H^s_3
  :=
  \bracket({
    H_3
  })_{a_1 a_2 a_3}
  e^{a_1}
  e^{a_2}
  e^{a_3}
  \;\;
  \in
  \Omega^3_{\mathrm{dR}}\bracket({
    \Sigma^{
      1,5\vert 2 \cdot \mathbf{8}_+
    }
  })
\end{equation}
(Hence in this case there is vanishing fermionic component, $H^0_3 = 0$.)

The brane flux analog of \cref{11DSuGraAslS4BianchiOnSuperspace} is now this statement:
\begin{proposition}[{\cite[Prop. 3.18]{GSS24-M5}}]
\label[proposition]
{M5FluxEquations}
On a $\sfrac{1}{2}$-BPS M5-embedding \cref{M5SuperEmbedding} into a super-space solution of 11D SuGra \textup{(\cref{11DSuGraAslS4BianchiOnSuperspace})}
the worldvolume super-flux density \cref{M5WorldvolumeSuperFLux} satisfies the $\mathfrak{l}_{_{S^4}}S^7$-Bianchi identity \cref{BianchiForM5WorldvolumeFlux} if the following non-linearly self-dual Maxwell equations hold rheonomically on the super-worldvolume: 
\begin{equation}
  \begin{aligned}
    &
    \Big\{
    \begin{alignedat}{2}
    \mathrm{d}
    \,
    H^s_3 
    &
    =
    \Phi^\ast G^s_4
    \;\;\;\;\;\;\;\;\;\;\;\;\;&&
    \textup{\color{darkblue}%
      ($\mathfrak{l}_{_{S^4}}S^7$-Bianchi)%
    }
    \end{alignedat}
    \\[6pt]
    \label
    {SuperspaceEOMForBField}
    \Leftrightarrow
    &
    \left\{
    \begin{alignedat}{2}
      \tfrac{1}{3!}
      \nabla_{[a_1}
      \bracket({
        H_3
      })_{a_2 a_3 a_4]}
      &
      =
      \tfrac{1}{4!}
      \bracket({
       \Phi^\ast G_4
      })_{a_1 \cdots a_4}
      &\;\;&
      \textup{\color{darkblue}%
        (Maxwell)%
      }
      \\[5pt]
      \bracket({
        H_3
      })_{abc}
      & =
      \tfrac
        {-4}
        {
          \rule{0pt}{8pt}
          1 
            - 
          \sfrac{2}{3}
          \,
          \mathrm{tr}
          ({
            \tilde H_3^2
            \cdot
            \tilde H_3^2
          })
        }
      \big(\,{
        \grayunderbrace{
          \rule[-4pt]{0pt}{0pt}%
          \tilde H_{abc}
        }{ 
          \substack{\mathclap{
            \textup{self-dual}
          }} 
        }
        + 
        \grayunderbrace{
          2\bracket({
            \tilde H_3^2
          })_a^{a'}
          \tilde H_{a'bc}
        }{ 
          \substack{
            \textup{anti self-dual}
          } 
        }
      }\, \big)
      &&
      \textup{\color{darkblue}%
        (``self-duality'')%
      }
      \\
      \mathrlap{
        \textup{
          rheonomically
        }
      }
    \end{alignedat}
    \right.
  \end{aligned}
\end{equation}
\end{proposition}

\begin{remark}
[Comparison to the literature]
\label[remark]
{M5SuperembeddingComparisonToLiterature}
\begin{itemize}[itemsep=-2pt]
    \item[\textbf (i)]
  This kind of statement goes back to \cite{HoweSezgin1997}, further discussed in \cite[\S 5.2]{Sorokin2000}, with special emphasis on the subtle nature of the ``self-duality'' condition in \cite{HoweSezginWest1997} (cf. \cite[Rem. 3.19]{GSS24-M5}).
     \item[\textbf (ii)] Our statement above assumes (cf. \cite[Rem. 3.20]{GSS24-M5}) the $\sfrac{1}{2}$-BPS embedding condition \cref{TransversalSuperembeddingCondition}, which is stronger than the super-embedding condition considered there (as it forces the bosonic shear component $\tau$ to vanish, cf. \cite[(108) \& Rem. 3.17]{GSS24-M5}). 
  
     \item[\textbf (iii)] The concretely known examples of M5 super-embeddings do satisfy this stronger condition \parencites[Ex. 3.14]{GSS24-M5}{GSS25-Embedding}. But we are not claiming that the weaker condition should be disregarded, just that it leads (cf. \cite[Rem. 3.20]{GSS24-M5}) to a fermionic correction to the twisted Maxwell equation \cref{SuperspaceEOMForBField} already on superspace.%
     \footnote{
       This in contrast to the situation for 11D SuGra, where the equations of motion on superspace have the plain bosonic form \cref{11DSugraSuperspaceEoM}, with fermionic corrections appearing only in the translation to bosonic spacetime as discussed in \cref{ComputingFermionicSourceTerms}.
     }
  \end{itemize}
\end{remark}

%%%%%%%%%%%%%%
\paragraph
{Conclusion}
%%%%%%%%%%%%%%
  In summary, we have arrived at the remarkable statement that 11D SuGra with M5-probes is put on-shell by asking:
  \begin{enumerate}
    \item 
    the super-spacetime $X^{1,10\vert \mathbf{32}}$ to be super-torsion free,
    \item 
    the super-worldvolume
    $\inlinetikzcd{ \Sigma^{1,5\vert 2\cdot \mathbf{8}_+} \ar[r, "{ \Phi }"] \& X^{1,10\vert \mathbf{32}}}$ to be a $\sfrac{1}{2}$-BPS super-embedding,
    \item 
    the 
    super-flux densities \cref{SuperFluxFormsIn11D,M5WorldvolumeSuperFLux} to satisfy the
    $\big({\inlinetikzcd{\mathfrak{l}_{_{S^4}}S^7 \ar[r, "{ \mathfrak{l}h_{\mathbb{H}} }"] \& \mathfrak{l}S^4} }\big)$-Bianchi identity \cref{lS7RelativeToS4}:
    \begin{equation}
    \label
    {lhHBianchisOnSuperFluxes}
      \left(
      \begin{aligned}
        & H^s_3,
        \\
        & 
          G^s_4, G^s_7
      \end{aligned}
      \right)
      \in
      \Omega^1_{\mathrm{cl}}
      \bracket({
        \Phi;
        \mathfrak{l}h_{\mathbb{H}}
      })
      \mathrlap{\,.}
    \end{equation}
  \end{enumerate}
This result condenses the entire dynamics essentially to a constraint of purely cohomological nature \cref{lhHBianchisOnSuperFluxes}. As such, this formulation opens the door to the quantization of these flux densities in nonabelian cohomology; this is the topic of \cref{Completions}.

%%%%%%%%%%%%%%

%%%%%%%%%%%%%%
%%%%%%%%%%%%%%
\section
{Completions}
\label
{Completions}
%%%%%%%%%%%%%%

We have established in \cref{Supergravity} that, over super-torsion free super-spacetimes $X^{1,d\vert \mathbf{N}}$:
\begin{equation}
  \begin{tikzcd}[
    ampersand replacement=\&,
    sep=0pt,
    /tikz/column 3/.append style={anchor=base west}
  ]
    \left(
    \begin{aligned}
      & G^s_4
      \\
      & G^s_7
    \end{aligned}
    \right)
    \&\in\&
    \Omega^1_{\mathrm{cl}}\bracket({
      X^{1,10\vert\mathbf{32}};
      \mathfrak{l}S^4
    })
    \& \Leftrightarrow \&
    \text{11D SuGra}
    \\
    \left(
    \begin{aligned}
      & F^s_2
      \\
      & F^s_4
      \\
      & F^s_6
      \\
      & H^s_3
      \\
      & H^s_7
    \end{aligned}
    \right)
    \& \in \&
    \Omega^1_{\mathrm{cl}}\bracket({
      X^{1,9
        \vert
        \mathbf{16} \oplus \overline{16} };
      \mathfrak{l}
      \mathrm{Cyc}(S^4)
    })
    \& \Leftrightarrow \&
    \text{10D IIA SuGra}.
  \end{tikzcd}
\end{equation}

Or rather, this is \emph{1850s-style} SuGra, at the level at which Maxwell understood electromagnetism: where the (higher) gauge field content is reflected only in its flux densities (the higher analogs of the Faraday tensor).

Our aim is to promote this to \emph{1930s-style} SuGra, where the higher gauge fields are understood to be global/IR completions of local gauge potential data via gluing data (cf. \cref{GluingDataSchematics}) subject to flux quantization laws, in a generalization of how Dirac globally completed the electromagnetic field in order to identify the topological (monopole) charges that source it.

Dedicated lecture notes on the general procedure of global IR completion of higher Maxwell-type gauge theories are in the companion article \cite{SS26-HigherGauge}, quick exposition of the required higher geometry is in \cite{Schreiber2025}; here we gloss over mathematical details of \emph{higher stacks} (for which see \parencites{FSS15-Stacky}[\S 1]{FSS23-Char}) and instead try to connect more to the physics literature.

%%%%%%%%%%%%%%%
\paragraph
{Gauged Sets}
%%%%%%%%%%%%%%%

We need to properly deal with the fact that the higher gauge fields in higher-dimensional SuGra are subject to the \emph{higher gauge principle}: Field configurations are only defined up to gauge equivalences, which themselves are only defined up to gauge-of-gauge-equivalences, which themselves are only defined up to ever higher gauge equivalences. 

The gauge principle entails that collections of gauge fields \emph{do not form an ordinary set} --- because in an ordinary set any pair of elements is either equal or distinguishable. Rather, collections of gauge fields form ``gauged sets'' $\mathcal{G}$, where any pair of their elements $e_1, e_2 \in \mathcal{G}$ comes with a set $\mathcal{G}(e_1, e_2)$ 
of gauge transformations $t$ between them (empty if they are not gauge equivalent). So a generic gauged set schematically looks as follows:
\begin{equation}
\label
{GaugedSetSchematics}
  \mathcal{G}
  =
  \left\{
  \begin{tikzcd}[
    column sep=20pt
  ]
    &
    s_2
    \ar[
      dr,
      "{ 
        t' 
      }"{description}
    ]
    \\
    s_1
    \ar[
      ur,
      "{ t }"{description}
    ]
    \ar[
      ur,
      <-,
      bend left=40,
      "{ 
        t^{-1} 
      }"{description}
    ]
    \ar[
      rr,
      "{
        t' \circ t
      }"{description}
    ]
    &&
    s_3
  \end{tikzcd}
  \;,\,
  \begin{tikzcd}
    s_4
    \ar[
      out=55,
      in=180-55,
      looseness=5,
      shift left=2pt,
      "{ \,\mathclap{g}\, }"{description}
    ]
  \end{tikzcd}
  \,,\,
  \cdots
  \right\}.
\end{equation}

Since these gauge transformations have an associative composition rule $(-)\circ(-)$, include all identities and all have inverses $(-)^{-1}$, gauged sets are exactly small \emph{groupoids} \cref{GroupoidOfYangMillsFields}.

%%%%%%%%%%%%%%
\paragraph
{Higher Gauged Sets}
%%%%%%%%%%%%%%
If there are gauge-of-gauge transformations, then this story repeats: Instead of a set of gauge transformations between any pair of elements $s, s'$, we have a gauged set $\mathcal{G}(s,s')$ of gauge-of-gauge transformations. An example of a gauged set of gauge transformations between $s,s'$ may look as follows, schematically:
\begin{equation}
  \mathcal{G}(s, s')
  =
  \left\{
  \begin{tikzcd}
    s
    \ar[
      rr,
      bend left=60,
      "{ t }"{description, name=s}
    ]
    \ar[
      rr,
      bend right=60,
      "{ t'' }"{description, name=t}
    ]
    \ar[
      from=s,
      to=t,
      bend left=20,
      Rightarrow
    ]
    \ar[
      rr,
      crossing over,
      "{ 
        t'
      }"{description, name=i, pos=.27}
    ]
    \ar[
      from=s,
      to=i,
      shorten=2pt,
      Rightarrow
    ]
    \ar[
      from=i,
      to=t,
      shorten=2pt,
      Rightarrow
    ]
    &&
    s'
  \end{tikzcd}
  \right\}
  \mathrlap{\,.}
\end{equation}
But this means that the elements $s$ themselves now form a ``second-order gauged set'', known as a \emph{2-groupoid}.

Proceeding in this manner, sets involving higher-order gauge transformations up to order $n$ are known as \emph{$n$-groupoids} and generally as \emph{$\infty$-groupoids} if there is no bound on the order of higher gauge transformations. The neatest mathematical model for such higher gauged sets is known as \emph{simplicial sets} which are \emph{Kan fibrant}, for short: \emph{Kan simplicial sets} or just \emph{Kan complexes}.

Concretely:
\begin{definition}
[Higher gauged sets]
\label[definition]
{HigherGaugedSets}
 A \emph{simplicial set} $\mathcal{G}$ is:
  \begin{enumerate}

  \item
  A sequence $\bracket({\mathcal{G}_n})_{n \in \mathbb{N}}$ of sets of order=$n$ gauge transformations going between $(n-1)$-transformations arranged in the shape
  of an $n$-dimensional triangle,
  called an \emph{$n$-simplex} $\Delta^n$:
  \[
    \def\arraycolsep{10pt}
    \begin{array}{c|c|c|c|c}
    \begin{tikzcd}[
      row sep=25pt,
      column sep=13
    ]
      \scalebox{.7}{
        \color{gray}
        $(\Delta^0)$
      }
    s
   \end{tikzcd}
   &
   \begin{tikzcd}[
     row sep=25pt,
     column sep=13
   ]
      \scalebox{.7}{
        \color{gray}
        $(\Delta^1)$
      }
    s
    \ar[
      rr,
      "{ g }"
    ]
    &&
    s'
   \end{tikzcd}
   &
   \begin{tikzcd}[
     row sep=25pt,
     column sep=13
   ]
    &
    |[alias=two]|
    s'
    \ar[
      dr,
      "{ h }"
    ]
    \\
    s
    \ar[
      ur,
      "{ 
          \scalebox{.9}{
            \color{gray}
            $(\Delta^2)$
            \;\;\;\,
          }
        g 
      }"
    ]
    \ar[
      rr,
      "{ h \circ g }"{swap},
      "{\ }"{name=onethree}
    ]
    &&
    s''
    \ar[
      from=two,
      to=onethree,
      Rightarrow,
      shorten <=5pt
    ]
  \end{tikzcd}  
  &
  \begin{tikzcd}
    &[-10pt]&
    &[-50pt] 
    |[alias=two]|
    s'
    \ar[
      ddr,
      "{\ }"{name=twofour, swap}
    ]
    &[-5pt]
    \\
    \\
    |[alias=one]|
    s
    \ar[
      rrrr,
      "{\ }"{name=onefour},
    ]
    \ar[
      rrrr,
      shift left=7pt,
      shorten <=35pt,
      shorten >= 21pt,
      Rightarrow,
      nfold=3
    ]
    \ar[
      drr,
      "{\  }"{name=onethree}
    ]
    \ar[
      uurrr,
      "{
        \mathllap{
          \scalebox{.9}{
            \color{gray}
            $(\Delta^3)$
          }
        }
      }"{yshift=15pt}
    ]
    &&&&
    |[alias=four]|
    s'''
    \mathrlap{\,,}
    \\[-3pt]
    &&
    |[alias=three]|
    s''
    \ar[
      urr,
      "{\ }"{name=one}
    ]
    \ar[
      from=two, 
      to=onefour,
      Rightarrow, 
      crossing over,
      shorten=5pt
    ]
    \ar[
      from=two, 
      to=onethree,
      Rightarrow, 
      shorten <=10pt,
      shorten >=1pt,
      crossing over
    ]
    \ar[
      from=three, 
      to=onefour,
      Rightarrow, 
      crossing over,
      shorten=5pt
    ]
    \ar[
      from=two,
      to=three,
      crossing over,
      "{\ }"{name=twothree}
    ]
    \ar[
      from=three, 
      to=twofour,
      shorten <= 10pt,
      shorten >= 2pt,
      Rightarrow, 
      crossing over,
      end anchor={[xshift=5pt]}
    ]
  \end{tikzcd}
  &
  \cdots
  \\
  \substack{
    \text{elements}
  }
  &
  \substack{\color{darkblue} 
    \text{gauge}
    \\
    \text{transformations}
  }
  &
  \substack{\color{darkgreen} 
    \text{gauge-of-gauge}
    \\
    \text{transformations}
  }
  &
  \substack{\color{darkorange} 
    \text{gauge-of-gauge-of-gauge}
    \\
    \text{transformations}
  }
  &
  \end{array}
  \]

    \item 
    \emph{Face maps}
    $\inlinetikzcd{ \mathcal{G}_{n+1} \ar[r, "{ \partial_i }"] \& \mathcal{G}_n }$ sending each such $(n+1)$-simplex to its $i$th face (for $0 \leq i \leq n+1$):
    \[
      \begin{tikzcd}[column sep=large]
        \cdots
        \ar[
          r,
          shift right=12pt,
          dotted
        ]
        \ar[
          r,
          shift right=4pt,
          dotted
        ]
        \ar[
          r,
          shift left=4pt,
          dotted
        ]
        \ar[
          r,
          shift left=12pt,
          dotted
        ]
        &
        \mathcal{G}_2
        \ar[
          r,
          shift right=12pt,
          "{ 
            \partial_0 
          }"{description, pos=.45}
        ]
        \ar[
          r,
          "{ 
            \partial_1 
          }"{description, pos=.45}
        ]
        \ar[
          r,
          shift left=12pt,
          "{ 
            \partial_2 
          }"{description, pos=.45}
        ]
        &
        \mathcal{G}_1
        \ar[
          r,
          shift right=6pt,
          "{ 
            \partial_0 
          }"{description, pos=.45}
        ]
        \ar[
          r,
          shift left=6pt,
          "{ 
            \partial_1 
          }"{description, pos=.45}
        ]
        &
        \mathcal{G}_0
      \end{tikzcd}
    \]
    and satisfying the (geometrically evident) relation
    \begin{equation}
      \label
      {SimplicialIdentitiesForFaceMaps}
      \partial_i 
        \circ
      \partial_j
      =
      \partial_{j-1} 
        \circ
      \partial_i
      \;\;
      \text{  if $i < j$}.
    \end{equation}
    
    This structure may be thought of as the non-linear generalization of the structure of a chain complex of abelian groups:
    When the sets $\mathcal{G}_n$ happen to carry the structure of abelian groups such that the $\partial_i$ are group homomorphisms, then \cref{SimplicialIdentitiesForFaceMaps} implies (and is essentially equivalent to) that
    $
      \partial
      :=
      \sum_i (-1)^i \partial_i
    $ 
    satisfies $\partial \circ \partial = 0$.

    \item 
    \emph{Degeneracy maps}
    $\inlinetikzcd{ \mathcal{G}_{n} \ar[r, "{ \mathrm{s}_i }"] \& \mathcal{G}_{n+1} }$,
    sending each $n$-simplex to an identity $(n+1)$-simplex on it,
    as expressed by further \emph{simplicial identities}:
    \begin{equation}
      \mathrm{s}_i 
        \circ 
      \mathrm{s}_j
      =
      \mathrm{s}_j 
        \circ 
      \mathrm{s}_{i - 1}
      \;\;
      \text{ if $i > j$}
      \;\;\;
      \text{and}
      \;\;\;
      \partial_i 
        \circ 
      \mathrm{s}_j
      =
      \begin{cases}
        \mathrm{s}_{j - 1} \circ \partial_i
        & 
        \text{ if $i < j$}
        \\
        \mathrm{id}
        &
        \text{ if $i \in \{j, j+1\}$}
        \\
        \mathrm{s}_j \circ \partial_{i - 1}
        &
        \text{ if $i > j+1$}.
      \end{cases}
    \end{equation}
   \end{enumerate}

A simplicial set $\mathcal{G}$ is \emph{Kan} if:
\begin{itemize}
\item 
for all tuples of its $n$-simplices that form the boundary of an $(n+1)$-simplex except for one missing face, a candidate missing face and an $(n+1)$-simplex filling this boundary also exist in $\mathcal{G}$.
\end{itemize}

A map between (Kan) simplicial sets $\inlinetikzcd{\mathcal{G}_\bullet \ar[r, "{ f }"] \& \mathcal{G}'_\bullet}$ is a sequence of maps of sets $\big(\inlinetikzcd{ \mathcal{G}_n \ar[r, "{ f_n }"] \& \mathcal{G}'_n }\big)_{n \in \mathbb{N}}$ compatible with the face and degeneracy maps.

There is a fairly evident notion of higher homotopy groups $\pi_\bullet$ of a Kan simplicial set, and a map is called a (weak simplicial homotopy) \emph{equivalence} if it induces an isomorphism on all homotopy groups (this is the non-linear version of \emph{quasi-isomorphism} between chain complexes).

\end{definition}

As a slogan:
\begin{standout}
  Higher gauged sets
  are 
  \emph{$\infty$-groupoids}, which
  are presented by 
  \emph{Kan simplicial sets}, which are to be thought of
  as ``non-linear chain complexes'' in non-negative degrees.
\end{standout}

\begin{example}
  \begin{enumerate}
    \item A set $S$ regarded as a higher gauged set with no non-trivial gauge and higher gauge transformations is the constant Kan simplicial set $S_\bullet$ with $S_n \defneq S$, for all $n \in \mathbb{N}$, meaning that there are just the identity gauge transformations on the elements $s \in S$, only the identity gauge-of-gauge transformations on these identity gauge transformations, and so on. 

    \item A small groupoid $\mathcal{G}$  regarded as a higher gauge set with no non-trivial higher gauge transformations
    is the Kan simplicial set with
    \begin{enumerate}
     \item  $\mathcal{G}_0$ the set of objects of $\mathcal{G}$,

     \item $\mathcal{G}_1$ the set of morphisms of $\mathcal{G}$,

     \item $\mathcal{G}_{n \geq 1}$ the set of sequences of length $n$ of composable morphisms of $\mathcal{G}$, 
     corresponding to the identity higher gauge transformations between these and their composites.
    \end{enumerate}

   \item For a topological space $X$, its \emph{path $\infty$-groupoid} is the Kan simplicial set $\mathrm{Sing}(X)_\bullet$ whose elements are the points of $X$, whose gauge transformations are the continuous paths in $X$, whose gauge-of-gauge transformations are the continuous 2-dimensional paths in $X$, and so on.

  Remarkably: All Kan simplicial sets arise as path $\infty$-groupoids of topological spaces, up to equivalence, in fact the path $\infty$-groupoid construction identifies the homotopy theory of topological spaces with that of (Kan) simplicial sets.
  
  For this reason, one may (and we will) notationally conflate homotopy types of topological spaces with their path $\infty$-groupoids.
  \end{enumerate}
\end{example}

%%%%%%%%%%%%%%
\paragraph
{Literature}
%%%%%%%%%%%%%
For exposition of simplicial sets: \cite{Friedman2012}. For comprehensive discussion of simplicial homotopy theory: \cite{GoerssJardine2009}. Digest streamlined towards our needs here: \cite[\S 1]{FSS23-Char}.

%%%%%%%%%%%%%%%
\subsection
{Potentials}
\label
{Potentials}
%%%%%%%%%%%%%%%

%%%%%%%%%%%%%%%%%
\subsubsection
{$C$-Field Potentials in the 11D Bulk}
\label
{OnCFieldPotentials}
%%%%%%%%%%%%%%%%

%%%%%%%%%%%%
\paragraph
{Gauged Set of Local $C$-Field Potentials}
%%%%%%%%%%%

We begin with considering the local situation over a (super-)chart $U$.
Here:
\begin{definition}
\label[definition]
{CFieldGaugeStructure}
For fixed (super-)flux densities 
\[
  \begin{aligned}
    G_4 & 
    \in
    \Omega^4_{\mathrm{dR}}\bracket({
      U
    })
    \\
    G_7 & 
    \in
    \Omega^7_{\mathrm{dR}}\bracket({
      U
    })
  \end{aligned}
  \;\;\;
  \text{s.t.}
  \;\;\;
  \begin{aligned}
    \mathrm{d}\, G_4
    & = 0
    \\
    \mathrm{d}\, G_7
    & =
    \tfrac{1}{2} G_4 G_4 \,,
  \end{aligned}
\]
take:
\begin{enumerate}

\item
  \emph{$C$-field gauge potentials} to be:
  \begin{equation}
  \label{LocalCFieldPotential}
    \begin{aligned}
      C_3 
      & 
      \in
      \Omega^3_{\mathrm{dR}}
      \bracket({
        U
      })
      \\
      C_6
      & 
      \in
      \Omega^6_{\mathrm{dR}}
      \bracket({
        U
      })
    \end{aligned}
    \;\;\;
    \text{s.t.}
    \;\;\;
    \begin{aligned}
      \mathrm{d}\, C_3
      & =
      G_4
      \\
      \mathrm{d}\, C_6
      & =
      G_7 - 
      \tfrac{1}{2}
      C_3 G_4
      \mathrlap{\,,}
    \end{aligned}
  \end{equation}

\item
  \emph{$C$-field gauge transformations}
  between pairs of potentials to be:
  \begin{equation}
  \label
  {LocalCFieldGaugeTransformations}
    \begin{aligned}
      C_2 
      & 
      \in
      \Omega^2_{\mathrm{dR}}
      \bracket({
        U
      })
      \\
      C_5
      & 
      \in
      \Omega^5_{\mathrm{dR}}
      \bracket({
        U
      })
    \end{aligned}
    \;\;\;
    \text{s.t.}
    \;\;\;
    \begin{aligned}
      \mathrm{d}\, C_2
      & =
      C'_3 - C_3
      \\
      \mathrm{d}\, C_5
      & =
      C'_6 - C_6
      -
      \tfrac{1}{2}
      C'_3 C_3
      \mathrlap{\,,}
    \end{aligned}
  \end{equation}

\item 
  \emph{$C$-field gauge-of-gauge transformations}
  between parallel pairs of gauge transformations
  \begin{equation}
  \begin{tikzcd}[column sep=large]
    \bracket({
      C_3, C_6
    })
    \ar[
      rr,
      bend left=30pt,
      "{\color{darkblue}
        ({
          C_2, C_5
        })
      }"{
        description,
        name=s
      }
    ]
    \ar[
      rr,
      bend right=30pt,
      "{\color{darkblue}
        ({
          C'_2, C'_5
        })
      }"{
        description,
        name=t
      }
    ]
    \ar[
      from=s,
      to=t,
      Rightarrow,
      "{ \color{darkgreen}
        (C_1, C_4) 
      }"{description}
    ]
    &&
    \bracket({
      C'_3, C'_6
    })
  \end{tikzcd}
\end{equation}
to be:
  \begin{equation}
    \label
    {LocalCFieldGaugeOfGaugeTransformations}
    \begin{aligned}
      C_1 
      & 
      \in
      \Omega^1_{\mathrm{dR}}
      \bracket({
        U
      })
      \\
      C_4
      & 
      \in
      \Omega^4_{\mathrm{dR}}
      \bracket({
        U
      })
    \end{aligned}
    \;\;\;
    \text{s.t.}
    \;\;\;
    \begin{aligned}
      \mathrm{d}\, C_1
      & =
      C'_2 - C_2
      \\
      \mathrm{d}\, C_4
      & =
      C'_5 - C_5
      \mathrlap{\,.}
    \end{aligned}
  \end{equation}
\end{enumerate}
\end{definition}

\begin{lemma}
\label[lemma]
{GroupoidOfLocalCFieldGaugeTransformations}
The local $C$-field gauge potentials \cref{LocalCFieldPotential} with  local $C$-field gauge transformations \cref{LocalCFieldGaugeTransformations} between them form a groupoid \cref{GroupoidOfYangMillsFields} 
under the following composition operation, 
for fixed parameter $\lambda \in \mathbb{R}$:
\begin{equation}
\label
{CompositionOfCFieldGaugeTransformations}
  \begin{tikzcd}[
    row sep=20pt,
    column sep=65pt
  ]
    &
    \bracket({
      C'_3, C'_6
    })
    \ar[
      dr,
      "{ \color{darkblue}
        (C'_2, C'_5)
      }"{sloped}
    ]
    \\
    \bracket({
      C_3, C_6
    })
    \ar[
      ur,
      "{\color{darkblue}
        (C_2, C_5)
      }"{sloped}
    ]
    \ar[
      rr,
      "{\color{darkblue}
        (C'_2, C'_5)
        \circ_{_{\lambda}}
        (C_2, C_5)
      }",
      "{\color{darkblue}
        :=
        \,
        \left(
          \begin{subarray}{l}
          C_2 + C'_2
          ,
          \\
          C_5 + C'_5 -\sfrac{1}{2} \,   
          C_2 \mathrm{d}C'_2
          +
    \sfrac{\lambda}{2}\,
         \mathrm{d}( C_2 C'_2 )  
    \end{subarray}
        \right)
      }"{swap}
    ]
    && 
    \bracket({
      C''_3, C''_6
    })
    \mathrlap{\,.}
  \end{tikzcd}
\end{equation}
Moreover, for pairs of parameters $\lambda,\lambda' \in \mathbb{R}$ these composition laws are related by $C$-field gauge-of-gauge transformations \cref{LocalCFieldGaugeOfGaugeTransformations}:
$$
  \begin{tikzcd}[
    column sep=55pt
  ]
    \bracket({
      C_3, C_6
    })
    \ar[
      rr,
      bend left=30pt,
      "{\color{darkblue}
        (C'_2, C'_5)
        \circ_{_{\lambda}}
        (C_2, C_5)
      }"{
        description,
        name=s
      }
    ]
    \ar[
      rr,
      bend right=30pt,
      "{\color{darkblue}
        (C'_2, C'_5)
        \circ_{_{\lambda'}}
        (C_2, C_5)
      }"{
        description,
        name=t
      }
    ]
    \ar[
      from=s,
      to=t,
      Rightarrow,
      "{ \color{darkgreen}
        \scaledbracket({
          0, 
          \sfrac{(\lambda' - \lambda)}{2}
          \,
          C_2 C'_2
        })
      }"{description}
    ]
    &&
    \bracket({
      C''_3, C''_6
    })\,.
  \end{tikzcd}
$$
\end{lemma}
\begin{proof}
It is a matter of straightforward computation to check the following.
\begin{enumerate}

\item
The composition \cref{CompositionOfCFieldGaugeTransformations} is well-defined:
\[
  \begin{aligned}
    & 
    \mathrm{d}\bracket({
      C_5 + C'_5 -  
    \sfrac{1}{2}\,
          C_2 \mathrm{d}C'_2
          +
    \sfrac{\lambda}{2}\,
         \mathrm{d}( C_2 C'_2 )
    })
    \\
    & =
    C''_6 - C_6
    -
    \tfrac{1}{2}
    C''_3 C'_3
    -
    \tfrac{1}{2}
    C'_3 C_3
    -\tfrac{1}{2}
    \bracket({
      C'_3 - C_3
    })
    \bracket({
      C''_3- C'_3
    })
    \\
    & =
    C''_6 - C_6 - 
    \tfrac{1}{2}
    C''_3 C_3
    \mathrlap{\,,}
  \end{aligned}
\]
\item 
and is unital, with identity morphisms $(0,0)$. 

\item It is associative:
\[
  \hspace{-.7cm}
  \begin{aligned}
    &
    \bracket({
       C_5 + C'_5 -    \tfrac{1}{2}\, 
          C_2 \mathrm{d}C'_2
          +
    \tfrac{\lambda}{2}\,
         \mathrm{d}( C_2 C'_2 )  
    })
    +
    C''_5
    -
    \tfrac{1}{2}\bracket({
      C_2 + C'_2
    }) 
    \mathrm{d}C''_2 
    +
    \tfrac{\lambda}{2}\mathrm{d}\big(
    \bracket({
      C_2 + C'_2
    })
    C''_2 \big)
    \\
    & =
    C_5 + C'_5 + C''_5
    -
    \tfrac{1}{2}C_2 
    \mathrm{d}C'_2
    -
    \tfrac{1}{2}
    C_2
    \mathrm{d}
    C''_2
    -
    \tfrac{1}{2}
    C'_2
    \mathrm{d}
    C''_2
    +
    \tfrac{\lambda}{2}\dd(C_2 C'_2)
    +
    \tfrac{\lambda}{2} \dd (
    C_2
    C''_2 )
    +
    \tfrac{\lambda}{2} \dd( 
    C'_2
    C''_2 )
    \\
    & =
    C_5 +
    \bracket({
      C'_5 + C''_5
      -
      \tfrac{1}{2}
      C'_2 \mathrm{d}C''_2
      +
      \tfrac{\lambda}{2}\dd(
      C'_2 C''_2
    )})- 
    \tfrac{1}{2}
    C_2 \mathrm{d} 
    \bracket({ 
       C'_2 + C''_2 
    })
    +
    \tfrac{\lambda}{2}
    \mathrm{d}\big( C_2
    \bracket({
      C'_2 + C''_2 
    }) \big)
    \mathrlap{\,,}
  \end{aligned}
\]

\item
and has inverses given by
\begin{equation}
\label
{InverseOfLocalCFieldGaugeTransformation}
  \bracket({
    C_2, C_5
  })^{-1}
  =
  \bracket({
    -C_2,
    -C_5 
    + 
    \bracket({
      \lambda- \tfrac{1}{2}
    })\, 
    C_2 \mathrm{d}C_2
  })
  \mathrlap{\,,}
\end{equation}
which are well-defined because:
\[
  \mathrm{d}\bracket({
    -C_5 
      + 
    \bracket({
      \lambda-\tfrac{1}{2}
    })
    C_2 \mathrm{d}C_2
  })
  =
  C_6 - C'_6
  -
  \tfrac{1}{2}C_3C'_3
  +
  \bracket({
    \lambda-\tfrac{1}{2}
  })
  \smash{
  \grayunderbrace{
  \bracket({\mathrm{d}C_2})^2
  }{0}
  }
  \mathrlap{\,,}
\]
and indeed satisfy:
\[
    C_5
    +
    \bracket({
      - C_5
      +
      \bracket({
        \lambda - \tfrac{1}{2}
      })
      C_2 \mathrm{d}C_2
    })
    +
    \tfrac{\lambda}{2}
    C_2 \mathrm{d}\bracket({-C_2})
    -
    \tfrac{1-\lambda}{2}
    \bracket({-C_2})\mathrm{d}C_2
    =
    0
    \mathrlap{\,.}
\]
\end{enumerate}

Finally, the difference in composition for two parameter values is indeed:
  \begin{align*}
    &
    \tfrac{\lambda'}{2}
    C'_2\mathrm{d}C_2
    -
    \tfrac{1-\lambda'}{2}
    C_2 \mathrm{d}C'_2
    -
    \bracket({
    \tfrac{\lambda}{2}
    C'_2\mathrm{d}C_2
    -
    \tfrac{1-\lambda}{2}
    C_2 \mathrm{d}C'_2    
    })
    \\
    & =
    \tfrac{\lambda' - \lambda}{2}
    \bracket({
      C'_2\mathrm{d}C_2
      +
      C_2 \mathrm{d}C'_2
    })
    \\
    & = 
    \mathrm{d}
    \bracket({
      \tfrac{\lambda'-\lambda}{2}
      C_2 C'_2
    })
    \mathrlap{\,.}
    \qedhere
  \end{align*}
\end{proof}

However, while they exist as such for any $\lambda$, the groupoids of local $C$-field gauge transformations from \cref{GroupoidOfLocalCFieldGaugeTransformations} are not physically meaningful, because they refer to gauge transformations as such without involving the gauge-of-gauge-equivalences between these. To fully reflect the latter, we are to consider a \emph{2-groupoid} and ultimately a \emph{3-groupoid} of local $C$-field gauge potentials, including also the third-order gauge transformations. 

Alternatively, to stay within the realm of 1-groupoids for the time being, we may quotient out the gauge-of-gauge transformations:
\begin{lemma}
\label[lemma]
{GroupoidOfClassesOfCFieldGaugeTransformations}
  The local $C$-field gauge potentials \cref{LocalCFieldPotential} with
  local gauge-of-gauge-equivalence classes 
  \cref{LocalCFieldGaugeOfGaugeTransformations}
  $[C_2, C_5]$
  of
  local $C$-field gauge transformations 
  $(C_2, C_5)$
  \cref{LocalCFieldGaugeTransformations} 
  between them form a groupoid \cref{GroupoidOfYangMillsFields} 
  under the unique composition operation
  that is given on representatives by \cref{CompositionOfCFieldGaugeTransformations} for any value of $\lambda \in \mathbb{R}$:
\begin{equation}
\label
{CompositionOfClassesOfCFieldGaugeTransformations}
  \begin{tikzcd}[
    row sep=10pt,
    column sep=60pt
  ]
    &
    \bracket({
      C'_3, C'_6
    })
    \ar[
      dr,
      "{\color{darkblue}
        [C'_2, C'_5]
      }"{sloped}
    ]
    \\
    \bracket({
      C_3, C_6
    })
    \ar[
      ur,
      "{\color{darkblue}
        [C_2, C_5]
      }"{sloped}
    ]
    \ar[
      rr,
      "{\color{darkblue}
        [C'_2, C'_5]
        \circ
        [C_2, C_5]
      }",
      "{\color{darkblue}
        :=
        \,
        \scaledbracket[{
          (C'_2, C'_5)
          \circ_{_{\lambda}}
          (C_2, C_5)
        }]
      }"{swap}
    ]
    && 
    \bracket({
      C''_3, C''_6
    })
    \mathrlap{\,.}
  \end{tikzcd}
\end{equation}
\end{lemma}
\begin{proof}
For a pair of gauge-of-gauge-transformed composable gauge transformations 
\[
  \begin{aligned}
  \bracket({
    \tilde C_2, 
    \tilde C_5
  })
  &
  :=
  \bracket({
    C_2 
    + 
    \mathrm{d} C_1
    ,\,
    C_5 + \mathrm{d} C_4
  })
  \\
  \bracket({
    \tilde C'_2, 
    \tilde C'_5
  })
  &
  :=
  \bracket({
    C'_2 
    + 
    \mathrm{d} C'_1
    ,\,
    C'_5 + \mathrm{d} C'_4
  })
  \mathrlap{\,,}
  \end{aligned}
\]
we check that their composition is gauge-of-gauge-equivalent to that of the untransformed transformations:
\[
  \begin{aligned}
    &
    \bracket({
      C_5 + \mathrm{d}C_4
    })
    +
    \bracket({
      C'_5 + \mathrm{d}C'_4
    })
    +
    \tfrac{\lambda}{2}
    \bracket({
      C'_2 + \mathrm{d}C'_1 
    })
    \mathrm{d}
      C_2
    -
    \tfrac{1-\lambda}{2}
    \bracket({
      C_2 + \mathrm{d}C_1
    })
    \mathrm{d}
      C'_2
    \\
    & =
    C_5 + C'_5 +
    \tfrac{\lambda}{2}
    C'_2 \mathrm{d}C_2
    -
    \tfrac{1-\lambda}{2}
    C_2 \mathrm{d}C'_2
    +
    \mathrm{d}\bracket({
      C_4 + C'_4
      +
      \tfrac{\lambda}{2}
      C'_1 \mathrm{d}C_2
      -
      \tfrac{1-\lambda}{2}
      C_1 \mathrm{d}C'_2
    })
    \mathrlap{\,,}
  \end{aligned}
\]
and that their inverses \cref{InverseOfLocalCFieldGaugeTransformation} are gauge-of-gauge-equivalent to those of the untransformed transformations:
\begin{align*}
    &
    -
    \bracket({
      C_5 + \mathrm{d}C_4
    })
    +
    \bracket({
      \lambda - \tfrac{1}{2}
    })
    \bracket({C_2 + \mathrm{d}C_1 }) \mathrm{d}C_2
    \\
    & =
    - C_5 
    + 
    \bracket({
      \lambda - \tfrac{1}{2}
    })
    C_2 \mathrm{d}C_2
    +
    \mathrm{d}\bracket({
      - C_4 
      +
      \bracket({
        \lambda - \tfrac{1}{2}
      })
      C_1 \mathrm{d}C_2
    })
    \mathrlap{\,.}
    \qedhere
\end{align*}
\end{proof}

%%%%%%%%%%%%%
\paragraph
{Comparison to the Literature}
%%%%%%%%%%%%%

While the above definitions of gauge and gauge-of-gauge transformations of $C$-field potentials are \emph{plausible}, they are at this point not systematically \emph{derived} from an overarching principle.  

In fact, previous literature is ambiguous on their definition:
\begin{enumerate}

 \item The $C$-field gauge transformations considered in \parencites[(2.12)]{CoimbraEtAl2014} coincide exactly with our \cref{LocalCFieldGaugeTransformations}.

 \item Previously, in \cite[(21)]{BandosEtAl2004} the authors take the local $C$-field gauge transformations to be as follows (just adjusted for our normalization convention):
 \begin{equation}
   \begin{aligned}
     C_2 
     & 
     \in
     \Omega^2_{\mathrm{dR}}
     \bracket({
       U
     })
     \\
     C_5
     & 
     \in
     \Omega^5_{\mathrm{dR}}
     \bracket({
       U
     })
   \end{aligned}
   \;\;\;
   \text{s.t.}
   \;\;\;
   \begin{aligned}
     \mathrm{d}\, 
     C_2
     & 
     =
     C'_3 - C_3
     \\
     \mathrm{d}\,
     C_5
     & =
     C'_6 - C_6
     +
     \tfrac{1}{2}
     C_2 G_4
     \mathrlap{\,.}
   \end{aligned}
 \end{equation}
 This coincides with \cref{LocalCFieldGaugeTransformations} up to gauge-of-gauge transformation \cref{LocalCFieldGaugeOfGaugeTransformations}:
 \begin{equation}
  \left(
   \begin{aligned}
     &C_2,
     \\[-2pt]
     &C_5
   \end{aligned}
   \right)
   \,\mapsto\,
  \left(
   \begin{aligned}
     &C_2,
     \\[-2pt]
     &
     C_5
     -
     \mathrm{d}
     \bracket({
       \tfrac{1}{2}
       C_2 C'_3  
     })
   \end{aligned}
   \right)
   \mathrlap{,}
 \end{equation}
 and hence agrees on gauge-of-gauge-equivalence classes \cref{CompositionOfClassesOfCFieldGaugeTransformations}.

 \item 
 However, the original authors \parencites[(2.4)]{CremmerEtAl1998}[(3.3)]{LavrinenkoEtAl999}[(14)]{KalkkinenStelle2003} on this subject (also \parencites[(4.9)]{Sati2010}[\S 2.3]{Baraglia2012}{Rosabal2026}) considered  ``gauge transformations'' that are essentially different when taken at face value,
  parameterized by \emph{shifts} of the gauge potentials:
 \begin{equation}
 \label
 {NonExactShiftsOfCFieldGaugePotential}
 \begin{aligned}
   \Lambda_3
   & 
   \in 
   \Omega^3_{\mathrm{cl}}\bracket({
     U
   })
   \\
   \Lambda_6
   & 
   \in 
   \Omega^6_{\mathrm{cl}}\bracket({
     U
   })
 \end{aligned}
 \;\;\;
 \text{s.t.}
 \;\;\;
 \begin{aligned}
   \Lambda_3 & = C'_3 - C_3
   \\
   \Lambda_6 & =
   C'_6 - C_6 
   -
   \tfrac{1}{2}
   \Lambda_3 C_3\,.
 \end{aligned}
 \end{equation}
 This definition \cref{NonExactShiftsOfCFieldGaugePotential} concerns closed (and hence exact on the chart $U$) shifts of gauge potentials, hence 
 it may be understood as defining not the gauge transformations themselves (which ought to be parametrized by potentials of $\Lambda_3$ and $\Lambda_6$) but the condition on potentials of \emph{being gauge equivalent}, at least over a chart.
 % We would argue that  \cref{NonExactShiftsOfCFieldGaugePotential} is not the right notion of gauge transformation, since all known gauge transformation laws are parametrized by potentials for the shifts. Or else, that this definition \cref{NonExactShiftsOfCFieldGaugePotential} does not define the gauge transformations themselves but just the condition on potentials of \emph{being gauge equivalent}. This is under the assumption that we really are over a chart (which is never stated in previous literature): Over a general manifold, \cref{NonExactShiftsOfCFieldGaugePotential} admits \emph{non-exact shifts}, which is ``clearly'' not right for a definition of gauge transformations.
\end{enumerate}

Hence what one needs is a general theory that allows one to systematically identify the nature of higher gauge transformations. This is what we turn to next, following \cite[\S 2.14]{GSS24-SuGra}.

%%%%%%%%%%%%%%%
\paragraph
{Relating to Higher Concordances of $\mathfrak{l}S^4$-Valued Forms}
%%%%%%%%%%%%%%%

To this end, we next observe (with \parencites{GSS24-SuGra,GSS24-M5}) that the systematic way to derive the higher gauge structure of higher gauge potentials is by reducing them to higher-order concordances $\shape \mathbf{\Omega}^1_{\mathrm{cl}}\bracket({U;\mathfrak{a}})$ \cref{ShapeOfSmoothSetOfClosedFormsInIntro}
of closed differential forms with coefficients in the characteristic $L_\infty$-algebra $\mathfrak{a}$ 
\cref{TheCharacteristicLInfinityAlgebra}
of the duality-symmetric Bianchi identities.

Let us denote the groupoid established in \cref{GroupoidOfClassesOfCFieldGaugeTransformations}
as follows (the \emph{1-truncation} $[-]_1$ of the groupoid from \cref{GroupoidOfLocalCFieldGaugeTransformations}):
\begin{equation}
\label
{1GroupoidOfLocalCFieldPotentials}
  \begin{aligned}
  \bracket[{
    C\mathrm{FldPtntl}\bracket({
      U
    })
  }]_1
  &
  :
  \colorbox{lightgray}{$
  \begin{subarray}{l}
  \text{
    $C$-Field gauge potentials on the chart $U$
  }
  \\
  \text{
    with gauge-of-gauge-equivalence classes
  }
  \\
  \text{
    of gauge transformations between them,
  }
  \end{subarray}
  $}
  \end{aligned}
\end{equation}
where we leave the fixed flux density $(G_4, G_7)$ notationally implicit.

Recall from the overview \cref{Globalization} that just from the datum of the characteristic $L_\infty$-algebra \cref{TheCharacteristicLInfinityAlgebra} we obtain a higher groupoid of gauge potentials for given flux densities on a chart, this being the homotopy fiber:
\begin{equation}
\label
{3GroupoidOfCFieldPotentialsViaConcordances}
  \mathfrak{a}
  \mathrm{Ptntl}\bracket({U})
  :=
  \mathrm{fib}
  \big({
  \inlinetikzcd{
  \{ \vec F \}
  \ar[r]
  \&
  \shape 
  \mathbf{\Omega}^1
    _{\mathrm{cl}}\bracket({U; \mathfrak{a}})
  }
  }\big)
  \mathrlap{\,.}
\end{equation}

In order to stay within the realm of 1-groupoids, for the purpose of these lecture notes, we consider again the 1-truncated 1-groupoid obtained by quotienting out 2-morphisms:
\begin{equation}
\label
{1GroupoidOfCharValuedGaugePotentials}
  \bracket[{
    \mathfrak{a}
    \mathrm{Ptntl}\bracket({U})
  }]_1
  \defneq
  \Big[{
  \mathrm{fib}
  \big({
  \inlinetikzcd{
  \{ \vec F \}
  \ar[r]
  \&
  \shape 
  \mathbf{\Omega}^1
    _{\mathrm{cl}}\bracket({U; \mathfrak{a}})
  }
  }\big)
  }\Big]_1
  \mathrlap{\,.}
\end{equation}

The point is that the groupoid \cref{1GroupoidOfCharValuedGaugePotentials}, while unwieldy at face value, is \emph{canonically defined} from just the datum of the duality-symmetric Bianchi identities, and provides the right gauge structure for general flux quantization (\cref{Charges}). Therefore the following statement is remarkable:
\begin{proposition}
\label[proposition]
{ConcordancesEquivalentToGaugeTransf}
The canonically defined groupoid \cref{1GroupoidOfCharValuedGaugePotentials} of local $\mathfrak{a}$-potentials --- for $\mathfrak{a} \defneq \mathfrak{l}S^4$ the characteristic $L_\infty$-algebra of 11D SuGra \textup{(\cref{CharacteristicLInfinityOf11DSuGra})} --- is equivalent \cref{EquivalenceOfCategories} to the explicit groupoid \cref{1GroupoidOfLocalCFieldPotentials} of local $C$-field gauge potentials \textup{(\cref{CFieldGaugeStructure,GroupoidOfClassesOfCFieldGaugeTransformations})}:
\begin{equation}
  \begin{tikzcd}
  \bracket[{
  \bracket({\mathfrak{l}S^4})
  \mathrm{Ptntl}
  \bracket({U})
  }]_1
  \ar[
    r,
    "{ \sim }"
  ]
  &
  \bracket[{
    C\mathrm{FldPtntl}\bracket({U})
  }]_1 \,.
  \end{tikzcd}
\end{equation}
\end{proposition}
\begin{proof}
  This is essentially the statement of 
  \cite[Prop. 2.48]{GSS24-SuGra} (cf. also \cite{Banerjee2025-Potentials}):
  Items (i) and (ii) there say that we have a functor that is surjective on objects, and item (iii) there says that this functor is full. It remains to show that it is faithful, which is equivalent to saying that before quotienting the construction extends to be full on gauge-of-gauge transformations. This one readily checks by the same method as used in the proof there.
\end{proof}

This statement is the best one can hope for in view of previous literature, where gauge-of-gauge transformations etc. have not been discussed.
Accordingly, this means that we may regard the untruncated higher groupoid $\bracket({\mathfrak{l}S^4})\mathrm{Ptntl}\bracket({U})$
\cref{3GroupoidOfCFieldPotentialsViaConcordances} as the proper local higher gauge structure of $C$-field potentials. 

%%%%%%%%%%%%%%%%%
\subsubsection
{NS/RR-Field Potentials in the 10D IIA Bulk}

By the isomorphism of $\infty$-groupoids from {\textbf{(ii)}} of Prop. \ref{NonTrivialTorusBundleToroidification/Oxidation}, it follows that the local gauge structure of the 11D C-field determines the local gauge structure of the (duality-symmetric) fields for all descendant supergravity theories. In particular, restricting to the set of $S^1$-invariant super-fluxes in 11D, this implies that the basic decompositions of all the formulas from \cref{OnCFieldPotentials} yield the corresponding gauge potentials and (higher) gauge transformations for the lower-dimensional theories (cf. \cite[Prop. 2.4]{GiotopoulosSati2026}).

For instance, in the case of IIA supergravity, this immediately yields the NS/RR local gauge structure encoded in the cyclification $\mathrm{cyc}(\mathfrak{l}S^4)$. Namely,
for fixed NS/RR (super-)flux  densities 
\[
  \begin{aligned}
    F_2 & 
    \in
    \Omega^2_{\mathrm{dR}}\bracket({
      U
    })
    \\
    H_3 & 
    \in
    \Omega^3_{\mathrm{dR}}\bracket({
      U
    })
     \\
    F_4 & 
    \in
    \Omega^4_{\mathrm{dR}}\bracket({
      U
    })
     \\
    F_6 & 
    \in
    \Omega^6_{\mathrm{dR}}\bracket({
      U
    })
     \\
    H_7 & 
    \in
    \Omega^7_{\mathrm{dR}}\bracket({
      U
    })
  \end{aligned}
  \;\;\;
  \text{s.t.}
  \;\;\;
  \begin{aligned}
  \mathrm{d}\, F_2
    & = 0
    \\
    \mathrm{d}\, H_3
    & = 0
    \\
    \mathrm{d}\, F_4
    & =
    \tfrac{1}{2} H_3 F_2
    \\
    \mathrm{d}\, F_6
    & = - 
   H_3 F_4
    \\
    \mathrm{d}\, H_7
    & =
    \tfrac{1}{2} F_4 F_4 + F_2 F_6 
  \end{aligned}
\]
then:
\begin{enumerate}
 \item
 \emph{NS/RR-field gauge potentials} are given by:
 \begin{equation}
  \label{LocalNS/RRFieldPotential}
    \begin{aligned} A_1 
      & 
      \in
      \Omega^1_{\mathrm{dR}}
      \bracket({
        U
      })
      \\
      B_2 
      & 
      \in
      \Omega^2_{\mathrm{dR}}
      \bracket({
        U
      })
      \\
      A_3 
      & 
      \in
      \Omega^3_{\mathrm{dR}}
      \bracket({
        U
      })
      \\
      A_5
      & 
      \in
      \Omega^5_{\mathrm{dR}}
      \bracket({
        U
      })
       \\
      B_6
      & 
      \in
      \Omega^6_{\mathrm{dR}}
      \bracket({
        U
      })
    \end{aligned}
    \;\;\;
    \text{s.t.}
    \;\;\;
    \begin{aligned}
    \mathrm{d}\, A_1
      & =
      F_2
      \\
      \mathrm{d}\, B_2
      & =
      H_3
      \\
      \mathrm{d}\, A_3
      & =
      F_4 - F_2 B_2
      \\
      \mathrm{d}\, A_5
      & =
      F_6 
    +\tfrac{1}{2}A_3  H_3 + \tfrac{1}{2} B_2 F_4
      \\
      \mathrm{d}\, B_6
      & =
       H_7 
    -
    \tfrac{1}{2} A_3  F_4 - F_2  A_5
      \mathrlap{\,.}
    \end{aligned}
  \end{equation}

 \item
 \emph{NS/RR-field gauge transformations} between pairs of potentials are of the form:
\begin{equation}
  \label
{LocalNS/RRFieldGaugeTransformations}
    \begin{aligned}
      N_0 
      & 
      \in
      \Omega^0_{\mathrm{dR}}
      \bracket({
        U
      })
      \\
        N_1 
      & 
      \in
      \Omega^1_{\mathrm{dR}}
      \bracket({
        U
      })
      \\
    M_2 
      & 
      \in
      \Omega^2_{\mathrm{dR}}
      \bracket({
        U
      })
      \\
      N_4 
      & 
      \in
      \Omega^4_{\mathrm{dR}}
      \bracket({
        U
      })
      \\
      M_5
      & 
      \in
      \Omega^5_{\mathrm{dR}}
      \bracket({
        U
      })
    \end{aligned}
    \;\;\;
    \text{s.t.}
    \;\;\;
    \begin{aligned}
    \mathrm{d}\, N_0
      & =
      A'_1 - A_1
      \\
      \mathrm{d}\, N_1
      & =
      B'_2 - B_2
      \\
      \mathrm{d}\, M_2
      & =
      A'_3 - A_3 + F_2 N_1 
      \\
      \mathrm{d}\, N_4
      & =
      A'_5 - A_5 - \tfrac{1}{2} (B'_2  A_3 - A'_3  B_2)
      \\
      \mathrm{d}\, M_5
      & =
      B'_6 - B_6 + F_2  N_4 - \tfrac{1}{2} A'_3   A_3 
      \mathrlap{\,.}
    \end{aligned}
  \end{equation}

 \item
 \emph{NS/RR-field gauge-of-gauge transformations} are of the form:
  \begin{equation}
    \label
{LocalNS/RRFieldGaugeOfGaugeTransformations}
    \begin{aligned}
      K_0 
      & 
      \in
      \Omega^0_{\mathrm{dR}}
      \bracket({
        U
      })
      \\
      L_1 
      & 
      \in
      \Omega^1_{\mathrm{dR}}
      \bracket({
        U
      })
      \\
      K_3 
      & 
      \in
      \Omega^3_{\mathrm{dR}}
      \bracket({
        U
      })
      \\
      L_4
      & 
      \in
      \Omega^4_{\mathrm{dR}}
      \bracket({
        U
      })
    \end{aligned}
    \;\;\;
    \text{s.t.}
    \;\;\;
    \begin{aligned}
      \mathrm{d}\, K_0
      & =
      N_1' - N_1
      \\
      \mathrm{d}\, L_1
      & =
      M_2' - M_2 - F_2 K_0
      \\
      \mathrm{d}\, K_3
      & =
      N'_4 - N_4
      \\
      \mathrm{d}\, L_4
      & =
      M'_5 - M_5 - F_2 K_3
      \mathrlap{\,.}
    \end{aligned}
  \end{equation}
\end{enumerate}

This pattern propagates similarly in an algorithmic manner to even higher gauge transformations, and moreover to the local gauge-structure of lower dimensional SuGras.
%%%%%%%%%%%%%%%%

%%%%%%%%%%%%%%%%%
\subsubsection
{$B$-Field Potentials on M5-Probes}
%%%%%%%%%%%%%%%%

The discussion of local $C$-field gauge structure (\cref{OnCFieldPotentials}) generalizes to the presence of M5-brane probes $\inlinetikzcd{ \Sigma^{1,5\vert 2\cdot \mathbf{8}_+} \ar[r, "{ \Phi }"] \& X^{1,10\vert \mathbf{32}} }$ (\cref{Probes}). 

By analysis of the higher concordances of $\mathfrak{l}_{_{S^4}} S^7$-valued differential forms from \cref{BianchiForM5WorldvolumeFlux} and \cref{M5FluxEquations}
one finds \parencites[\S 4.1]{GSS24-M5}[\S 3.1]{Banerjee2025-Potentials} in addition to the bulk field data from \cref{CFieldGaugeStructure}, the following data (on a compatible chart $U$ of $\Sigma^{1,5\vert 2\cdot \mathbf{8}_+}$) for fixed M5-worldvolume flux density
\begin{equation}
  H_3 \in \Omega^3_{\mathrm{dR}}
  \;\;
  \text{s.t.}
  \;\;
  \mathrm{d}\, H_3
  =
  \Phi^\ast G_4
  \mathrlap{\,.}
\end{equation}
\begin{enumerate}
 \item
 $B$-field gauge potentials
 alongside the $C$-field potentials $(C_3, C_6)$ from \cref{LocalCFieldPotential} are of the form:
 \begin{equation}
 \label
 {LocalBFieldPotential}
   B_2 \in \Omega^2_{\mathrm{dR}}(U)
   \;\;
   \text{s.t.}
   \;\;
   \mathrm{d}\, B_2
   =
   H_3 - C_3
   \mathrlap{\,.}
 \end{equation}

 \item
 $B$-field gauge transformations alongside the $C$-field gauge transformations $(C_2, C_5)$ from \cref{LocalCFieldGaugeTransformations} are of the form:
 \begin{equation}
   \mathrm{d}\, B_1
   =
   B'_2 - B_2 + C_2
   \mathrlap{\,.}
 \end{equation}

 \item
 $B$-field gauge-of-gauge transformations alongside the $C$-field gauge-of-gauge transformations $(C_1, C_4)$ from \cref{LocalCFieldGaugeOfGaugeTransformations} are of the form:
 \begin{equation}
   \mathrm{d}\, B_0
   =
   B'_1 - B_1 + C_1
   \mathrlap{\,.}
 \end{equation}
\end{enumerate}

\medskip

With this analysis in hand, we next find the admissible global IR completions of 11D SuGra (with M5-probes), and the corresponding topological brane charges: \cref{Charges}.

%%%%%%%%%%%%%%%%%%%
\subsection
{Charges}
\label
{Charges}
%%%%%%%%%%%%%%%%%%

We have now seen that the local on-shell field and higher gauge structure of 11D SuGra is entirely  controlled by the M-theory gauge $L_\infty$-algebra $\mathfrak{l}S^4$ (\cref{CharacteristicLInfinityOf11DSuGra}), which in the presence of M5-brane probes is enhanced to the fibration $\inlinetikzcd{ \mathfrak{l}_{_{S^4}}S^7 \ar[r, "{ \mathfrak{l}h_{\mathbb{H}} }"] \& \mathfrak{l}S^4 }$ \cref{lS7RelativeToS4}. Here we finally explain how this characterization allows for the global IR-completion of the higher gauge fields, thereby defining the topological charges of the singular (black) branes of the theory (the \emph{M-brane} charges).  

Dedicated lecture notes on the mathematical details of this claim are in the companion article \cite{SS26-HigherGauge}, to which we refer the interested reader. Here we give a slightly sketchier motivation from Dirac charge quantization and then highlight the nature of the topological brane charges that accompany each such choice of global completion.

%%%%%%%%%%%%%
\paragraph
{Comparison to the Literature: Models of the $C$-field}
%%%%%%%%%%%%%

The idea that the $C$-field of 11D SuGra needs a global definition originates with \parencites[(3.3)]{Evslin_2003}{AschieriJurco2004}[\S2.7]{HopkinsSinger2005}{DFM2007} (the latter authors speak of a ``model of the $C$-field''), further considered in \cite[\S 2]{Sati2010} and \parencites[(2.1.15)]{FSS14-7D}[\S 4.1]{FSS15-ModuliStack}[(2.1)]{DonagiWijnholt2023}: In these articles, the magnetic flux density $G_4$ is quantized in a (shifted) form of (differential) ordinary integral cohomology (cf. recalled in \cref{OnHigherDiracChargesInDiffOrdCohomology}), while quantization of the electric flux density $G_7$ is not considered (but see \cite{Sati2005}). Accordingly, the peculiar mixed electromagnetic nature of the gauge transformations \cref{LocalCFieldGaugeTransformations} of the local $C$-field potentials (\cref{CFieldGaugeStructure}) is not reflected in these models.  
 
On this point compare the famous proposal of quantizing the RR-flux of 10D II SuGra in topological K-theory (cf. \parencites{Freed2002}[\S 3]{Szabo2013}{GS22-KTheory}):  While not traditionally advertised as such, this is tacitly a joint quantization of both the magnetic RR-fluxes, $F_{\leq 5}$, and of the electric RR-fluxes, $F_{\geq 5}$ (cf. \cite[\S 2.4 \& \S 4.1]{SS25-Flux}).
If 10D IIA SuGra is still to originate from 11D (cf. \cref{SomeReductionsOf11DSuGra}),  globally subject to this flux quantization, then the quantization of $G_7$ in 11D cannot be ignored.
Despite the title of \cite{DMW2003,DMW2000}, this problem has not been addressed before \cite{BaSS26-UnstableK, GiotopoulosSati2026}.

The general notion of what it means to globally complete a (Maxwell-type) higher gauge theory by proper electromagnetic flux quantization was laid out in \parencites{SS24-Phase}, based on the results of \cite{FSS23-Char} and surveyed in \cite{SS25-Flux,SS25-Complete,SS26-HigherGauge}. The application of this method to the $C$-Field in 11D SuGra is considered in: \parencites{FSS15-M5WZW,FSS20-H,FSS21-Hopf,SS20-Tad,GS21,FSS22-Twistorial,Grady2025,BaSS26-UnstableK} 
following \cite[\S2.5]{Sati2018}.

%%%%%%%%%%%%%%%%%
\subsubsection
{Higher Dirac charges in Differential Ordinary Cohomology}
\label
{OnHigherDiracChargesInDiffOrdCohomology}
%%%%%%%%%%%%%%%%%

The one special case of flux quantization that has become more widely appreciated is the higher generalization of Dirac charge quantization, 
where we have a flux density $F_{n+1}$ (in degree $n+1$) subject to the simple Bianchi identity
\begin{equation}
  \mathrm{d}\,
  F_{n+1}
  =
  0
  \mathrlap{\,,}
\end{equation}
and where the corresponding gauge field is globally completed as a 
cocycle in \emph{differential ordinary cohomology}, equivalently modeled by \emph{Deligne cohomology}, \emph{Cheeger-Simons differential characters} and \emph{higher $\mathrm{U}(1)$-bundle gerbes}.

We now briefly review a construction of differential ordinary cohomology that naturally leads to the more general differential nonabelian cohomology that we need for higher-dimensional supergravity.

%%%%%%%%%%%%%%
\paragraph
{Presheaves of Chain Complexes}
%%%%%%%%%%%%%%
To understand this flux quantization, we consider \emph{{\v C}ech cohomology} with coefficients in \emph{sheaves of chain complexes} (exposition for physicists: \cite{Alvarez1985}). The idea of this is the following:

By a \emph{presheaf of chain complexes} $A_\bullet$, we mean here an assignment which to each (super-) Cartesian space $U$ --- hence to each chart of any given (super-) spacetime --- assigns a chain complex in non-negative degrees of abelian groups $A_n(U)$:
\begin{equation}
\label
{PresheafOfChainComplexes}
  A_\bullet(U)
  \defneq
  \Big(
    \inlinetikzcd{
      \cdots
      \ar[
        rr,
        "{ \partial^A_2 }"
      ]
      \&\&
      A_2(U)
      \ar[
        rr,
        "{ \partial^A_1 }"
      ]
      \&\&
      A_1(U)
      \ar[
        rr,
        "{ \partial^A_0 }"
      ]
      \&\&
      A_0(U)
    }
  \Big)
\end{equation}
and to each map of charts a chain map in the other direction
\begin{equation}
\label
{FunctorialityOfPresheafOfChainComplexes}
  \begin{tikzcd}[
    row sep=0pt, 
    column sep=3pt
  ]
    U 
    \ar[
      dd,
      "{ \phi }"
    ]
      &\mapsto&
    A_\bullet(U)
    \\
    &\mapsto&
    \\
    U' &\mapsto&
    A_\bullet\bracket({U'})
    \ar[
      uu,
      "{ \phi^\ast }"
    ]
  \end{tikzcd}
\end{equation}
such that composition of maps and identity maps are respected.

Key examples include: 
\begin{enumerate}
\item
The $k$-fold suspension ($k \in \mathbb{N}$) of an abelian group $A$, regarded as a constant presheaf in degree $k$:
\begin{equation}
\label
{ShiftedAbelianGroupComplexOfSheaves}
  A[k](-)
  :=
  \big(
    \inlinetikzcd{
      A
      \ar[r, "{ 0 }"]
      \&
      0
      \ar[r, "{ 0 }"]
      \&
      0
      \ar[r, ]
      \&
      \cdots
      \ar[r]
      \&
      0
    }
  \big)
  \mathrlap{\,,}
\end{equation}
and, generally,
\begin{equation}
  A_\bullet[k]
  :=
  \begin{cases}
    A_{\bullet-k}
    &
    \text{if $\bullet-k \geq 0$}
    \\
    0 & \text{otherwise.}
  \end{cases}
\end{equation}

\item 
The $k$-shifted de Rham complex:
\begin{equation}
\label{DeRhamComplexOfSheaves}
  \Omega^{k-\bullet}_{\mathrm{dR}}(-)
  :=
  \Big(
    \inlinetikzcd{
      \Omega^0_{\mathrm{dR}}(-)
      \ar[r, "{ \mathrm{d} }"]
      \&
      \cdots
      \ar[r, "{ \mathrm{d} }"]
      \&
      \Omega^{k-2}_{\mathrm{dR}}(-)
      \ar[r, "{ \mathrm{d} }"]
      \&
      \Omega^{k-1}_{\mathrm{dR}}(-)
      \ar[r, "{ \mathrm{d} }"]
      \&
      \Omega^k
        _{\mathcolor{purple}{\mathrm{cl}}}(-)
    }
  \Big)
\end{equation}
ending in the sheaf of \emph{closed} $k$-forms.
\end{enumerate}

%%%%%%%%%%%%%
\paragraph
{Natural Chain Maps}
%%%%%%%%%%%%%
A map of presheaves of chain complexes (\emph{natural chain map}) is, for each chart, a sequence of maps that commute with the differentials,
\begin{equation}
  \begin{tikzcd}[row sep=10pt,
    column sep=15pt
  ]
  A_\bullet(U)
  \ar[
    d,
    "{ f_\bullet(U) }"{swap}
  ]
  &[-14pt]
  \defneq 
  \Big(
      \cdots
      \ar[
        r,
        "{ \partial^A_2 }"
      ]
      &[15pt]
      A_2(U)
      \ar[
        r,
        "{ \partial^A_1 }"
      ]
      \ar[
        d,
        "{ f_2(U) }"{swap, pos=.45}
      ]
      &[15pt]
      A_1(U)
      \ar[
        r,
        "{ \partial^A_0 }"
      ]
      \ar[
        d,
        "{ f_1(U) }"{swap, pos=.45}
      ]
      &[15pt]
      A_0(U)
      \ar[
        d,
        "{ f_0(U) }"{swap, pos=.45}
      ]
  &[-24pt]
  \Big)
  \\
  B_\bullet(U)
  &
  \defneq
  \Big(
      \cdots
      \ar[
        r,
        "{ \partial^B_2 }"
      ]
      &
      B_2(U)
      \ar[
        r,
        "{ \partial^B_1 }"
      ]
      &
      B_1(U)
      \ar[
        r,
        "{ \partial^B_0 }"
      ]
      &
      B_0(U)
      &[-17pt]
  \Big)
  \mathrlap{\,,}
  \end{tikzcd}
\end{equation}
and with the pullbacks \cref{FunctorialityOfPresheafOfChainComplexes}.
This is called a \emph{natural quasi-isomorphism} (also: \emph{global quasi-isomorphism}, since we do not require passage to stalks), denoted ``$\inlinetikzcd{{} \ar[r, "{ \sim }"] \& {}}$", if it induces isomorphisms on all chain homology groups for all $U$. 

For example, there is a canonical chain inclusion
from $\mathbb{R}[n]$ \cref{ShiftedAbelianGroupComplexOfSheaves} into $\Omega^{n-\bullet}_{\mathrm{dR}}$ \cref{DeRhamComplexOfSheaves},
\begin{equation}
  \begin{tikzcd}[row sep=10pt,
    column sep=4pt
  ]
  \mathbb{R}[n]
  \ar[
    d,
  ]
  &
  \defneq
  \big(
  &[-12pt]
      \mathbb{R}
      \ar[
        r,
        "{ 0 }"
      ]
      \ar[
        d,
        hook
      ]
      &[10pt]
      \cdots
      \ar[
        r,
        "{ 0 }"
      ]
      &[15pt]
      0
      \ar[
        r,
        "{ 0 }"
      ]
      \ar[
        d,
        "{ 0 }"{swap}
      ]
      &[15pt]
      0
      \ar[
        r,
        "{ 0 }"
      ]
      \ar[
        d,
        "{ 0 }"{swap}
      ]
      &[15pt]
      0
      \ar[
        d,
        "{ 0 }"{swap}
      ]
  &[-12pt]
  \big)
  \\
  \Omega^{n-\bullet}_{\mathrm{dR}}
  &
  \defneq
  \big(
  &
      \Omega^0_{\mathrm{dR}}
      \ar[
        r,
        "{ \mathrm{d} }"
      ]
      &
      \cdots
      \ar[
        r,
        "{ \mathrm{d} }"
      ]
      &
      \Omega^{n-2}_{\mathrm{dR}}
      \ar[
        r,
        "{ \mathrm{d} }"
      ]
      &
      \Omega^{n-1}_{\mathrm{dR}}
      \ar[
        r,
        "{ \mathrm{d} }"
      ]
      &
      \Omega^n_{\mathrm{cl}}
      &
  \big)
  \mathrlap{\,,}
  \end{tikzcd}
\end{equation}
which is a natural quasi-isomorphism, by the Poincar{\'e} lemma.

One says that a diagram of presheaves of chain complexes of the form shown in the following on the left here is a \emph{homotopy fiber product} of the form shown on the right:
\begin{equation}
\label
{HomotopyPullbackOfChainComplexes}
  \begin{tikzcd}[row sep=small]
    X_\bullet 
      \times_{\widehat{B}_\bullet} 
    \widehat{Y}_\bullet
    \ar[r]
    \ar[d]
    \ar[
      dr,
      phantom,
      "{ \lrcorner }"{pos=.1}
    ]
    &
    \widehat{Y}_\bullet
    \ar[
      d,
      ->>
    ]
    \ar[
      r,
      <-,
      "{ \sim }"
    ]
    &
    Y_\bullet
    \ar[d]
    \\
    X_\bullet
    \ar[r]
    &
    \widehat{B}_\bullet
    \ar[
      r,
      <-,
      "{ \sim }"
    ]
    &
    B_\bullet
  \end{tikzcd}
  \;\;\;
  \Rightarrow
  \;\;\;
  \begin{tikzcd}[row sep=small, 
    column sep=25pt
  ]
    X_\bullet
    \times^h_{B_\bullet}
    Y_\bullet
    \ar[r]
    \ar[d]
    &
    Y_\bullet
    \ar[d]
    \ar[
      dl,
      Rightarrow,
      shorten=15pt
    ]
    \\
    X_\bullet
    \ar[r]
    &
    B_\bullet
    \mathrlap{\,.}
  \end{tikzcd}
\end{equation}

%%%%%%%%%%%%%
\paragraph
{{\v C}ech Cohomology}
%%%%%%%%%%%%
Given a presheaf of chain complexes $A_\bullet$ \cref{PresheafOfChainComplexes}, then for $X$ a \text{(super-)} manifold and $\inlinetikzcd{ U_i \ar[r, hook, "{ \iota_i }"] \& X }$ a good open cover, meaning an open cover such that all chart intersections
\begin{equation}
  U_{i_1 \cdots i_k}
  :=
  U_{i_1} 
    \cap \cdots \cap
  U_{i_k}
\end{equation}
are themselves charts (if not empty), the \emph{{\v C}ech cohomology} of $X$ with coefficients in $A_\bullet$ is the chain homology of the \emph{{\v C}ech complex} with entries
\begin{equation}
  C_n\bracket({X;A_\bullet})
  :=
  \bigoplus_{l-k = n}
  \,
  \bigoplus_{i_0 \cdots i_k}
  A_{l}\bracket({
    U_{i_0 \cdots i_k}
  })  
\end{equation}
and differential
\begin{equation}
  \begin{aligned}
    \bracket({
      \partial a
    })_{i_0 \cdots i_{k+1}}
    & :=
    \bracket({
      \partial^A a
      +
      (-1)^n
      \delta a
    })_{i_0 \cdots i_{k+1}}
    \\[2pt]
    & :=
      \partial^A 
      a_{i_0 \cdots i_{k+1}}
      +
      (-1)^n
      \sum_{
        0 \leq j \leq k+1
      }
      (-1)^j
      a_{
        i_0 \cdots 
          i_{j-1} i_{j+1}
          \cdots
          i_{k+1}
      }
      \vert_{
        U_{i_0 \cdots i_{k+1}}
      }
      \mathrlap{\,.}
  \end{aligned}
\end{equation}
We write:
\begin{equation}
\label
{CechCohomology}
  H^n\bracket({
    X;
    A_\bullet
  })
  :=
  \frac
   {
     Z_0\bracket({
       X;
       A_\bullet[n]
      })
  }
  {
    B_0\bracket({
      X; 
      A_\bullet[n]
    })
  }
  :=
  \frac{
    \mathrm{ker}
    \big(
    \inlinetikzcd{
      C_0\bracket({X;A_\bullet[n]})
      \ar[r, "{ \partial_{-1} }"]
      \&
      C_{-1}\bracket({X;A_\bullet[n]})
    }
    \big)
  }{
    \mathrm{im}
    \big(
    \inlinetikzcd{
      C_1\bracket({X;A_\bullet[n]})
      \ar[r, "{ \partial_0 }"]
      \&
      C_{0}\bracket({X;A_\bullet[n]})
    }
    \big)    }
    \mathrlap{\,.}
\end{equation}

Examples of {\v C}ech cohomologies:
\begin{enumerate}
\item
\emph{Integral ordinary cohomology} is {\v C}ech cohomology with coefficients in $\mathbb{Z}$ \cref{ShiftedAbelianGroupComplexOfSheaves};
\begin{equation}
\label
{IntegralCohomologyAsCechCohomology}
  H^n\bracket({
    X; \mathbb{Z}
  })
  \simeq
  H^0\bracket({
    X; \mathbb{Z}[n]
  })
  \mathrlap{\,.}
\end{equation}
The cocycles are
\begin{equation}
  \big(
    n_{i_0 i_1 \cdots i_n}
    \in
    \mathbb{Z}
  \big)_{\! i_j \in I}
  \text{ s.t. }
  \left\{
    \begin{aligned}
    & n_{i_1 i_2 i_3 \cdots i_{n+1}}
    \\
    - &
    n_{i_0 i_2 i_3 \cdots i_{n+1}}
    \\
    + &
    n_{i_0 i_1 i_3 \cdots i_{n+1}}
    \\
    - & \cdots
    \end{aligned}
    =
    0
  \right.
  \text{
    on
    $U_{i_0 i_1 \cdots i_{n+1}}$.
  }
\end{equation}
For instance, consider flat $1+d$-dimensional spacetime $\mathbb{R}^{1,d}$ with the worldvolume of a $p$-brane removed. One finds:
\begin{equation}
\label
{OrdinaryCohomologyOfHigherDiracMonopoles}
  H^{n}\big({
    \underbrace{
    \mathbb{R}^{1,d}
    \setminus
    \mathbb{R}^{1,p}
    }_{
    \mathclap{
    \mathbb{R}^{1,p}
    \,\times\, 
    \mathbb{R}_{> 0}
    \,\times\,
    S^{d-p-1}    
    }
    }
    \,;
    \mathbb{Z}
  }\big)
  \simeq
  H^n\bracket({
    S^{d-p-1}
    ;
    \mathbb{Z}
  })
  =
  \begin{cases}
    \mathbb{Z}
    &
    \text{ if $n = d-p-1$}
    \\
    0 & \text{ otherwise.}
  \end{cases}
\end{equation}
These are the possible charges of \emph{higher Dirac monopoles} with charges in ordinary cohomology. We are next after  higher gauge fields that reflect these charges.

\item
Given a sheaf $A$ of abelian groups regarded as a chain complex concentrated in degree 0, then 
\begin{equation}
  H^0\bracket({X; A})
  \simeq
  A(X)
\end{equation}
is simply the group of global sections of that sheaf. For instance
\begin{equation}
  H^0\bracket({
    X; 
    \Omega^{k}_{\mathrm{dR}}
  })
  \simeq
  \Omega^k_{\mathrm{dR}}\bracket({X})
  \mathrlap{\,.}
\end{equation}

\end{enumerate}

%%%%%%%%%%%%%%
\paragraph
{Deligne Cohomology and Dirac Charge Quantization}
%%%%%%%%%%%%%

Differential ordinary cohomology is, as a {\v C}ech cohomology theory, represented by the smooth \emph{Deligne complex} (cf. \parencites[\S I.5]{Brylinski1993}{Gajer1997}, also \parencites[\S 2.5]{FSS13-CupCS}[\S 3.1]{FSS15-Stacky}{Grady2017}{Grady2019}), which is a kind of combination of \cref{DeRhamComplexOfSheaves} and \cref{ShiftedAbelianGroupComplexOfSheaves}, but ending in not-necessarily closed forms:
\begin{equation}
\label
{TheDeligneComplex}
  \mathbb{Z}[n+1]_{\mathrm{dff}}(-)
  :=
  \Big(
    \inlinetikzcd{
      \mathbb{Z}
      \ar[r, hook]
      \&
      \Omega^0_{\mathrm{dR}}(-)
      \ar[r, "{ \mathrm{d} }"]
      \&
      \cdots
      \ar[r, "{ \mathrm{d} }"]
      \&
      \Omega^{n-2}_{\mathrm{dR}}(-)
      \ar[r, "{ \mathrm{d} }"]
      \&
      \Omega^{n-1}_{\mathrm{dR}}(-)
      \ar[r, "{ \mathrm{d} }"]
      \&
      \Omega^n_{\mathrm{dR}}(-)
    }
  \Big)
  \mathrlap{\,.}
\end{equation}
The {\v C}ech cohomology with these coefficients is \emph{differential ordinary cohomology} in the guise of smooth \emph{Deligne cohomology}:
\begin{equation}
  H^{n+1}_{\mathrm{dff}}\bracket({
    X;
    \mathbb{Z}
  })
  \simeq
  H^0\bracket({
    X;
    \mathbb{Z}[n+1]_{\mathrm{dff}}
  })
  \mathrlap{\,.}
\end{equation}

For instance, a 0-cycle with coefficients in $\mathbb{Z}_{\mathrm{dff}}[2]$ is
\begin{equation}
\label
{0CocycleInCechDeligne2Cohomology}
\hspace{-2.5mm} 
  \left(
  \begin{aligned}
    A_{i_0} 
    &
    \in
    \Omega^1_{\mathrm{dR}}\bracket({
      U_{i_0}
    })
    \\
    \lambda_{i_0 i_1}
    & 
    \in
    \Omega^0_{\mathrm{dR}}\bracket({
      U_{i_0 i_1}
    })
    \\
    n_{i_0 i_1 i_2}
    &
    \in
    \mathbb{Z}
  \end{aligned}
  \! \right)_{\!\!\mathrlap{i_j \in I}}
  \;\;\,\text{s.t.} \,
  \left\{
  \begin{alignedat}{2}
    \mathrm{d}\lambda_{i_0 i_1}
    &
    +
    A_{i_1} 
      - 
    A_{i_0}
    &&
    = 
    0
    \;
    \text{
      on
      $U_{i_0 i_1}$
    }
    \\
    n_{i_0 i_1 i_2}
    &+ 
    \lambda_{i_1 i_2}
    -
    \lambda_{i_0 i_2}
    +
    \lambda_{i_0 i_1}
    && 
    =
    0
    \text{
      on
      $U_{i_0 i_1 i_2}$
    }
    \\
    & 
    n_{i_1 i_2 i_3}
    - 
    n_{i_0 i_2 i_3}
    +
    n_{i_0 i_1 i_3}
    -
    n_{i_0 i_1 i_2}
    &&
    = 0
    \text{
      on
      $U_{i_0 i_1 i_2 i_3}$.
    }
  \end{alignedat}
  \right.
\end{equation}
This is exactly the {\v C}ech data for the electromagnetic field under ordinary flux quantization. 

Similarly, a 1-chain between a pair of such 0-cycles is
\begin{equation}
\hspace{-2mm} 
  \left(
  \begin{aligned}
    a_{i_0} 
    & \in 
    \Omega^0_{\mathrm{dR}}\bracket({
      U_{i_0}
    })
    \\
    \kappa_{i_0 i_1} & \in
    \mathbb{Z}
  \end{aligned}
  \right)_{\!\!\mathrlap{i_j \in I}}
  \;\;\;\;\text{s.t.}\;
  \left\{
  \begin{alignedat}{5}
    \mathrm{d}\, a_{i_0}
    &
    &&=
    A'_{i_0} - A_{i_0}
    &\;&
    \text{on
      $U_{i_0}$
    }
    \\
    \kappa_{i_0 i_1}
    &
    -
    a_{i_1} + a_{i_0}
    &&
    =
    \lambda'_{i_0 i_1} - 
    \lambda_{i_0 i_1}
    &&
    \text{on
      $U_{i_0 i_1}$
    }
    \\
    &
    -\kappa_{i_1 i_2}
    +
    \kappa_{i_0 i_2}
    -
    \kappa_{i_0 i_1}
    &&
    =
    n'_{i_0 i_1 i_2}
    - 
    n_{i_0 i_1 i_2}
    &&
    \text{on
      $U_{i_0 i_1 i_2}$,
    }
  \end{alignedat}
  \right.
\end{equation}
which is the global data of a gauge transformation between a pair of globally completed electromagnetic fields.

\begin{figure}[htb]
\caption{\label{DeligneComplexAsPullback}
Shown is how the Deligne complex \cref{TheDeligneComplex} completes a (homotopy) pullback diagram \cref{HomotopyPullbackOfChainComplexes} of chain complexes which presents the homotopy pullback \cref{DeligneComplexAsHomotopyPullback}, cf. \cite[proof of Prop. 9.5]{FSS23-Char}.
}
\centering
\adjustbox{
  rndfbox=4pt
}{
$
  \begin{tikzcd}[column sep=50pt,
    ampersand replacement=\&
  ]
    \left(
    \begin{aligned}
      &
      \;\mathbb{Z}
      \\
      &
      \;
      \rotatebox[origin=c]{-90}
        {$\hookrightarrow$}
      \\
      &
      \Omega^0_{\mathrm{dR}}
      \\
      &
      \downarrow\!\!
      \mathrlap{\scalebox{.7}{$\mathrm{d}$}}
      \\
      &
      \Omega^1_{\mathrm{dR}}
      \\
      &
      \downarrow\!\!
      \mathrlap{\scalebox{.7}{$\mathrm{d}$}}
      \\
      &
      \;\,
      \vdots
      \\
      &
      \downarrow\!\!
      \mathrlap{\scalebox{.7}{$\mathrm{d}$}}
      \\
      &
      \Omega^{n-1}_{\mathrm{dR}}
      \\
      &
      \downarrow\!\!
      \mathrlap{\scalebox{.7}{$\mathrm{d}$}}
      \\
      &
      \Omega^n_{\mathrm{dR}}
    \end{aligned}
    \right)
    \ar[
      dr,
      phantom,
      "{ \lrcorner }"{pos=.4}
    ]
    \ar[
      d,
      "{
        \scaledbracket({
          \substack{
            0
            \\
            0
            \\
            \vdots
            \\
            0
            \\
            \mathrm{d}
          }
        })
      }"
    ]
    \ar[
      r,
      "{ i_1 }"
    ]
    \&[20pt]
    \left(
    \def\arraycolsep{0pt}
    \def\arraystretch{1.05}
    \begin{array}{ccc}
      \;\mathbb{Z}
      &
      \oplus
      &
      \Omega^0_{\mathrm{dR}}
      \\
      \rotatebox[origin=c]{-90}
        {$\hookrightarrow$}
      &
      \rotatebox[origin=c]{+45}
        {$\overset{\!\!\pm\mathrm{id}}{\longleftarrow}$}
      &
      \downarrow\!\!
      \mathrlap{\scalebox{.7}{$\mathrm{d}$}}
      \\
      \Omega^0_{\mathrm{dR}}
      &\oplus&
      \Omega^1_{\mathrm{dR}}
      \\
      \downarrow\!\!
      \mathrlap{\scalebox{.7}{$\mathrm{d}$}}
      &
      \rotatebox[origin=c]{+45}
        {$\overset{\!\!\mp\mathrm{id}}{\longleftarrow}$}
      &
      \downarrow\!\!
      \mathrlap{\scalebox{.7}{$\mathrm{d}$}}
      \\
      \Omega^1_{\mathrm{dR}}
      &
      \oplus
      &
      \Omega^2_{\mathrm{dR}}
      \\
      \downarrow\!\!
      \mathrlap{\scalebox{.7}{$\mathrm{d}$}}
      &
      \rotatebox[origin=c]{+45}
        {$\overset{\!\!\pm\mathrm{id}}{\longleftarrow}$}
      &
      \downarrow\!\!
      \mathrlap{\scalebox{.7}{$\mathrm{d}$}}
      \\
      \vdots
      &
      \,
      \vdots
      &
      \vdots
      \\
      \downarrow\!\!
      \mathrlap{\scalebox{.7}{$\mathrm{d}$}}
      &
      \rotatebox[origin=c]{+45}
        {$\overset{\!\!-\mathrm{id}}{\longleftarrow}$}
      &
      \downarrow\!\!
      \mathrlap{\scalebox{.7}{$\mathrm{d}$}}
      \\
      \Omega^{n-1}_{\mathrm{dR}}
      &\oplus&
      \Omega^n_{\mathrm{dR}}
      \\
      \downarrow\!\!
      \mathrlap{\scalebox{.7}{$\mathrm{d}$}}
      &
      \rotatebox[origin=c]{+45}
        {$\overset{\!\!+\mathrm{id}}{\longleftarrow}$}
      &
      \\
      \Omega^n_{\mathrm{dR}}
    \end{array}
    \right)
    \ar[
      d,
      ->>,
      "{
        \scaledbracket({
          \substack{
            \mathrm{pr}_2
            \\
            \mathrm{pr}_2
            \\
            \vdots
            \\
            \mathrm{pr}_2
            \\
            \mathrm{d}
          }
        })
      }"
    ]
    \&[20pt]
    \left(
    \begin{aligned}
      &
      \;\mathbb{Z}
      \\
      &
      \downarrow
      \\
      &
      \;
      0
      \\
      &
      \downarrow
      \\
      &
      \;
      0
      \\
      &
      \downarrow
      \\
      &
      \;\,
      \vdots
      \\
      &
      \downarrow
      \\
      &
      \;
      0
      \\
      &
      \downarrow
      \\
      &
      \;
      0
    \end{aligned}
   \, \right)
    \ar[
      l,
      "{ \sim }"{swap}
    ]
    \ar[
      d,
      "{
        \bracket({
        \substack{
          i
          \\
          0
          \\
          \vdots
          \\
          0
          \\
          0
        }
        })
      }"
    ]
    \\[35pt]
    \left(
    \begin{aligned}
      &
      \;\,0
      \\
      &
      \,
      \downarrow
      \\
      &
      \;\,
      0
      \\
      &
      \,
      \downarrow
      \\
      &
      \;\,
      0
      \\
      &
      \,
      \downarrow
      \\
      &
      \;\,
      \vdots
      \\
      &
      \,
      \downarrow
      \\
      &
      \;\,
      0
      \\
      &
      \,
      \downarrow
      \\
      &
      \Omega^{n+1}_{\mathrm{cl}}
    \end{aligned}
    \right)
    \ar[
      r,
      "{ 
        \bracket({
        \substack{
          0
          \\
          0
          \\
          \vdots
          \\
          0
          \\
          \mathrm{id}
        }
        })
      }"{description}
    ]
    \&
    \left(
    \begin{aligned}
      &
      \Omega^0_{\mathrm{dR}}
      \\
      &
      \downarrow\!\!
      \mathrlap{\scalebox{.7}{$\mathrm{d}$}}
      \\
      &
      \Omega^1_{\mathrm{dR}}
      \\
      &
      \downarrow\!\!
      \mathrlap{\scalebox{.7}{$\mathrm{d}$}}
      \\
      &
      \;\,
      \vdots
      \\
      &
      \downarrow\!\!
      \mathrlap{\scalebox{.7}{$\mathrm{d}$}}
      \\
      &
      \Omega^{n-1}_{\mathrm{dR}}
      \\
      &
      \downarrow\!\!
      \mathrlap{\scalebox{.7}{$\mathrm{d}$}}
      \\
      &
      \Omega^n_{\mathrm{dR}}
      \\
      &
      \downarrow\!\!
      \mathrlap{\scalebox{.7}{$\mathrm{d}$}}
      \\
      &
      \Omega^{n+1}_{\mathrm{cl}}
    \end{aligned}
    \right)
    \&
    \left(
    \begin{aligned}
      &
      \;\mathbb{R}
      \\
      &
      \downarrow
      \\
      &
      \;
      0
      \\
      &
      \downarrow
      \\
      &
      \;
      0
      \\
      &
      \downarrow
      \\
      &
      \;\,
      \vdots
      \\
      &
      \downarrow
      \\
      &
      \;
      0
      \\
      &
      \downarrow
      \\
      &
      \;
      0
    \end{aligned}
   \, \right)
    \ar[
      l,
      "{ \sim }"{swap}
    ]
  \end{tikzcd}
$
}
\end{figure}

Under the cohomology operation induced by the evident chain map projecting out the integer piece in \cref{TheDeligneComplex},
\begin{equation}
  \begin{tikzcd}[
    row sep=0pt
  ]
    \mathbb{Z}[n+1]_{\mathrm{dff}}
    \ar[r, "{ \boldsymbol{\chi} }"]
    &
    \mathbb{Z}[n+1]
    \\
    H^{n+1}_{\mathrm{dff}}\bracket({X; \mathbb{Z}})
    \ar[r, "{ \chi }"]
    &
    H^{n+1}\bracket({X; \mathbb{Z}})
    \mathrlap{\,,}
  \end{tikzcd}
\end{equation}
these differential cocycles project onto the integral cocycles \cref{IntegralCohomologyAsCechCohomology}. In particular, we see that underlying the globalized electromagnetic field \cref{0CocycleInCechDeligne2Cohomology} are integrally quantized charges of magnetic monopoles \cref{OrdinaryCohomologyOfHigherDiracMonopoles}. This is the modern form of Dirac's original magnetic charge/flux quantization argument.

%%%%%%%%%%%%
\paragraph
{Differential Ordinary Cohomology as a Fiber Product}
%%%%%%%%%%%%

The point for our purposes now is a slick reformulation of Deligne cohomology in a way that generalizes to flux quantization of Bianchi identities with non-linear self-couplings.
To this end, we observe (in \cref{DeligneComplexAsPullback}) that the Deligne complex 
\cref{TheDeligneComplex} is a homotopy fiber product \cref{HomotopyPullbackOfChainComplexes}
of this form:
\begin{equation}
\label
{DeligneComplexAsHomotopyPullback}
  \begin{tikzcd}[row sep=15pt,
    column sep=25pt]
    \mathbb{Z}[n]_{\mathrm{dff}}
    \ar[
      r
    ]
    \ar[d]
    \ar[
      dr,
      phantom,
      "{ \lrcorner }"{pos=0}
    ]
    &
    \mathbb{Z}[n]
    \ar[d]
    \ar[
      dl,
      Rightarrow,
      shorten=10pt
    ]
    \\
    \Omega^{n}_{\mathrm{cl}}
    \ar[r]
    &
    \mathbb{R}[n]
  \end{tikzcd}
  \;\;\;
  \text{hence}
  \;\;\;
  \mathbb{Z}[n]_{\mathrm{dff}}
  \simeq
  \Omega^n_{\mathrm{dR}}
  \underset{\mathbb{R}[n]}{\times^h}
  \mathbb{Z}[n]
  \mathrlap{\,.}
\end{equation}

This makes nicely manifest how passage to Deligne cohomology is the result of combining 
\begin{enumerate}
\item
closed differential forms (abelian flux densities) with 
\item
integral cohomology classes (the corresponding Dirac charges), 
\end{enumerate}
compatibly with their joint image in real cohomology.

\clearpage

%%%%%%%%%%%%%
\subsubsection
{General Charges in Differential Nonabelian Cohomology}
\label
{OnGeneralChargesInDiffNonabCohomology}
%%%%%%%%%%%%%

We now need to generalize the above situation (\cref{OnHigherDiracChargesInDiffOrdCohomology}) from chain complexes to their non-linear generalization to higher gauged sets (\cref{HigherGaugedSets}).

%%%%%%%%%%%%%
\paragraph
{From Chain Complexes to Higher Gauged Sets}
%%%%%%%%%%%%%

As we generalize from higher Maxwell/Dirac theory with its linear Bianchi identities $\mathrm{d} F_{n+1} = 0$, to general higher Maxwell\emph{-type} Bianchi identities (\cref{OnHigherMaxwellTypeBianchiIdentities}) characterized by non-abelian $L_\infty$-algebras $\mathfrak{a}$ \cref{TheCharacteristicLInfinityAlgebra}, the set of solutions to these Bianchis is no longer an abelian group $\Omega^{n+1}_{\mathrm{cl}}(X) \simeq \Omega^1_{\mathrm{cl}}\bracket({X; \mathbb{R}[n]})$, but just a set $\Omega^1_{\mathrm{cl}}\bracket({X;\mathfrak{a}})$: The sum of solutions to a non-linear equation like \cref{CFieldEoM} is no longer a solution, in general.

Hence we are to generalize from presheaves of chain complexes \cref{PresheafOfChainComplexes} to presheaves of Kan simplicial sets (cf. \cref{FromChainComplexesToSimplicialSets}).

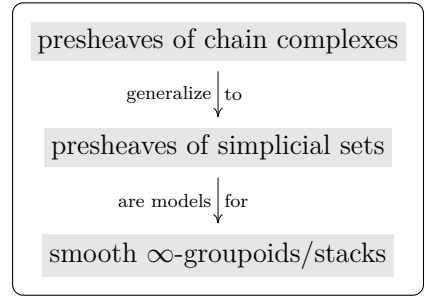
\begin{SCfigure}[1][htb]
\caption{\label
{FromChainComplexesToSimplicialSets}
The generalization from linear to non-linear Bianchi identities requires generalizing the presheaves (over Cartesian charts) of chain complexes, which describe ordinary Dirac flux quantization, to presheaves of (Kan) simplicial sets. 
}
\centering
\adjustbox{scale=0.88, 
  rndfbox=4pt
}{
\begin{tikzcd}[
  row sep=20pt
]
\colorbox{lightgray}{$
\text{presheaves of chain complexes} 
$}
\ar[
  d,
  "{
    \text{generalize}
  }"{swap},
  "{
    \text{to}
  }" 
]
\\
\colorbox{lightgray}{$
\text
{presheaves of simplicial sets}
$}
\ar[
  d,
  "{
    \text{are models}
  }"{swap},
  "{
    \text{for}
  }"
]
\\
\colorbox{lightgray}{$
\text
{smooth $\infty$-groupoids/stacks}
$}
\end{tikzcd}
}
\end{SCfigure}

Regarded under stalkwise homotopy equivalence, these presheaves of simplicial sets on charts are called \emph{smooth $\infty$-groupoids} (or \emph{smooth $\infty$-stacks}), cf. \parencites[\S 1]{FSS23-Char}{FSS15-Stacky}[\S 4]{SS26-Bun}{SS25-Orient}.

For instance, the chain complex $A[n]$ is really a stand-in for the Eilenberg-MacLane space traditionally denoted $K\bracket({A,n})$, which we denote $B^n A$ (cf. \cref{Homotopy}).
A classical result (cf. \cite[Ex. 2.1]{FSS23-Char}) identifies this as the classifying space for ordinary cohomology, in that ordinary cohomology classes are equivalently homotopy classes of maps into such a EM-space:
\begin{equation}
  H^n\bracket({
    X;
    \mathbb{Z}
  })
  \simeq
  \pi_0
  \,\mathrm{Map}\bracket({
    X, B^n \mathbb{Z}
  })
  \mathrlap{\,.}
\end{equation}
This shows that the space $B^n \mathbb{Z}$ indeed plays the same role as the chain complex $\mathbb{Z}[n]$ in \cref{IntegralCohomologyAsCechCohomology}.

In this sense, the above presheaves of chain complexes generalize to presheaves of spaces.
In this more general context, the homotopy pullback \cref{DeligneComplexAsHomotopyPullback} hence reads equivalently:
\begin{equation}
  \begin{tikzcd}[row sep=15pt,
    column sep=25pt
  ]
    \bracket({
      B^n \mathbb{Z}
    })_{\mathrm{dff}}
    \ar[
      r
    ]
    \ar[d]
    \ar[
      dr,
      phantom,
      "{ \lrcorner }"{pos=.1}
    ]
    &
    B^n \mathbb{Z}
    \ar[d]
    \ar[
      dl,
      Rightarrow,
      shorten=15pt
    ]
    \\
    \Omega^{n}_{\mathrm{cl}}
    \ar[r]
    &
    B^n \mathbb{R}
    \mathrlap{\,.}
  \end{tikzcd}
\end{equation}

In this formulation, cocycles in the cohomology theories represented by the vertices of this diagram are represented by maps into these spaces:
\begin{equation}
\label
{AbelianDiracQuantizationDiagrammatically}
  \begin{tikzcd}[row sep=small]
    &
    & 
    B^n \mathbb{Z}
    \ar[dd]
    \\
    X
    \ar[
      urr,
      bend left=20,
      "{ \chi }"{
        swap,
        name=charge
      }
    ]
    \ar[
      dr,
      "{ F_{n} }"{
        name=flux
      }
    ]
    \ar[
      Rightarrow,
      from=charge,
      to=flux,
      "{\color{darkgreen} \widehat{A} }"
    ]
    \\
    &
    \Omega^{n}_{\mathrm{cl}}
    \ar[
      r
    ]
    &
    B^n \mathbb{R}
    \mathrlap{\,.}
  \end{tikzcd}
\end{equation}

%%%%%%%%%%%%%
\paragraph
{IR-Completed Fields in Differential Nonabelian Cohomology}
%%%%%%%%%%%%%
This perspective now makes it clear how to generalize Dirac flux quantization from linear to non-linear Bianchi identities characterized by possibly non-abelian  $L_\infty$-algebras $\mathfrak{a}$. We need:
\begin{enumerate}

  \item 
  a classifying space $\mathcal{A}_{\mathbb{R}}$ for $\mathfrak{a}$-valued non-abelian de Rham cohomology,

  \item 
  a canonical map
  \[
    \inlinetikzcd{
      \Omega^1_{\mathrm{cl}}(-;\mathfrak{a})
      \ar[r]
      \&
      \mathcal{A}_{\mathbb{R}}
    }
    \mathrlap{\,,}
  \]
  which classifies sending $\mathfrak{a}$-valued closed forms to their non-abelian de Rham cohomology class,

  \item
  a classifying \emph{nonabelian character map}
  \[
    \begin{tikzcd}
      \mathcal{A}
      \ar[
        r,
        "{ \mathbf{ch}^{\mathcal{A}} }"
      ]
      &
      \mathcal{A}_{\mathbb{R}}
      \mathrlap{\,,}
    \end{tikzcd}
  \]
  which extracts exactly those aspects of $\mathcal{A}$-valued charges that are reflected by $\mathfrak{a}$-valued differential forms,
\end{enumerate}
for any classifying space $\mathcal{A}$ which is compatible with the given characteristic $L_\infty$-algebra $\mathfrak{a}$, in that
\begin{equation}
\label
{CompatibilityForCLassifyingSpace}
  \mathcal{A}_{\mathbb{R}}
    \simeq 
  \shape \Omega^1_{\mathrm{dR}}(-;\mathfrak{a})
  \mathrlap{\,.}
\end{equation}

The fact that such \emph{character maps in nonabelian cohomology} exist, and how they are constructed, is the content of the eponymous monograph \cite{FSS23-Char}, using classical results of rational homotopy theory. Expository surveys are given in \parencites[\S 3]{SS25-Flux}[\S 4.1.2]{SS26-HigherGauge}.
This shows that the compatibility condition \cref{CompatibilityForCLassifyingSpace} is equivalent to the demand that the \emph{real Whitehead-bracket} $L_\infty$-algebra $\mathfrak{l}(-)$ of $\mathcal{A}$ \cref{WhiteheadBracketLInfinityAlgebra} is $\mathfrak{a}$:
\begin{equation}
\label
{WhBracketCoincidingWithCharacteristic}
  \mathfrak{l}\mathcal{A}
  \simeq
  \mathfrak{a}
  \mathrlap{\,.}
\end{equation}

With this in hand, the now evident generalization of the previous situation \cref{AbelianDiracQuantizationDiagrammatically} yields the desired global data for a higher gauge field with duality-symmetric Bianchi identities characterized by $\mathfrak{a}$ and charges classified by $\mathcal{A}$ subject to the compatibility condition \cref{CompatibilityForCLassifyingSpace,WhBracketCoincidingWithCharacteristic}; such a field configuration is the data of the dashed maps in a diagram of smooth $\infty$-groupoids (\cref{FromChainComplexesToSimplicialSets}) of the following form (cocycles in \emph{differential $\mathcal{A}$-cohomology} \cite[Def. 9.3]{FSS23-Char}): 
\begin{equation}
\label{DifferentialNonabelianCohomologyDiagram}
\hspace{-3mm}
  \begin{tikzcd}[
   row sep=0pt, column sep=5pt
  ]
    &[10pt]
    &[10pt]
    &
    \mathcal{A}
    \ar[
      dd,
      "{ \mathbf{ch}^{\mathcal{A}} }"
    ]
    \\[35pt]
    X
    \ar[
      urrr,
      dashed,
      bend left=20,
      "{ \chi }"{
        swap,
        name=charge,
        pos=.3
      },
      "{
        \text{\color{darkblue}charges}
      }"{sloped, pos=.3}
    ]
    \ar[
      ddr,
      dashed,
      "{ \vec F }"{
        name=flux
      },
      "{
        \text{\color{darkblue}fluxes}
      }"{sloped, swap}
    ]
    \ar[
      Rightarrow,
      from=charge,
      to=flux,
      dashed,
      "{ \widehat{A} }"{swap},
      "{
        \text{\color{darkgreen}potentials}
      }"{sloped, , yshift=-2pt, swap}
    ]
    \\
    &&&
    \mathcal{A}_{\mathbb{R}}
    \\
    &
    \Omega^{1}_{\mathrm{cl}}({
      -;\mathfrak{a}
    })
    \ar[
      r
    ]
    &
    \shape \Omega^1_{\mathrm{cl}}({
      -;\mathfrak{a}
    })
    \ar[
      ur,
      phantom,
      "{ \simeq }"{sloped}
    ]
  \end{tikzcd}
  \;\;
  \Leftrightarrow
  \;
  \begin{tikzcd}[sep=35pt]
    X
    \ar[
      r, 
      dashed,
      "{
        (\vec F, \hat A, \chi)
      }"
    ]
    &
    \mathcal{A}_{\mathrm{dff}}
    :=
    \Omega^1_{\mathrm{cl}}(
      -;
      \mathfrak{a}
    )
    \underset{
      \mathcal{A}_{\mathbb{R}}
    }{
      \times^h
    }
    \mathcal{A}
    \mathrlap{\,.}
  \end{tikzcd}
\end{equation}

\begin{example}
  The global completions of 11D SuGra are parameterized by classifying spaces $\mathcal{A}$ whose $\mathbb{R}$-rational homotopy type is that of the 4-sphere \cref{MTheoryGaugeAlgebra}:
  \begin{equation}
  \label
  {AdmissibilityConditionFor11DSuGra}
    \mathfrak{l}\mathcal{A}
    \simeq
    \mathfrak{l}S^4
    \,.
  \end{equation}
  For any such choice, a global $C$-field configuration according to \cref{DifferentialNonabelianCohomologyDiagram} is given with respect to an open cover $\{ \inlinetikzcd{ U_i \ar[r, "{ \iota_i }"] \& X } \}_{i \in I}$ of (super-) spacetime by {\v C}ech data of the following form:
  \begin{enumerate}
    \item
    On charts $U_i$ potentials \cref{LocalCFieldPotential} for the restriction of the flux densities to these charts:
    \begin{equation}
      \left(
      \begin{aligned}
        C^{(i)}_3
        \\
        C^{(i)}_6
      \end{aligned}
      \right)
      \;\; \text{ s.t. } \;\;
      \begin{aligned}
        \mathrm{d}\, C^{(i)}_3 & = G_4
        \\
        \mathrm{d}\, C^{(i)}_6 & = G_7 
        - \tfrac{1}{2} C^{(i)}_3 G_4 
        \mathrlap{\,.}
      \end{aligned}
    \end{equation}

    \item
    On double intersections $U_{i j}$ gauge transformations \cref{LocalCFieldGaugeTransformations} between the restrictions of the gauge potentials on $U_i$ and $U_j$, respectively:
    \begin{equation}
      \begin{aligned}
        \mathrm{d}\, C^{(ij)}_2
        & =
        C_3^{(j)} - C_3^{(i)}
        \\
        \mathrm{d}\, C^{(ij)}_5
        & =
        C_6^{(j)} - C_6^{(i)}
        -
        \tfrac{1}{2}C^{(j)}_3
        C_3^{(i)}
        \mathrlap{\,.}
      \end{aligned}
    \end{equation}

    \item 
    On triple intersections $U_{i j k}$ gauge-of-gauge transformations \cref{LocalCFieldGaugeOfGaugeTransformations} 
        \[
      \begin{tikzcd}[
        row sep=50pt
      ]
        &
        \left(
        \substack{
          C_3^{(j)}
          \\
          C_6^{(j)}
        }
        \right)
        \ar[
          dr,
          "{
            \left(\color{darkblue}
            \substack{
              C_2^{(jk)}
              \\
              C_5^{(jk)}
            }
            \right)
          }"
        ]
        \ar[
          d,
          Rightarrow,
          shorten=5pt,
          "{\color{darkgreen}
            \left(
            \substack{
              C_1^{(ijk)}
              \\
              C_4^{(ijk)}
            }
            \right)
          }"{description}
        ]
        \\
        \left(
        \substack{
          C_3^{(i)}
          \\
          C_6^{(i)}
        }
        \right)
        \ar[
          ur,
          "{\color{darkblue}
            \left(
            \substack{
              C_2^{(ij)}
              \\
              C_5^{(ij)}
            }
            \right)
          }"
        ]
        \ar[
          rr,
          "{\color{darkblue}
            \left(
            \substack{
              C_2^{(ik)}
              \\
              C_5^{(ik)}
            }
            \right)
          }"{swap}
        ]
        &{}&
        \left(
        \substack{
          C_3^{(k)}
          \\
          C_6^{(k)}
        }
        \right)
      \end{tikzcd}
      \;\;\;\;\;\;
      \text{on $U_{i j k}$}
    \]
    between the composition \cref{CompositionOfCFieldGaugeTransformations}, now for any (non-necessarily fixed) $\lambda \in \mathbb{R}$, of the restriction of the gauge transformations on $U_{i j}$ and $U_{jk}$ to the restriction of the gauge transformation on $U_{i k}$:
    \begin{equation}
      \begin{aligned}
        \mathrm{d}\,
        C_1^{(i j k)}
        & =
        C_2^{(i k)}
        -
        C_2^{(i j)} - C_2^{(j k)}
        \\
        \mathrm{d}\,
        C_4^{(i j k)}
        & =
        C_5^{(i k)}
        -
        C_5^{(i j)}
        -
        C_5^{(j k)}
        +
        \tfrac{1}{2}
        C_2^{(ij)}
        \mathrm{d}
        C_2^{(j k)}
        -
        \tfrac{\lambda}{2}\dd \big(
        C_2^{(i j)}
        C_2^{(jk)} \big) 
        \mathrlap{\,.}
      \end{aligned}
    \end{equation}
    
    \item 
    On higher intersections 
    the corresponding higher order differential form data \emph{and} compatibly further locally constant data of a cocycle in $\mathcal{A}$-cohomology.  
  \end{enumerate}

  This is the generalization for the 11D SuGra $C$-field of the {\v C}ech cohomology data \cref{0CocycleInCechDeligne2Cohomology} for the ordinary electromagnetic field. 

  However, in practice, one will tend to obtain the global $C$-field data \cref{DifferentialNonabelianCohomologyDiagram} using other methods. A slick construction of global $C$-field configurations (globalized in 4-Cohomotopy, cf. \cref{IRCompletionAccordingToHypothesisH}) analogous to multi-center ADHM instantons is given in \cite{GSS26-Conf}.

\end{example}

%%%%%%%%%%%%%%%%%%%
\subsubsection
{Brane Charges in Nonabelian Cohomology}
\label
{BraneChargesInNonabelianCohomology}
%%%%%%%%%%%%%%%%%%

The foremost effect of global IR-completion \cref{DifferentialNonabelianCohomologyDiagram} of SuGra theories is that it determines the topological charges carried by the singular branes of the theory (the higher analogs of Dirac monopoles). In particular, it determines 
\begin{enumerate}
\item
how these charges are discretized (``quantized''), reflecting indecomposable fundamental brane sources;
\item
the presence of torsion charges (a finite sum of which may vanish), reflecting the presence of \emph{fractional brane species}. 
\end{enumerate}

If $\mathcal{A}$ classifies the flux quantization law of the chosen IR completion \cref{WhBracketCoincidingWithCharacteristic}, then the charges \cref{DifferentialNonabelianCohomologyDiagram} are in the \emph{nonabelian cohomology} (\parencites[Def. 6.0.6]{Toen2002}[Def. 6]{Lurie2014}[\S 2]{FSS23-Char}, cf. \parencites[\S 1]{SS25-TEC}[\S 4]{SS26-Orb}) of spacetime with coefficients in $\mathcal{A}$:
\begin{equation}
  \begin{aligned}
    & 
    H^1\bracket({
      X; 
      \Omega \mathcal{A}
    })
    \;\;
    \substack{\text{(if $\mathcal{A}$ is connected)}}
    \\
    & \defneq
    H^0\bracket({
      X; 
      \mathcal{A}
    })
    \;\;
    \substack{\text{(generally)}}
    \\
    &
    :=
    \pi_0\, \mathrm{Map}\bracket({
      X; \mathcal{A}
    })
    \simeq
    \left\{
    \colorbox{lightgray}{$
      \substack{
        \text{Homotopy classes of}
        \\
        \text{maps $\inlinetikzcd{X \ar[r]\&  \mathcal{A}}$}
      }
      $}
    \right\}
    \mathrlap{\,.}
  \end{aligned}
\end{equation}
Concretely, if one measures these charges on near-horizon geometries of the black $p$-branes in $d$-dimensional space, in a generalization of \cref{OrdinaryCohomologyOfHigherDiracMonopoles}, then they take values in the homotopy groups $\pi_\bullet$ of the classifying space $\mathcal{A}$ (cf. \cite{SS23-Mf}, where we are now assuming that $\mathcal{A}$ is simply connected):
\begin{equation}
  \begin{aligned}
  H^1\bracket({
    \mathbb{R}^{1,d}
    \setminus
    \mathbb{R}^{1,p}
    ;
    \Omega\mathcal{A}
  })
  &
  \simeq
  H^1\bracket({
    S^{d-p-1};
    \Omega\mathcal{A}
  })
  \defneq
  \pi_0 \mathrm{Map}\bracket({
    S^{d-p-1},
    \mathcal{A}
  })
  \\
  & 
  \simeq
  \pi_{d-p-1}\bracket({
    \mathcal{A}
  })
  \mathrlap{\,.}
  \end{aligned}
\end{equation}

%%%%%%%%%%%%%
\paragraph
{Bulk Charges in 11D SuGra}
%%%%%%%%%%%%%

\begin{example}
[IR-Completion of 11D SuGra according to \emph{Hypothesis H}]
\label[example]
{IRCompletionAccordingToHypothesisH}
There is an IR-completion of 11D SuGra given by choosing in \cref{AdmissibilityConditionFor11DSuGra} the classifying space to be the 4-sphere itself:
\begin{equation}
\label
{HypothesisH}
  \mathcal{A}
  \defneq
  S^4
  \mathrlap{\,.}
\end{equation}
This choice was first considered in \parencites[\S 2.5]{Sati2018}[\S 4]{FSS15-M5WZW} and the hypothesis that this is the right choice for ``M-theory'' was called ``Hypothesis H'' (for quantization in co-\emph{H}omotopy theory) in \cite{FSS20-H,SS20-Tad,FSS21-Hopf} (where a list of consistency checks is presented supporting this hypothesis).

In any case, the choice \cref{HypothesisH} is ``weakly initial'' among IR-completions of 11D SuGra in that any other admissible CW-complex $\mathcal{A}$ \cref{AdmissibilityConditionFor11DSuGra} will be obtained by attaching further cells to the 4-sphere and hence will come with a canonical comparison map (including the 4-cell)
\begin{equation}
  \begin{tikzcd}
    S^4 
    \ar[r, "{ \iota }"]
    &
    \mathcal{A}
    \mathrlap{.}
  \end{tikzcd}
\end{equation}
which induces a canonical nonabelian cohomology operation
\begin{equation}
  \begin{tikzcd}
    \pi^4\bracket({X})
    \defneq
    H^1\bracket({
      X; \Omega S^4
    })
    \ar[
      r,
      "{ \iota_\ast }"
    ]
    &
    H^1\bracket({
      X; \Omega \mathcal{A}
    })
    \mathrlap{\,.}
  \end{tikzcd}
\end{equation}
from the nonabelian cohomology theory called \emph{4-Cohomotopy} (cf. \parencites[\S VII]{STHu59}[Ex. 2.7]{FSS23-Char}). This means that the M-brane charges according to Hypothesis H canonically map to the brane charges of any other IR-completion of 11D SuGra.

Now, according to Hypothesis H, fundamental M-brane charges are given by the homotopy groups of the 4-sphere. This implies, as shown in \cref{MBraneChargesAccordingToHypothesisH},  that Hypothesis H predicts integrally charged M2- and M5-branes, as it should be, among various fractional branes \cite[(22)]{SS23-Mf}.

\begin{table}[htb]
\caption{\label
{MBraneChargesAccordingToHypothesisH}
Possible charges of singular (``black'') M$p$-branes (as measured near-horizon) predicted by the IR-completion of 11D SuGra according to \emph{Hypothesis H} (\cref{IRCompletionAccordingToHypothesisH}).
Besides the expected integral charges carried by M2- and M5-branes, various fractional M-brane species may appear (which are invisible before IR-completion).
}
\adjustbox{scale=0.97,
  rndfbox=4pt
}{
\begin{tblr}{
  colspec = {rccccccccccl},
  row{1} = {bg=gray!30}
}
 $p =$ 
 & $0$
 & $1$
 & $2$ 
 & $3$
 & $4$ 
 & $5$ 
 & $6$
 & $7$ 
 & $8$
 & $9$
 \\
 & $0$ 
 & $0$ 
 & $\mathcolor{purple}{\mathbb{Z}}$
 & $0$ 
 & $0$ 
 & $\mathcolor{purple}{\mathbb{Z}}$
 & $0$ 
 & $0$
 & $0$ 
 & $0$
 &
 {\small (integral)}
 \\[-7pt]
 $
   \smash{
   \left.
   \begin{aligned}
   H^1\bracket({
   \mathbb{R}^{1,10}
   \setminus
   \mathbb{R}^{1,p};
   \Omega S^4
   })
   \\
   \simeq
   \pi_{9-p}\bracket({S^4})
   \end{aligned}
   \right\}
   }
 \simeq$
 & $\oplus$
 & $\oplus$
 & $\oplus$
 & $\oplus$
 & $\oplus$
 & $\oplus$
 & $\oplus$
 & $\oplus$
 & $\oplus$
 & $\oplus$
 \\[-2pt]
 & $\mathbb{Z}_{/2}^2$
 & $\mathbb{Z}_{/2}^2$
 & $\mathbb{Z}_{/12}$
 & $\mathbb{Z}_{/2}$
 & $\mathbb{Z}_{/2}$
 & $0$
 & $0$
 & $0$
 & $0$
 & $0$
 &
 {\small(fractional)}
\end{tblr}
}

\end{table}

If one feels that the fractional brane charges seen in \cref{MBraneChargesAccordingToHypothesisH} do not match desiderata on the global behavior of 11D SuGra, then one needs to choose another IR-completion given by another classifying space $\mathcal{A}$ of the rational homotopy type of the 4-sphere \cref{AdmissibilityConditionFor11DSuGra}.
\end{example}

%%%%%%%%%%%%%
\paragraph
{In the presence of Probe M5-Branes}
%%%%%%%%%%%%%

In the presence of a probe M5-brane worldvolume $\inlinetikzcd{ \Sigma^{1,5\vert 2 \cdot \mathbf{8}_+} \ar[r, "{ \Phi }"] \& X^{1,10\vert \mathbf{32}}}$ we have seen in \cref{M5braneProbes} that the joint bulk/brane Bianchi identities that put the theory on-shell are characterized by a  commuting diagram of this form \cite{GSS24-M5}:
\begin{equation}
\label
{M5BianchisAsCommutingDiagram}
  \begin{tikzcd}[column sep=huge]
    \Sigma^{1,5\vert 2 \cdot \mathbf{8}_+}
    \ar[
      r, 
      "{ H_3 }",
      dashed
    ]
    \ar[d, "{ \Phi }" ]
    &
    \Omega^1_{\mathrm{cl}}\bracket({
      -; 
      \mathfrak{l}_{_{S^4}}
      S^7
    })
    \ar[
      d,
      "{
        (\mathfrak{l}h_{\mathbb{H}})_\ast
      }"
    ]
    \\
    X^{1,10\vert \mathbf{32}}
    \ar[
      r, 
      "{  (G_4, G_7)  }",
      dashed
    ]
    &
    \Omega^1_{\mathrm{cl}}\bracket({
      -;
      \mathfrak{l}S^4
    })
  \end{tikzcd}
\end{equation}
saying that the $\mathbb{R}$-rational content of the fibration must be exactly that of the fibration of $L_\infty$-algebras characterizing the Bianchi identities \cref{M5BianchisAsCommutingDiagram}.

For such twisted relative Bianchi identities the global IR-completion from \cref{OnGeneralChargesInDiffNonabCohomology} accordingly generalizes from a single classifying space $\mathcal{A}$ for the bulk charges to a fibration $\inlinetikzcd{ \mathcal{B} \ar[r, "{ \mathcal{P} }"] \& \mathcal{A} }$ of classifying spaces, subject to the compatibility condition
\begin{equation}
  \mathfrak{l}\,
  \bracket({
    \inlinetikzcd{
      \mathcal{B}
      \ar[d, "{ \mathcal{B} }"]
      \\
      \mathcal{A}
    }
  \,})
  \simeq \;
  \inlinetikzcd{
    \mathfrak{l}_{_{S^4}}S^7
    \ar[d, "{ \mathfrak{l}h_{\mathbb{H}} }"]
    \\
    \mathfrak{l}S^4
    \mathrlap{\,.}
  }
\end{equation}
Given such a choice of fibration, the combined bulk/brane charges are classified by commuting diagrams of maps of this form \cite[(4.9)]{BaSS26-MString}:
\begin{equation}
\label
{RelativeClassifyingMaps}
  \begin{tikzcd}[row sep=10pt, column sep=35pt]
    \Sigma
    \ar[
      r,
      dashed
    ]
    \ar[
      d,
      "{ \Phi }"
    ]
    &
    \mathcal{B}
    \ar[
      d,
      "{ \mathcal{P} }"
    ]
    \\
    X
    \ar[
      r, dashed
    ]
    &
    \mathcal{A}
    \mathrlap{\,.}
  \end{tikzcd}
\end{equation}

\begin{example}
 The minimal IR-completion of probe M5-branes
 which is naturally compatible with the bulk IR-completion according to Hypothesis H (\cref{IRCompletionAccordingToHypothesisH}) is classified by the quaternionic Hopf fibration $h_{\mathbb{H}}$,
 with charges of the form (\parencites[\S 3.7]{FSS20-H}{FSS21-Hopf}[\S 4.2]{GSS24-M5}):
\begin{equation}
\label
{M5BraneChargesCommutingDiagram}
  \begin{tikzcd}[row sep=10pt, column sep=35pt]
    \Sigma
    \ar[
      r,
      dashed
    ]
    \ar[
      d,
      "{ \Phi }"
    ]
    &
    S^7
    \ar[
      d,
      "{ h_{\mathbb{H}} }"
    ]
    \\
    X
    \ar[
      r, dashed
    ]
    &
    S^4
    \mathrlap{\,.}
  \end{tikzcd}
\end{equation}

For instance, in the special case that the bulk charge is trivial, this implies that the charge carried by the ``self-dual'' field on the M5-brane is classified by the fiber of the quaternionic Hopf fibration and hence is in 3-Cohomotopy:
\[
  \begin{tikzcd}[row sep=small]
    \Sigma
    \ar[
      r,
      dashed
    ]
    \ar[
      d,
      "{ \Phi }"
    ]
    &
    S^3
    \ar[r]
    \ar[d]
    &
    S^7
    \ar[
      d,
      "{ h_{\mathbb{H}} }"
    ]
    \\
    X
    \ar[
      rr, 
      dashed,
      downhorup
    ]
    \ar[
      r,
    ]
    &
    \ast
    \ar[r]
    &
    S^4
    \mathrlap{\,.}
  \end{tikzcd}
\]
This reproduces the expected integral quantization of the charges carried by the ``self-dual string'' (the singular M-string), cf. \cref{MStringChargesAccordingToHypothesisH}.

\begin{table}[htb]
\caption{\label
{MStringChargesAccordingToHypothesisH}
Possible charges of singular (``black'') $q$-branes \emph{inside} an M5-brane probe (as measured near-horizon) predicted by the IR-completion of M5-probes of 11D SuGra according to the twisted relative extension \cref{M5BraneChargesCommutingDiagram} of  \emph{Hypothesis H} (\cref{IRCompletionAccordingToHypothesisH}).
Besides the expected integral charges carried by M1-strings, there is the possibility of fractional 0-brane charge (invisible before IR-completion).
}
\centering
\adjustbox{
  rndfbox=4pt
}{
\begin{tblr}{
  colspec = {rcccccl},
  row{1} = {bg=gray!30}
}
 $q =$ 
 & $0$
 & $1$
 & $2$ 
 & $3$
 & $4$ 
 \\
 & $0$ 
 & $\mathcolor{purple}{\mathbb{Z}}$
 & $0$
 & $0$ 
 & $0$ 
 &
 {\small (integral)}
 \\[-7pt]
 $
   \smash{
   \left.
   \begin{aligned}
   H^1\bracket({
   \mathbb{R}^{1,5}
   \setminus
   \mathbb{R}^{1,q};
   \Omega S^3
   })
   \\
   \simeq
   \pi_{4-p}\bracket({S^3})
   \end{aligned}
   \right\}
   }
 \simeq$
 & $\oplus$
 & $\oplus$
 & $\oplus$
 & $\oplus$
 & $\oplus$
 \\[-2pt]
 & $\mathbb{Z}_{/2}$
 & $0$
 & $0$
 & $0$
 & $0$
 &
 {\small(fractional)}
\end{tblr}
}

\end{table}

\[
  H^1\bracket({
    \mathbb{R}^{1,5}
    \setminus
    \mathbb{R}^{1,1}
    ;
    \Omega S^3
  })
  \simeq
  \pi_3\bracket({S^3})
  \simeq
  \mathbb{Z}
  \mathrlap{\,.}
\]

\end{example}

%%%%%%%%%%%%%%
\paragraph
{Chern-Simons-Type Charges}
%%%%%%%%%%%%%%

The above construction \cref{DifferentialNonabelianCohomologyDiagram} of global field configurations admits the following generalization: 

Given a choice of $L_\infty$-algebra inclusion 
\begin{equation}
\label
{CharacteristicLInfinityAlgebraInclusion}
  \begin{tikzcd}
    \mathfrak{a} 
    \ar[r, hook] 
    & 
    \tilde {\mathfrak{a}}
  \end{tikzcd}
\end{equation}
we may take the charges to be classified by spaces $\tilde {\mathcal{A}}$ such that \cref{WhBracketCoincidingWithCharacteristic}
\begin{equation}
  \mathfrak{l} \tilde {\mathcal{A}}
  \simeq
  \tilde {\mathfrak{a}}
  \mathrlap{\,,}
\end{equation}
and hence generalize the construction \cref{DifferentialNonabelianCohomologyDiagram} of global fields to
\begin{equation}
\label
{DifferentialNonabCohomologyWithCSFields}
\hspace{-4mm} 
  \begin{tikzcd}[
     row sep=1pt, 
     column sep=15pt
  ]
    &[10pt]
    &[10pt]
    &
    \tilde{\mathcal{A}}
    \ar[
      dd,
      "{ 
        \mathbf{ch}^{\tilde{\mathcal{A}}} 
      }"
    ]
    \\[35pt]
    X
    \ar[
      urrr,
      dashed,
      bend left=20,
      "{ \chi }"{
        swap,
        name=charge,
        pos=.3
      },
      "{
        \text{\color{darkblue}charges}
      }"{sloped, pos=.3}
    ]
    \ar[
      ddr,
      dashed,
      "{ \vec F }"{
        name=flux
      },
      "{
        \text{\color{darkblue}fluxes}
      }"{sloped, swap}
    ]
    \ar[
      Rightarrow,
      from=charge,
      to=flux,
      dashed,
      "{ \widehat{A} }"{swap},
      "{
        \text{\color{darkgreen}potentials}
      }"{sloped, , yshift=-2pt, swap}
    ]
    \\
    &&&
    \tilde {\mathcal{A}}_{\mathbb{R}}
    \\
    &
    \Omega^{1}_{\mathrm{cl}}({
      -;\mathfrak{a}
    })
    \ar[
      r
    ]
    &
    \shape \Omega^1_{\mathrm{cl}}({
      -;\tilde {\mathfrak{a}}
    })
    \ar[
      ur,
      phantom,
      "{ \simeq }"{sloped}
    ]
  \end{tikzcd}
  \;\;\;
  \Leftrightarrow
  \;\;\;
  \begin{tikzcd}[sep=35pt]
    X
    \ar[
      r,
      "{ (\vec F, \hat A, \chi) }"
    ]
    &
    \Omega^1_{\mathrm{cl}}(-;\mathfrak{a})
    \underset{
      \tilde {\mathcal{A}}_{\mathbb{R}}
    }{\times}
    \tilde {\mathcal{A}}
    \mathrlap{\,.}
  \end{tikzcd}
\end{equation}
This amounts to flux quantization in $\tilde {\mathcal{A}}$-cohomology while constraining the $\tilde{\mathfrak{a}}$-flux densities that are not in the subalgebra \cref{CharacteristicLInfinityAlgebraInclusion} to vanish.

\begin{example}
[Magnetized M5-Brane Probes]
\label[example]
{MagentizedM5BraneProbes}
  In the case of M5-branes \cref{M5BianchisAsCommutingDiagram} we may consider including the characteristic
  of the worldvolume flux into the Whitehead $L_\infty$-algebra of the \emph{twistor fibration} $t_{\mathbb{C}}$ (\cite{FSS22-Twistorial}):
  \begin{equation}
    \label
    {InclusionOflHopfFibIntolTwistorFib}
    \begin{tikzcd}[row sep=15pt, column sep=25pt]
      \mathfrak{l}_{_{S^4}}
      S^7
      \ar[
        d,
        "{ \mathfrak{l}h_{\mathbb{H}} }"
      ]
      \ar[
        r,
        hook
      ]
      &
      \mathfrak{l}_{_{S^4}}
      \mathbb{C}P^3
      \ar[
        d,
        "{
          \mathfrak{l}t_{\mathbb{C}}
        }"
      ]
      \\
      \mathfrak{l}S^4
      \ar[
        r,
        equals
      ]
      &
      \mathfrak{l}S^4
      \mathrlap{\,.}
    \end{tikzcd}
  \end{equation}

Here the Bianchi identities characterized by $\mathfrak{l}_{_{S^4}} \mathbb{C}P^3$ are
\[
  \begin{aligned}
    \mathrm{d}\, 
    \mathcolor{purple}{F_2}
    & = 0
    \\
    \mathrm{d}\, H_3
    & =
    \Phi^\ast G_4 
      - 
    \mathcolor{purple}{F_2 F_2}\,,
  \end{aligned}
\]
with an extra 2-flux density $F_2$, and the inclusion \cref{InclusionOflHopfFibIntolTwistorFib} enforces the further on-shell constraint
\begin{equation}
  F_2 = 0
  \mathrlap{\,.}
\end{equation}

Under the twisted relative version of the global field construction \cref{DifferentialNonabCohomologyWithCSFields} this means that besides the worldvolume $B$-field there appears also an ``$A$-field'' which has vanishing flux density but possibly topological global effects, like an abelian Chern-Simons field.

Indeed, the $\mathfrak{l}_{_{S^4}} \mathbb{C}P^3$ characteristic turns out (\cite[p. 38]{SS25-Srni}\cite[(83)]{Banerjee2025-Potentials}) to change the local $B$-field  potential \cref{LocalBFieldPotential} to
\begin{equation}
  \mathrm{d}\,
  B_2
  =
  H_3
  -
  C_3
  -
  \mathcolor{purple}{
  \mathrm{CS}\bracket({A_1})
  }
  \,,
\end{equation}
including the Chern-Simons form of a 1-form potential $A_1$. 
Expressions of this form for the ``M-theory 3-form'' $C_3$ have previously been considered in 
\cite[(2.1)]{DonagiWijnholt2023}\cite[(2.1.15)]{FSS14-7D} \cite[\S 4.1]{FSS15-ModuliStack}, following \cite{DFM2007}, and earlier in \cite[(3.3)]{Evslin_2003}.

Here (as for actual Chern-Simons theory), the Chern-Simons term actually vanishes on-shell (cf. \parencites[p. 7]{SS25-Seifert}[\S 2.2]{BaSS26-MString}), but the topological effect of the Chern-Simons-type field influences the topological charges seen in the IR-completion of the M5-brane. 

Concretely, the evident IR-completion of this situation is given by taking the classifying fibration to be the actual twistor fibration $\inlinetikzcd{ 
  \mathbb{C}P^3 
  \ar[
    r, 
    ->>, 
    "{ t_{\mathbb{C}} }"
  ] 
  \& S^4 
}$.

We mention some of the resulting effects in \cref{Applications} below.

\end{example}

%%%%%%%%%%%%%%%
\subsubsection
{Orbifold Charges in Equivariant Nonabelian Cohomology}
%%%%%%%%%%%%%%

%%%%%%%%%%%%%
\paragraph
{Equivariant Classifying Maps}
%%%%%%%%%%%%%
When generalizing (super-)spacetimes to orbifolds, the brane charges are measured in  \emph{orbifold cohomology} (cf. \cite{SS26-Orb}). For global orbifold quotients this is modeled by equipping both the domain spaces and the classifying spaces with continuous actions of a group $G$ and then constraining the classifying maps $c$ to be \emph{$G$-equivariant}
\begin{equation}
  \text{$G$-orbifold charges}
  \;\in\;
  \pi_0\,
  \mathrm{Map}\bracket({X,\mathcal{A}})^G
  :=
  \pi_0
  \Bigg\{\!
  \adjustbox{raise=-6pt}{
  \begin{tikzcd}
    X
    \ar[
      out=60,
      in=180-60,
      shift right=-1pt,
      looseness=4,
      "{ 
        \,\mathclap{G}\,
      }"{description}
    ]
    \ar[
      r,
      dashed,
      "{ c }"
    ]
    &
    \mathcal{A}
    \ar[
      out=60,
      in=180-60,
      looseness=4,
      shift right=-1pt,
      "{ 
        \,\mathclap{G}\,
      }"{description}
    ]
  \end{tikzcd}
  }
  \!\!\Bigg\}
  \mathrlap{,}
\end{equation}
in that for all $x \in X$ and $g \in G$ we have:
\begin{equation}
  c(g \cdot x) = g \cdot c(x)
  \mathrlap{\,.}
\end{equation}

An important consequence of this equivariance condition is that \emph{$G$-fixed points} 
\begin{equation}
  \label{GFixedLocus}
  X^G
  :=
  \bracketmid
    \{{x \in X}{\forall_g : g \cdot x = x}\}
  \subset 
  X
\end{equation}
must be mapped to $G$-fixed points:
\begin{equation}
  \label{InducedMapOnFixedPoints}
  \begin{tikzcd}[row sep=small, column sep=large]
    X
    \ar[
      out=60,
      in=180-60,
      shift right=-1pt,
      looseness=4,
      "{ 
        \,\mathclap{G}\,
      }"{description}
    ]
    \ar[
      r,
      dashed,
      "{ c }"
    ]
    &
    \mathcal{A}
    \ar[
      out=60,
      in=180-60,
      looseness=4,
      shift right=-1pt,
      "{ 
        \,\mathclap{G}\,
      }"{description}
    ]
    \\
    X^G
    \ar[u, hook]
    \ar[
      r,
      dashed,
    ]
    &
    \mathcal{A}^G
    \mathrlap{\,.}
    \ar[u, hook]
  \end{tikzcd}
\end{equation}
In particular, the relative mapping spaces \eqref{RelativeClassifyingMaps} between $G$-spaces become themselves $G$-spaces by the evident $G$-conjugation action, and their $G$-fixed loci are the spaces of relative equivariant maps:
\begin{equation}
  \label{DoublyRelativeEquivariantMappingSpace}
  \mathrm{Map}
  \bracket({
    \Phi
    ,
    \mathcal{P}
  })^G
  :=
  \biggggg\{ \!\!
  \begin{tikzcd}[row sep=17pt]
    \ar[
      out=53+45,
      in=180-53+45,
      shift left=6pt,
      looseness=4,
      "{ 
        \,\mathclap{G}\;
      }"{description}
    ]
    \phantom{A}
    \ar[
      d,
      ->>,
      shorten=-4pt,
      "{ \Phi }"
    ]
    \ar[
      r, 
      shorten=-4pt,
      dashed
    ]
    & 
    \phantom{A}
    \ar[
      out=53-45,
      in=180-53-45,
      shift left=6pt,
      looseness=4,
      "{ 
        \,\mathclap{G}\,
      }"{description}
    ]
    \ar[
      d,
      ->>,
      shorten=-4pt,
      "{ \mathcal{P} }"
    ]
    \\
    \phantom{A}
    \ar[
      out=53+135,
      in=180-53+135,
      shift left=5pt,
      looseness=4,
      "{ 
        \,\mathclap{G}\;
      }"{description}
    ]
    \ar[
      r,
      shorten=-4pt,
      dashed
    ]
    &
    \phantom{A}
    \ar[
      out=53-135,
      in=180-53-135,
      shift left=5pt,
      looseness=4,
      "{ 
        \,\mathclap{G}\;
      }"{description}
    ]
  \end{tikzcd}
  \!\! \biggggg\}.
\end{equation}
The joint homotopy classes of such equivariant maps are the (relative) orbifold charges (cf. \cite[(3)]{SS20-Tad}\cite[Fig. 7]{SS25-Orient}\cite{SS26-Orb}).

%%%%%%%%%%%%%%%%%%%
\subsection
{Applications}
\label
{Applications}
%%%%%%%%%%%%%%%%%%

We briefly highlight an application of IR-completed super-gravity with brane probes that concretely relates to contemporary questions in experimental topological quantum materials research. Here the relation between (IR-completed) super-gravity and (topological) experimental solid state physics is through ``geometric engineering'' of the latter:

%%%%%%%%%%%%%
\paragraph
{The Idea of Geometric Engineering}
%%%%%%%%%%%%%

By \emph{geometric engineering} of quantum field theories (\parencites{KatzVafa1996}{KatzKlemmVafa1996}{KatzVafa1997}{BelhajDrissiRasmussen2003}, cf. \parencites{GiveonKutasov1998}{Duplij2017}) one refers to the identification of a given QFT (or typically of suitable sectors in suitable limits) inside the worldvolume dynamics of branes propagating inside an auxiliary higher-dimensional ambient spacetime, typically at orbi-singularities. The idea is that field content and couplings of the QFT of interest thereby become encoded in the (orbifold) geometric nature of the brane configuration, thus allowing for possibly better analytic control. 

Geometric engineering is typically considered on single probe branes or on a small number $N$ of coincident branes. The limiting case of a large number of $\inlinetikzcd{N \ar[r] \& \infty}$ coincident branes, making a massive \emph{black} brane configuration with gravitational backreaction, is the topic of \emph{holography} (AdS/CFT duality \cite{AharonyEtAl1999}, cf. \parencites{HerzogKlebanov2006}{Baggioli2019}), where the black brane's worldvolume QFT may dually be described by (classical, in the appropriate limit) super-gravity of the ambient spacetime.

%%%%%%%%%%%%%%
\paragraph
{Literature on M5-Brane models in Solid State Physics}
%%%%%%%%%%%%%%

The idea to engineer strongly coupled condensed matter systems on M-branes in 11D SuGra came to prominence with the holographic model of superconductors due to
\cite{Herzog2007,GauntlettSonnerWiseman2009,GubserPufuRocha2010,GauntlettetalSonnerWiseman2010}, further discussed by \cite{DonosetalEtAl2013,DonosGauntlettPantelidou2013}, cf. review in \cite[\S 4.3]{Baggioli2019}.
This approach relies on a large $N$-limit. 

A suggestion that M5-branes may geometrically engineer topological quantum order was made in \cite{ChoGangKim2020}, aspects of which were further discussed in \cite{CuiQiuWang2023,BonettiSchaeferNamekiWu2024}.

Our rigorous geometric engineering of fractional quantum Hall topological order on $N=1$ IR-completed M5-branes at $A_n$-singularities is due \parencites{SS25-Seifert}[\S 4]{BaSS26-MString}[\S 5]{SS26-BBC}, based on results in \cite{SS25-AbelianAnyons,SS25-FQH}, survey in \cite{SS25-Srni}.

%%%%%%%%%%%%%%
\paragraph
{Geometric Engineering of FQH Order on IR-Completed M5-Probes}
%%%%%%%%%%%%%%
With topological charges properly defined on IR-completed M5-brane probes of 11D Sugra (\cref{BraneChargesInNonabelianCohomology}), we gain access to the proper \emph{geometric engineering} of topological phases of matter on single M5-branes.

In particular, after global completion of 11D SuGra in Cohomotopy  (\cref{IRCompletionAccordingToHypothesisH}), the topological quantum states of flux on ``magnetized'' M5-branes (\cref{MagentizedM5BraneProbes}) probing A-type orbisingularities reproduce in fine detail known and expected properties of topological quantum order of fractional quantum Hall (FQH) systems and make further predictions of experimental interest
(\cite{SS25-Rickles,SS25-Srni,SS25-ISQS29,SS26-BBC} based on \cite{SS25-AbelianAnyons,SS25-FQH}). 

Strikingly, beyond the topological sector, FQH excitations are observed, in the long-wavelength (IR) limit, to exhibit:
\begin{enumerate}
 \item \emph{$W_\infty$-symmetry} (cf. \parencites{IKS1992}{CTZ1993}{CTZ1994}[\S B]{Wang2023}[\S IV.B]{Du2025}{SS26-SDiff}), 

 hence area-preserving diffeomorphism symmetry, reflecting the incompressibility of the FQH liquid,
 
 \item \emph{super-symmetry} (cf. \parencites{Hasebe:2007}{GromovMartinecRyu2020}{Liu2024Resolving}{Pu2023Signatures}),

 under which the \emph{magneto-roton mode} is a spin-2 graviton-like super-partner to a spin-3/2 gravitino-like \emph{neutral graviton mode}. 
\end{enumerate}

We highlight that: \emph{These are exactly the characteristic symmetries also of super $p$-branes!} 

For super-symmetry this is clear.
The volume-preserving diffeomorphism symmetry of general super $p$-branes in light-cone gauge is due to \cite{BSTT1990}, see also \parencites{MatsuoShibusa2000}{HoKawaiLiao2026}.
The specific case of the M2-branes in Hamiltonian formulation was previously discussed in \parencites{FloratosIliopoulos1988}[p. 553]{dWHH1988}{dWMH1990}, discussion in covariant formulation is in \parencites{Katagiri2026a}{Katagiri2026b}. The specific case of the M5-brane is further discussed in \cite{BandosTownsend2008} and a special formulation for strong flux backgrounds is due to \cite{HIMS2008}, reviewed in \parencites[\S 3.2.1]{Ho2009}[\S 4.3]{HoMatsuo2016}.

This coincidence suggests that the geometric engineering of FQH systems on IR-completed M5-branes may be useful also beyond the topological sector.

%%%%%%%%%%%%%
\paragraph
{Literature on Effective Super-Gravity for FQH Systems}
%%%%%%%%%%%%%

A previous suggestion that FQH excitations are effectively described in the IR by super-gravity runs as follows:  
\begin{enumerate}
\item the bosonic excitation looks like a (chiral) \emph{graviton}: \cite{Liou2019} (theoretically) and \cite{Liang2024} (experimentally),

\item the fermionic excitation looks like a \emph{gravitino}: \parencites{Haldane2013}[p. 8]{Wang2023},

\item this pair of modes exhibits an emergent (broken) \emph{supersymmetry}: \cite{GromovMartinecRyu2020,Liu2024Resolving} (theoretically)  and \cite{Pu2023Signatures} (experimentally),

\item 
whose key properties may be modeled by \emph{supergravity}: \cite{NPBG2023,Du2025Chiral}.

\end{enumerate}
These proposals of effective SuGra for FQH excitations are arguably tentative and explorative. We suggest that the following rigorous anchoring at least of the topological sector of FQH systems modeled on IR-completed M5-branes may help clarify the situation.

%%%%%%%%%%%%
\paragraph
{IR-Completed M5-Branes Probing A-type singularities}
%%%%%%%%%%%%

We now make an M5-brane configuration concrete which, when IR-completed in twistorial Cohomotopy, provably engineers FQH topological order in fine detail.

An \emph{$A_{n-1}$-type orbi-singularity} has as isotropy group the finite cyclic group of order $n$, 
\[
  G 
    := 
  \CyclicGroup{n}
  \mathrlap{\,,}
\] 
and the orbifold is the quotient of  $\mathbb{C}^2$ by the action for which $[k] \in \CyclicGroup{n}$ acts as follows:
\begin{equation}
  \label{TheATypeAction}
  \begin{tikzcd}[
    column sep=-3pt
  ]
    \mathbb{H}
    \ar[
      out=60,
      in=180-60,
      shift left=1pt,
      looseness=4,
      "{ 
        \,\mathclap{G}\,
      }"{description}
    ]
    &\simeq_{{}_\mathbb{R}}&
    \mathbb{C}
    \ar[
      out=60,
      in=180-60,
      shift left=1pt,
      looseness=4,
      "{ 
        \,\mathclap{G}\,
      }"{description}
    ]
    &\times&
    \overline{\mathbb{C}}
    \ar[
      out=62,
      in=180-62,
      shift left=1pt,
      looseness=3,
      "{ 
        \,\mathclap{G}\,
      }"{description}
    ]
  \end{tikzcd}
  \hspace{15pt}
  :
  \hspace{15pt}
  \begin{tikzcd}
    {[k]} \cdot (z_1, z_2)
    :=
    \bracket({
      e^{2\pi\mathrm{i} k/n}
      z_1,
      \,
      e^{-2\pi\mathrm{i} k/n}
      z_2
    })
    \mathrlap{\,.}
  \end{tikzcd}
\end{equation}
As indicated, this action is in fact right $\mathbb{H}$-linear with respect to the identification $\mathbb{C}^2 \simeq_{{}_{\mathbb{R}}} \!\mathbb{H}$ of two copies of the complex numbers with the quaternions, being the left multiplication action with unit quaternions in the image of the inclusion
\begin{equation}
  \mathrm{U}(1)
  \subset 
  \mathrm{SU}(2)
  \simeq
  S({\mathbb{H}})
  \mathrlap{\,.}
\end{equation}
To reflect this domain orbifold in the coefficients of our flux-quantizing cohomology theory, we take the $G$-action on the classifying twistor fibration
$\inlinetikzcd{\mathbb{C}P^3 \ar[r, "{ t_{\mathbb{H}} }"] \& S^4 }$ (\cref{MagentizedM5BraneProbes}) to be ``of the same form'' as on spacetime, namely with $G$ acting as in \cref{TheATypeAction} on one of the two $\mathbb{C}^2 \simeq \mathbb{H}$-factors from which the projective spaces are defined:
\begin{equation}
  \begin{tikzcd}[
    column sep=-5pt,
    row sep=6pt
  ]
    &&
    \mathbb{C}P^3
    \ar[
      out=60,
      in=180-60,
      shift left=1pt,
      looseness=4,
      "{ 
        \,\mathclap{G}\,
      }"{description}
    ]
    \ar[
      rrrr, 
      ->>, 
      "{ t_{\mathbb{H}} }"
    ]
    \ar[d, equals]
    &&[27pt]&&
    S^4
    \ar[
      out=60,
      in=180-60,
      shift left=1pt,
      looseness=4,
      "{ 
        \,\mathclap{G}\,
      }"{description}
    ]
    \ar[d, equals]
    \\
    \mathllap{\big((}
    \mathbb{C}^2 
    \ar[
      out=60-180,
      in=-60,
      shift left=5pt,
      looseness=4,
      "{ 
        \,\mathclap{G}\,
      }"{description}
    ]
    &\times& \!\!
    \mathbb{C}^2)
    &\!\!\!\!\setminus \{0\}
    \big)\big/\mathbb{C}^\times
    \ar[
      r, 
      ->>,
      shorten >=10pt
    ]
    &[10pt]
    \mathllap{\big((}
    \mathbb{C}^2 
    \ar[
      out=60-180,
      in=-60,
      shift left=5pt,
      looseness=4,
      "{ 
        \,\mathclap{G}\,
      }"{description}
    ]
    &\times&
    \mathbb{C}^2
    )
    &\!\setminus \{0\}
    \big)\big/\mathbb{H}^\times
    \mathrlap{\,.}
  \end{tikzcd}
\end{equation}
With this, the resulting fibration of fixed loci \cref{InducedMapOnFixedPoints} is the fibration of the 2-sphere over the point:
\begin{equation}
  \label{FixedLocusInClassifyingFibrations}
  \begin{tikzcd}[
    column sep=35pt,
    row sep=7pt
  ]
    (\mathbb{C}P^3)^G
    \ar[
      d,
      equals
    ]
    \ar[
      r,
      ->>,
      "{
        t_{\mathbb{C}}^G
      }"
    ]
    &
    (S^4)^G
    \ar[
      d,
      equals
    ]
    \\
    S^2
    \ar[
      d,
      equals
    ]
    \ar[
      r,
      ->>
    ]
    &
    \ast
    \ar[
      d,
      equals
    ]
    \\
    \bracket({
      \mathbb{C}^2
      \setminus
      \{0\}
    })\big/\mathbb{C}^\times
    \ar[
      r,
      ->>
    ]
    &
    \bracket({
      \mathbb{C}^2
      \setminus
      \{0\}
    })\big/\mathbb{H}^\times
    \mathrlap{\,.}
  \end{tikzcd}
\end{equation}

Consider then an M5-brane worldvolume embedding $\Phi$ wrapping just the \emph{half-space} 
(also considered in \cite[p. 5]{CouzensLuscherSparks2025})
\begin{equation}
  \begin{tikzcd}[
    column sep=2pt
  ]
  \mathbb{C} 
    \ar[
      out=60,
      in=180-60,
      shift right=-1pt,
      looseness=4,
      "{ 
        \,\mathclap{G}\,
      }"{description}
    ]
    &\subset& 
  \mathbb{C}^2
    \ar[
      out=60,
      in=180-60,
      shift right=-1pt,
      looseness=3,
      "{ 
        \,\mathclap{G}\,
      }"{description}
    ]
  \end{tikzcd}
\end{equation}
of an A-type singularity \cref{TheATypeAction}, hence of the form
\begin{equation}
\label
{TheATypeOrbiConfiguration}
  \begin{tikzcd}[row sep=15pt, column sep=-4pt]
    \Sigma^{1,5}
    \ar[
      d,
      hook,
      "{ \Phi }"
    ]
    &:=&
    \Sigma^{1,3}
    \ar[
      d,
      equals
    ]
    &
    \times
    &
    \mathbb{C}
    \ar[
      out=60,
      in=180-60,
      shift right=-1pt,
      looseness=4,
      "{ 
        \,\mathclap{G}\,
      }"{description}
    ]
    \ar[
      d,
      equals
    ]
    \\
    X^{1,10}
    &:=&
    \Sigma^{1,3}
    &
    \times
    &
    \mathbb{C}
    \ar[
      out=60-180,
      in=-60,
      shift right=-1pt,
      looseness=4.7,
      "{ 
        \,\mathclap{G}\,
      }"{description}
    ]
    &\times&
    \overline{\mathbb{C}}
    \ar[
      out=60-180,
      in=-60,
      shift right=-1pt,
      looseness=4.2,
      "{ 
        \,\mathclap{G}\,
      }"{description}
    ]
    &\times&
    X^3
    \mathrlap{\,.}
  \end{tikzcd}
\end{equation}

This configuration is evidently $G$-equivariantly homotopic to its $G$-fixed locus \eqref{GFixedLocus} (by linearly contracting the non-compact singularity to the point, $\mathbb{C}^2 \sim \ast$). Since the equivariant relative mapping spaces \eqref{DoublyRelativeEquivariantMappingSpace} are equivariant homotopy-invariant, the orbifold charges reduce, by \eqref{InducedMapOnFixedPoints}, to plain charges on the fixed locus with $G$-fixed coefficients \eqref{FixedLocusInClassifyingFibrations}, cf. \parencites[p. 7]{SS25-Seifert}[\S 4]{SS25-TEC}:
\begin{equation}
\label
{ReducingTo2CohomotopyChargesOnM5}
  \Biggggg\{
  \hspace{-10pt}
  \begin{tikzcd}[
    ampersand replacement=\&,
    row sep=20pt
  ]
    \Sigma^{1,5}
    \ar[
      out=53+45,
      in=180-53+45,
      shift left=3pt,
      looseness=4,
      "{ 
        \,\mathclap{G}\;
      }"{description}
    ]
    \ar[
      d,
      ->>,
      "{ \Phi }"
    ]
    \ar[
      r, 
      dashed
    ]
    \& 
    \mathbb{C}P^3
    \ar[
      out=53-45,
      in=180-53-45,
      shift left=6pt,
      looseness=4,
      "{ 
        \,\mathclap{G}\,
      }"{description}
    ]
    \ar[
      d,
      ->>,
      "{ t_{\mathbb{C}} }"
    ]
    \\
    X^{1,10}
    \ar[
      out=53+135,
      in=180-53+135,
      shift left=5pt,
      looseness=4,
      "{ 
        \,\mathclap{G}\;
      }"{description}
    ]
    \ar[
      r,
      dashed
    ]
    \&
    S^4
    \ar[
      out=53-135,
      in=180-53-135,
      shift left=2pt,
      looseness=4,
      "{ 
        \,\mathclap{G}\;
      }"{description}
    ]
  \end{tikzcd}
  \hspace{-10pt}
  \Biggggg\}
  \simeq
  \big\{
  \begin{tikzcd}[
    sep=15pt
  ]
    \Sigma^{1,3}
    \ar[r, dashed]
    &
    S^2
  \end{tikzcd}
  \big\}
  \mathrlap{.}
\end{equation}

But this is, on the right, exactly the \emph{Hopfion} charge structure (\parencites{Wilczek1983}[\S II.C]{Forte1992})
of topological quantum order as seen via anyonic braiding phases (cf. \cite{Goldin2023}) in Fractional Quantum (Anomalous) Hall (FQ(A)H) systems.

Proceeding in this manner, one finds the following dictionary for geometric engineering of FQH quantum order on IR-completed M5-brane probes: 

\begin{description}

\item
\textbf{$A_n$-Singularities Engineer FQH Order:}

The above discussion shows how placing the IR-completed magnetized M5-brane on the half-space of an $A_n$-singularity collapses the classifying fibration to a 2-sphere. 

As a classifying space for topological charges, the latter turns out to encode in fine detail the expected structure of FQH topological order with abelian anyons, including the root-of-unity braiding phases that have in recent years been seen in experiment (beginning with \parencites{Bartolomei2020}{Nakamura2020}, for recent developments see \parencites{Ghosh2025}). 

The details of this analysis are due to \parencites{SS25-AbelianAnyons}{SS25-FQH}{SS25-FQAH}{SS25-ISQS29}{KSS26-HigherDimAnyons}, review includes \parencites{SS25-Srni}{SS25-ISQS29}{SS25-Rickles}.

On this backdrop one may now consider variants of the above brane configuration and thereby pass from retro-dicting FQH properties to making experimentally relevant predictions:

\item 
\textbf{Punctures Engineer Nonabelian Anyons:}

More concretely, the algebra of quantum observables on topological charges of the above form \cref{ReducingTo2CohomotopyChargesOnM5} is computed by \emph{topological light cone quantization} (\parencites{SS24-Obs}[\S 2]{SS25-Complete}) on 4D worldvolumes of the form
\begin{equation}
  \Sigma^{1,3}
  \simeq
  \mathbb{R}^{1,1}
  \times
  \Sigma^2
  \mathrlap{\,,}
\end{equation}
for $\Sigma^2$ a surface that models the FQH liquid.

In general, if this surface is not compact, then the \emph{solitonic} topological observables are measured via \emph{pointed} maps out of the one-point compactification of this surface (cf. \cite[\S 2.2]{SS25-Flux})
\[
  \big\{
    \inlinetikzcd{
      \Sigma^2_{\cpt}
      \ar[
        rr,
        dashed,
        "{
          \mathrm{pntd}
        }"{swap}
      ]
      \&\&
      S^2
    }
  \big\}
  \,,
\]
which encodes exactly the condition that charges \emph{vanish at infinity}.

But this means that if we consider the surface to be punctured, then the punctures become all identified with the point at infinity, modelling a situation where flux is expelled from their vicinity. In the laboratory this corresponds to doping of the FQH liquid by super-conducting islands expelling magnetic flux lines. But then the \emph{general covariance} (diffeomorphism equivariance) of the brane model predicts that the resulting charges transform under the mapping class group of the punctured surface --- which is a braid group signifying possibly non-abelian anyonic quantum states. This is the situation discussed in \cite{SS26-Islands}.

\item 
\textbf{M-Strings Engineer Edge Modes:}
One may further include \emph{M-string} probes \cite{HaghighatEtAl2015} of the M5-brane probe. The resulting system of iterated probe branes then admits \cite{BaSS26-MString} IR-completion classified by the quaternionic Hopf fibration factored through the twistor fibration. Under this completion the M-string turns out to geometrically engineer gapped nodal lines \cite[\S 2.2, Fig. 2.3]{SS26-Orb} exhibiting \cite{SS26-BBC} edge mode excitations of the kind (cf. \parencites[\S 2.5]{Wen1992}[\S 6.1.2]{Tong2016}) expected in FQH systems. 

\end{description}

This close match between the topological order of FQH systems and the topological charges of IR-completed M5-brane probes of 11D SuGra suggests that beyond its topological sector the IR-completed M5-brane model in 11D SuGra may geometrically engineer also the excitation spectrum  of FQH quantum liquids. This remains to be discussed elsewhere.

%%%%%%%%%%%%%%
\paragraph
{Conclusion}
%%%%%%%%%%%%%

IR-completions of super-gravity provide in particular the brane charge structure of the theory, including subtle topological effects beyond the charges that are reflected already by flux density forms. Beyond general theoretical interest, the above application to geometric engineering of FQH topological order highlights that this topological charge structure subsumes effects that can and have been measured in laboratory experiments on condensed matter systems. 

It is conceivable that better understanding of topological charge sectors may also inform questions encountered in high-energy particle physics, notably as concerns the non-perturbative QCD vacuum. For instance the best available physics-level explanation of confinement in QCD (one of the \emph{Millennium Problems} pronounced by the Clay Math Institute, cf. \cite{nLab:YMMassGap}) is the \emph{dual superconductor} model (cf. \cite{nLab:DualSuperconductorModel}), which argues that the compression of confining color flux tubes is caused by a  condensate of effective Dirac monopole charges. If true, this is a subtle topological effect whose proper discussion likely requires IR-complete theories.

%%%%%%%%%%%%%%

%%%%%%%%%%%%%%
\appendix

%%%%%%%%%%%%%%
\section
{Background}
%%%%%%%%%%%%%%

For reference in the main text here we briefly recall some background material with pointers to further reading.

%%%%%%%%%%%%%%%
\subsection
{Categories}
\label
{Categories}
%%%%%%%%%%%%%%

Super-geometry cannot be well-understood without a bare minimum of category-theoretic vocabulary (\cref{Superspace}). More serious higher category theory (homotopy theory, \cref{Homotopy}) is needed for understanding IR-completions of SuGra (\cref{Completions}). Therefore here we briefly recall some basics (for more cf. \cite{Geroch1985,AbramskyTzevelekos2011,Sc18-ToposLectures,Richter2020}).

%%%%%%%%%%%%%
\paragraph
{Arrows Denote Maps}
%%%%%%%%%%%%%

We are dealing with different kinds of \emph{objects} $X$, $A$, ... (spaces, super-spaces, algebras, super-algebras,...). An arrow like
\begin{equation}
  \inlinetikzcd{
    X
    \ar[rr, "{ f }"]
    \&\&
    A
  }
\end{equation}
denotes a \emph{map} (a mapping, a function, a transformation, ..., a \emph{morphism}) from one to another. This in the appropriate sense: 

If $X$ and $A$ are topological spaces then a map is a continuous map, if they are smooth manifolds then a map is a smooth map, if they are algebras then a map is an algebra homomorphism, if they are super-algebras then a map is a grade-preserving homomorphism, etc.

When maps are defined on elements (some kinds of maps are not!), then a \emph{barred arrow} ``$\mapsto$'' indicates where a map sends an element:
\begin{equation}
  \begin{tikzcd}[sep=0pt]
    X \ar[rr, "{ f }"] && A
    \\
    x &\mapsto& f(x)
    \mathrlap{\,.}
  \end{tikzcd}
\end{equation}

For any such notion of objects with maps between them to make good sense, we just need the most basic consistency conditions; the \emph{composition} of composable maps 
\begin{equation}
  \inlinetikzcd{
    X 
    \ar[
      r, 
      "{f}"
    ] 
    \ar[
      rr,
      uphordown,
      "{ g\circ f }"{description}
    ]
    \& 
    Y 
    \ar[r, "{ g }"] 
    \&
    Z
  }
\end{equation}
should:
\begin{enumerate}
\item 
  exist (be again a well-defined map of the appropriate kind),
\item 
  be associative 
\item and unital.
\end{enumerate}
The last one may feel too obvious to even state, but here we go: For every object $X$ there is a map $\inlinetikzcd{ X \ar[r, "{ \mathrm{id}_X }"] \& X}$ whose composition with any map $f$ out of or into $X$ is $f$.

Some jargon:
\begin{enumerate}
\item
When a map $f$ has an inverse under this composition, it is called an \emph{isomorphism}. 

\item
When we draw two paths of arrows between the same pair of objects (a \emph{diagram}), then we mean that the two composite maps agree (the \emph{diagram commutes}):
\begin{equation}
\label
{CommutingSquare}
  \begin{tikzcd}
    X 
      \ar[r, "{f_1}"] 
      \ar[d, "{g_1}"]
    & 
    A
      \ar[d, "{g_2}"]
    \\
    Y \ar[r, "{f_2}"] & B
  \end{tikzcd}
  \;\;\;\;
  \text{means that}
  \;\;
  g_2 
    \circ
  f_1
  =
  f_2
    \circ 
  g_1 
  \mathrlap{\,.}
\end{equation}
\end{enumerate}

%%%%%%%%%%%%%%
\paragraph
{Categories}
%%%%%%%%%%%%%

A collection of objects with such a kind of maps between them is called a \emph{category}. As in \emph{categorizing} objects: A \emph{kind} of objects.

The archetypical example is the category of sets:
\footnote{
A set-theoretic remark which the casual reader should ignore: The collection of all sets is not itself a set (\emph{Russell's paradox}) but a proper class. Categories in general have a proper class of objects. If a category happens to have an actual set of objects then it is called a \emph{small category}. Only groupoids are usually assumed to be small by default. 
}
\begin{equation}
  \label{CategoryOfSets}
  \begin{aligned}
    \colorbox{lightgray}{$ \mathrm{Set}$} : &
    \text{
      Objects are sets, maps are ordinary maps (functions).
    }
  \end{aligned}
\end{equation}

Basic examples of categories are \emph{concrete categories} of sets with extra structure. Beyond sets themselves \cref{CategoryOfSets}, examples of concrete categories are: 
\begin{subequations}
\label
{BasicExamplesOfConcrete}
  \begin{align}
    \label{CategoryOfTopologicalSpaces}
    \colorbox{lightgray}{$ \mathrm{TopSp}$}
    :
    &
    \text{
      Topological spaces with continuous maps between them.
    }
    \\
    \label{CategoryOfSmoothManifolds}
    \colorbox{lightgray}{$ \mathrm{SmthMfd} $}
    :
    &
    \text{
    Smooth manifolds with smooth maps
    between them.
    }
    \\
    \label
    {CategoryOfVectorSpaces}
    \colorbox{lightgray}{$ \mathrm{Vec}_{_{\mathbb{K}}} $}
    : 
    &
    \text{
      $\mathbb{K}$-Vector spaces with linear maps between them.
    }
    \\
    \colorbox{lightgray}{$ \mathrm{Alg}_{_{\mathbb{K}}} $}
    :
    &
    \text{
      Associative $\mathbb{K}$-algebras with homomorphisms between them.
    }
    \\
   \colorbox{lightgray}{$  \mathrm{CAlg}_{_{\mathbb{K}}} $}
    :
    &
    \text{
      Commutative $\mathbb{K}$-algebras with homomorphisms between them.
    }
  \end{align}
\end{subequations}

But categories do not need to be concrete in this way. Notably, super-manifolds (\cref{SmoothSupermanifoldsAsSmoothSuperSet}) are not just sets with extra structure; the maps between them are \emph{super-fields}:
\begin{equation}
    \label{CategoryOfSuperhManifolds}
    \begin{aligned}
  \colorbox{lightgray}{$   \mathrm{SprMfd} $}
    :
    &
    \text{
    Super-manifolds with super-functions between them (\cref{Superspace}).
    }
    \end{aligned}
\end{equation}

Another example of non-concrete categories is the category of gauged sets \cref{GaugedSetSchematics} of Yang-Mills gauge field configurations with gauge transformations between them (cf. \cite{BeniniSchenkelSchreiber2018}):
\begin{equation}
\label
{GroupoidOfYangMillsFields}
  \begin{aligned}
   \colorbox{lightgray}{$  G\mathrm{Cnn}_X $}
    & :%
    \text{
      $G$-Yang-Mills fields on $X$
      with gauge transformations 
      between them.
    }
  \end{aligned}
\end{equation}

Such gauge transformations are invertible. Generally, invertible maps in a category are called \emph{isomorphisms}, denoted $\inlinetikzcd{ c \ar[r, "{ \sim }"] \& c' }$, and may be thought of as generalized gauge transformations.
The isomorphisms from an object $c$ in a category $\mathcal{C}$ to itself form a group, the \emph{symmetry group} or \emph{automorphism group} of that object:
\begin{equation}
\label
{AutomorphismGroup}
  \mathrm{Aut}_{\mathcal{C}}(c)
  :=
  \Big\{\!
  \inlinetikzcd{
    c
    \ar[
      out=55+90,
      in=180-55+90,
      looseness=5pt,
      shift left=4pt,
      "{ \sim }"{
        sloped, 
        yshift=-1pt
      }
    ]
  }
  \Big\}
  \mathrlap{\,.}
\end{equation}
Therefore, 
categories in which \emph{all} maps are isomorphisms are called \emph{groupoids}.

For each group $G$ there is a groupoid with a single abstract object that has $G$ as its automorphisms, the \emph{delooping groupoid} of $G$:
\begin{equation}
\label
{DeloopingGroupoid}
  \begin{aligned}
 \colorbox{lightgray}{$    \mathbf{B}G $}
    &
    :
    \text{
      Single abstract object with 
      $G$ worth of automorphisms.
    }
  \end{aligned}
\end{equation}

An abstract class of examples: With $\mathcal{C}$ any category, there is its \emph{opposite category}:
\begin{equation}
\label
{OppositeCategory}
  \begin{aligned}
 \colorbox{lightgray}{$    \mathcal{C}^{\mathrm{op}} $}
    :
    \text{
      Objects of $\mathcal{C}$ with 
      maps those of $\mathcal{C}$ but regarded as pointing the other way.
    }
  \end{aligned}
\end{equation}
The objects of $\mathcal{C}^{\mathrm{op}}$ are also called the \emph{formal duals} of the objects of $\mathcal{C}$.
For instance, the basic example of the \emph{duality between geometry and algebra} (cf. \parencites[\S2]{Street2007}[\S 3.1]{Sc18-Durham}[Tbl. 4]{JurcoEtAl2019}{Scholze2026-Geometry}) is:
\begin{equation}
\label
{AffineSchemes}
  \begin{aligned}
  \colorbox{lightgray}{$   \mathrm{CAlg}_{_{\mathbb{K}}}
      ^{\mathrm{op}} $}
    :
    &
    \text{
      Affine schemes over $\mathbb{K}$
      with scheme maps between them.
    }
  \end{aligned}
\end{equation}
This duality is the basis for defining super-Cartesian spaces $\mathbb{R}^{n\vert q}$ (\cref{Charts}).

Other general constructions: For a pair of categories $\mathcal{C}$ and $\mathcal{D}$ we have their product and their disjoint union categories:
\begin{subequations}
\begin{align}
  \label
  {ProductCategory}
    \mathcal{C}
    \times
    \mathcal{D}
    :
    &
    \text{
      Pairs of objects, from $\mathcal{C}$ and $\mathcal{D}$, 
      with pairs of maps between them.
    }
    \\
  \label
  {DisjointUnionCatgegory}
    \mathcal{C}
    \sqcup
    \mathcal{D}
    :
    &
    \text{
      Objects from either $\mathcal{C}$ or $\mathcal{D}$, with their maps between them.
    }
\end{align}
\end{subequations}

%%%%%%%%%%%%
\paragraph
{Functors}
%%%%%%%%%%%%

Functors are the maps between categories $\mathcal{C}$, $\mathcal{D}$ themselves. A functor $\inlinetikzcd{ \mathcal{C} \ar[r, "F"] \& \mathcal{D}}$ sends objects to objects and maps to maps, such that composition and identities are preserved:
\begin{equation}
\label
{FunctorPreservingComposition}
  \begin{tikzcd}[
    row sep=0pt, column sep=15pt
  ]
    \mathcal{C}
    \ar[rr, "{ F }"]
    &&
    \mathcal{D}
    \\
    c_1 
    \ar[
      dd,
      "{f}"
    ]
    \ar[
      dddd,
      leftvertright,
      "{ g \circ f }"{
        rotate=90,
        pos=.5,
        description,
      }
    ]
    &\mapsto& F(c_1)
    \ar[
      dd,
      "{ F(f) }"{swap}
    ]
    \ar[
      dddd,
      rightvertleft,
      "{ F(g \circ f) }"{
        rotate=90,
        pos=.5,
        description,
      }
    ]
    \\
     &\mapsto&
    \\
    c_2 
    \ar[
      dd,
      "{g}"
    ]
    &\mapsto& F(c_2)
    \ar[
      dd,
      "{ F(g) }"{swap}
    ]
    \\
     &\mapsto&
    \\
    c_3 &\mapsto& F(c_3)
  \end{tikzcd}
\end{equation}

For example, forming tangent bundles is a functor from the category of smooth manifolds \cref{CategoryOfSmoothManifolds} to itself, sending smooth maps to their derivatives:
\begin{equation}
  \begin{tikzcd}[
    sep=0pt
  ]
    \mathrm{SmthMfd}
    \ar[
      rr,
    ]
    &&
    \mathrm{SmthMfd}
    \\
    X
    \ar[
      dd,
      "{ f }"
    ]
    &\mapsto& 
    T X
    \ar[
      dd,
      "{ d f }"'{swap}
    ]
    \\
      &\mapsto&
    \\
    Y 
      &\mapsto& 
    T Y
    \mathrlap{\,.}
  \end{tikzcd}
\end{equation}
That this indeed respects composition \cref{FunctorPreservingComposition} is equivalently the \emph{chain rule} of calculus.

A basic example: Functors from a delooping groupoid \cref{DeloopingGroupoid} of a group $G$ to vector spaces \cref{CategoryOfVectorSpaces} are linear $G$-representations $\rho$ on a vector space $V$:
\begin{equation}
\label
{LinearRepAsFunctor}
  \begin{tikzcd}[row sep=0pt, column sep=5pt]
    \mathbf{B}G
    \ar[
      rr,
    ]
    &&
    \mathrm{Vec}_{_{\mathbb{K}}}
    \\
    \bullet
    \ar[
      dd,
      "{ g }"
    ]
    &\mapsto&[6pt]
    V
    \ar[
      dd,
      "{ \rho(g) }"
    ]
    \\[2pt]
    &\mapsto&
    \\[-2pt]
    \bullet
    &\mapsto&
    V
    \mathrlap{\,.}
  \end{tikzcd}
\end{equation}

An important abstract example is the \emph{hom-functor} of a category $\mathcal{C}$, which sends pairs $c,d$ of its objects to the \emph{hom set} $\mathrm{Hom}_{\mathcal{C}}(c, d)$ of maps $\inlinetikzcd{c \ar[r] \& d}$, and pairs of maps to their pre/postcomposition:
\begin{equation}
\label
{HomFunctor}
  \begin{tikzcd}[row sep=0pt, column sep=5pt]
    \mathcal{C}^{\mathrm{op}}
    \times
    \mathcal{C}
    \ar[
      rr,
      "{ 
        \mathrm{Hom}_{\mathcal{C}} 
      }"
    ]
    &&
    \mathrm{Set}
    \\
    (c_1, d_1)
    \ar[
      dd,
      <-,
      shift right=7pt,
      "{ (\gamma }"'
    ]
    \ar[
      dd,
      phantom,
      "{ , }"{pos=.57}
    ]
    \ar[
      dd,
      shift left=7pt,
      "{ \delta) }"
    ]
    &\mapsto&
    \mathrm{Hom}_{\mathcal{C}}\bracket({
      c_1, d_1
    })
    \ar[
      dd,
      "{
        \substack{
          f
          \\
          \rotatebox[origin=c]{-90}{$\mapsto$}
          \\
          \delta \circ f \circ \gamma
        }
      }"
    ]
    \\[+10pt]
    &\mapsto&
    \\[+10pt]
    (c_2, d_2)
    &\mapsto&
    \mathrm{Hom}_{\mathcal{C}}\bracket({
      c_2, d_2
    })
    \mathrlap{\,.}
  \end{tikzcd}
\end{equation}

Some important jargon:
\begin{enumerate}

\item
A functor that induces bijections between all hom-sets is said to exhibit a \emph{full subcategory inclusion} and to be \emph{fully faithful}, denoted:
\begin{equation}
\label{FullSubcategoryEmbedding}
  \inlinetikzcd{
    \mathcal{C}
    \ar[
      r, 
      hook, 
      "{ F }"
    ]
    \&
    \mathcal{D}
  }
  \;\;\;\;
  \Leftrightarrow
  \;\;\;\;
  \forall_{c_1, c_2}
  \;
  \inlinetikzcd{
    \mathrm{Hom}_{\mathcal{C}}(c_1,c_2)
    \ar[
      r,
      "{ F }",
      "{ \sim }"'
    ]
    \&
    \mathrm{Hom}_{\mathcal{D}}\bracket({
      F(c_1),
      F(c_2)
    })
    \mathrlap{\,.}
  }
\end{equation}

\item A functor that is surjective on isomorphism classes of objects is called \emph{essentially surjective}.

\item A functor that is both essentially surjective as well as fully-faithful is called an \emph{equivalence} of categories, denoted: 
\begin{equation}
\label
{EquivalenceOfCategories}
  \inlinetikzcd{
    \mathcal{C}
    \ar[r, "F", " \sim "{swap} ]
    \&
    \mathcal{D}
  }
  \;\;\;\;\;
  \Leftrightarrow
  \;\;\;\;\;
  \text{$F$ is essentially surjective and fully-faithful}
\end{equation}
\end{enumerate}
Such equivalence is the notion of \emph{meta-gauge transformations} between categories themselves.

For example,
every groupoid $\mathcal{G}$ is equivalent to its \emph{skeleton}, the disjoint union \cref{DisjointUnionCatgegory} 
over isomorphism classes $[c]$
of delooping groupoids \cref{DeloopingGroupoid} of automorphism groups \cref{AutomorphismGroup}:
\begin{equation}
\label
{SkeletalizationOfGroupoids}
  \inlinetikzcd{
    \bigsqcup_{ [c] }
    \mathbf{B}
    \mathrm{Aut}_{\mathcal{G}}(c)
    \ar[
      r,
      "{ \sim }"
    ]
    \&
    \mathcal{G}
    \mathrlap{\,,}
  }
\end{equation}
via the functor which includes the full subcategory on one representative object $c$ in each isomorphism class $[c]$.
This expresses the fact that gauge equivalent objects are operationally indistinguishable.%
\footnote{
  Beware that as stated, \cref{SkeletalizationOfGroupoids} holds for the \emph{geometrically discrete} groupoids discussed so far, not in general for topological and Lie groupoids: Because a set with topology, even in the absence of any isomorphisms between its elements, is not in general the disjoint union of its elements.
}

%%%%%%%%%%%%
\paragraph
{Natural Transformations}
%%%%%%%%%%%%

Natural transformations are the maps between parallel functors 
$
  \inlinetikzcd{
    X 
      \ar[r, shift left=2pt, "{F}"]
      \ar[r, shift right=2pt, "{G}"']
    \&
    A
}$. 
A transformation $\inlinetikzcd{ F \ar[r, "{ \phi }"] \&  G}$ is an assignment of maps $\phi(c)$ of $\mathcal{D}$ to objects $c$ of $\mathcal{C}$ such that all the following square diagrams commute \cref{CommutingSquare}:
\begin{equation}
  \begin{tikzcd}[
    column sep=10pt
  ]
    c_1 
    \ar[
      d,
      "{ f }"
    ]
      &\mapsto&
    F(c_1)
    \ar[
      d,
      "{ F(f) }"
    ]
    \ar[
      rr, 
      "{ \phi(c_1) }"
    ]
    &&[12pt]
    G(c_1)
    \ar[
      d,
      "{ G(f) }"
    ]
    \\
    c_2 
      &\mapsto&
    F(c_2)
    \ar[
      rr, 
      "{ \phi(c_2) }"
    ]
    &&
    G(c_2)
    \mathrlap{\,.}
  \end{tikzcd}
\end{equation}

These transformations compose by composing their components:
\begin{equation}
  \phi_2 \circ \phi_1 (c)
  \defneq
  \phi_2(c) \circ \phi_1(c)
  \mathrlap{\,.}
\end{equation}

Therefore, functors $\inlinetikzcd{\mathcal{C} \ar[r] \& \mathcal{D}}$ are the objects of a category whose maps are natural transformations, called the \emph{functor category} between $X$ and $A$:
\begin{equation}
\label
{FunctorCategory}
  \begin{aligned}
  \mathrm{Fnctr}\bracket({
    \mathcal{C},
    \mathcal{D}
  })\!\!:
  &
  \text{
  Functors $\inlinetikzcd{\mathcal{C} \ar[r] \& \mathcal{D}}$ with natural transformations between them.}
  \end{aligned}
\end{equation}

For instance, natural transformations between functors $\rho$, $\rho'$ from a delooping groupoid to vector spaces are intertwiners $\phi$ between the corresponding linear $G$-representations \cref{LinearRepAsFunctor}:
\begin{equation}
  \begin{tikzcd}[
    column sep=10pt
  ]
    \bullet 
    \ar[
      d,
      "{ g }"
    ]
      &\mapsto&
    V
    \ar[
      d,
      "{ \rho(g) }"
    ]
    \ar[
      rr, 
      "{ \phi }"
    ]
    &&[12pt]
    V'
    \ar[
      d,
      "{ \rho'(g) }"
    ]
    \\
    \bullet
      &\mapsto&
    V
    \ar[
      rr, 
      "{ \phi }"
    ]
    &&
    V'
    \mathrlap{\,.}
  \end{tikzcd}
\end{equation}
Therefore the functor category from $\mathbf{B}G$ to $\mathrm{Vec}_{_{\mathbb{K}}}$ is the \emph{representation category} of $G$:
\begin{equation}
\label
{RepGAsFunctorCategory}
  \begin{aligned}
    \mathrm{Fnctr}\bracket({
      \mathbf{B}G,
      \mathrm{Vec}_{_{\mathbb{K}}}
    })
    & :
    \text{
      Linear $G$-representations with 
      intertwiners between them.
    }
  \end{aligned}
\end{equation}

A key re-incarnation of functor categories are \emph{presheaf categories}, see \cref{CategoryOfPresheaves} below.

%%%%%%%%%%%%
\paragraph
{Adjunctions}
%%%%%%%%%%%%

Adjunctions are what category theory is really about (\emph{duality} and \emph{universal constructions}) --- but since we avoid entering actual category theory in these lecture notes (instead just using the basic concept formation), we only mention them briefly.

A pair of back-and-forth functors being \emph{adjoint}, denoted
\begin{equation}
\label
{AdjointFunctors}
  \begin{tikzcd}
    \mathcal{C}
    \ar[
      rr,
      shift left=5pt,
      "{ L }"
    ]
    \ar[
      rr,
      phantom,
      "{ \bot }"{scale=.7}
    ]
    \ar[
      rr,
      <-,
      shift right=5pt,
      "{ R }"'
    ]
    &&
    \mathcal{D}
    \mathrlap{\,,}
  \end{tikzcd}
\end{equation}
means that there is a natural isomorphism (``forming adjuncts'')
\begin{equation}
  \begin{tikzcd}
  \mathrm{Hom}_{\mathcal{D}}\bracket({
    L(-),
    -
  })
  \ar[
    r,
    <->,
    "{ \widetilde{(-)} }"
  ]
  &
  \mathrm{Hom}_{\mathcal{C}}\bracket({
    -,
    R(-)
  })
  \mathrlap{\,,}
  \end{tikzcd}
\end{equation}
between the hom-functors \cref{HomFunctor} precomposed with the \emph{left adjoint} functor $L$ on the left and with the \emph{right adjoint} functor $R$ on the right, respectively.

The adjuncts of identities are called the \emph{units} and \emph{counits} of the adjunction, respectively, denoted:
\begin{subequations}
  \begin{align}
    \label{AdjunctionUnit}
    \inlinetikzcd{
      c
      \ar[
        r,
        "{ \eta_c }"
      ]
      \&
      R \circ L(c)
    }
    &
    := 
    \widetilde {\mathrm{id}_{L(c)}}
    \\
    \label{AdjunctionCounit}
    \inlinetikzcd{
      L \circ R(d)
      \ar[
        r,
        "{ \epsilon_d }"
      ]
      \&
      d
    }
    &
    := 
    \widetilde {\mathrm{id}_{R(d)}}
    \mathrlap{\,.}
  \end{align}
\end{subequations}

%%%%%%%%%%%%%
\paragraph
{Presheaves}
%%%%%%%%%%%%%

Given a small category $\mathcal{C}$, the functor category
 \cref{FunctorCategory}
from its opposite \cref{OppositeCategory} to sets \cref{CategoryOfSets}, is called the category of \emph{presheaves} over $\mathcal{C}$:
\begin{equation}
\label
{CategoryOfPresheaves}
  \mathrm{PSh}\bracket({
    \mathcal{C}
  })
  :=
  \mathrm{Fnctr}\bracket({
    \mathcal{C}^{\mathrm{op}},
    \mathrm{Set}
  })
  \mathrlap{\,.}
\end{equation}
If we think of the objects of $\mathcal{C}$ as some kind of spaces, then we may think of presheaves $\mathbf{X} \in \mathrm{PSh}\bracket({\mathcal{C}})$ as assigning sets of maps $\mathbf{X}(c)$ from $c$ to a generalized space $\mathbf{X}$ \emph{probeable} by the spaces in $\mathcal{C}$, and defined by these sets of probes (cf. exposition in \cite{Schreiber2025,GS25-FieldsI}).

In particular, the original spaces give such presheaves/generalized spaces by the \emph{Yoneda embedding}:
\begin{equation}
  \begin{tikzcd}[
    sep=0pt
  ]
    \mathcal{C}
    \ar[
      rr,
      hook,
      "{ y }"
    ]
    &&
    \mathrm{PSh}\bracket({
      \mathcal{C}
    })
    \\
    c 
      &\mapsto&
    \mathrm{Hom}_{\mathcal{C}}(-,c)
    \mathrlap{\,.}
  \end{tikzcd}
\end{equation}

With this it turns out, remarkably, that the sets $\mathbf{X}(c)$ which we \emph{thought of} as maps from $c$ to $\mathbf{X}$ actually \emph{are} these maps --- this is the statement of the  ubiquitous \emph{Yoneda lemma}, saying that there are natural isomorphisms
\begin{equation}
  \begin{tikzcd}
    \mathrm{Hom}_{_{\mathrm{PSh}(\mathcal{C})}}
    \bracket({
      y(-), \mathbf{X}
    })
    \ar[
      r,
      "{ \sim }"
    ]
    &
    \mathbf{X}(-)
    \mathrlap{\,.}
  \end{tikzcd}
\end{equation}

%%%%%%%%%%%%%
\paragraph
{Sites}
%%%%%%%%%%%%%

Often the probe spaces $c \in \mathcal{C}$ have a notion of how several \emph{glue to} or \emph{cover} another one. 
To make this precise: A \emph{coverage} on a small category $\mathcal{C}$ is 
\begin{enumerate}
\item
for each object $c \in \mathcal{C}$ a set of \emph{covering sets} of morphisms $\{ \inlinetikzcd{ c_i \ar[r, "{ \iota_i }"] \& c }\}_{i \in I}$,

\item
such that the ``pullbacks'' of these covering sets along any map $\inlinetikzcd{ c' \ar[r] \& c }$ may be refined by coverings $\big\{ \inlinetikzcd{ c'_{i'} \ar[r, "{ \iota'_{i'} }"] \& c' }\big\}_{i' \in I'}$ in that 
\begin{enumerate}

\item 
for every $i \in I$ there exists $i' \in I'$ and a map $\inlinetikzcd{c'_{i'} \ar[r, dashed] \& c_i}$
\item
such that the following diagram commutes \cref{CommutingSquare}:
\begin{equation}
  \begin{tikzcd}[row sep=15pt, column sep=25pt]
    c'_{i'}
    \ar[
      r,
      dashed
    ]
    \ar[
      d,
      "{ \iota'_{i'} }"
    ]
    & 
    c_i
    \ar[
      d,
      "{ \iota_i }"
    ]
    \\
    c'
    \ar[r]
    &
    c
    \mathrlap{\,.}
  \end{tikzcd}
\end{equation}
\end{enumerate}
\end{enumerate}
A category $\mathcal{C}$ equipped with such a coverage is called a \emph{site}.

For example, the category $\mathrm{SmthMfd}$ \cref{CategoryOfSmoothManifolds} becomes a site with covering sets the \emph{open covers}.

%%%%%%%%%%%%%
\paragraph
{Sheaves}
%%%%%%%%%%%%%

If we still think of presheaves $\mathbf{X}$ over $\mathcal{C}$ \cref{CategoryOfPresheaves} as generalized spaces probeable by objects of $\mathcal{C}$, then in view of a given coverage on $\mathcal{C}$ (as above) this picture is only consistent if probes of $\mathbf{X}$ respect this notion of gluing. 

Concretely, for $\{\inlinetikzcd{ c_i \ar[r, "{ \iota_i }"] \& c }\}_{i \in I}$ a covering set, one says that plots $\bracket\{{\phi_i \in \mathbf{X}(c_i)}\}_{i \in I}$ are \emph{compatible} (match where covering probes intersect) if for all $i, i' \in I$ and pairs of dashed maps making the following diagram commute  
\begin{equation}
  \begin{tikzcd}[row sep=15pt, column sep=25pt]
    d
    \ar[
      r,
      dashed,
      "{ f }"
    ]
    \ar[
      d,
      dashed,
      "{ f' }"
    ]
    & c_i
    \ar[
      d,
      "{ \iota_i }"
    ]
    \\
    c_{i'}
    \ar[
      r,
      "{ \iota_{i'} }"
    ]
    &
    c
  \end{tikzcd}
  \;\;\;\;
  \text{we have}
  \;\;\;\;
  \mathbf{X}\bracket({f})\bracket({\phi_i})
  =
  \mathbf{X}\bracket({f'})\bracket({\phi_{i'}})
  \mathrlap{\,.}
\end{equation}

The presheaf $\mathbf{X}$ is called a \emph{sheaf} if for each such compatible family there exists a glued plot $\phi \in \mathbf{X}(c)$ in that 
\begin{equation}
  \forall_{i \in I}
  \;\;
  \mathbf{X}({ \iota_i })(\phi)
  =
  \phi_i
  \mathrlap{\,.}
\end{equation}
The full subcategory \cref{FullSubcategoryEmbedding} of sheaves among presheaves is denoted
\begin{equation}
\label
{CategoryOfSheaves}
  \inlinetikzcd{
   \mathrm{Sh}\bracket({\mathcal{C}})
   \ar[r, hook]
   \&
   \mathrm{PSh}
   \bracket({\mathcal{C}})
  }
  \mathrlap{\,.}
\end{equation}

This is how we conceive of general super-spaces in \cref{Spaces}, by taking $\mathcal{C}$ to be the site of super-Cartesian spaces $\mathbb{R}^{n\vert q}$.

%%%%%%%%%%%%%%
\subsection
{\texorpdfstring
  {$L_\infty$-Algebra}
  {L-infinity Algebra}
}
\label
{LInfinityAlgebra}
%%%%%%%%%%%%%%

We briefly recall the dual incarnation of finite-type $L_\infty$-algebras in the guise of their Chevalley-Eilenberg dgc-algebras (cf. \parencites{LadaStasheff1992}[\S 6.1]{SatiSchreiberStasheff2009}[\S 4]{FSS23-Char}). This dual formulation is much more concise and transparent than the ``direct'' one. In particular, it makes immediate the generalization to super-$L_\infty$-algebras \parencites{FSS15-WZW}[\S 3]{FSS2019}[\S 2.1]{GSS25-TD}, which in this dual guise are (mis-)named ``FDAs'' in the supergravity literature  (cf. \cite{CDF1991,CastellaniDAuria2025}).

Our ground field is the real numbers, $\mathbb{R}$. By \emph{graded} we mean $\mathbb{Z}$-graded (though all constructions considered here happen to be just $\mathbb{N}$-graded). We say \emph{gc-algebras} for graded-commutative algebras and \emph{dgc-algebras} for differential graded-commutative algebras with differential of degree +1.

%%%%%%%%%%%%
\paragraph
{The base case of ordinary Lie algebras}
%%%%%%%%%%%
For $\bracket({\mathfrak{g},[-,-]})$ a finite-dimensional Lie algebra, recall the construction of its \emph{Chevalley-Eilenberg algebra} $\mathrm{CE}(\mathfrak{g})$:

Write $\mathfrak{g}^\vee := \mathfrak{g}^\ast$ for the linear dual of the underlying vector space and 
\begin{equation}
\label
{GrassmannAlgebraOnDualLieAlgebra}
  \wedge^\bullet \mathfrak{g}^\vee
  :=
  \mathrm{Sym}\bracket({
    \mathfrak{g}^\vee[1]
  })
\end{equation}
for its \emph{Grassmann algebra}, the free graded-commutative algebra for which $\wedge^n \mathfrak{g}^\vee$ is the subspace of degree $n$.

A graded derivation $\mathrm{d}$ on $\wedge^\bullet \mathfrak{g}^\vee$ is thus determined by its restriction $\inlinetikzcd{ \wedge^1 \mathfrak{g}^\vee \ar[r] \& \wedge^\bullet \mathfrak{g}^\vee }$. Taking this restriction to be the linear dual of the Lie bracket,
$
  \mathrm{d}
    _{\vert \wedge^1 \mathfrak{g}^\vee}
  :=
  [-,-]^\ast
$
hence yields a graded derivation of degree 1. One checks that its nilpotency is equivalent to the Jacobi identity of the bracket:
\begin{equation}
  \mathrm{d}^2 = 0
  \;\;
  \Leftrightarrow
  \;\;
  \text{
   Jacobi identity on $[-,-]$}
  \mathrlap{\,.}
\end{equation}
The resulting dgc-algebra
is the \emph{Chevalley-Eilenberg algebra} of the Lie algebra:
\begin{equation}
\label
{CEAlgebraOfLieAlgebra}
  \mathrm{CE}\bracket({
    \mathfrak{g},
    [-,-]
  })
  :=
  \bracket({
    \wedge^\bullet
    \mathfrak{g}^\vee,
    \mathrm{d} 
    :=
    [-,-]^\ast
  })
  \mathrlap{\,.}
\end{equation}
The key observation now is that this construction is a \emph{fully faithful} functor on the category $\mathrm{LieAlg}^{\mathrm{fd}}_{_{\mathbb{R}}}$ of finite-dimensional Lie algebras: 
\begin{equation}
  \begin{tikzcd}
    \mathrm{LieAlg}
      ^{\mathrm{fd}}
      _{_{\mathbb{R}}}
    \ar[r, hook, "{ \mathrm{CE} }"]
    &
    \mathrm{dgcAlg}^{\mathrm{op}}
     _{_{\mathbb{R}}}
    \mathrlap{\,.}
  \end{tikzcd}
\end{equation}

This means that we may work (op-)equivalently with the dgc-algebras $\mathrm{CE}(\mathfrak{g})$. But for these there is an evident generalization, simply by allowing the vector space $\mathfrak{g}$ to be $\mathbb{N}$-graded:

%%%%%%%%%%%%
\paragraph
{The case of finite-type $L_\infty$-algebras}
%%%%%%%%%%%%

For $\mathfrak{g}$ an $\mathbb{N}$-graded vector space of finite type (meaning that each degree $\mathfrak{g}_n$ is a finite-dimensional vector space), write $\mathfrak{g}^\vee$ for its degreewise dual vector space, hence with 
\begin{equation}
  \bracket({
    \mathfrak{g}^\vee
  })_n
  :=
  \bracket({\mathfrak{g}_n})^\ast
\end{equation}
and write
\begin{equation}
\label
{GradedGrassmannAlgebra}
  \wedge^\bullet \mathfrak{g}^\vee
  :=
  \mathrm{Sym}\bracket({
    \mathfrak{g}^\vee[1]
  })
\end{equation}
for the graded Grassmann gc-algebra, which in low degrees is
\begin{equation}
  \bracket({
    \wedge^\bullet \mathfrak{g}^\vee
  })_1
  =
  \wedge^1 \mathfrak{g}^\vee_0
  \,,
  \;\;\;
  \bracket({
    \wedge^\bullet \mathfrak{g}^\vee
  })_2
  =
  \wedge^2 \mathfrak{g}^\vee_0
  \oplus
  \wedge^1 \mathfrak{g}^\vee_1
  \,,
  \;\;
  \text{etc.}
\end{equation}

Then the category of (connective) finite-type $L_\infty$-algebras is the full subcategory of $\mathrm{dgcAlg}^{\mathrm{op}}_{_{\mathbb{R}}}$ on those whose underlying gc-algebra is of the form \cref{GradedGrassmannAlgebra}:
\begin{equation}
\label
{LInfinityFaithfulInDGCAlgOp}
  \begin{tikzcd}
    L_\infty \mathrm{Alg}^{\mathrm{ft}}
      _{_{\mathbb{R}}}
    \ar[
      r,
      hook,
      "{ \mathrm{CE} }"
    ]
    &
    \mathrm{dgcAlg}^{\mathrm{op}}
      _{_{\mathbb{R}}}
    \mathrlap{\,.}
  \end{tikzcd}
\end{equation}
Unwinding this: Since a differential $\mathrm{d}$  on \cref{GradedGrassmannAlgebra} is again determined by its restriction to $\wedge^1 \mathfrak{g}^\vee$, it is given by a sequence of linear duals of $n$-ary graded-skew symmetric linear maps for $n \in \mathbb{N}$:
\begin{equation}
  \mathrm{d}_{\vert
    \wedge^1\mathfrak{g}^\vee
  }
  =
  [-]^\ast
  +
  [-,-]^\ast
  +
  [-,-,-]^\ast
  +
  \cdots
  :
  \inlinetikzcd{
    \wedge^1 \mathfrak{g}^\vee
    \ar[
      r
    ]
    \&
    \wedge^1 \mathfrak{g}^\vee
    \oplus
    \wedge^2 \mathfrak{g}^\vee
    \oplus
    \wedge^3 \mathfrak{g}^\vee
    \oplus 
    \cdots
    ,
  }  
\end{equation}
and the nilpotency condition $\mathrm{d}^2 = 0$ is hence equivalently a (somewhat complicated looking) generalization of the Jacobi identity on this system of brackets. Such a system of brackets on graded vector spaces $\mathfrak{g}$ that  satisfies these conditions is exactly what makes $\bracket({ \mathfrak{g}, \bracket({[-], [-,-], [-,-,-], \cdots  }) })$ an $L_\infty$-algebra.

This means that finite-type $L_\infty$-algebras $\mathfrak{h}$ may be defined by giving a list of generators and differential relations for their CE-algebras, like this:
\begin{equation}
\label{GenericLInfinityPresentation}
  \mathrm{CE}(\mathfrak{h})
  \;\simeq\;
  \mathrm{Free}_{\mathrm{dgcAlg}}  
  \bracket({
    (e^i)_{i \in I}
  })
  \big/
  \big(
    \mathrm{d}\, e^i
    =
    P^i(\vec e\,)
  \big)_{i \in I}
  \mathrlap{\,,}
\end{equation}
for graded-symmetric polynomials $\bracket({P_i})_{i \in I}$ in variables $e^i$ of degrees $\mathrm{deg}(e^i)$.

%%%%%%%%%%%%%
\paragraph
{Examples of $L_\infty$-Algebras}
%%%%%%%%%%%%%

The simplest nontrivial example is:
\begin{example}
  For $n \in \mathbb{N}_{> 0}$ the \emph{line Lie $n$-algebra} $\mathbb{R}[n]$ is that whose CE-algebra has a single generator in degree $n$ which is closed:
  \begin{equation}
    \mathrm{CE}\bracket({
      \mathbb{R}[n]
    })
    \simeq
    \mathrm{Free}_{\mathrm{dgcAlg}}
    ({
      e_n
    })
    \big/
    \big(
      \mathrm{d}\, e_n = 0
    \big)
    \mathrlap{\,.}
  \end{equation}
\end{example}

The key class of examples for the purpose of flux quantization is the following (which subsumes all connected \emph{nilpotent} $L_\infty$-algebras of finite type):
\begin{example}
[Real Whitehead-bracket algebras (cf {\cite[Prop. 5.13, Rem. 5.4]{FSS23-Char}})]
\label[example]
{WhiteheadBracketLInfinityAlgebra}
Let $\mathcal{A}$ be a connected and simply-connected topological space of rational finite type (meaning that all its rational cohomology groups are finite-dimensional). Then there is an $L_\infty$-algebra, 
\begin{equation}
  \mathfrak{l}\mathcal{A}
  \in 
  L_\infty \mathrm{Alg}
    ^{\mathrm{ft}}
    _{_{\mathbb{R}}}
  \mathrlap{\,,}
\end{equation}
unique up to isomorphism,
whose Chevalley-Eilenberg algebra \cref{LInfinityFaithfulInDGCAlgOp}:
\begin{enumerate}
  \item 
  is generated from the $\mathbb{R}$-rationalized homotopy groups of $\Omega \mathcal{A}$:
  \begin{equation}
    \bracket({
      \mathfrak{l}\mathcal{A}
    })_\bullet
    \simeq
    \pi_\bullet\bracket({
      \Omega \mathcal{A}
    })
    \otimes_{_{\mathbb{Z}}}
    \mathbb{R}
    \mathrlap{\,,}
  \end{equation}

  \item 
  is quasi-isomorphic to the dgc-algebra $\Omega^\bullet_{\mathrm{PS}}\bracket({\mathcal{A}})$ of ``piecewise smooth'' differential forms on $\mathcal{A}$, namely of differential forms on all singular simplices $\inlinetikzcd{ \Delta^n_{\mathrm{top}} \ar[r, "{ \phi }"] \& \mathcal{A} }$ ($n \in \mathbb{N}$) compatible with the face and degeneracy maps,
  hence in particular:

  \item 
  has
  cochain cohomology isomorphic to the real cohomology of $\mathcal{A}$:
  \begin{equation}
    H^\bullet\bracket({
      \mathrm{CE}\bracket({
        \mathfrak{l}\mathcal{A}
      })
    })
    \simeq
    H^\bullet\bracket({
      \mathcal{A};
      \mathbb{R}
    })
    \mathrlap{\,.}
  \end{equation}

\end{enumerate}

This $\mathrm{CE}\bracket({\mathfrak{l}\mathcal{A}})$ is called the \emph{minimal Sullivan model} of $\mathcal{A}$ (here: over $\mathbb{R}$). The binary bracket $[-,-]$ of $\mathfrak{l}\mathcal{A}$ is the \emph{Whitehead bracket} on the $\mathbb{R}$-rationalized homotopy groups, and the higher brackets are the corresponding \emph{higher Whitehead brackets}.

Special cases of this class of examples may be found listed in \parencites{Menichi2015}[p. 21]{SS25-Flux}.
\end{example}

An important class of constructions of new $L_\infty$-algebras from a given one is the content of \cref{RationalCyclification,RationalToroidification}.

%%%%%%%%%%%%%
\paragraph
{Closed $L_\infty$-Algebra Valued Forms}
%%%%%%%%%%%%%
Consider $X \in \mathrm{SmthMfd}$.
Given again an ordinary Lie algebra $\mathfrak{g}$ of finite dimension, it is immediate that the $\mathfrak{g}$-valued differential 1-forms on $X$ are equivalently gc-algebra homomorphisms out of the dual Grassmann algebra \cref{GrassmannAlgebraOnDualLieAlgebra}, like this:
\begin{equation}
  \Omega^1_{\mathrm{dR}}\bracket({
    X;
    \mathfrak{g}
  })
  \simeq
  \mathrm{Hom}_{\mathrm{gcAlg}}
  \bracket({
    \wedge^\bullet \mathfrak{g}^\vee,
    \Omega^\bullet_{\mathrm{dR}}\bracket({
      X
    })
  })
  \mathrlap{\,.}
\end{equation}
Moreover, it is immediate that the \emph{closed} (meaning: \emph{flat}, \emph{Maurer-Cartan}) $\mathfrak{g}$-valued forms among these correspond exactly to those gc-homomorphisms which are in fact dgc-homomorphisms out of the Chevalley-Eilenberg algebra \cref{CEAlgebraOfLieAlgebra}: 
\begin{equation}
  \begin{aligned}
  \Omega^1_{\mathrm{cl}}\bracket({
    X;
    \mathfrak{g}
  })
  &
  :=
  \bracketmid\{{
    A 
      \in 
    \Omega^1_{\mathrm{dR}}\bracket({
      X;
      \mathfrak{g}
    })
  }{
    \mathrm{d}\, A 
    + 
    \tfrac{1}{2}
    [A \wedge A]
    = 0
  }\}
  \\
  & 
  \simeq
  \mathrm{Hom}_{\mathrm{gcAlg}}
  \bracket({
    \mathrm{CE}(\mathfrak{g}),
    \Omega^\bullet_{\mathrm{dR}}\bracket({
      X
    })
  })
  \mathrlap{\,.}
  \end{aligned}
\end{equation}

Taking this characterization as the definition, it immediately generalizes to the case where $\mathfrak{g}$ is a finite-type $L_\infty$-algebra. We say that dgc-homomorphisms out of its Chevalley-Eilenberg algebra \cref{LInfinityFaithfulInDGCAlgOp} into the de Rham dgc-algebra are the closed (flat) $\mathfrak{g}$-valued differential forms:
\begin{equation}
\label
{ClosedLInfinityValuedDifferentialForms}
  \Omega^1_{\mathrm{cl}}\bracket({
    X;
    \mathfrak{g}
  })
  :=
  \mathrm{Hom}_{\mathrm{dgcAlg}}
  \bracket({
    \mathrm{CE}(\mathfrak{g}),
    \Omega^\bullet_{\mathrm{dR}}(X)
  })
  \mathrlap{\,.}
\end{equation}
Dually, these are maps of \emph{$L_\infty$-algebroids} \cite[\S A.1]{SSS12} out of the \emph{tangent Lie algebroid} $T X$ of $X$ into the 
$L_\infty$-algebra $\mathfrak{g}$:
\begin{equation}
  \Omega^1_{\mathrm{cl}}\bracket({
    X; 
    \mathfrak{g}
  })
  \simeq
  \mathrm{Hom}_{
    L_\infty\mathrm{Algbrd}
  }\bracket({
    T X
    ,
    \mathfrak{g}
  })
\end{equation}
(viewed as fibered over $X$ and over the point $*$, respectively), since $\mathrm{CE}(TX) \simeq \Omega^\bullet_\mathrm{dR}(X)$. 

While this is well known in itself, we emphasize the role of such closed $L_\infty$-valued forms as flux densities that solve duality-symmetric higher Maxwell-type Bianchi identities (\cref{HigherMaxwellEquations,Characteristics}).

%%%%%%%%%%%%%
\paragraph
{$L_\infty$-Algebraic Dimensional Reduction/Oxidation Adjunction}
%%%%%%%%%%%%%

The following \cref{ExtCycAdjunction} is the description of \emph{double dimensional reduction} (and its inverse: ``oxidation'') for closed $L_\infty$-algebra valued differential forms, going back to \cite[\S 3]{FSS18-TD}.

First consider the following classes of examples of $L_\infty$-algebras:

\begin{example}
[$L_\infty$-cyclification {{\parencites[Def. 3.3]{FSS18-TD}[Def. 2.23]{GSS25-TD}}}]
\label[example]
{RationalCyclification}
 Given $\mathfrak{h} \,\in\, \mathrm{sL}_\infty\mathrm{Alg}^{\mathrm{ft}}_\FR$ with presentation \cref{GenericLInfinityPresentation},
its {\textit cyclification} $\mathrm{cyc}(\mathfrak{h}) \,\in\, \mathrm{sL}_\infty\mathrm{Alg}^{\mathrm{ft}}_\FR$ is given by
\begin{equation}
    \label{cykCE}
    \mathrm{CE}\big(
      \mathrm{cyc}(\mathfrak{h})
    \big)
    \;
    :=
    \;
    \mathrm{Free}_{\mathrm{dgcAlg}}
    \biggg(
      \adjustbox{
        raise=-4pt
      }{$
      \def\arraystretch{1.5}
      \begin{array}{l}
        (e^i)_{i \in I}, \;
        \grayoverbrace{
          \omega_2
        }{
          \mathclap{
          \mathrm{deg} 
          \,=\,
          (2,\mathrm{evn})
          }
        },
        \\
        (\,
        \grayunderbrace{
          \mathrm{s} e^i
        }{
         \mathclap{
           \scalebox{.7}{$
           \def\arraystretch{.9}
           \begin{array}{c}
           \mathrm{deg}
           \,=\,
           \\
           \mathrm{deg}(e^i) - (1,\mathrm{evn})
           \end{array}
           $}
          }
        }
        )_{i \in I}
      \end{array}
      $}
    \biggg)
    \Big/
    \left(
    \def\arraystretch{1}
    \def\arraycolsep{2pt}
    \begin{array}{lcl}
      \mathrm{d}\, 
      \omega_2
      &=&
      0
      \\
      \mathrm{d}\, e^i
      &=&
      \mathrm{d}_{\mathfrak{h}}
      \, e^i
      +
      \omega_2\, 
      \mathrm{s} e^i
      \\
      \mathrm{d}\, \mathrm{s} e^i
      &=&
      -
      \mathrm{s}(
        \mathrm{d}_{\mathfrak{h}}
        \,
        e^i
      )
    \end{array}
   \!\! \right)
    ,
  \end{equation}
  where in the last line on the right the shift $s$ is uniquely extended to a super-graded derivation of degree $(-1, \mathrm{evn})$.
\end{example}
More generally
\begin{example}
[$L_\infty$-$n$-toroidification {\parencites[Thm. 2.6]{SatiVoronov2025}[Def. 2.38]{GSS25-TD}}]
\label[example]
{RationalToroidification}
 For $n \in \mathbb{N}$ the $n$-\textit{toroidification} of $\mathfrak{h} \,\in\, \mathrm{sL}_\infty\mathrm{Alg}^{\mathrm{ft}}_\FR$ with presentation \cref{GenericLInfinityPresentation}
  is given by
  \begin{align}
    \label{CEAlgebraOfNToroidification}
  &
  \hspace{-4mm} 
   \mathrm{CE}\big(
      \mathrm{tor}^n(\mathfrak{h})
    \big)
    \;
    \simeq
    \\
    &
   \hspace{-4mm}  
    \FR_\dd
    \bigg[
      \smash{
      \big(
        \grayunderbrace{
        \dir{r}{\omega}_2
        }{
          \mathclap{
            \adjustbox{scale=.6}{$
            \def\arraystretch{1.1}
            \begin{array}{c}
              \mathrm{deg}
              \,=
              \\
              (2,\mathrm{evn})
            \end{array}
          $}
          }
        }
      \big)_{r =1}^n
      }
      ,\;
      \smash{
      \big(
        \grayunderbrace{
        \dir{i_r}{\mathrm{s}}
        \cdots
        \dir{i_2}{\mathrm{s}} \,
        \dir{\hspace{5pt}i_1}{\mathrm{s}} \, 
        e^i
        }{
          \mathclap{
          \adjustbox{
            scale=.6
          }{$
          \def\arraystretch{.9}
          \def\arraycolsep{0pt}
          \begin{array}{c}
            \mathrm{deg}
            \,=\,
            \\
            \mathrm{deg}(e^i)
            -
            (r, \mathrm{evn})
           \end{array}
         $}
         }
        }
      \big)_{
        \adjustbox{
          scale=.6
        }{$
        \def\arraystretch{1.1}
        \def\arraycolsep{0pt}
        \begin{array}{l}
          i \in I,
          \;\;\;\;
          0 \leq r \leq n,
          \\
          n \geq 
            i_r > \cdots > i_2 > i_1
        \geq 1
        \end{array}
      $}
      }
      }
    \,\bigg]
    \Big/
    \left(
    \def\arraystretch{1.3}
    \def\arraycolsep{2pt}
    \begin{array}{l}
      \mathrm{d}\,
      \dir{r}{\omega}_2
      =
      0
      \\[-2pt]
      \mathrm{d}
      \,
      e^i
      =
      \mathrm{d}_{\mathfrak{h}}
      e^i
      \,+\,
      \sum_{r=1}^n
      \,
      \dir{r}{\omega}_2
      \, 
      \dir{r}{\mathrm{s}}
      e^i
      \\
      \mathrm{d}\circ 
      \dir{r}{\mathrm{s}}
      =
      -
      \dir{r}{\mathrm{s}}
      \circ
      \mathrm{d}
      ,\;\;
      \dir{r}{\mathrm{s}}
        \circ 
      \dir{r'}{\mathrm{s}}
      =
      -
      \dir{r'}{\mathrm{s}}
      \circ
      \dir{r}{\mathrm{s}}
    \end{array}
    \right)
    .
    \nonumber 
  \end{align}

\smallskip 
\noindent
Here, on the right of \eqref{CEAlgebraOfNToroidification} the differential $\mathrm{d}$ is extended to the shifted generators by the rule that it graded-commutes with all the shift operators $\dir{r}{\mathrm{s}}$, which in turn are regarded as uniquely extended to graded derivations of degree $\mathrm{deg} = -1$, anti-commuting among each other. 
\end{example}

\begin{proposition}
[{\parencites[Thm. 3.8]{FSS18-TD}[Prop. 2.40]{GSS25-TD}}]
\label[proposition]
{ExtCycAdjunction}
Let $\mathfrak{g}$ be a (super-) $L_\infty$-algebra of finite type supplied with $n$ distinct 2-cocycles $\{\dir{k}{c}_1\}_{k=1,\cdots, n}$, which define its $n$-fold central extension \, $\dir{1\cdots n}{\widehat{\mathfrak{g}}}$\;\;:
\[
  \mathfrak{g} \,\in\, \mathrm{sLieAlg}_\infty
  \,,
  \hspace{.4cm}
  \dir{1}{c}_1
  ,\,
  \cdots,
  \dir{n}{c}_1
  \;\in\;
  \mathrm{CE}(\mathfrak{g})
  \,,
  \;\;\;
  \mathrm{d}\, 
  \dir{k}{c}_1
  \,=\,
  0
  \hspace{1cm}
  \begin{tikzcd}[row sep=small]
    \dir{1\cdots n}{\widehat{\mathfrak{g}}}
    \ar[
      d
    ]
    \\
    \mathfrak{g}
    \ar[
      r,
      "{
        \dir{1\cdots n}{c}_1
      }"
    ]
    &
    b\mathbb{R}^n 
    \, .
  \end{tikzcd}
\]
Then, for any target super $L_\infty$-algebra $\mathfrak{h}$ of finite type with $n$-toroidification $\mathrm{tor}^n(\mathfrak{h})$ \textup{(\cref{RationalToroidification})}, there is a canonical bijection of sets of $L_\infty$-homomorphisms :
\begin{equation}\label{nToroidificationHomIsomorphism}
 \begin{tikzcd}[column sep=30pt]
  \Big\{\;\,
      \dir{1\cdots n}{\widehat{g}}\;
      \ar[
        rr,
        "{ f }"
      ]
      &&
      \mathfrak{h}
  \Big\}
 \end{tikzcd}
 \begin{tikzcd}[column sep=45pt]
    \ar[
      r,
      shift left=5pt,
      "{
        \overset{
          \textup{\color{darkgreen}%
            \textbf{reduction}}
        }{
         \mathrm{rdc}_{\;
           \dir{1\cdots n}{c}_1
         }
        }
      }",
    ]
    \ar[
      r,
      phantom,
      "{ \sim }"
    ]
    \ar[
      r,
      shift right=4pt,
      <-,
      "{
        \underset{
          \smash{\textup{\color{darkgreen}%
            \textbf{oxidation}
          }}
        }{
        \mathrm{oxd}_{
          \;
          \dir{1\cdots n}{c}_1
        }
        }
      }"{swap}
    ]
    &[20pt]
    {}
  \end{tikzcd}
  \Bigg\{
    \begin{tikzcd}[
      row sep=-6pt, 
      column sep=20pt
    ]
      \mathfrak{g}
      \ar[
        rr,
        "{
          \widetilde f
        }"
      ]
      \ar[
        dr,
        "{
         \dir{1\cdots n}{c}_1
        }"{swap}
      ]
      &&
      \mathrm{tor}^n(\mathfrak{h})
      \ar[
        dl,
        "{
          \dir{1\cdots n}{\omega}_2
        }"
      ]
      \\
      &
      b \mathbb{R}^n
    \end{tikzcd}
 \Bigg\},
\end{equation}
  given by
\begin{equation}
  \label{nToroidificationHomBijection}
  \hspace{-4mm} 
  \left.
  \begin{tikzcd}[sep=0pt]
    \dir{1\cdots n}{\widehat{\mathfrak{g}}}
    \ar[
      rr,
      "{ f }"
    ]
    &&
    \mathfrak{h}
    \\
    \alpha^i_{\mathrm{bas}}
    + 
    \;
    \sum\limits_{\mathclap{\substack{1\leq r \leq n \\ 1 \leq i_1< \cdots < i_r \leq n}}}
    \;
    \dir{i_r}{e}\cdots \dir{i_1}{e} \cdot \alpha^i_{i_r \cdots i_1} 
    &\longmapsfrom&
    e^i
  \end{tikzcd}
  \right\} 
%  \hspace{.1cm}
  \leftrightsquigarrow
%  \hspace{.1cm}
  \left\{
  \begin{tikzcd}[
    row sep=-1pt, 
    column sep=0pt]
    \mathfrak{g}
    \ar[
      rr,
      "{ 
        \widetilde{f} 
      }"
    ]
    &&
    \mathrm{tor}^n(\mathfrak{h})
    \\
    \alpha^i_{\mathrm{bas}}
    &\longmapsfrom&
    e^i
    \\
    (-1)^{r(r+1)/2} \cdot \alpha^i_{i_r \cdots i_1}
    &\longmapsfrom&
    \dir{i_r}{\mathrm{s}}\cdots \dir{i_1}{\mathrm{s}}e^i
    \\
    \dir{r}{c}_1
    &\longmapsfrom&
    \dir{r}{\omega}_2
    \mathrlap{\,.}
  \end{tikzcd}
  \right.
\end{equation}
\end{proposition}

The promotion of this phenomenon of \emph{double dimensional reduction} from characteristic $L_\infty$-algebras to full-blown supergravity theories is the topic of \cref{Reductions}.

%%%%%%%%%%%%%%%
\subsection
{Homotopy}
\label
{Homotopy}
%%%%%%%%%%%%%%%

In its modern generality of $\infty$-category theory:
\begin{standout}
Homotopy theory is the mathematics of the higher gauge principle.
\end{standout}
Here objects are defined only up to gauge transformations (\emph{homotopy equivalences}) that themselves are only defined up to gauge-of-gauge transformations (\emph{homotopies}), which themselves are only defined up to higher-order gauge transformations (higher homotopies), and so on.

Introductions include \parencites{Richter2020}[\S 1]{FSS23-Char}.
For the present lectures, we gloss over actual homotopy theory and here just highlight a handful of the most fundamental notions mentioned in the main text.

%%%%%%%%%%%%%
\paragraph
{Topological Homotopy}
%%%%%%%%%%%%%

In its original form, homotopy theory is about deformations of continuous maps between topological spaces. An introduction aimed at physicists includes \parencites{Schwarz1994}{Pal2019}.

\begin{enumerate}
  \item
  A \emph{homotopy} $\eta$ between a parallel pair of maps $\inlinetikzcd{ X \ar[r, shift left=2pt, "{ f }"] \ar[r, shift right=2pt, "{ g }"']\& Y }$ between a pair of topological spaces, suggestively denoted
  \[
    \begin{tikzcd}
      X
      \ar[
        rr,
        bend left=30,
        "{ f }",
        "{ \  }"'{name=s}
      ]
      \ar[
        rr,
        bend right=30,
        "{ g }"',
        "{ \  }"{name=t}
      ]
      \ar[
        from=s,
        to=t,
        Rightarrow,
        "{ \eta }"
      ]
      &&
      Y
      \mathrlap{\,,}
    \end{tikzcd}
  \]
  is itself given by a continuous map out of the cylinder $[0,1] \times X $ over $X$ which continuously interpolates between $f$ and $g$:
  \begin{equation}
  \label
  {TopologicalLeftHomotopy}
    \begin{tikzcd}
      {[0,1]}
      \times
      X 
      \ar[r, "{ \eta }"]
      &
      Y
    \end{tikzcd}
    \;\;
    \text{s.t.}
    \;\;
    \left\{
    \begin{aligned}
      \eta(0,-) & = f(-)
      \\
      \eta(1,-) & = g(-)
      \mathrlap{\,.}
    \end{aligned}
    \right.
  \end{equation}

  \item

  For example, if $X = \ast$ is the singleton, then maps $\inlinetikzcd{\ast \ar[r] \& Y}$ are points of $Y$ and homotopies between them are continuous paths in $Y$. 

  Being connected by a continuous path is an equivalence relation on points of $Y$, and 
  the set of \emph{connected components} of $Y$ is the set of these equivalence classes:
  \begin{equation}
  \label{ConnectedComponents}
    \pi_0(Y)
    :=
    \{
      \inlinetikzcd{
        \ast \ar[r, dashed] \& Y
      }
    \}/\mathrm{homotopy}
    \mathrlap{\,.}
  \end{equation}

  \item 
  
  The set of all maps from $X$ to $Y$ is itself naturally a topological space, the \emph{mapping space}
  \begin{equation}
    \mathrm{Map}\bracket({
      X, Y
    })
    \defneq
    \{
      \begin{tikzcd}[sep=15pt]
        X
        \ar[r, dashed]
        &
        Y
      \end{tikzcd}
    \}
    \,,
  \end{equation}
  which is such that the continuous paths inside it are the homotopies \cref{TopologicalLeftHomotopy}. Hence its connected components \cref{ConnectedComponents} are exactly the \emph{homotopy classes} of maps 
  \begin{equation}
    \pi_0
    \, 
    \mathrm{Map}\bracket({
      X,Y
    })
    \simeq
    \{
      \begin{tikzcd}[sep=15pt]
        X
        \ar[r, dashed]
        &
        Y
      \end{tikzcd}
    \} 
    /
    \mathrm{homotopy}
    \mathrlap{\,.}
  \end{equation}
\end{enumerate}

%%%%%%%%%%%%%%

% endmatter

%%%%%%%%%%%%%%%%%%%%
\printbibliography
%%%%%%%%%%%%%%%%%%%%

\end{document}